\newif\ifarxiv\arxivtrue  % whether to produce index, for arXiv version
\documentclass[acmsmall]{acmart}

\AtBeginDocument{%
  }

\setcopyright{acmcopyright}
\copyrightyear{202X}
\acmYear{202X}
\acmDOI{XXXXXXX.XXXXXXX}

\acmJournal{TOPLAS}
\acmVolume{1}
\acmNumber{1}
\acmArticle{1}
\acmMonth{1}

\received{October 2020}
\received[revised]{May 2021, December 2021, March 2022}
\received[accepted]{July 2022}

\usepackage{subcaption} %% For complex figures with subfigures/subcaptions
\usepackage{mathpartir} % essential macro
\usepackage{color}
\usepackage{stackrel} % subscript on arrow
\usepackage[only,owedge,llfloor,rrfloor,lightning,lbag,rbag,Longmapsto,llbracket,rrbracket]{stmaryrd}
\usepackage{whyrellang}

\usepackage{thmtools} % for restatable theorems, to put proofs in appendix.
\usepackage{thm-restate}
\usepackage{wrapfig}
\usepackage{diagrams}

\usepackage{relsize}

\usepackage{tikz} % for diagrams
\usetikzlibrary{arrows,decorations.pathmorphing,backgrounds,positioning,fit,calc}

\newcommand{\specSym}{\leadsto}
\newcommand{\flowtr}[4]{#2:\: #1\specSym #3\:[#4]} % pre, CODE, post, frame
\newcommand{\flowty}[3]{#1\specSym #2\:[#3]} % flow type - spec without the method (was flowtrTy)
\newcommand{\flowtyf}[2]{#1\specSym #2} % flow type without frame condition
\newcommand{\flowtyC}[3]{#1\mathord{\specSym}#2\:[#3]} % \flowty compact (formatting)

\newcommand{\rflowtr}[4]{#2:\: #1\rspecSym #3\:[#4]} % pre, CODE, post, frame
\newcommand{\rflowty}[3]{#1\rspecSym #2\:[#3]} % flow type - spec without the method
\newcommand{\rflowtyf}[2]{#1\rspecSym #2} % flow type - frameless

\newcommand{\rHPflowtr}[6]{\proves^{#1}_{#2}\rflowtr{#3}{#4}{#5}{#6}} % provable
\newcommand{\rHVflowtr}[6]{\models^{#1}_{#2}\rflowtr{#3}{#4}{#5}{#6}} % valid

\newcommand{\HPflowtr}[6]{\proves^{#1}_{#2}\flowtr{#3}{#4}{#5}{#6}} % subscripted turnstile
\newcommand{\HVflowtr}[6]{\models^{#1}_{#2}\flowtr{#3}{#4}{#5}{#6}} % subscripted turnstile

\newcommand{\conjInv}{\mathbin{\owedge}}

\newcommand{\prePostImply}{\Rrightarrow} % pre- and post-implication of relational specs 

\newcommand{\Img}{\mbox{\large\textbf{`}}}
\newcommand{\allfields}{\mathsf{any}} % R.* is all fields 

\newcommand{\emptymod}{\emptyeff} % generic name for "main" module with empty effect
\newcommand{\emptyeff}{\text{\protect\tiny$\bullet$}}
\newcommand{\eff}{\varepsilon} % effect
 
\newcommand{\effe}{\eta}
\newcommand{\wri}[1]{\keyw{wr}\,#1} % write effect; \wr already defined
\newcommand{\rd}[1]{\keyw{rd}\,#1} % read effect; 
\newcommand{\rw}[1]{\keyw{rw}\,#1} % read + write abbrev
\newcommand{\ol}{\overline} % succinct notations for repeating elements

\renewcommand{\emptyset}{\varnothing}
\renewcommand{\setminus}{\backslash} 
\newcommand{\powerset}{\mathbb{P}}

\newcommand{\mifthenelse}[3]{#1\;\MIF\;#2\;\MELSE\;#3}

\newcommand{\MIF}{\mathsf{if}}

\newcommand{\MELSE}{\mathsf{else}}

\newcommand{\proves}{\vdash}

\newcommand{\aftqua}{.\:}
\newcommand{\all}[2]{\forall #1 \aftqua #2}
\newcommand{\some}[2]{\exists #1 \aftqua #2}
\newcommand{\unioneff}[2]{(\mathord{+} #1 \aftqua #2)}

\newcommand{\subst}[3]{\substScript{#1}{#2}{#3}}
\newcommand{\substScript}[3]{{#1}^{#2}_{#3}}
\newcommand{\union}{\mathbin{\mbox{\small$\cup$}}}
\newcommand{\intersect}{\mathbin{\mbox{\small$\cap$}}}
\newcommand{\imp}{\Rightarrow}

\newcommand{\eqdef}{\mathrel{\,\hat{=}\,}}

\newcommand{\means}[1]{\llbracket\, #1 \,\rrbracket}

\newcommand{\keyw}[1]{\ensuremath{\mathsf{#1}}} 
\newcommand{\Region}{\keyw{rgn}}
\newcommand{\sing}[1]{\{#1\}} % singleton region
\newcommand{\varblock}[2]{\keyw{var}~ #1 ~\keyw{in}~ #2} %\var{x:T}{body}

\newcommand{\NEW}{\keyw{new}}

\newcommand{\WHILE}{\keyw{while}}
\newcommand{\IF}{\keyw{if}}
\newcommand{\FI}{\keyw{fi}}
\newcommand{\DO}{\keyw{do}}
\newcommand{\OD}{\keyw{od}}
\newcommand{\bWHILE}{{\color{blue}\WHILE}}
\newcommand{\bIF}{{\color{blue}\IF}}
\newcommand{\bTHEN}{{\color{blue}\THEN}}
\newcommand{\bFI}{{\color{blue}\FI}}
\newcommand{\bDO}{{\color{blue}\DO}}
\newcommand{\bOD}{{\color{blue}\OD}}

\newcommand{\THEN}{\keyw{then}}

\newcommand{\new}[1]{\NEW\;#1}
\newcommand{\skipc}{\keyw{skip}}
\newcommand{\Endcall}{\keyw{ecall}}
\newcommand{\Endvar}{\keyw{evar}}
\newcommand{\Endlet}{\keyw{elet}} 
\newcommand{\seqc}[2]{\ensuremath{{#1}\:;{#2}}} 
\newcommand{\ifc}[3]{\keyw{if}\ {#1}\ \keyw{then}\ {#2}\ \keyw{else}\ {#3}}
\newcommand{\whilec}[2]{\keyw{while}\ {#1}\ \keyw{do}\ {#2}}
\newcommand{\BE}{\mathbin{=}}
\newcommand{\letcom}[3]{\keyw{let}~#1 \BE #2~\keyw{in}~#3} % no recursion in relRL
\newcommand{\self}{\keyw{self}}
\newcommand{\mself}{\mathit{self}} % self in math 
\newcommand{\meth}{\keyw{meth}}
\newcommand{\INT}{\keyw{int}}
\newcommand{\BOOL}{\keyw{bool}}

\newcommand{\NULL}{\keyw{null}}

\newcommand{\cfg}{\mathit{cfg}} % configuration, 
\newcommand{\Align}{\mathconst{align}} % alignment of traces 
\newcommand{\allowTo}{\mathord{\to}} % was \rightsquigarrow
\newcommand{\allowDep}{\mathord{\Rightarrow}} 
\newcommand{\successorTo}{\hookrightarrow}

\newcommand{\AllowedDep}[8]{#1,#2\overset{#8}\allowDep #3,#4 \models^{#6}_{#7} #5}

\newcommand{\uequiv}{\mathrel{\cong}} % unconditional equivalence of commands or biprograms

\newcommand{\trans}[1]{\mathrel{\overset{{#1}}{\protect{\raisebox{0pt}[.6ex][0pt]{$\longmapsto$}}}}}

\newcommand{\tranStar}[1]{\mathrel{\overset{{#1}}{\protect{\raisebox{0pt}[.6ex][0pt]{$\longmapsto$}}}%
              \raisebox{.6ex}[0pt][0pt]{\small$*$}}}

\newcommand{\biTrans}[1]{\mathrel{\overset{{#1}}{\protect{\raisebox{0pt}[.8ex][0pt]{$\Longmapsto$}}}}}
\newcommand{\biTranStar}[1]{\mathrel{\overset{{#1}}{\protect{\raisebox{0pt}[.8ex][0pt]{$\Longmapsto$}}}%
              \raisebox{.6ex}[0pt][0pt]{\small$*$}}}

\newcommand{\configm}[3]{\langle #1,\: #2,\: #3\rangle} % with method env't 
\newcommand{\configc}[1]{\langle #1 \rangle} % just the code 

\newcommand{\mathconst}[1]{\text{\textsl{#1}}} % font for fixed math entities
\newcommand{\bnd}{\mathconst{bnd}}
\newcommand{\mdl}{\mathconst{mdl}} % module of a method name
\newcommand{\topm}{\mathconst{topm}} % module of a control stack
\newcommand{\imports}{\preceq} % 
\newcommand{\irimports}{\prec} % irreflexive version of \imports
\newcommand{\preOfSpec}{\mathconst{pre}} % precondition of a spec

\newcommand{\Active}{\mathconst{Active}} 
\newcommand{\agree}{\mathconst{Agree}} % agreement modulo effect
\newcommand{\rlocs}{\mathconst{rlocs}} 
\newcommand{\wlocs}{\mathconst{wlocs}} 
\newcommand{\locations}{\mathconst{locations}} 
\newcommand{\Lagree}{\mathconst{Lagree}} 
\newcommand{\freshRefs}{\mathconst{freshRefs}} 
\newcommand{\freshLocs}{\mathconst{freshL}}

\newcommand{\snap}{\mathconst{snap}}
\newcommand{\Asnap}{\mathconst{Asnap}}
\newcommand{\Bsnap}{\mathconst{Bsnap}}
\newcommand{\Bmon}{\mathconst{Bmon}}

\newcommand{\fields}{\mathconst{Fields}} 
\newcommand{\Class}{\mathconst{DeclClass}} 
\newcommand{\fresh}{\mathconst{Fresh}} % allocator function
\newcommand{\varfresh}{\mathconst{FreshVar}} % allocation of local vars 

\newcommand{\update}[3]{[#1\, |\, #2\scol\, #3]} % function override - Reynolds' notation
\newcommand{\scol}{\mathord{:}} % spaceless colon
\newcommand{\extend}[3]{[#1 \mathord{+} #2\scol\, #3]} 
\newcommand{\drop}[2]{#1\mathbin{\!\upharpoonright\!}#2} % remove from domain of map

\newcommand{\Type}{\keyw{type}} % predicate on regions - program syntax
\newcommand{\type}{\mathconst{Type}} % dynamic type of an object in semantics - vs \Type
\newcommand{\Vars}[1]{\mathconst{Vars}(#1)} % domain of a state; use \dom for math maps; was  \Dom{ } previous RL papers 

\newcommand{\dom}{\mathconst{dom}\,} % domain of a mapping 
\newcommand{\rng}{\mathconst{rng}\,}

\newcommand{\Default}[1]{\mathconst{default}(#1)}
\newcommand{\semNull}{\mathconst{null}} % was \NIL
\newcommand{\Refset}{\mathconst{Ref}}
\newcommand{\True}{\keyw{true}}
\newcommand{\False}{\keyw{false}}

\newcommand{\Rall}[1][]{\mathit{R}}

\newcommand{\Emp}{\emptyset} 
\newcommand{\lloc}{\keyw{alloc}} % was \All
\newcommand{\old}{\keyw{old}} 
\newcommand{\disj}[2]{#1\mathbin{\mbox{\#}}#2} % disjointness 
\newcommand{\ind}[2]{#1 \indSymbol #2}

\newcommand{\indSymbol}{\mathbin{\cdot\mbox{\small\textbf{/}}.}}

\newcommand{\ftpt}{\mathconst{ftpt}} % footprint function
\newcommand{\fra}[2]{#1\mathrel{\keyw{frm}}#2} % frames judgement

\newcommand{\rprel}[2]{\Rprel{\pi}{#1}{#2}}
\newcommand{\Rprel}[3]{#2\stackrel{#1}{\sim}#3}

\newcommand{\RprelT}[3]{#2\stackrel{#1}{\approx}#3} % isomorphic states
\newcommand{\RprelTS}[3]{#2 \approxeq _{#1} #3} % isomorphic result sets

\newcommand{\configr}[5]{\langle #1,\: #2|#3,\:  #4|#5\rangle}

\newcommand{\syncbi}[1]{\lfloor #1 \rfloor}
\newcommand{\Syncbi}[1]{\llfloor #1 \rrfloor} % from stmaryrd.sty
\newcommand{\splitbi}[2]{(#1|#2)}
\newcommand{\Splitbi}[2]{(#1\mid#2)} % wider version of \splitbi
\newcommand{\smallSplitSym}{\mbox{\tiny$|$}} % for use with \ifcbi and \whilecbi
\newcommand{\ifcbi}[3]{\keyw{if}\ {#1}\ \keyw{then}\ {#2}\ \keyw{else}\ {#3}}
\newcommand{\whilecbi}[2]{\keyw{while}\ {#1}\ \keyw{do}\ {#2}}
\newcommand{\letcombi}[3]{\keyw{let}~#1 \BE #2~\keyw{in}~#3} 
\newcommand{\varblockbi}[2]{\keyw{var}~ #1 ~\keyw{in}~ #2} %\var{x:T}{body}

\newcommand{\splitbir}[2]{(#1 |^\text{\tiny\(\!\triangleright\)} #2)}
\newcommand{\Splitbir}[2]{(#1 \mid^\text{\tiny\(\!\triangleright\)} #2)} % larger 

\newcommand{\whilecbiA}[3]{\keyw{while}\ {#1} \cdot {#2}\ \keyw{do}\ {#3}}

\newcommand{\Fault}{\lightning}

\def\rightharpoonupfill{$\mathsurround=0pt \mathord- \mkern-6mu
  \cleaders\hbox{$\mkern-2mu \mathord- \mkern-2mu$}\hfill
  \mkern-6mu \mathord\rightharpoonup$}

\def\leftharpoonupfill{$\mathsurround=0pt \mathord\leftharpoonup \mkern-6mu
  \cleaders\hbox{$\mkern-2mu \mathord- \mkern-2mu$}\hfill
  \mkern-6mu \mathord-$}

\def\overleftharpoonup#1{\vbox{\ialign{##\crcr
      \leftharpoonupfill\crcr\noalign{\kern-1pt\nointerlineskip}
      $\hfil\displaystyle{#1}\hfil$\crcr}}}

\def\overrightharpoonup#1{\vbox{\ialign{##\crcr
      \rightharpoonupfill\crcr\noalign{\kern-1pt\nointerlineskip}
      $\hfil\displaystyle{#1}\hfil$\crcr}}}

\def\overleftrightharpoonup#1{\vbox{\ialign{##\crcr
      \leftharpoonupfill\hspace*{-.7em}\rightharpoonupfill\crcr\noalign{\kern-1pt\nointerlineskip}
      $\hfil\displaystyle{#1}\hfil$\crcr}}}
\newcommand{\Left}[1]{\overleftharpoonup{#1}} % projections on various syntax; no longer formulas
\newcommand{\Right}[1]{\overrightharpoonup{#1}}

\newcommand{\rTow}{\mathconst{r2w}}
\newcommand{\wTor}{\mathconst{w2r}}
\newcommand{\writes}{\mathconst{wrs}}
\newcommand{\reads}{\mathconst{rds}}
\newcommand{\locEq}{\mathconst{locEq}} % relational part
\newcommand{\LocEq}{\mathconst{LocEq}} % complete spec 
\newcommand{\written}{\mathconst{wrttn}}

\newcommand{\rspecSym}{\ensuremath{\mathrel{\text{\small$\thickapprox\hspace*{-.4ex}>$}}}}

\renewcommand{\P}{\mathcal{P}}  % ALERT renew
\renewcommand{\S}{\mathcal{S}} % ALERT renew
\newcommand{\M}{\mathcal{M}}  % relation - local coupling
\newcommand{\N}{\mathcal{N}}  % relation - local coupling
\newcommand{\Q}{\mathcal{Q}} 
\newcommand{\R}{\mathcal{R}} 

\newcommand{\Z}{\mathbb{Z}} 

\newcommand{\weave}{\looparrowright} % weaves to 

\newcommand{\Agr}{\ensuremath{\mathbb{A}}}
\newcommand{\lorbi}{\lor}
\newcommand{\eqbi}{\mathrel{\ddot{=}}}

\newcommand{\later}{\Diamond} % 'possibly'
\newcommand{\always}{\mathord{\text{\small$\Box$}}}

\newcommand{\lexop}{\text{\color{blue}\textbf{$\langle\hspace*{-2.2pt}[$}}} % open left expression
\newcommand{\lexcl}{\text{\color{blue}\textbf{$\langle\hspace*{-2.3pt}]$}}} % close left expression
\newcommand{\rexop}{\text{\color{blue}\textbf{$[\hspace*{-2.4pt}\rangle$}}} % open right expression
\newcommand{\rexcl}{\text{\color{blue}\textbf{$]\hspace*{-2.2pt}\rangle$}}} % close right expression

\newcommand{\leftex}[1]{ \text{\small$\langle\hspace*{-2.2pt}[$} #1 \text{\small $\langle\hspace*{-.65ex}]$}
}
\newcommand{\rightex}[1]{ \text{\small$[\hspace*{-2.5pt}\rangle$} #1 \text{\small$]\hspace*{-2.2pt}\rangle$} 
}

\newcommand{\leftF}[1]{\leftex{#1}}
\newcommand{\rightF}[1]{\rightex{#1}}

\newcommand{\Both}[1]{\ensuremath{\mathbb{B}} #1}

\newtheorem{remark}{Remark}
\renewcommand{\phi}{\varphi}

\newcommand{\dt}[1]{\index{#1}\textbf{\emph{#1}}} % defined term

\newcommand{\rn}[1]{\textsc{#1}} % rule names

\newcommand{\RL}{\textsc{RL}}
\newcommand{\RLI}{\textsc{RLI}}
\newcommand{\RLII}{\textsc{RLII}}
\newcommand{\RLIII}{\textsc{RLIII}} % readRL 

\definecolor{light-gray}{gray}{0.88}
\definecolor{dark-gray}{gray}{0.25}
\newcommand{\ghostbox}[1]{\colorbox{light-gray}{#1}} % grey background in math

\newcommand{\WhyRel}{WhyRel}

\newcommand{\gmid}{\mathrel{{\color{blue}\mid}}}

\newcommand{\ruleSOF}{
\protect\inferrule*[left=SOF]{ 
  \Phi,\Theta \HPflowtr{}{M}{P}{C}{Q}{\eff}  \\
  \models \fra{\bnd(N)}{I} \\
  N\in \Theta\\ N\neq M \\ 
  \all{m\in\Phi}{\mdl(m)\not\imports N} \\
  \mbox{$C$ binds no $N$-method}
}{ \Phi,(\Theta \conjInv I) \HPflowtr{}{M}{P\land I}{C}{Q\land I}{\eff} }
}
\newcommand{\ruleCtxIntro}{
\protect\inferrule*[left=CtxIntro]
{ \Phi\HPflowtr{}{M}{P}{A}{Q}{\eff}  \\
  P \imp \ind{\bnd(\mdl(m))}{\eff} \\ 
  P \imp \ind{\bnd(\mdl(m))}{\rTow(\eff)}  
}{ \Phi,\: m:\flowty{R}{S}{\effe} \HPflowtr{}{M}{P}{A}{Q}{\eff} }
}

\newcommand{\rulerWhile}{
\protect\inferrule*[left=rWhile]
{
\Phi\rHPflowtr{}{}{\Q\land\neg\P\land\neg\P'\land\leftF{E}\land\rightF{E'}}{CC}{\Q}{\eff|\eff'} \\
\Phi \rHPflowtr{}{}{\Q\land\P\land\leftF{E}}{\splitbi{\Left{CC}}{\skipc}}{\Q}{\eff|\emptyeff} \\
\Phi \rHPflowtr{}{}{\Q\land\P'\land\rightF{E'}}{\splitbi{\skipc}{\Right{CC}}}{\Q}{\emptyeff|\eff'} \\
\ind{\unioneff{N\in\Phi,N\neq M}{\bnd(N)}}{\rTow(\ftpt(E))} \\ 
\ind{\unioneff{N\in\Phi,N\neq M}{\bnd(N)}}{\rTow(\ftpt(E'))} \\ 
\Q\imp E\eqbi E' \lorbi (\P\land\leftF{E}) \lorbi (\P'\land\rightF{E'}) \\
\eff\mbox{ is }\Left{\Q}/\eff\mbox{-immune}\\
\eff'\mbox{ is }\Right{\Q}/\eff'\mbox{-immune}
}{
\Phi \rHPflowtr{}{}{\Q}{\whilecbiA{E\smallSplitSym E'}{\P\smallSplitSym \P'}{CC}}{\Q\land\leftF{\neg E}\land\rightF{\neg E'}}{\eff,\ftpt(E)|\eff',\ftpt(E')}
}}

\newcommand{\ruleLink}{
\protect\inferrule*[left=Link]  
{
\Phi,\Theta \proves_{\mdl(m_i)} B_i: \Theta(m_i) \\
\Phi,\Theta \proves_{\emptymod} C: \flowty{P}{Q}{\eff} \\
\dom(\Theta)=\ol{m} \\
\all{N\in\Phi,L\in\Theta}{N\not\imports L} \\
\all{N,L}{ N\in\Theta \land N\irimports L \imp L\in(\Phi,\Theta)} 
} 
{
\Phi\proves_{\emptymod} \letcom{\ol{m}}{\ol{B}}{C} : \flowty{P}{Q}{\eff} 
}}

\newcommand{\ruleCall}{
\protect\inferrule*[left=Call\quad]
{}{ m\scol \flowty{P}{Q}{\eff} \proves_{\emptymod} m(): \flowty{P}{Q}{\eff} }
}

\newcommand{\ruleModIntro}{
\protect\inferrule*[left=ModIntro]
{ \Phi \proves_{\emptymod} A: \flowty{P}{Q}{\eff}  \\ 
  \Phi \proves_{\emptymod} A: \flowty{P\land \Bsnap_M}{\Bmon_M}{\eff}  \\ 
  \mbox{if $M\in\Phi$ then $A$ is a call}
} {
  \Phi \proves_M A: \flowty{P}{Q}{\eff}  
}}

\newcommand{\ruleIf}{
\protect\inferrule*[left=If]
 { \Phi \HPflowtr{}{M}{ P \land E}{C_1}{Q}{\eff}\\
   \Phi \HPflowtr{}{M}{P \land \neg E}{C_2}{Q}{\eff} \\
   \ind{\unioneff{N\in\Phi,N\neq M}{\bnd(N)}}{\rTow(\ftpt(E)) }
 } 
 { \Phi \HPflowtr{}{M}{P}{\ifc{E}{C_1}{C_2}}{P'}{\eff,\ftpt(E)} }
}

\newcommand{\ruleCtxIntroInOne}{
\protect\inferrule*[left=CtxIntroIn1]
{ \Phi\HPflowtr{}{M}{P}{C}{Q}{\eff}  \\
\mdl(m)\in\Phi 
} {
\Phi,m\scol\flowty{R}{S}{\effe} \HPflowtr{}{M}{P}{C}{Q}{\eff}  
}}

\newcommand{\ruleCtxIntroInTwo}{
\protect\inferrule*[left=CtxIntroIn2]
{ \Phi \proves_M A : \flowty{P}{Q}{\eff}  \\  \mdl(m)=M 
\\ A \mbox{ is not a call}
}{ \Phi,m\scol\flowty{R}{S}{\effe} \HPflowtr{}{M}{P}{A}{Q}{\eff}  
}}

\newcommand{\ruleCtxIntroCall}{
\protect\inferrule*[left=CtxIntroCall]
{ \Phi \proves_M p(): \flowty{P}{Q}{\eff}  \\ 
  \Phi \proves_M p(): \flowty{P\land \Bsnap_N}{\Bmon_N}{\eff}  \\ 
  N = \mdl(m) \\
  \mdl(p)\imports\mdl(m)
} {
  \Phi,m\scol\flowty{R}{S}{\effe} \HPflowtr{}{M}{P}{p()}{Q}{\eff}  
}}

\newcommand{\ruleConseq}{
\protect\inferrule*[left=Conseq]{ 
\Phi\HPflowtr{}{M}{P}{C}{Q}{\eff} \\
P_1\imp P\\ 
Q\imp Q_1\\ 
P_1\models \eff \leq \eff_1
}{ 
\Phi\HPflowtr{}{M}{P_1}{C}{Q_1}{\eff_1}
}}

\newcommand{\rulerLink}{%
\protect\inferrule*[left=rLink]{ 
\Phi,\Theta \proves_{\emptymod} \Syncbi{C} : \rflowty{\P}{\Q}{\eff} \\
\Phi,\Theta \proves_{\mdl(m)} \splitbi{B}{B'} : \Theta_2(m)  \\
\Phi_0,\Theta_0 \proves_{\mdl(m)} B : \Theta_0(m)  \\
\Phi_1,\Theta_1 \proves_{\mdl(m)} B' : \Theta_1(m) \\
\delta = \unioneff{L\in(\Phi,\Theta)}{\bnd(L)} \\
(\Phi,\Theta) \prePostImply \LocEq_{\delta}(\dot{\Phi},\dot{\Theta}) \\
\P \imp pre(\locEq_{\delta}(\flowty{P}{Q}{\eff})) \\
\all{N\in\Phi,L\in\Theta}{N\not\imports L} \\
\all{N,L}{ N\in\Theta \land N\irimports L \imp L\in(\Phi,\Theta)} \\
\mbox{$C$ is let-free} 
}{ 
\Phi \proves_{\emptymod} 
  \letcombi{m}{\splitbi{B}{B'}}{\Syncbi{C}} : \rflowty{\P}{\Q}{\eff} 
}}

\newcommand{\rulerSOF}{%
\protect\inferrule*[left=rSOF]{ 
\LocEq_\delta(\Phi,\Theta) \proves_M \Syncbi{C} :  \locEq_\delta(\flowty{P}{Q}{\eff})\\
\models \fra{ \bnd(N) | \bnd(N) }{ \N } \\
\N\imp\always\N \\
N\neq M \\ 
N\in \Theta\\ 
\all{m\in\Phi}{\mdl(m)\not\imports N} \\
\delta = \unioneff{L\in(\Phi,\Theta),L\neq M}{  \bnd(L) } \\
\mbox{$C$ is let-free} 
}{ 
    \LocEq_\delta(\Phi), \, \LocEq_\delta(\Theta)\conjInv\N  \proves_M
    \Syncbi{C} \: : \: \locEq_\delta (\flowty{P}{Q}{\eff}) \conjInv \N
}}

\newcommand{\rulerLocEq}{%
\protect\inferrule*[left=rLocEq]{ 
	\Phi\HPflowtr{}{M}{P}{C}{Q}{\eff}\\
	P\models \wTor(\eff)\leq\reads(\eff)\\
\delta=\unioneff{N\in\Phi,N\neq M}{ \bnd(N) } \\
\mbox{$C$ is let-free}
}{ 
\LocEq_\delta(\Phi) \proves_M \Syncbi{C}:\locEq_\delta(\flowty{P}{Q}{\eff})
}}

\newcommand{\rulerIf}{%
\protect\inferrule*[left=rIf]
{
\Phi \rHPflowtr{}{M}{\P\land\leftF{E}\land\rightF{E'}}{CC}{\Q}{\eff|\eff'} \\
\Phi \rHPflowtr{}{M}{\P\land\leftF{\neg E}\land\rightF{\neg E'}}{DD}{\Q}{\eff|\eff'} \\
\P\imp E\eqbi E' \\
\delta=\unioneff{N\in\Phi,N\neq M}{\bnd(N)} \\
\ind{\delta}{\rTow(\ftpt(E))} \\
\ind{\delta}{\rTow(\ftpt(E'))} \\
}{
\Phi \rHPflowtr{}{M}{\P}{\ifcbi{E\smallSplitSym E'}{CC}{DD}}{\Q}{\eff,\ftpt(E)|\eff',\ftpt(E')}
}}

\newcommand{\rulerConseq}{
\protect\inferrule*[left=rConseq]{ 
 \Phi\rHPflowtr{}{}{\P}{CC}{\Q}{\eff|\eff'}\\
 \R\imp \P \\
 \Q \imp \S \\
 \P\models (\eff|\eff')\leq(\effe|\effe') 
}{ 
 \Phi\rHPflowtr{}{}{\R}{CC}{\S}{\effe|\effe'}
}}

\newcommand{\rulerCall}{
\protect\inferrule*[left=rCall]{
  \Phi_0 \proves m() : \Phi_0(m) \\
  \Phi_1 \proves m() : \Phi_1(m) 
}{ 
   \Phi \proves \syncbi{m()} : \Phi_2(m) 
}}

\ifarxiv
\makeindex
\fi

\begin{document}

%% Title information

\title{A Relational Program Logic with Data Abstraction and Dynamic Framing
\ifarxiv{\normalsize \quad [version with index]}\fi 
}

        %% [Short Title] is optional;
                                        %% when present, will be used in
                                        %% header instead of Full Title.
%\titlenote{with title note}             %% \titlenote is optional;
                                        %% can be repeated if necessary;
                                        %% contents suppressed with 'anonymous'
%\subtitle{Subtitle}                     %% \subtitle is optional
%\subtitlenote{with subtitle note}       %% \subtitlenote is optional;
                                        %% can be repeated if necessary;
                                        %% contents suppressed with 'anonymous'

%% Author information
%% Contents and number of authors suppressed with 'anonymous'.
%% Each author should be introduced by \author, followed by
%% \authornote (optional), \orcid (optional), \affiliation, and
%% \email.
%% An author may have multiple affiliations and/or emails; repeat the
%% appropriate command.
%% Many elements are not rendered, but should be provided for metadata
%% extraction tools.

%% Author with single affiliation.
%% (for author with two affiliations and emails - see template.tex)

\author{Anindya Banerjee}
\email{anindya.banerjee@imdea.org}
\affiliation{
\institution{IMDEA Software Institute}
\city{Pozuelo de Alarcon}
\state{Madrid}
\country{Spain}}
\orcid{0000-0001-9979-1292}
\author{Ramana Nagasamudram}
\email{rnagasam@stevens.edu}
\affiliation{
\institution{Stevens Institute of Technology}
\city{Hoboken}
\state{New Jersey}
\country{USA}}
\author{David A. Naumann}
\email{naumann@cs.stevens.edu}
\affiliation{
\institution{Stevens Institute of Technology}
\city{Hoboken}
\state{New Jersey}
\country{USA}}
\orcid{0000-0002-7634-6150}
\author{Mohammad Nikouei}
\email{snikouei@stevens.edu}
\affiliation{
\institution{Stevens Institute of Technology}
\city{Hoboken}
\state{New Jersey}
\country{USA}}

%\authorsaddresses{} % OVERRIDE TO SAVE SPACE ON PAGE 1
\renewcommand{\shortauthors}{Banerjee, Nagasamudram, Naumann, Nikouei}

\begin{abstract}
Dedicated to Tony Hoare.
\\[2ex]
In a paper published in 1972 Hoare articulated the fundamental notions of
hiding invariants and simulations.
Hiding: invariants on encapsulated data representations need not be mentioned in
specifications that comprise the API of a module.
Simulation: correctness of a new data representation and implementation
can be established by proving simulation between the old and new implementations
using a coupling relation defined on the encapsulated state.
These results were formalized semantically and for a simple model of state,
though the paper claimed this could be extended to encompass dynamically
allocated objects.
In recent years, progress has been made towards formalizing the claim,
for simulation, though mainly in semantic developments.
In this article, hiding and simulation are combined with the idea in Hoare's 1969 paper:
a logic of programs.
For an object-based language with dynamic allocation,
we introduce a relational Hoare logic
with stateful frame conditions
that formalizes encapsulation, hiding of invariants, and
couplings that relate two implementations.
Relations and other assertions are expressed in first-order logic.
Specifications can express a wide range of relational properties such as
conditional equivalence and noninterference with declassification.
The proof rules facilitate relational reasoning by means of convenient alignments and
are shown sound with respect to a conventional operational semantics.  A derived proof rule for equivalence of linked programs directly embodies representation independence.
Applicability to representative examples
is demonstrated using an SMT-based implementation.
\end{abstract}




\maketitle

\section{Introduction}\label{sec:intro}

Data abstraction has been a cornerstone of software development methodology since the seventies. Yet it is surprisingly difficult to achieve in a reliable manner in modern programming languages that permit manipulation of the global heap via dynamic allocation, shared mutable objects, and callbacks. Aliasing can violate conventional syntactic means of encapsulation (modules, classes, packages, access modifiers) and therefore can undercut the fundamental guarantee of abstraction: equivalence of client behavior under change of a module's data structure representations.

The theory of data abstraction is well-known since Hoare's seminal paper~\cite{Hoare:data:rep}.
Its main ingredients are
the encapsulation of effects,
hidden invariants (that is, private invariants that do not appear in a method's interface specifications, so that clients are exempt from having to establish them for calls to the method),
and relational reasoning: coupling relations and simulations.
%% First, one posits a coupling relation that relates the hidden invariants of two versions of an internal data structure. Secondly, one shows the \emph{abstraction theorem}~\cite{Reynolds84}: the coupling relation is a simulation, to wit, the relation is preserved by the (interface) method implementations. Finally, equivalence of client behaviors is justified by encapsulation of effects: the coupling relation is the identity on observable state, a property termed \emph{identity extension}~\cite{Reynolds84}.
Hoare's paper provides a semantic formalization of these ideas using a simple model of state and it claims that the ideas can be extended to encompass dynamically allocated objects.

The justification of Hoare's claim is a primary focus of this article, which is in the context of two strands of recent work. One strand has made progress on automating proofs of conditional equivalence and relational properties in general, based on automated theorem proving (e.g., SMT) and techniques to decompose relational reasoning by expressing alignment of executions in terms of ``product programs''.
The other strand has made progress towards formalizing Hoare's claim
in semantic theories of representation independence (simulation and logical relations).
This article brings the strands together using the idea in Hoare's 1969 paper~\cite{Hoare69}: a logic of programs.
In this way we address three goals:

\textbf{Modular reasoning about relational properties of object-based programs.}
Such properties include not just equivalence but many others such as noninterference.
Conditional equivalence, for example, is needed to justify bug fixes and refactorings (regression verification), taking into account
preconditions that capture usage context.
Conditional noninterference expresses information flow security policies with declassification;
similar dependency properties express context conditions for compiler optimizations.
Modular reasoning requires \emph{procedural abstraction,} i.e., reasoning about code under hypotheses in the form of method contracts.
It requires \emph{local reasoning}, based on frame conditions.
And it requires \emph{data abstraction}, based on program modules and encapsulated data representations.

\textbf{Automated reasoning.}
We aim to facilitate verification using what have been called auto-active verification tools~\cite{LeinoM10} like Why3 and Dafny.
Users may be expected to provide source level annotations (contracts and data invariants)  and alignment hints (to decompose relational reasoning) but are not expected to guide proof tactics or provide full functional specifications.
The latter is a key point. It is difficult for developers to formulate full functional specs
of applications and libraries, and such specs would often need mathematical types not amenable to automated provers.  Experience shows the value of weak specs of input validity and data structure consistency.
Frame conditions are particularly useful for the developer and for the reasoning system~\cite{HatcliffLLMP12}.
% and this includes designating the internal data meant to be encapsulated.

\textbf{Foundational justification.}
We aim for tools that yield strong evidence of correctness based on accurate program semantics.
In this article we consider sequential programs at the source level, with idealizations---unbounded integers, heap, stack---that often are used to simplify specs and facilitate automated theorem proving.
We carefully model dynamic allocation at the level of abstraction of garbage-collected languages such as Java and ML.  The ultimate goal is tools for languages used in practice, for which semantics should be
machine-checked and based on the compiler and machine model.

\paragraph{Summary of the state of the art with respect to these goals.}
To position our work we give a quick summary; thorough discussion with citations can be found in Section~\ref{sec:related}.

There are several mature automated verifiers for \emph{unary} (non-relational) verification,
including local reasoning by separation logic and by stateful frame conditions (``dynamic frames''),
based on SMT solvers and other techniques for proof automation including inference of annotations
and decentralized invariants~\cite{RegLogJrnI,Filliatre2021} %(called localized global invariants in \RLI)
to lessen the need for induction.
 % [CEGAR, abs int].
While abstract data types are commonly supported in specifications, encapsulation of heap structures remains a difficult challenge.
For relational reasoning, there has been good progress in automation;
this  has made clear the need for both lockstep 
alignment of subcomputations using relational formulas
and ``asynchronous'' alignments using unary reasoning. 
%For relational reasoning, there has been good progress in automation
%which has made clear the need for both lockstep and ``asynchronous'' alignment of subcomputations.
Automated verifiers have varying degrees of foundational justification, but a standard technique is well established:
verification conditions are based on a Hoare logic which in turn is proved sound.

The semantic theory of data abstraction is well understood for a wide range of languages, mostly focused on syntactic means of encapsulation including type polymorphism, but also considering state-based notions like ownership using specialized types or program annotations.  These theories account for heap encapsulation and simulation but have not been well connected with general program reasoning: in brief, they say why simulation implies program equivalence but do not say how to prove simulation.
Some of this theory has been incorporated in interactive verification tools, for example based on the Coq proof assistant.
In such a setting, the powerful ambient logic makes it possible to express all the theory, and recent work includes
relational program logics that feature local reasoning and hiding.
These works focus on concurrency and higher order programs, and have many complications needed to address those challenges---far from the simplicity of first-order specs supported by automated provers
and accessible to ordinary developers.

\paragraph{Our contribution, in a nutshell.}
This article presents a full-featured, general \emph{relational program logic} that supports \emph{modular reasoning} about both unary and relational properties of object-based programs.
The logic formalizes state-based encapsulation and the hiding of invariants and coupling relations,
including a proof rule for equivalence of linked programs which directly embodies the theory of representation independence.
The logic uses a form of product program,\footnote{Some authors restrict the term 
``product'' to mean a representation that is itself a program.
Our usage is looser, encompassing representations like pairs of programs~\cite{Francez83} and our custom syntax.}
called ``biprogram'', to designate alignments of subprograms
to facilitate use of simple relational assertions that are  amenable to automated proof.
The verification conditions are all first-order, without need for inductive predicates,
and \emph{amenable to SMT-based automation}.
A \emph{foundational justification} is provided: detailed soundness proofs
with respect to standard operational semantics.

%% The logic supports local reasoning, including small axioms and frame rules (for both unary and relational fragments of the logic), hypothetical specifications of methods, and methodologies like ownership. The logic's assertion language allows expressing relations, assertions, and frame conditions in first order logic.  Typical relations include those that describe agreement between unbounded pointer structures while allowing for differences in object allocation on the heap; those that relate similarly- and differently-structured programs; and those that relate partially- and fully-aligned loop iterations. The logic provides means to express and manipulate syntactic alignment of subprograms (termed ``biprograms'' in the sequel) to facilitate simple relational assertions that are inductive and amenable to automated proof. Such alignment has been recognized in the literature as an extension of Floyd's inductive assertions method to relational reasoning: the points of alignment are the control points where relational properties can be asserted.
%
%Applicability to representative examples of data abstraction is demonstrated using an SMT-based automated verifier.

\paragraph{Outline and reader's guide.}

Section~\ref{sec:synop} summarizes the problem, the approach taken, and the contributions of this article.
Section~\ref{sec:Usyntax} presents most of the syntactic ingredients of the unary logic,
including effect expressions, unary specs and correctness judgments.
Novel syntactic elements are explained informally via examples
and an extended example illustrates encapsulation and modular linking.

Section~\ref{sec:biprograms} first presents the syntactic ingredients of the relational logic---biprograms, relation formulas, relational specs and correctness judgments---and then presents a series of 
examples to illustrate alignment, relations on heap structures, and relational modular linking.

%The immediately following \dn{was sec:examples} gives some illustrative examples of relational reasoning and encapsulation, highlighting features of specs and biprograms that facilitate reasoning.
%% One might prefer for all the syntactic ingredients to be introduced incrementally, but this arrangement
%% brings most syntactic definitions together in one place.
%% It caters both for readers already familiar with prior work on region logic and for readers who are not.

After Sects.~\ref{sec:synop}--\ref{sec:biprograms}, readers who are not interested in semantic details may wish to skip to Section~\ref{sec:unaryLog} which presents the rules of the unary logic, and then skip again to Section~\ref{sec:rellog} which presents the rules of the relational logic, including the modular linking rule and its derivation from simpler rules.

Section~\ref{sec:unarySem} defines the semantics of programs and unary correctness judgments;
it is based on standard small-step semantics but we need a number of notions concerning agreement and dependency, leading to the novel and subtle semantics of encapsulation.
Section~\ref{sec:biprog} gives the semantics of biprograms and relational correctness.
Section~\ref{sec:cases} sketches the use of a prototype tool to evaluate viability of the logic's proof obligations for SMT-based verification.
Section~\ref{sec:related} surveys related work and
Section~\ref{sec:discuss} concludes.
% with our inability to craft an envoi for our epic journey.

%The appendix provides proofs as well as a glossary (Sect.~\ref{sec:app:guide}) and index.
%dn{ACM may remove index and even glossary} 

A lengthy appendix provides proofs and additional details, none of which should be needed
to understand the contents of the article. Nonetheless, cross-references to the appendix are included.
There is also a glossary of symbols and a table of metavariables (Section~\ref{sec:app:guide}). 
The article is self-contained but includes some remarks to cater for readers who are familiar with prior work on region logic on which we build.

\section{Synopsis}\label{sec:synop}

\subsection{Modular reasoning about relational properties}\label{sec:modrel}

\begin{wrapfigure}{R}{0.33\textwidth}  % r allows float, R means exactly here
% ALERT blank line is needed before module Cell
\begin{lstlisting}   % Not whyrel code; this is meant to be easily read

module MCell
   class Cell
   meth Cell(c: Cell) /* constructor */
   meth cget (c: Cell) : int /* pure */
   meth cset (c: Cell, v: int)
      requires { c <> null }
      ensures { cget(c) = v }
\end{lstlisting}
\caption{Example interface.}\label{fig:MCell}
\vspace*{-3ex} % ALERT hack to avoid space on following page
\end{wrapfigure}
To introduce the problem addressed in this article, we begin by sketching Hoare's story about proofs of correctness of data representations.
%One novelty of our logic is that the whole story is embodied in proof rules.
Often a software component is revised with the intent to improve some characteristic such as performance while preserving its functional behavior.  As a minimal example consider this program in an idealized object-based language,
with integer global variables \whyg{x,y}.
\begin{lstlisting}[basicstyle=\linespread{0.5}\sffamily]
  var c: Cell in c := new Cell; x := x+1; cset(c,x); y := cget(c)
\end{lstlisting}
It is a client of the interface in Figure~\ref{fig:MCell}.
An obvious implementation of the module\footnote{Classes are instantiable.
For our purposes, modules are static~\cite{OHearnYangReynoldsInfoToplas,RegLogJrnII},
like packages in Java and other languages.}
is for class \whyg{Cell} to declare an integer field \whyg{val} that stores the value.
Suppose we change the implementation: store the negated value, in a field named \whyg{f}, and let \whyg{cget} return its negation.
Client programs like the one above should not be affected by this change, at the usual level of abstraction (e.g, ignoring timing).
To be specific, we have equivalence of the two programs obtained by linking the client with one or the other implementation of the module.
(Equivalence means equal inputs lead to equal outputs.)
This has nothing to do with the specific client. The point of data abstraction is to free the client programmer from dependence on internal representations, and to free the library programmer from needing to reason about specific clients.

The (relational) reasoning here is familiar in practice and in theories of \dt{representation independence}.
There is a coupling relation that connects the two data representations; in this case, for corresponding object references 
$o,o'$ of type \whyg{Cell}, 
\begin{equation}\label{eq:coupleCell}
\mbox{the value of field $o'.\code{f}$ is the negation of $o.\code{val}$.}
\end{equation}
This relation is maintained, by paired execution of the two implementations, for each method of the module and for all instances of the class.
The fields are encapsulated within the module, so a client can neither falsify the relation nor behave differently from related states since the visible part of the relation is the identity.

Figure~\ref{fig:align} depicts steps of two executions of the example client, linked with alternate implementations of the methods it calls.
The top line indicates a relation between the initial states of the left and right executions.
The client's precondition $P$ holds in both ($\Both{}$),
and the initial states agree ($\Agr$) on the part of the state that is client-visible.
Unknown to the client, the module coupling relation $\M$
is established by the constructors and can be assumed in reasoning about
the calls, provided the method's implementations preserve the relation.
A client step, like \whyg{x:=x+1} here, should preserve $\M$ for reasons of encapsulation.
The bottom line indicates agreement on the final result.
Each method has alternate implementations; the ones for \whyg{cset}
are labelled (as $B,B'$) for expository purposes.

%% \begin{wrapfigure}{r}{0\textwidth} % r allows float, R means exactly here
%% \begin{footnotesize}
%% \begin{diagram}[h=3.7ex,w=4em]
%% {}         & \hLine^{\Both P \land \Agr vis}       & {} \\
%% %\dDotsto^C &                       & \dDotsto_C \\
%% \dDotsto &                       & \dDotsto \\
%% {}         & \hLine^{\Agr vis\land\M} & {}\\
%% \dDotsto^B &                       & \dDotsto_B' \\
%% {}         & \hLine_{\Agr vis\land\M} & {} \\
%% %\dDotsto^C &                       & \dDotsto_C \\
%% \dDotsto &                       & \dDotsto \\
%% {}         & \hLine_{\Agr vis}       & {} \\
%% \end{diagram}
%% \vspace*{-2ex}
%% \end{footnotesize}
%% \caption{Alignment}
%% \label{fig:align}
%% \end{wrapfigure}

% ALERT - hard-coded fonts; \code and \whyg macros don't work in diagram
\begin{figure}[h]
\begin{footnotesize}
\begin{tabular}{ll}
\begin{diagram}[h=3.7ex,w=1em]
{}         & \hLine^{\Both P \land \Agr vis}     & {} \\
\dDotsto^{c:=\new{\mathsf{Cell}}}  &                       & \dDotsto_{c:=\new{\mathsf{Cell}}} \\
{}         & \hLine^{\Agr vis\land\M}            & {}\\
\dDotsto^{x:=x+1}  &                           & \dDotsto_{x:=x+1} \\
{}         & \hLine^{\Agr vis\land\M}            & {}\\
\dDotsto^{\mathsf{cset}(c,x)}_B &                          & \dDotsto_{\mathsf{cset}(c,x)}^{B'} \\
{}         & \hLine^{\Agr vis\land\M}            & {} \\
\dDotsto^{y:=\mathsf{cget}(c)} &                         & \dDotsto_{y:=\mathsf{cget}(c)} \\
{}         & \hLine^{\Agr vis}                  & {} \\
\end{diagram}
&
\parbox{.5\textwidth}{
$\Both P$ ---both initial states satisfy $P$
\\[1ex]
$\Agr vis$ ---two states agree on client-visible locations
\\[1ex]
$\M$ ---coupling relation on encapsulated locations
\\[1ex]
$B,B'$ ---alternate implementations of a method
}% parbox
\end{tabular}
\end{footnotesize}
\caption{Two executions, with relations between aligned points.}
\label{fig:align}
\end{figure}

In this work, we introduce a logic in which one can specify relational properties such as the preservation of a coupling relation by the two implementations $B,B'$, as well as equivalence of the two linked programs for a client $C$.  Moreover the equivalence can be inferred directly from the preservation property.
Equivalence is expressed in local terms, referring just to the part of the state that $C$ acts on:
In the example client program, the pre-relation is agreement on the value of \whyg{x} and the post-relation is agreement on \whyg{y}.
If $C$ is part of a larger context then a relational frame rule
can be applied to infer that relations on separate parts of the state are also maintained by $C$ as discussed later.

\paragraph{Encapsulation.}

The above reasoning depends crucially on encapsulation, and many programming languages have features intended to provide encapsulation.
In unary verification, encapsulation serves to protect invariants on internal data structures.
It is well known,
and often experienced in practice,
that references and mutable state can break encapsulation in conventional languages like Java and ML.
There has been considerable research on methodologies using type annotations and assertions to enforce disciplines including ownership for the sake of encapsulation and local reasoning.
%In this work we focus on heap encapsulation,
%and formalize a logic with minimal syntactic mechanisms for encapsulation,
%, relying on minimal logical notions to cater for automation,
This work focuses on heap encapsulation, without commitment to any specific discipline, but provides a framework in which such disciplines can be used.

In this article, encapsulation is at the granularity of a module, not a class or object.
Thus the implementation of a method \whyg{cswap(c, d: Cell)} that swaps the values of two cells can exploit that the cells have the same internal representation.
However, it is often useful for each instance of an abstraction, say a cell or a stack, to ``own''
some locations that are separate from those of other instances,
so we can do framing at the granularity of an instance.
This is manifest in frame conditions, as we will see for \whyg{cset},
and it is also manifest in invariants.  For example, a module for stacks implemented using
linked lists has the invariant that distinct stacks use disjoint list nodes.

%% One implementation
%% \begin{lstlisting}
%%   class Cell { val: int }
%%   proc Cell(c: Cell) { v := 0 } /* superfluous, as fields and variables are 0-initialized */
%%   function cget (c: Cell) : int { return c.val }
%%   proc cset (c: Cell, v: int) { c.val := v }
%% \end{lstlisting}

%% Another implementation:
%% \begin{lstlisting}
%%   class Cell { f: int }
%%   proc Cell(c: Cell) { }
%%   function cget (c: Cell) : int { return - c.f }
%%   proc cset (c: Cell, v: int) { c.f := - v }
%% \end{lstlisting}

Let us sketch how encapsulation and module invariants can be formalized in a unary logic.
The linking of a client $C$ with a method implementation $B$ can be represented by a simple construct,  $\letcom{m}{B}{C}$ that binds $B$ to method name $m$.
(For clarity we ignore parameters and consider a single method rather than simultaneous linkage
of several methods.)
The \emph{modular linking rule} looks as follows,
where we use notation $C:\flowtyf{P}{Q}$ instead of the usual Hoare triple $\{P\}C\{Q\}$
(for partial correctness).\footnote{Following
   O'Hearn et al~\cite{OHearnYangReynoldsInfoToplas,RegLogJrnII},
   we use the term modular for information hiding, not just procedural abstraction.}
\begin{equation}\label{eq:mismatch}
\inferrule{ % *[left=MLink]{
m:\flowtyf{R}{S} \: \proves \: C: \flowtyf{P}{Q} \\
m:\flowtyf{(R\land I)}{(S\land I)} \: \proves \: B: \flowtyf{(R\land I)}{(S\land I)} \\
}{
\proves \letcom{m}{B}{C} : \flowtyf{P}{Q}
}
\end{equation}
The first premise says $C$ is correct under the hypothesis that $m$ satisfies the spec $\flowtyf{R}{S}$.
(The general form allows other hypotheses, which are retained in the conclusion.)
The second premise says the body $B$ of $m$ satisfies a different spec,
$\flowtyf{R\land I}{S\land I}$ (and assumes the same, as needed in case of recursive calls to $m$ in $B$).
The spec $\flowtyf{R}{S}$ should be understood as the interface on which $C$ relies---indeed, $C$ is \emph{modularly correct} in the sense that it satisfies its spec when linked with any correct implementation of $m$, so $C$ never calls $m$ outside its specified precondition $R$.
In the verification of $B$, the internal invariant $I$ can be assumed initially and must be reestablished.  The invariant is \emph{hidden} from clients of the module.

%since $C$ might call $m$ in states where $I$ does not hold.
%because $C$ wrote the internal locations on which $I$ depends.

As displayed, rule (\ref{eq:mismatch}) is obviously unsound 
because $C$ might write a location on which $I$ depends and then call $m$ in a state where $I$ does not hold.  
The idea is to prevent that by encapsulation, for which we are required to
\begin{list}{}{}
\item[(E1)] delimit the module's ``internal locations'',
\item[(E2)] ensure that the module's private invariant $I$ depends only on those locations,
\item[(E3)] frame the effects of $C$ and ensure its writes are separate from the internal locations,
and
\item[(E4)] arrange that $I$ is established initially (e.g., by module initialization and object constructors).
\end{list}

\paragraph{Relational modular linking} 

Encapsulation licenses more than just the hiding of invariants.
Once the requirements (E1)--(E4) are met in a way that makes (\ref{eq:mismatch}) sound, we can contemplate
the adaptation of (\ref{eq:mismatch}) to relational reasoning and in particular proving equivalence of two linkages, $\letcom{m}{B}{C}$ and $\letcom{m}{B'}{C}$.  
The labels (E1)--(E4) are used to also refer to the requirements as adapted to relational reasoning.

The two linkages
cannot be expected to behave identically: $B$ and $B'$ typically have different internal state on which they act differently.
What can be expected is that from initial states that are equivalent in terms of client-visible locations,
the two linkages yield final states that are equivalent on visible locations,
as indicated by the deliberately vague ``$vis$'' in Figure~\ref{fig:align}.
We say equivalent states because $B$ and $B'$ may do different allocations;
so the resulting heap structure should be isomorphic but need not be identical.
(For many purposes one wants to reason at the source language level of abstraction,
ignoring differences due to timing, code size, and absolute addresses; that is our focus.)
Given that we have framing (E3), it suffices to establish ``local equivalence'' in the sense
that initial agreement on locations readable by $C$ leads to final agreement on locations writable by $C$---and on freshly allocated locations.
Agreement on other visible locations should then follow.

We write $\splitbi{B}{B'}: \rflowtyf{\R}{\S}$,
for relations $\R$ and $\S$ on states, to say
that pairs of terminated executions of programs $B$ and $B'$, from states related by $\R$,
end in states related by $\S$.
For example, $\splitbi{C}{C}: \rflowtyf{\Agr x}{\Agr y}$ says two runs of $C$ from states that agree on the value of $x$ end in states that agree on the value of $y$.
The relational generalization of (\ref{eq:mismatch}) is a \emph{relational modular linking rule}
of this form:
\begin{equation}\label{eq:mismatchR}
\inferrule{ % *[left=rMLink]{
m:\flowtyf{R}{S} \: \proves \: C: \flowtyf{P}{Q} \\
m: \ldots % \rflowtyf{ (\Both{R}\land\Agr in\land \M) }{ (\Both{S}\land \Agr out \land \M) }
   \;\proves\;
%   \splitbi{B}{B'}: \rflowtyf{ (\Both{R}\land\Agr in\land \M) }{ (\Both{S}\land \Agr out \land \M) }
   \splitbi{B}{B'}:\, \rflowtyf{ \Both{R}\land\Agr in\land \M \, }{ \, \Both{S}\land \Agr out \land \M }
}{
\proves
\Splitbi{\letcom{m}{B}{C}}{\letcom{m}{B'}{C}} :
   \rflowtyf{ \Both{P}\land \Agr vis }{ \Both{Q}\land \Agr vis } \\
}
\end{equation}
The first premise is unary correctness of $C$ assuming the interface spec of $m$ as in rule (\ref{eq:mismatch}).
The conclusion of (\ref{eq:mismatchR}) expresses local equivalence of the two linkages,
under precondition $P$.
%(We let the spec arrow $\rspecSym$ bind loosely and omit needless parentheses, now that the reader is getting familiar with the notations.)
The second premise relates the two implementations $B$ and $B'$
and is meant to say that if the client-visible ``input'' locations
are in agreement then the resulting visible outputs are in agreement.
In addition, a relation $\M$ is conjoined to the pre- and post-condition.
A coupling relation $\M$ usually has three conjuncts: it says the left state satisfies some
invariant $I$ on the internal state used by $B$,
the right state satisfies invariant $I'$ on the internal state
used by $B'$, and there is some connection between the internal states.
% DN trying to deliberately introduce left/right terminology
(We often use ``left'' and ``right'' in connection with two programs, states, or executions to be related.)
The hypothesis for $m$ in the second premise is the same spec as proved for $\splitbi{B}{B'}$,
following the pattern in (\ref{eq:mismatch}).
We elide that hypothesis for readability: relational reasoning involves two of everything and the notations quickly become cluttered!
As with the modular linking rule (\ref{eq:mismatch}),
the relational modular linking rule (\ref{eq:mismatchR})
is unsound unless we satisfy requirements (E1)--(E4).
For relational reasoning, (E2) and (E4) are adapted to relations,
and (E3) is strengthened to ensure separation for reads,
as one would expect to avoid dependence on internal representations.

\paragraph{Alignment.}

%A relational spec is interpreted in terms of the paired initial and final states.
One technique for proving some relation on final states is to leverage functional specs:
a strong constraint on the output values, such as $out=f(in)$ for some mathematical function $f$, entails that initial agreement on $in$ leads to final agreement on $out$.  But the need to find and prove functional specs can often be avoided through judicious \emph{alignment}
of intermediate points in execution.  This technique is used to prove soundness of (\ref{eq:mismatchR}).  To illustrate, consider an instantiation of the general rule in which the three methods in Figure~\ref{fig:MCell} are bound simultaneously
(\whyg{cset}, \whyg{cget}, and the \whyg{Cell} constructor).
We show that two executions of the example client can be aligned as in Figure~\ref{fig:align}, with the indicated relations holding at the aligned points.  After the two constructor calls, the resulting states should agree on visible locations and be related by the coupling, according to the premise proved for the constructor.  From any pair of states related by $\Agr x\land\Agr c\land \M$, two executions of \whyg{x:=x+1} maintain agreement on visible variables including $x$, and according to (E3) this step in the client code is not touching internal locations on which $\M$ depends, so $\M$ continues to hold.  From any pair of states related by $\Agr vis\land \M$, a pair of calls to \whyg{cset} results in states  related, by the premise for \whyg{cset}.  Similarly for \whyg{cget}.
In fact $\M$ relates the final states in Figure~\ref{fig:align} but we omit it there,
to emphasize that it is an ingredient of proof rather than the property of ultimate interest.
%This particular alignment is effective: while it allows for unary pre- and post-conditions, it does not require precise functional specs; in the extreme case the postconditions $Q$ and $S$ in the rule can be simply $true$.

In a good alignment, most of the intermediate relations are
agreements ($\Agr$) that amount to simple equalities connecting values in locations of the two states.
Finding and exploiting good alignments is essential in order to leverage automatic theorem provers.
For \whyg{cset(c,v)} in Figure~\ref{fig:MCell},
the first implementation is \whyg{c.val:= v; return c.val} and the second is
\whyg{c.f:= -v; return -c.f}.
If we align their executions at the semicolons,
we can assert the coupling relation (\ref{eq:coupleCell}) at that point,
by unary reasoning about the effect of the two field updates.
Again by unary reasoning about the return expressions
we get that the same values are returned, as needed for the final
agreement on visible variable $y$.
Alignment does not eliminate the need for unary/functional reasoning, but rather
reduces it to small program fragments for which precise semantics can be computed by a theorem prover.

Alignment can be expressed by means of a product program, that is, a program, or some kind of automaton, whose executions correspond to paired executions of the given programs.  We call this well known technique the \emph{product principle}:
to prove a correctness judgment $\splitbi{C}{C'}: \rflowtyf{\R}{\S}$
relating programs $C$ and $C'$,  it suffices to prove the spec for some product program whose executions cover the executions of $C$ and $C'$.

To emphasize the role of alignment we consider another example,
not about representation independence but about secure information flow.
The following program acts on a linked list of integer values, where each node
has a boolean field, \whyg{pub}, meant to indicate that this value is public.
\begin{equation}\label{eq:sumpub}
sumpub: \qquad 
\mbox{\lstinline{s:=0;  p:=head;  while p <> null do if p.pub then s:=s+p.val fi;  p:=p.nxt od}}
\end{equation}
We want to specify and prove that this does not reveal any information about non-public values.
Suppose we can define $listpub(p)$ to be the mathematical list of public values reached from \whyg{p}.
To express that the final value of $s$ depends only on public elements of the list
we use the spec $\rflowtyf{\Agr listpub(p)}{\Agr s}$.
The program satisfies the unary spec
$\flowtyf{true}{s=sum(listpub(head))}$,
and any program that satisfies this must also satisfy
$\rflowtyf{\Agr listpub(head)}{\Agr s}$.
But we can prove the relational spec without recourse to the unary spec.
At points in execution where two runs have passed the same number of public nodes,
the relation $\Agr s\land \Agr listpub(p)$ holds;
this suggests an alignment where it suffices to use
relational invariant $\Agr s\land \Agr listpub(p)$.
Adding the same value to $s$ on both sides maintains $\Agr s$ and
there is no need to reason that $s$ is the sum of previously traversed public values.
The same relational invariant should suffice if sum is replaced by a more complicated function.
The alignment can be described as follows: consider an iteration just on the left (resp.\ right), if the next left (resp.\ right) node is not public; and simultaneous execution of the body on both sides, if both next nodes are public.

We cannot in fact define $listpub$ as a function of $p$, owing to the possibility of cycles in the heap.
Instead we use an inductive relation when we work out the details of this example Section~\ref{sec:weave}.

\paragraph{Summary of ingredients needed.}

To achieve the three goals in Section~\ref{sec:intro} we need:
\begin{itemize}
\item A unary logic of functional correctness under hypotheses (for procedure-modularity), that supports framing (for local reasoning) and encapsulation (for hiding and abstraction).
To support a wide range of programming patterns, the logic should support reasoning in terms of
encapsulation at the granularity of an object which ``owns'' some internal state, say representing an instance of an ADT.  It should also support reasoning at the granularity of a module,
% say, all the instances of a queue class, which share some internal representation to efficiently support moving elements between queues or resource accounting.
where many instances of multiple classes may share the internal representation.
It should encompass flexible patterns of sharing in data structures and between clients and components.

\item A relational logic with framing and encapsulation,
in which the relation formulas in specs and intermediate assertions are sufficiently expressive to describe data structures with dynamically allocated objects.  Agreement ``modulo renaming'' is needed to reason at the level of abstraction of Java/ML which provide reference equality and preclude arithmetic comparisons and operations on pointers,
to express local equivalence and other relations.
The logic must provide means to reason with alignments that admit simple intermediate relations.
Examples like the $sumpub$ program in (\ref{eq:sumpub}) show the need to use state-dependent alignments
in addition to alignments of control structure.
\end{itemize}
These ingredients need to be provided in ways that facilitate verification tools that leverage
automated provers especially SMT solvers.  Reasoning under hypotheses is straighforward to implement,
but effective expression of specs and alignment is less obvious.

%% \dn{At this point we could mention that Iris provides the listed ingredients in some form, in particular ReLoC,
%% but with the aim of supporting concurrency and higher order programs, for which reason its spec notations and underlying semantics are quite elaborate and far from facilitating implementation with the level of automation provided by systems like Dafny, Why3, VeriFast, Viper, etc. discussed in Sect.~\ref{sec:related}.
%% But it's been said already in the quick summary of state of the art.
%% }

\subsection{An approach based on region logic} \label{sec:approach}

Our relational logic is based on prior work in which ghost state is used in frame conditions to  describe sets of heap locations.
This approach, dubbed \emph{dynamic frames}~\cite{KassiosFM},
has been shown to be amenable to SMT-based automated reasoning in verification tools~\cite{SmansFAC10,Leino10,RosenbergBN12,piskac2014grasshopper}, and
shown to be effective in expressing relations on dynamically allocated data structures~\cite{AmtoftBB06,BanerjeeNN16}.
In particular we build on a series of articles on \emph{region logic} (\RL); it provides a methodologically neutral basis for heap encapsulation with sufficient generality for sequential first-order object-based programs featuring callbacks between modules.
We refer to key articles as \RLI~\cite{RegLogJrnI}, \RLII~\cite{RegLogJrnII},
and \RLIII~\cite{BanerjeeNN18},\index{RLI}\index{RLII}\index{RLIII}
and summarize key ideas in the following.

\paragraph{Framing.}

In current tools, the most common form of frame condition is a ``modifies clause'' that lists some expressions, meant to designate the writable locations.  A reads clause is similar. In the formalization of \RL, specifications are written in the compact form $\flowty{pre}{post}{\text{\emph{frame}}}$
where the effect expressions in the frame condition are tagged by keywords $\wri{}$ and $\rd{}$ to designate writables and readables.
We use $\rw$ to abbreviate the possibility to both read and write.
In this work, a \dt{region} is a set of object references.
For example, a possible spec of \whyg{cset(c,v)} is
$\flowty{c\neq\NULL}{cget(c)=v}{\rw{\{c\}\Img\allfields}}$
where the postcondition refers to the mathematical interpretation of the pure method \whyg{cget}~(as in \RLIII).
The singleton region $\sing{c}$ is used in the frame condition.
In the \emph{image expression } $\{c\}\Img\allfields$, the token $\allfields$ is a
data group~\cite{LeinoEtalDG} that abstracts from field names.
Concrete field names can also be used in image expressions, e.g., $\sing{c}\Img val$.
This example designates a single location, which may as well be written $c.val$.
But the image notation can be used for larger sets of heap locations.
For variable $r$ of type region,
$r\Img val$ designates the set of $val$ fields of all \whyg{Cell} objects in $r$.
So $\rd{r\Img val}$ in a frame condition allows any of these fields to be read.

Following separation logic, \RL\ features local reasoning in the form
of a \emph{frame rule}, but achieves this with ordinary first-order assertions.
For an example, strengthening the precondition of \whyg{cset(c,v)}
gives $\flowty{c\neq\NULL\land d\neq c}{cget(c)=v}{\rw{\{c\}\Img\allfields}}$.
The frame rule lets us add $d.val=z$ to the pre- and post-condition.
%$\flowty{c\neq\NULL\land d\neq c\land d.val=0}{cget(c)=v\land d.val=0}{\rw{\{c\}\Img\allfields}}$,
Why? Because the condition $d.val=z$ cannot be falsified:
the writes allowed by the frame condition are separate from
what is read\footnote{For a formula's meaning to depend on a location is different
from a program reading the location during execution.
However, these two notions have closely related extensional semantics based on agreement between states.
So, following the \RL\ articles, we use the terminology and notation of read effects for both.}
by the formula $d.val=z$.
In case of the variables $d$ and $z$, this is a matter of checking that $d$ and $z$ are not writable.  %(which could even be checked by inspection of the code).
Distinctness of field names can be used similarly.
But here, $\rw{\{c\}\Img\allfields}$ allows that $c.val$ can be written
and $val$ also occurs in the formula $d.val=z$.
Separation holds because the regions $\sing{c}$ and $\sing{d}$ are disjoint, written
$\disj{\sing{c}}{\sing{d}}$, which follows from precondition $d\neq c$.
As in the frame rule of separation logic~\cite{OHearnRY01},
this reasoning is inherently state dependent;
separation would not hold if variables $d$ and $c$ held the same reference.
Our frame rule has this form:
\begin{equation}\label{eq:frameRule}
\begin{array}{l}
\mbox{from} \quad C:\flowty{P}{Q}{\eff} \quad \mbox{infer} \quad  C:\flowty{P\land R}{Q\land R}{\eff} \\
\mbox{provided that locations read by $R$ are separate from locations writable according to $\eff$.}
\end{array}
\end{equation}
In the frame rule of \RL, 
separation is expressed by a conjunction of set disjointness formulas derived syntactically from 
the frame condition $\eff$ and the read effects of $R$.
In this example, the relevant effects are $\wri{c.val}$ and $\rd{d.val}$
and there is a single disjointness formula: $\disj{\sing{c}}{\sing{d}}$.
This formula is obtained by applying the separator function $\indSymbol$ introduced later, in Figure~\ref{fig:sepdef}.
% In the rule (Fig.~\ref{fig:proofrulesU}),
% the requisite disjointness formulas are obtained by a simple syntactic operation.

\paragraph{Encapsulation.}

\RLII\ features \emph{dynamic boundaries}, in which
the idea of dynamic frame is adapted to encapsulation for module interfaces.
The dynamic boundary of a module is simply an effect expression
that designates the locations meant to be internal to the module.
Technically, it is a read effect, in keeping with its role to cover
the footprint of the module invariant.
In addition to the usual meaning of a partial correctness judgment,
there is an additional obligation: the program must not \emph{write} locations within the boundary of any module other than its own module.

\tikzset{
  pics/sobj/.style args={#1}{
    code={
      \node[draw,rounded corners,minimum height=1cm,minimum width=1.25cm] (-main) {};
      \node[above of=-main,xshift=-0.75cm,yshift=-0.35cm] (-name) {#1};
      \node[] (-text) {\whyg{Stack}};
    }
  }
}

\tikzset{
  nobj/.pic={
    \node[draw,rounded corners,minimum height=0.75cm,minimum width=1cm] (-main) {};
    \node[] (-text) {\whyg{Node}};
  }
}

\begin{wrapfigure}{r}{0.40\textwidth} % r allows float, R means exactly here
  \centering
  \begin{tikzpicture}[scale=0.75,transform shape,>=stealth]
    \pic (stack-1) at (0,0) {sobj={$s_1$}};
    \pic (node-11) at (2,0) {nobj};
    \pic (node-12) at (3.5,0) {nobj};
    \pic (node-13) at (5,0) {nobj};

    \draw[->] (stack-1-main.east) to (node-11-main.west);
    \draw[->] (node-11-main.east) to (node-12-main.west);
    \draw[->] (node-12-main.east) to (node-13-main.west);

    \pic (stack-2) at (0,-1.75) {sobj={$s_2$}};
    \pic (node-21) at (2,-1.75) {nobj};
    \pic (node-22) at (3.5,-1.75) {nobj};

    \draw[->] (stack-2-main.east) to (node-21-main.west);
    \draw[->] (node-21-main.east) to (node-22-main.west);

    \node[rectangle,draw=blue,thick,dashed,rounded corners,
    minimum width=2cm,minimum height=3.5cm,
    left of=stack-1-main,xshift=0.95cm,yshift=-0.75cm] (pool-main) {};
    \node[above of=pool-main,yshift=1cm,xshift=-0.5cm] (pool-text) {$pool$};

    \node[rectangle,draw=blue,thick,dashed,rounded corners,
    minimum width=4.5cm,minimum height=1.2cm,
    left of=node-11-main,xshift=2.5cm] (s1-rep-main) {};
    \node[above of=s1-rep-main,yshift=-0.2cm,xshift=-1.5cm] (s1-rep-main-text) {$s_1.rep$};

    \node[rectangle,draw=blue,thick,dashed,rounded corners,
    minimum width=3cm,minimum height=1.2cm,
    left of=node-21-main,xshift=1.75cm] (s2-rep-main) {};
    \node[above of=s2-rep-main,yshift=-0.2cm,xshift=-0.75cm] (s2-rep-main-text) {$s_2.rep$};
  \end{tikzpicture}
\caption{The pool and rep idiom.}\label{fig:pool-and-rep}
\end{wrapfigure}

For the example module \whyg{MCell}, the dynamic boundary (omitted from Figure~\ref{fig:MCell}) is formulated in terms of a ghost variable, $pool$, of type region.  The postcondition of the \whyg{Cell} constructor says the new cell is added to $pool$.
The boundary is $\rd{pool},\rd{pool\Img\allfields}$, so
clients must not write the variable $pool$ or any field of an object in $pool$.
One could as well achieve this effect using module-scoped field names, so let us briefly consider a less degenerate example: a module for stacks.

In addition to ghost variable $pool$ containing all instances of the stack class,
that class would have a ghost field $rep$ of type region.
In an implementation using linked lists, each stack's list nodes would be in its $rep$,
and the module invariant would specify some ``object invariant''
for each stack together with its nodes.
This is depicted in Figure~\ref{fig:pool-and-rep}.
In an implementation using arrays, $rep$ would contain the stack's array,
and the module invariant would express some condition that holds for each stack object and its array.
Of course there is a single interface for the module.
Method frame conditions
will refer to $pool$ and $rep$, and not expose implementation details.
To facilitate per-instance framing, an invariant like
$s\neq t\imp \disj{s.rep}{t.rep}$ is used, which says the representations for distinct stacks are disjoint.
A suitable dynamic boundary is
$\rd{pool}, \rd{pool\Img\allfields}, \rd{pool\Img rep\Img\allfields}$.
It designates fields of the stack objects in $pool$ and also fields of all their rep objects.
(Array slots can be viewed as fields.)
The mentioned invariant enables use of the frame rule to consider updates of a single instance,
and it is suitable to be included in the module interface for use by clients.
(Either as explicit conjunct in method pre- and post-conditions,
or declared as a \emph{public invariant} for syntactic sugar.)
For example, \whyg{s.push(n)} writes $s.rep\Img\allfields$;
in states where $s\neq t$ this preserves the value of \whyg{t.top()} which reads $t.rep\Img \allfields$---and preservation holds in virtue of frame conditions, without recourse to postconditions that specify functional behavior.

In summary, a module interface comprises a collection of method specs,
%\dn{optionally including a public invariant,} 
and a dynamic boundary. A module implementation maintains an internal invariant $I$, the footprint of which
should be framed by the boundary.
The invariant $I$ should be such that it follows from the initial conditions of the main program.
For example, universal quantification over elements of $pool$ holds when $pool$ is empty.
An alternate approach is to require clients to call a module initializer.

\paragraph{Modular linking.}

Following the lead of O'Hearn et al.~\cite{OHearnYangReynoldsInfoToplas},
the logic in \RLII\ derives a modular linking rule like (\ref{eq:mismatch})
from two simpler rules:
An obviously-sound rule for the linking construct ($\letcom{m}{B}{C}$)
and a \emph{second order frame rule}
that accounts for hiding of invariants on encapsulated state.
A minimalistic formalization of modules is used, to keep the focus on the main ideas.
The unary correctness judgment takes the form
$\Phi \proves_M C: \flowty{P}{Q}{\eff}$
with $M$ the name of the module in which $C$ is to be used.
It says that, under hypotheses $\Phi$ and precondition $P$, command $C$ stays within the effects $\eff$ and
establishes $Q$ if it terminates---and in addition, $C$ \emph{respects} the boundaries
of any modules in $\Phi$ other than its own module $M$.
This formalizes requirement (E3).
In \RLII, ``respect of dynamic boundaries'' means not writing locations inside them.
In the present article, we must strengthen respect to prohibit reading, to ensure that $C$ has
no dependency---neither reads nor writes---on the internal representation of modules other than its own.
%This strengthening is nontrivial: it is difficult to get the semantics right,
%and while most proof rules need little or no revision they must all be proved sound
%for the new semantics.

\subsection{Relational region logic}\label{sec:rrl}

Our relational specs have the form $\rflowty{\P}{\Q}{\eff|\eff'}$ where $\P$ (resp.\ $\Q$)
is the relational pre- (resp.\ post-)condition.  There is a separate frame condition $\eff$ for the left execution and $\eff'$ for the right.  Often those are the same, in which case
we abbreviate as $\rflowty{\P}{\Q}{\eff}$.
The meaning of frame conditions and encapsulation is the same as in the unary logic.
Leaving effects aside, there are several ways one could interpret a spec
$\splitbi{C}{C'}: \rflowty{\P}{\Q}{\eff|\eff'}$
in regards to termination.
All ways consider a pair of initial states, say $\sigma,\sigma'$, that satisfy $\P$.
The ``$\forall\exists$ interpretation'' says that for every execution of $C$ from $\sigma$,  terminating in a state $\tau$, there is an execution of $C'$ from $\sigma'$ that terminates in a state related to $\tau$ by $\Q$.
The $\forall\exists$ interpretation asserts relative termination and 
caters for nondeterminacy.
The ``$\forall\forall$ interpretation''
was already mentioned just before (\ref{eq:mismatchR}):
every pair of terminating runs of $C$ and $C'$ from $\P$-related states end in $\Q$-related states.
The $\forall\forall$ form is fine for deterministic programs which is what we consider, and it is simpler, so we use it.
%For both relational and unary judgments our semantics disallows runtime faults, regardless of whether execution terminates.

For relation formulas we build directly on image expressions.
Agreements are interpreted in terms of a partial bijection
between the dynamically allocated references of the left and right states,
as commonly used to account for bijective renaming of references %% (i.e., bijective renaming)
at the Java/ML level of abstraction~\cite{BanerjeeNaumann02c,BanerjeeNaumannJFP05,BartheR05,Beringer11}; we call these \emph{refperms}.
For region expression $G$, the relation $\Agr G\Img f$
asserts agreement on $f$-fields for objects in $G$ that correspond
according to the refperm.
We do not require every allocated reference to be in the refperm: this is important,
to specify relational properties that allow differences in allocation behavior.
Examples of such differences include internal data structures
and reasoning about secure information flow 
(under low branch condition, allocated locations can be added to the refperm, but not under high branch condition).

We formulate the logic in terms of an explicit representation for product programs which designate alignments.
The biprogram form $\splitbi{C}{C'}$ indicates no alignment except for the initial and final states. 
Other biprogram forms express, for example, that iterations of a loop are to be aligned in lockstep,
or conditionally as needed for the $sumpub$ example (\ref{eq:sumpub}).
For the implementations of \whyg{cset}, the alignment described earlier
is expressed as \whyg{(c.val:= v | c.f:= -v); (return c.val | return -c.f)}.

A judgment for $\splitbi{C}{C'}$ directly entails the expected relation between unary executions of commands $C$ and $C'$ (as confirmed by our adequacy theorem). 
The choice to use a different alignment of $C$ with $C'$ is formalized by an explicit
proof rule.
The rule is formulated in terms of a \emph{weaving relation} that connects a biprogram with a more tightly aligned version, typically chosen because it admits use of simpler relational assertions.
The rule says that properties of the woven program hold also for $\splitbi{C}{C'}$.

Given that we confine attention to sequential code, it seems natural to expect that
programs are deterministic, but we also aim for reasoning at the source code level
abstraction---for which determinacy is unrealistic owing to dynamic allocation!
The behavior of an allocator typically depends on things that are not visible
at the source level.
There is no need to make unrealistic assumptions.
Our program semantics allows that the allocator may be nondeterministic (while not assuming that it is ``maximally nondeterministic'' as often done
in the literature).
Our program semantics is \emph{quasi-deterministic} in the sense that outcomes are unique up to bijective renaming of references.
Our relation formulas do not allow pointer arithmetic or comparisons other than equality,
so they are invariant under renaming.
These design decisions entail some complications in the technical development, but
ensure that interesting programs do provably satisfy expected $\forall\forall$ properties.

As already mentioned,
the unary modular linking rule (\ref{eq:mismatch}) is derived (in \RLII)
from two simpler rules:
a basic linking rule, where assumed and proved specs match exactly, together with a second order frame rule.
Our novel relational modular linking rule (\ref{eq:mismatchR}) is derived from
a relational linking rule, a relational second order frame rule, and a third rule.
The third rule lifts a unary correctness judgment to a relational judgment that says
a program is locally equivalent to itself.  For this to be proved, it is stated in a stronger form: a program can be aligned with itself in lockstep such that local equivalence holds at each intermediate step.

As for the goal of foundational justification,
our approach is to work directly with a conventional operational semantics for unary correctness, for which we formulate a semantics of encapsulation.
The biprogram semantics is based directly on that,
so that soundness for rules in the relational logic has a direct connection---adequacy theorem---to unary semantics.
One benefit from carrying out the development in terms of this elementary semantics
is
%to open the door to proving correctness of a verifier based on automated provers.
%Another benefit is
that one can see that most of the soundness proofs can be adapted
easily to total correctness (both runs always terminate)
and to relative termination (right run terminates whenever left does).

\subsection{Contributions}

We highlight the following contributions.

\emph{A unary logic for modular reasoning about sequential object-based programs using first-order assertions.}
The key contribution and most difficult definition to get right is the extensional semantics of encapsulation, which is part of the meaning of correctness judgments.
Small-step operational semantics is used so we can define what it means for a given
step to be outside the boundaries of all modules but its own.
We build on the semantics in \RLII\ but completely revamp it to handle encapsulation of 
reads in addition to writes.
Dynamic boundaries are taken from \RLII; 
most of the proof rules of \RLII\ need little or no revision, but they must all be re-proved for the new semantics. 
Owing to the need for quasi-determinacy (for $\forall\forall$ extensional semantics
of read effects),
the new semantics of hypothetical judgments quantifies over possible denotations (called \emph{context interpretations}) rather than a single ``least refined'' denotation as
in \RLII\ and in O'Hearn et al~\cite{OHearnYangReynoldsInfoToplas}.
%The proof rules enforce encapsulation by conditions similar to those in \RLII,
%but as the semantics is different new soundness proofs are needed.
We present detailed soundness proofs of the key rules (Theorem~\ref{thm:unarysoundness}).

\emph{A relational logic.}
The logic relies on unary judgments for reasoning about atomic commands and for enforcing encapsulation.
Relational assertions are first-order formulas.
Our presentation focuses on data abstraction,
because this is the first relational logic to embody representation independence as a proof
rule using only first-order means.
But the logic is general, with a full range of rules that facilitate reasoning with convenient alignments.

% (though it is also surprising how incomplete many published logics are, in terms of which alignment patterns can be used).
We present detailed soundness proofs of the key rules (Theorem~\ref{thm:sound}).
Formally, judgments of the relational logic give properties of biprograms;
the adequacy Theorem~\ref{thm:biprogram-soundness}
connects those properties with the expected properties
in terms of paired unary executions in standard semantics (the product principle).

\emph{Demonstration of suitability for automation via case studies in a prototype relational verifier.}
The prototype translates biprograms and verification conditions specific to our logic, which are all first-order, into Why3 code and lemmas, proved using SMT solvers~(\url{why3.lri.fr}).
The modular linking rules (unary and relational) are implemented by generating suitable Why3 specs for the programs involved.
The case studies include noninterference, program transformations,
and representation independence.

\subsection{About the proofs}

The most difficult technical result is the \emph{lockstep alignment lemma} (Lemma~\ref{lem:rloceq}).
It brings together the semantics of encapsulation in the unary logic,
which involves a single context interpretation,
with the semantics of relational correctness---which involves three context interpretations, to account for un-aligned calls as well as aligned calls and relational specs.

The direct use of small-step semantics makes for lengthy soundness
proofs that require, in some cases, intricate inductive hypotheses.
But transition semantics is a critical ingredient
for a first-order definition of heap encapsulation.
It was quite difficult to arrive at rules for relational linking and second order framing that are provably sound.  Several variations on the semantics of encapsulation turned out to be sound for the unary linking and second order frame rules but failed to validate a sufficiently strong lockstep alignment property on which relational linking can be based.

Aside from lockstep alignment, the soundness proofs for linking rely on denotational semantics which in turn relies on quasi-determinacy.  This property is also used to establish embedding/projection results on which the adequacy theorem is based.

The semantics of correctness judgments is extensional in the sense that it refers only to behavior in a standard transition semantics---no instrumentation artifacts.  Like in \RLII, it does rely on use of transition semantics in order to express that control is currently within a specific module and outside the boundaries of other modules in scope.
This affects which program transformations are correctness-preserving;
more on this in Section~\ref{sec:uequiv}.
%This in turn means that some program transformations are not correctness-preserving, or are harder to prove so, in our semantics; more on this in Sect.~\ref{sec:uequiv}.

Once the right definitions, lemmas, and induction hypotheses have been determined, the soundness proofs %simply
go by induction on traces, with many details to check. We relegate them to appendices.

\subsection{Current limitations}\label{sec:curlim}

The formal development omits some features that were handled in the prior works on which we build: parameters,
%recursion,
private methods, constructor methods, pure methods for abstraction in specs.
These are all compatible with the formal development;
all are implemented in the prototype and used in exposition.
The theory is compatible with standard forms of encapsulation based on scoping mechanisms
(e.g., module scoped variables),
which for practical purposes should be leveraged as much as possible; for simplicity we refrain from
formalizing such mechanisms.\footnote{Specs involving explicit footprints are more
  verbose than those based on separation logic,
  and our minimalist formalization of modules increases verbosity.
  This article does not propose concrete syntax for practical use,
  but the issue is addressed in some related work (Section~\ref{sec:related}).
}  % was sec:relatedRL
The prototype also supports public invariants;
as noted in connection with the stack example, these are
important for client reasoning about boundaries using patterns like ownership.
Public invariants need not be formalized in the theory, as they can
be explicitly included in method specs.

The simplicity of our semantic framework (e.g., standard semantics of formulas and programs) may facilitate foundational justification of a verifier,
but we have not formally proved the correctness of our prototype.

There are two technical limitations.
First, the semantics of encapsulation and the proved rules
handle collections of modules with both import hierarchy and callbacks.
But the key rules for relational linking and
relational second order framing (\rn{rSOF}) only handle simultaneous linking of a collection of modules.
This is enough to model linking as implemented in a verifier.
However, one may hope for a theory that accounts for distinct inference steps that successively link different layers of hierarchy, as in our unary logic.
To achieve this, the lockstep alignment lemma needs to be strengthened to ensure
agreements for already-linked methods.
This requires to further complicate an already intricate theory.
In this article we just sketch the issue (Section~\ref{sec:nestedX}).

Second, the current formulation has a technical condition (boundary monotonicity) that prevents release of encapsulated locations,
in the sense of reasoning with specs that describe outward ownership transfer.
(Inward transfer is fine.)
Modules can create new objects for clients, as in the shared handle objects for priority queues,
one of our running examples.
But a location that has been within the boundary must stay there.
Overcoming this restriction, or finding
idiomatic specification patterns that dodge it, is left to future work.
Both inward and outward transfer are possible in \RLII\ (an example is in Section~2.2 of that article).

Addressing the limitations is the subject of ongoing and future work.

%\section{Syntax of programs, biprograms, and their specifications}\label{sec:syntax}
\section{Programs: their syntax and specifications}\label{sec:Usyntax}

\begin{figure}[t]
  \begin{lstlisting}
class Pnode { val: int; key: int; sibling: Pnode; child: Pnode; prev: Pnode; }

class Pqueue { head: Pnode; size: int; ghost rep: rgn; }

meth Pqueue (self:Pqueue) =
  self.rep := {null}; pool := pool ** {self};

meth insert (self:Pqueue, val:int, key:int): Pnode =
   result := new Pnode(val, key);
   self.rep := self.rep ** {result};
   if self.head = null then self.head := result;
   else self.head := link(self, self.head, result) fi;
\end{lstlisting}
%/*internal method*/
%\vspace*{-2ex}
  \caption{Excerpts of priority queue (PQ) implementation (in the syntax of our prototype).}\label{fig:PQueue1}
\end{figure}

This section defines the syntax of programs and their unary specifications and correctness judgments.
Subsections~\ref{sec:progtype}--\ref{sec:ucorr} 
collect together almost all the syntactic forms
and definitions concerning syntax, using a few examples to explain unusual things.
Section~\ref{sec:eg:encap} gives more holistic examples
to illustrate how the syntax is used and why we need various syntactic elements,
focusing on how requirements (E1)--(E4) for encapsulation in Section~\ref{sec:modrel} are expressed and checked.

\subsection{Programs and Typing}\label{sec:progtype}

A running example is introduced in Figure~\ref{fig:PQueue1}.
We consider the priority queue module \whyg{PQ} which exposes a class whose instances represent priority queues that store integer values and priorities, referred to as ``keys'' (smaller key means higher priority)~\cite{Weiss}.
Our implementations (based on~\cite{Weiss}) use pairing heaps, where each queue contains a $head$ field that points to a \whyg{Pnode} object and each \whyg{Pnode} contains $sibling$, $prev$, and $child$ fields that point to other \whyg{Pnode}s.
The $rep$ field of a queue is used to hold 
references to the objects notionally owned by the queue.

%all objects that are reachable
%from the queue's $head$, or have been reachable but were removed by \whyg{deleteMin}.

The syntax of programs in our formal development is in Figure~\ref{fig:bnf}.
The grammar includes biprograms, to which we return in Section~\ref{sec:biprograms}.
% moved: A class is just a named record type.
Field read and write commands are written with dereferencing implicit,
as in Java (though using the symbol $:=$) and are desugared to have a single heap access which simplifies proof rules.
The $\keyw{let}$ construct, featured in the modular linking rule (\ref{eq:mismatch}),
represents scoped method declarations.\footnote{We use the short term ``method'' for what should properly be 
called procedure. The term ``method'' usually implies dynamic dispatch which is beyond the scope of this article.}
Some examples, like Figure~\ref{fig:PQueue1}, use the syntax of our prototype, in which keyword \whyg{meth} corresponds to the $\keyw{let}$ construct.
Examples use some syntax sugars implemented in our prototype,
e.g., invocation of method \whyg{link} in an update of field \whyg{self.head} (Figure~\ref{fig:PQueue1}).
A method named after a class (e.g, \whyg{Pqueue}) is meant to be used as a constructor,
i.e., invoked on a newly allocated object, the fields of which are initialized with default values (null for classes, $\Emp$ for regions).

To lessen the need for uninteresting transitions in program semantics,
we equate certain syntactic forms.
For example, there is no transition from $(\skipc;C)$ to $C$ because we consider them to be the same syntactic object, see Figure~\ref{fig:synident}.
Working with syntax trees up to (i.e., quotiented by) syntactic equivalence
is done in the previous \RL\ articles and elsewhere.\footnote{See, e.g.,~\cite{AptOld3}.
   We use the symbol $\equiv$ because
   it is used for structural congruences in process algebra,
   which have the same purpose of streamlining the transition system.}
We sometimes use the symbol $\equiv$ for equality of other syntactic forms, like variables, just to emphasize that they are syntactic.

\begin{figure}[t]
  \begin{small}
    \[\begin{array}{l@{\hspace{.5em}}l@{\hspace{.2em}}r@{\hspace{.3em}}l}
  \multicolumn{4}{l}{m\in \mathconst{MethName}\qquad x,y,r\in \mathconst{VarName} \qquad f,g\in \mathconst{FieldName} \qquad   K\in \mathconst{DeclaredClassName} } \\[1ex]
\mbox{(Classes)} &   &::=&
\keyw{class} ~ K ~ \{ \ol{f\scol T} \}
   \quad\mbox{(overline indicates finite lists)} \\
\mbox{(Types)} & T &::=&
\INT \gmid \BOOL \gmid \Region \gmid K \quad \mbox{(and math types, in specs and ghost code)} % \OBJECT \gmid K
\\[.2ex]
\mbox{(Prog.\ expr.)}& E &::=& x \gmid n \gmid \NULL \gmid {E\otimes E }
\quad \mbox{where $n$ is in $\Z$ and $\otimes$ is in $\{=,+,-,*,\geq,\land,\ldots\}$ }\\[.2ex]
\mbox{(Region\ expr.)} & G  & ::= & x \gmid \Emp \gmid \sing{E} \gmid G\Img f \gmid G/K \gmid G \otimes G
  \quad \mbox{where $\otimes$ is in $\{\cup,\cap,\setminus\}$ }\\[.2ex]
\mbox{(Expressions)} & F & ::= & E \gmid G  \\[.2ex]
\mbox{(Atomic com.)} & A & ::= & \skipc \gmid m() \gmid x := F \gmid x := \new{K} \gmid x := x.f \gmid x.f := x  \\[.2ex]
\mbox{(Commands)} & C &::=& A \gmid \letcom{m()}{C}{C} \gmid \ifc{E}{C}{C} \gmid \whilec{E}{C} \gmid \seqc{C}{C}
\gmid \varblock {x\scol T}{C}  \\[.2ex]

\mbox{(Biprograms)} & CC &::=& \splitbi{C}{C} \gmid \syncbi{A}
  \gmid \letcombi{m()}{\splitbi{C}{C}}{CC}
  \gmid \varblockbi{x\scol T\smallSplitSym x\scol T}{CC}
\gmid \seqc{CC}{CC} \\[.2ex]
&&&
\gmid \ifcbi{E\smallSplitSym E}{CC}{CC}
\gmid \whilecbiA{E\smallSplitSym E}{\P\smallSplitSym \P}{CC}
\\[.2ex]
\multicolumn{4}{l}{
  \mbox{Syntax sugar: $\whilecbi{E\smallSplitSym E'}{CC}$
abbreviates $\whilecbiA{E\smallSplitSym E'}{\False\smallSplitSym \False}{CC}$.}}
\\[.2ex]
\multicolumn{4}{l}{
  \mbox{Identifiers: $B,C,D$ for commands, $BB,CC,DD$ for biprograms.}}

\end{array}\]
\end{small}
%\hrule
%\vspace*{-2ex}
\caption{Programs and biprograms. For relation formulas $\P$ see Figure~\ref{fig:relFormulas}.}
\label{fig:bnf}
\end{figure}

\begin{figure}[t]
\begin{small}
\( (\skipc;C) \equiv C \quad (C;\skipc) \equiv C \quad (C_0;C_1);C_2 \equiv C_0;(C_1;C_2) \)  \\
\( \splitbi{\skipc}{\skipc} \equiv \syncbi{\skipc} \quad \syncbi{\skipc};CC \equiv CC
\quad CC;\syncbi{\skipc} \equiv CC
   \quad (CC_0;CC_1);CC_2 \equiv CC_0;(CC_1;CC_2) \)
\end{small}
%\vspace*{-1ex}
\caption{Syntactic equivalence \ghostbox{$\equiv$} of programs and biprogams.}
\label{fig:synident}
\end{figure}

Programs and specs are typed in a conventional way.
A \dt{typing context} $\Gamma$ maps variable names to data types
and method names to the token $\meth$, written as usual as lists,
e.g., $x\scol T,y\scol T,m\scol \meth$.
(In the formalization we omit method parameters and results.)
Various definitions refer to a typing context typically meant to be the global
variables, including ghost variables which may be of type $\Region$ (region).
We do not formalize ghost variables as such~\cite{FilliatreGP16,RegLogJrnI}.

The idea of ghost code is to instrument a program with extra state for
the sake of reasoning, in such a way that the termination and behavior of 
the original program is not affected.  This can be formalized in terms of
a rule for elimination of ghost state~\cite{OwickiGries,FilliatreGP16,RegLogJrnI}.
We refrain from doing so in this article; the additions would not be illuminating.

A class is just a named record type.  In the formal development
we assume an ambient \dt{class table} that declares some class types and the types of their fields.  For simplicity this has global scope.
%In our implementation there is a separate class table on the left and on the right, but in this article, w.l.o.g.\ we use the same table on both sides.
We assume that field names in different class declarations are distinct, so any declared field $f$ determines a unique class, $\Class(f)$,\index{$\Class$}
that declares it, and also a type, which we write $f:T$.

Section~\ref{sec:approach} introduced the region expressions used in frame conditions.
In addition to (mutable) variables of type region, there are set operations like union, singleton,
subtraction ($\setminus$), and image expressions. 
The expression $\sing{x}$ denotes the singleton set containing the value of $x$.
For $G$ a region expression, the image expression $G\Img f$ is the empty region if $f:\INT$.
If $f$ is of some class type, $G\Img f$ is the set of current values of $f$-fields
%of objects whose reference is in $G$.  
of objects (i.e., object references) in $G$.  
For $f$ of type $\Region$ the image is the union of the field values.  For example, in the idiom using global variable $pool:\Region$
containing some objects with field $rep:\Region$, the image
$pool\Img rep$ is the union of their $rep$ fields.
The type restriction expression $G/K$ denotes the elements of $G$ of type $K$ (which excludes null).

As usual in program logics, field access and update is limited to the primitive forms
$x:=y.f$ and $x.f:=y$. 
In specs and ghost code, a dereference chain like $x.f.g.h$ (for reference type fields) can be expressed by the region expression $\sing{x}\Img f\Img g\Img h$; if $x$ is null the value is the empty set.

\begin{wrapfigure}{R}{0.40\textwidth}  % r allows float, R means exactly here
\begin{small}
\(
  \inferrule{ \Gamma\proves E:K }
  { \Gamma\proves \sing{E}:\Region}
\quad
  \inferrule{\Gamma \proves G: \Region }
  {\Gamma \proves G\Img f: \Region}
\)
\end{small}
\caption{Region expression typing (selected).}\label{fig:rgnTyping}
\end{wrapfigure}

Owing to the simple model of classes,
the notation \ghostbox{$G\Img\allfields$} can be defined as shorthand for $G\Img\ol{f}$ where
$\ol{f}$ is the list of all field names.
An implementation can support user-defined data groups which
can be used to abstract from specific sets of fields~\cite{LeinoEtalDG}.

The typing rules for expressions and commands are straightforward and omitted,
with the exception of those in Figure~\ref{fig:rgnTyping}.
We highlight those because
we allow  $f$ in an image expression $G\Img f$ to have any type;
as noted above, its value is empty unless $f$ has region or class type.\footnote{Typing in \RLI,\RLII\ is slightly more restrictive.}

Program variables are partitioned into two sets, ordinary variables and \dt{spec-only variables}.\footnote{As
  in \RLII,
   we rely on a partition of ordinary variables into \dt{locals}, which are bound by
   $\keyw{var}$ (and in \RLII\ also method parameters), and \dt{globals}; but
   we ignore the distinction where possible.
   Also, typing rules impose the \dt{hygiene property} that variable and method names
   are not re-declared; this facilitates modeling of states and environments as maps.
   } %footnote
The distinguished variable $\lloc:\Region$ is an ordinary variable, but it is treated specially: It is present in all states, and is automatically updated in the transition semantics by the transition for $\NEW$, so in every state its value is exactly the set of allocated references.
Spec-only variables are used in specs to ``snapshot'' initial values for reference in the postcondition.
Spec-only variables do not occur in code, even ghost code, or in effects.\footnote{Spec-only variables are also used in \RLII. But here we also disallow the use of $\lloc$ in ghost code, which was not necessary in \RLII, so we have additional need to snapshot $\lloc$.}
In our prototype, ``old'' expressions are used to abbreviate the use of snapshot variables~\cite{LeavensBR06}.

Commands are typed in a context $\Gamma$.
We omit the straightforward rules for typing of commands,
except to note that a call $\Gamma\proves m()$ is well formed only if $m:\meth$ is in $\Gamma$.
To streamline the formal development we omit parameters for methods; by-value parameters can be handled straightforwardly as in \RLII\ and \RLIII.\footnote{As in those works,
   we also disallow $\keyw{let}$-commands inside let-bound commands and biprograms:
   in $\letcom{m}{B}{C}$ there must be no $\keyw{let}$ in $B$.
   (By modeling only top-level method declarations, we simplify the semantics.)
   We also disallow free occurrences of local variables in $B$; thus in
   $\varblock{x\scol T}{\letcom{m}{B}{C}}$  the module code $B$ can't refer to $x$.
  In practice, let is only used outermost.
} %footnote

Program expressions $E$ are heap independent.
For expressions of reference type, the only constant is $\NULL$ and the only operation is equality test, written $=$.
Region expressions can depend on the heap but are always defined.
Null dereference faults only occur in the primitive load and store commands $x:=y.f$ and $x.f:=y$.
By contrast, if $x$ is null then $\sing{x}\Img f$ is defined to be empty.

\subsection{Modules}\label{sec:modules}

Assume given a set \mathconst{ModName} of module names, and map
\index{$\mdl$}
$\mdl:\mathconst{MethName}\to\mathconst{ModName}$
that associates each method with its module.
Usually we use letters $M,N,L$ for module names,
but there is a distinguished module name, $\emptymod$,
that serves both as main program and as default module in the proof rules for atomic commands.
Assume given a preorder $\imports$ (read ``imports'') on \mathconst{ModName},
which models the reflexive transitive closure of the import relation
of a complete program. %, where $M$ imports $N$ if code of $M$ calls a method of $N$.
We write $\irimports$ for the irreflexive part.
Cycles are allowed, as needed for interdependent modules that respect each other's encapsulation boundaries.
A module interface includes a spec for each method.
The function $\bnd$ \index{$\bnd$} from $\mathconst{ModName}$ to effect expressions
associates each module with its dynamic boundary, which is thus part of its interface along with its method specs.
This lightweight formalization of modules is adapted from \RLII\ (its Section~6.1).

For the \whyg{PQ} interface in Figure~\ref{fig:PQintr}, $\mdl(\code{insert})=\code{PQ}$.
In one of our case studies, the main program implements Dijkstra's single-source shortest-paths (SSSP) algorithm, as a client of \whyg{PQ} and another module \whyg{Graph}.
The import relations are then $\emptymod\irimports\code{PQ}$ and
$\emptymod\irimports \code{Graph}$.

A module $M$ specifies a dynamic boundary $\bnd(M)$.
The boundary can be expressed using regions and data groups for abstraction,
to cater for implementations that have differing internals.
This is why there is a single type, $\Region$, for sets of references of any type.
Well-formedness conditions for boundaries are defined in Section~\ref{sec:unaryspec}.

A proper module system would include module-scoped variables and fields that need not be part of the interface and need not be the same in different implementations of a module $N$.
%In our simplified setup, if two implementations need to use different internal variables, all of them can be included in $\bnd(N)$.
Our simplified  formulation streamlines the formal development, because we do not need
syntax, typing contexts, etc.\ for a full-fledged module calculus, nor correctness judgments for modules.
But this comes at a price: some well-formedness conditions on correctness judgments (in the following subsections) and side conditions (in proof rules) merely serve to express lexical scoping that could be handled more neatly using a proper module system.

% [Mar-15-22] RN: spec for insert originally also had val >= 0 and key >= 0,
% but these aren't necessary.
\begin{figure}[t]
\begin{lstlisting}
module PQ =
  public pool: rgn
  boundary { pool, pool`any, pool`rep`any }

  meth Pqueue (self: Pqueue)                  /* constructor */
  meth isEmpty (self: Pqueue) : bool
  meth findMin (self: Pqueue) : Pnode

  meth insert (self: Pqueue, val: int, key: int) : Pnode
     requires { self <> null /\ self iin pool }
     ensures  { not (isEmpty(self)) /\ result iin self.rep /\ result.val = val /\ result.key = key }
     writes { {self}`any, self.rep`any, alloc } reads { {self}`any, self.rep`any, alloc }

  meth deleteMin (self: Pqueue)
  meth decreaseKey (self: Pqueue, handle: Pnode, key: int)
end
\end{lstlisting}
%\vspace*{-2ex}
  \caption{Priority queue interface \whyg{PQ}, eliding private methods and most specs.}
  \label{fig:PQintr}
\end{figure}

\subsection{Unary specifications}
\label{sec:unaryspec}

We assume a first-order signature providing primitive type, function,
and predicate symbols for use in specs and in ghost code.
Predicate formulas are in Figure~\ref{fig:preds}.
The  \dt{points-to} relation $x.f=E$ says that $x$ is non-null and the value of field $f$ equals
the value of $E$.
For examples, see the postcondition of \whyg{insert} in Figure~\ref{fig:PQintr}.
The predicate $\Type(G,\ol{K})$ says that every non-null reference in $G$ has
one of the class types in the list $\ol{K}$.

Typing of unary predicate formulas $P$ is straightforward.
For example, the points-to formula $x.f=E$ is well formed (\dt{wf}) in $\Gamma$
provided $\Gamma(x)$ is some type $K$ that declares $f:T$ and $E$ has type $T$.
An expression $E$ counts as an atomic formula if it has type $\BOOL$;
this includes equality tests.
The signature may include equality at other math types, with standard interpretation.

Quantifiers at a class type $K$ range over allocated references of type $K$.
The logic does not require quantification at type $\Region$ but we include it
to simplify the grammar.
It is often useful to bound the range of quantification at reference type to a specific region,
in the form $\all{x:K}{x\in G \imp P}$, to facilitate framing.
(This is explored in \RLI.)
In sugared form: $\all{x:K\in G}{P}$.

\begin{figure}[t]
\begin{small}
\(
\begin{array}{lrl@{\hspace{.2em}}}
P & ::= & E \mid x.f=E \mid G  \subseteq G \mid \Type(G,\ol{K}) \mid R(\ol{F})
            \quad (\mbox{atomic formulas, where $R$ is in the signature}) \\[.2ex]
& & \mid P \land P \mid P\imp P \mid (\all{x:T}{P})
\\[.2ex]
\multicolumn{3}{l}{
  \mbox{Syntax sugar: } \disj{G}{H} \eqdef G\intersect H \subseteq \sing{\NULL}
\mbox{ and } x\in G \eqdef  \sing{x}\subseteq G 
  \mbox{ and standard defs of $\neg$, $\lor$, and $(\some{x:T}{P})$.}
} \\
\multicolumn{3}{l}{\mbox{Precedence: $\land$ binds more tightly than $\imp$ and less tightly than relations like $=,\subseteq$.}} \\
\multicolumn{3}{l}{\mbox{Associativity: $P\imp Q\imp R$ means $P\imp(Q\imp R)$.}}
\end{array}
\)
\end{small}
%\vspace*{-1ex}
\caption{State predicates. For expression forms $E$, $F$ and $G$ see Figure~\ref{fig:bnf}.}
\label{fig:preds}
\end{figure}

\paragraph{Effect expressions.}

A \dt{spec} \ghostbox{$\flowty{P}{Q}{\eff}$} comprises
precondition $P$, postcondition $Q$, and frame condition $\eff$.
Frame conditions are \dt{effect expressions} $\eff$, defined by
\begin{equation}\label{eq:effects}
\begin{array}{llcl}
\mbox{(Left-expression)}   & LE & ::= & x \gmid G\Img f \\
\mbox{(Effect expression)} & \eff & ::= & \rd{LE} \gmid \wri{LE} \gmid \eff,\eff \gmid \emptyeff
\end{array}
\end{equation}
\dt{Left-expressions}, $LE$, are a subset of expressions (category $F$ in Figure~\ref{fig:bnf}).
They have l-values, as discussed below, and are used in effects and in agreement formulas.\footnote{For readers familiar with prior \RL\ articles:
Effect expressions are exactly the same
as in previous articles; we have changed the grammar for clarity.}
An effect $\eff$ is wf in $\Gamma$ provided each of its left-expressions is.
% DN hyphenate left-expression to avoid confusion with other uses of left.

Notation: Besides $\eff$ we often use identifiers $\effe$ and $\delta$ for effect expressions.
We use the short term \dt{effect} for effect expressions, including compound ones like
$\rd{x},\wri{x},\wri{\sing{x}\Img f}$.
The singleton image $\wri{\sing{x}\Img f}$ can be abbreviated as $\wri{x.f}$.
We use the abbreviation $\rw{}$ to mean $\rd{}$ and $\wri{}$.
The empty effect is given explicit notation $\emptyeff$
for clarity in certain parts of the development, but we omit it when confusion seems unlikely.
We often treat compound effects as sets of atomic reads and writes.
We also omit repeated tags, e.g., $\rd{x,y}$ abbreviates $\rd{x},\rd{y}$;
and then reads are separated from writes by semicolon, e.g.,
$\rd{x,y};\wri{z,w}$.

\paragraph{l-value and r-value.}

In common usage, the term r-value refers to the meaning of an expression in contexts like the right side of an assignment. For those expressions allowed on the left of an assignment, the l-value is the location to be assigned and the r-value is 
the current contents of that location~\cite{Strachey00}.
In our language there are two forms of mutable location: variables 
and heap locations.  A heap location is a pair $(o,f)$ where $o$ is an object reference and $f$ a field name; we write the pair as $o.f$.

We identify a subset of expressions, called left-expressions (\ref{eq:effects}),
which have an l-value ---in addition to the r-values described in Sec.~\ref{sec:progtype} (and formalized in Figure~\ref{fig:rexpsem}).
In general, the l-value of a left-expression designates a set of locations.
In frame conditions, left-expressions are interpreted for their l-values 
as is common in spec languages. 
(Note that our left-expression form $G\Img f$ is not an assignment target.)

In the write effect $\wri{x}$, the l-value of expression $x$ is a single location, the variable $x$ itself, independent of the current state.
For the left-expression  $\sing{x}\Img f$, the l-value is again a single location, namely 
$o.f$ where $o$ is the r-value of $x$ in the current state ---unless that value is null, in which case the l-value is the empty set.

Consider a variable $r:\Region$.
The l-value of $r\Img f$ is the set of $o.f$ where $o$ is a non-null reference
that is an element of the current value of $r$.  (We may say ``object in $r$'' to be casual.)

What about the l-value of $r\Img f \Img g$? It is the set of $o.g$ where $o$ is a non-null reference
in the region $r\Img f$---that is, $o$ is an element of the r-value of $r\Img f$.
In case $f$ has type $\INT$, that region is empty.
In case $f$ has some class type $K$, the region $r\Img f$ is the set of contents of $f$ fields of objects in $r$.  So, for $o.g$ to be in the l-value of $r\Img f\Img g$ means $o$ is the value in $p.f$ for some non-null reference $p$ in $r$.

Suppose instead that $f$ has type $\Region$.
Then the r-value of $r\Img f$ is defined to be the union of the values of the $f$-fields of objects in $r$.
(We use the union in order to avoid sets of sets.)
So, for $o.g$ to be in the l-value of $r\Img f\Img g$ means $o$ is an element of the set $p.f$ for
some non-null $p$ in $r$.

\begin{sloppypar}
In general, the l-value of a left-expression is dependent on the state, for the values of variables and for the values of fields of allocated objects.
For example, consider the private method, \whyg{link}, used internally by \whyg{insert}
(Figure~\ref{fig:PQueue1}). The ascribed effect of method \whyg{link} is
$\rw{\sing{\self}\Img rep\Img child,
  \sing{\self}\Img rep\Img sibling,
  \sing{\self}\Img rep\Img prev}$.
Here, $\{\self\}\Img rep$ is used for its r-value which is a set of objects in the $rep$ field
(the same as $\self.rep$), and the left-expression
$\{\self\}\Img rep\Img child$ is used in the effect to refer to the locations of the child fields of all the \whyg{Pnode}s in $\self\Img rep$.
\end{sloppypar}

\paragraph{Dynamic boundary and operations on effects.}

For expressions and atomic formulas, read effects can be computed syntactically by the \dt{footprint function}, $\ftpt$, defined in Figure~\ref{fig:foot}.
For example, the private invariant for the \whyg{PQ} module (Figure~\ref{fig:PQintr})
includes $q.rep\Img prev\subseteq q.rep$.  Its footprint, computed by $\ftpt$, is
$\rd{q}, \rd{\sing{q}\Img rep}, \rd{\sing{q}\Img rep\Img prev}$,
which can be abbreviated as $\rd{q, \sing{q}\Img rep, q.rep\Img prev}$.
It has a closure property, framed reads, that will play a role in reasoning about encapsulation.

% includes the conjunct $rep\Img sibling\subseteq rep \land rep\Img child\subseteq rep$.
% Its footprint, computed by $\ftpt$, is $\rd rep, \rd rep\Img sibling, \rd rep\Img child$.

%\begin{wrapfigure}{r}{0.40\textwidth}  % r allows float, R means exactly here
\begin{figure}[t]
\begin{small}
\(
\begin{array}{lll}
\ftpt(x) &\eqdef& \rd{x}  \\[.2ex]
\ftpt(\Emp) &\eqdef& \emptyeff  \\[.2ex]
\ftpt(\sing{E}) &\eqdef& \ftpt(E) \\[.2ex]
\ftpt(G/K) &\eqdef& \ftpt(G)   \\[.2ex]
\ftpt(G\Img f) &\eqdef& \rd{G\Img f}, \ftpt(G)   \\[.2ex]
\ftpt(F_1 \odot F_2) &\eqdef& \ftpt(F_1), \ftpt(F_2)
\quad\mbox{ for $\odot$ in $\{\union\;,\intersect\;,\setminus\;, +\;,-\}$ } \\[.2ex]
%\\
%\multicolumn{3}{l}{
% \quad\mbox{ for $\odot$ in $\{\union\;,\intersect\;,\setminus\;,\subseteq\}$ }} \\[.2ex]
\ftpt(G_0\subseteq G_1) &\eqdef& \ftpt(G_0),\ftpt(G_1)\\[.2ex]
\ftpt(x.f = F) &\eqdef& \rd{x},\rd{\sing{x}\Img f}, \ftpt(F)\\[.2ex]
\ftpt(E=E') &\eqdef& \ftpt(E), \ftpt(E')
\end{array}
\)
\end{small}
%\vspace*{-1ex}
\caption{Footprints of expressions and atomic formulas.}
\label{fig:foot}
\end{figure}

\begin{definition}[\textbf{framed reads; candidate dynamic boundary}]
\label{def:framedreadsDyn}
An effect $\eff$ has \dt{framed reads} provided that for every $\rd{G\Img f}$ in $\eff$,
its footprint $\ftpt(G)$ is in $\eff$.
A \dt{candidate dynamic boundary} is an effect that has framed reads, has no write effects,
and has no spec-only or local variables.
\end{definition}

In addition to the well-formedness assumption that the module import relation, $\imports$, is a preorder, we also assume that every declared boundary, $\bnd(M)$, is a candidate dynamic boundary.
The distinguished \dt{default module} name $\emptymod$ has empty boundary: $\bnd(\emptymod)=\emptyeff$.
For a finite set $X\subseteq\mathconst{ModName}$, we use the abbreviation
\ghostbox{$\unioneff{N\in X}{\bnd(N)}$}
for the catenation (union) of the boundaries.
Note that such combined boundaries are themselves candidate dynamic boundaries.
For \whyg{PQ}, the dynamic boundary, $\bnd(\code{PQ})$,
is $\rd{pool, pool\Img\allfields, pool\Img rep\Img\allfields}$.

The syntactic operation of \dt{effect subtraction},
\ghostbox{$\eff\setminus\effe$},
is used to formulate local equivalence specs;
in particular we subtract a dynamic boundary from a method's frame condition.
Subtraction is defined as follows.
First, put $\eff$ and $\effe$ into the following normal form:\footnote{After replacing
the data group $\allfields$ with the fields it stands for.}
No field occurs outermost in more than one field read or more than one field write.
This can be achieved by merging $\rd{G\Img f},\rd{H\Img f}$ into $\rd{(G\union H)\Img f}$
and likewise for write.
(Occurrences of field images within $G$ and $H$, not being outermost, are untouched.)
Assuming $\eff,\effe$ are in normal form,
define $\eff\setminus\effe$ to be $(\delta_0,\delta_1,\delta_2,\delta_3)$ where
\begin{equation}\label{eq:effsubtract}
\hspace*{-2em}\begin{array}{l}
   \delta_0 =  \{ \rd{x} \mid \rd{x}\in\eff \mbox{ and }\rd{x}\notin\effe \} \\
   \delta_1 =  \{ \rd{G\Img f} \mid \rd{G\Img f} \in \eff \mbox{ and }
                                     \effe \mbox{ has no $f$ read} \}
                \union \{ \rd{(G\setminus H)\Img f} \mid
                                \rd{G\Img f} \in \eff \mbox{ and } \rd{H\Img f} \in \effe \}
\end{array}
\end{equation}
and $\delta_2,\delta_3$ are defined the same way for writes.
For example, let $r$ and $s$ be region variables.
Then $(\rd{r},\rd{s},\rd{(r\union s)\Img nxt},\rd{r\Img val})\setminus
       (\rd{r},\rd{\sing{x}\Img nxt})$
is $\rd{s},\rd{((r\union s)\setminus\sing{x})\Img nxt},\rd{r\Img val}$.

% Taken from RLIII
\begin{figure}[t]
\(
\begin{array}{lcl}
\ind{\rd{G_1\Img f}}{\wri{G_2\Img g}} &=& \mbox{ if $f\equiv g$ or $f\equiv\allfields$ or $g\equiv\allfields$ then $\disj{G_1}{G_2}$
  else $\True$} \\
\ind{\rd{y}}{\wri{x}} &=&
 \mbox{ if $x\equiv y$ then $\False$ else $\True$}\\
%\ind{\rd{\lloc}}{\wri{\lloc}} &=& false \\
%\ind{\rd{G_1\Img f}}{\al{G_2}} &=& \mbox{ if $\lloc$ occurs in $G_1$ then $\disj{\lloc}{G_2}$ else $true$} ???\\
\ind{\delta}{\eff} &=& \True \quad\mbox{for all other pairs of atomic effects} \\
\ind{\delta}{\eff} &=& \True \quad\mbox{in case $\delta$ or $\eff$ is empty}\\
\ind{(\eff,\delta)}{ \effe} &=& (\ind{\eff}{\effe}) \land (\ind{\delta}{\effe}) \\
\ind{\delta}{(\eff,\effe)}&=& (\ind{\delta}{\eff})\land(\ind{\delta}{\effe})
\end{array}
\)
 \caption{The separator function \ghostbox{$\indSymbol$} is defined by recursion on effects.}
 \label{fig:sepdef}
\end{figure}

The separator function \ghostbox{$\indSymbol$}, 
mentioned in connection with the frame rule (\ref{eq:frameRule})
%\dn{mentioned in connection with the frame rule [ref eqn earlier]; drop forward ref?}
is defined by structural recursion on effects (Figure~\ref{fig:sepdef}).\footnote{This is unchanged from prior work (\RLI,\RLII).  The data group ``$\allfields$''
can be expanded to all the field names.  Computing
$\ind{ \rd{G\Img f} }{ \wri{H\Img\allfields}}$
yields the formula $\disj{G}{H}$.}
Given effects $\eff,\effe$ it generates a formula $\ind{\eff}{\effe}$ that implies
the read effects in $\eff$ are disjoint locations from the writes in $\effe$.
Please note that $\indSymbol$ is not syntax in the logic; it's a function in the metalanguage that is used to obtain formulas, dubbed \dt{separator formulas}, from effects.
For example, $\ind{\rd{r\Img nxt}}{\wri{r\Img val}}$ is the formula $true$ and
$\ind{\rd{r\Img nxt}}{\wri{s\Img nxt}}$ is the disjointness formula\footnote{Note that
   $\disj{r}{s}$ allows $r$ and/or $s$ to contain null; this is ok because there are
   no heap locations based on null.}
$\disj{r}{s}$.
Note that $\ind{\eff}{\effe}$ is identical to
$\ind{\reads(\eff)}{\writes(\effe)}$
where $\reads$\index{$\reads$} keeps just the read effects and
$\writes$\index{$\writes$} the writes.
The separator function can  be used to obtain disjointness conditions for
two read effects, say $\eff$ and $\effe$, by using the function we call
$\rTow$\index{$\rTow$} which discards write effects and changes reads to writes,
as in $\ind{\eff}{\rTow(\effe)}$.
Function $\wTor$\index{$\wTor$} does the opposite.
The upcoming Example~\ref{ex:PQframe} shows a use of $\indSymbol$ and the frame rule.

\subsection{Unary correctness judgments}\label{sec:ucorr}

On the way to formalizing correctness judgments, we first consider specs.
Spec-only variables are implicitly scoped over the spec but not explicitly declared.

\begin{definition}[\textbf{wf spec}]
\label{def:wfspec}
A spec $\flowty{P}{Q}{\eff}$ is well formed (\dt{wf}) in context $\Gamma$ if
\begin{itemize}
\item $\Gamma$ has no spec-only variables, and $\eff$ is wf in $\Gamma$.

\item $P$ and $Q$ are wf in $\Gamma,\hat{\Gamma}$,
for some $\hat{\Gamma}$ that declares only spec-only variables.\footnote{Here is what is needed to formalize method parameters. 
They can be referenced in the pre- and postcondition.
The frame must not allow write of a parameter, for the usual reason in Hoare logic that
the postcondition should refer to the initial value.
The frame should not allow read of a parameter: The call rule reflects that what is read is the argument expression in the call.
The linking rule allows the body of a method to read its parameters (see \RLIII).
} %footnote
\item In $P$, every occurrence of a spec-only variable $s$ is in an equation
$s=F$ that is a top-level conjunct of $P$,
where $F$ has no spec-only variables;
and every spec-only variable in $Q$ occurs in $P$.
%\dn{The only reason to make $\lloc$ not spec-only is to facilitate writing the last bullet.}
\end{itemize}
\end{definition}
The last item says spec-only variables are used as ``snapshot'' variables.\footnote{In
Def.~\ref{def:wfspec}, $\hat{\Gamma}$  is uniquely determined from the other conditions.
This is why we can leave types of spec-only variables implicit.
Their scope is also not explicit, but in the semantics they are scoped over the pre- and post-states.
We can refer to ``the spec-only variables of $P$'' as a succinct way to refer to those used in the spec.
} %footnote
In this article, the $'$ symbol is often used for identifiers on the right side of a pair,
so we avoid it for other decorative purposes, instead using $\hat{hats}$ and $\dot{dots}$.

A \dt{hypothesis context} $\Phi$ (context, for short)
maps some procedure names to specs
and is written as a comma-separated list of entries $m: \flowty{P}{Q}{\eff}$.

A \dt{correctness judgment}  has the form
\ghostbox{
\(
\Phi \proves^\Gamma_M  C: \flowty{P}{Q}{\eff}
\)
}
where $\Phi$ is a hypothesis context and $M$ is a module name.  The judgment is for code of the \dt{current module} $M$.
We distinguish two kinds of method calls in $C$:
\dt{environment calls} are those where a called method is bound by let within $C$;
the others, \dt{context calls}, are those where a called method is specified in $\Phi$.
Informally, the correctness judgment says executions of $C$ from $P$-states read and write only as allowed by $\eff$, and $Q$ holds in the final state if execution terminates.
A context call to $m$ in $\Phi$ may involve reading and writing encapsulated state for the module, $\mdl(m)$, of $m$, and these effects must be allowed by $\eff$.
Commands are given small step semantics, with bodies of let-bound methods kept in an environment.
The judgment also says that, aside from context calls, steps of $C$ must neither read nor write locations encapsulated by any module in $\Phi$ except its own module $M$.  These conditions must hold for any correct implementation of $\Phi$, so the judgment expresses ``modular correctness''~\cite{Leavens-Naumann15}.

Typically, in a judgment $\Phi\proves_M C:\ldots$
we will have $M\imports N$ for each $N$ in $\Phi$ (i.e., each $N$ for which some $m$ in $\Phi$ has $\mdl(m)=N$). However, we do not want to say $\Phi$ must contain every $N$ with $M\imports N$, because we use ``small axioms''~\cite{OHearnRY01} to specify atomic commands, which are stated in terms of the minimum relevant context. % (e.g., rule \rn{Alloc} in Fig.~\ref{fig:proofrulesU}).
Additional hypotheses can be added using ``context introduction'' rules
with side conditions that enforce encapsulation, as discussed in Sects.~\ref{sec:eg:encap} and~\ref{sec:encap}.
At the point in a proof where a client $C$ is linked with implementations of its context $\Phi$, the judgment for $C$ will include all methods of the modules in $\Phi$, and all transitive imports.

Because we are not formalizing a separate calculus of modules and module judgments,
some module-related scoping and typing conditions are associated with correctness judgments for commands.
The lack of an explicit binder for the spec-only variables of a spec also requires some care.

\begin{definition}[\textbf{wf correctness judgment}]
\label{def:wfjudge}
\index{wf}
A correctness judgment
\(\Phi \proves^\Gamma_M  C: \flowty{P}{Q}{\eff} \)
is wf if
\begin{itemize}
\item $\Phi$ is wf, i.e., each spec in $\Phi$ is wf in $\Gamma$ and they
have disjoint spec-only variables.\footnote{The latter condition loses no generality, since spec-only variables have scope over a single spec,
and distinctness helps streamline notation in some soundness proofs.}
\item No spec-only variables, nor $\lloc$, occur in $C$.
\item No methods occur in $\Gamma$, and $C$ is wf\footnote{\label{fn:label}Strictly speaking,
  we assume that for any subprogram of the form $\ifc{E}{C}{D}$, we have $C\nequiv D$.
  This loses no generality: it can be enforced using labels, or through the addition
  of   dummy assignments.  This is needed in order to express, in the definitions for
  encapsulation (Def.~\ref{def:valid}), that two executions follow exactly the same
  control path. } %footnote
in the typing context that extends $\Gamma$ to declare the methods in $\Phi$.
\item for all $N$ with $N\in\Phi$ or $N=M$, the candidate dynamic boundary $\bnd(N)$ is wf in $\Gamma$.
\item $\flowty{P}{Q}{\eff}$ is wf in $\Gamma$,
and its spec-only variables are distinct from those in $\Phi$.
\end{itemize}
\end{definition}

For example,
\[ m: \flowty{\True}{x>0}{\rw{x}}
\proves^{x:\INT,y:\INT}_{\emptymod}
x:=0;m() : \flowty{ x\leq 0 }{x > 0}{\rw{x}}
\]
is a wf judgment; in particular we have the typing  $x\scol\INT,y\scol\INT,m\scol\meth\proves x:=0;m()$.

% conjInv used to be defined here

\begin{example}
\label{ex:PQversions}
This example illustrates boundaries and specs.
To specify the priority queue ADT (Figure~\ref{fig:PQintr}),
we use an ownership idiom mentioned earlier (Section~\ref{sec:approach}). A ghost variable $pool:\Region$ is used to keep track of queue instances and each queue's $rep$ field contains objects it
notionally owns.
For a particular implementation, the private invariant includes conditions that imply all allocated queues have valid representations.

In one of our case studies we verify two implementations of the \whyg{PQ} module using pairing heaps~\cite{Weiss},
both using objects of class \whyg{Pnode}.
The private invariant of both versions includes the condition that for each $q\in pool$,
$q.rep\Img sibling \union q.rep\Img prev \union q.rep\Img child \subseteq q.rep$.
This says the $rep$ of $q$ is closed under these field images.
An interesting feature of this example is that clients manipulate \whyg{Pnode} references,
as ``handles'' returned by \whyg{insert}, but must respect encapsulation by not reading or writing the fields.

The leaves of the pairing heap are represented using $\NULL$ for the child in one implementation
and using references to a sentinel \whyg{Pnode} in the other.
One benefit of using sentinels is that certain checks for $\NULL$ can be avoided;
our motivation is simply to exemplify two different but similar data structures.

As per Figure~\ref{fig:PQintr} the dynamic boundary, $\bnd(\code{PQ})$,
is $\rd{pool, pool\Img\allfields, pool\Img rep\Img\allfields}$.
To reason that operations on one priority queue have no effect on others,
the public invariant expresses  disjointness
following the idiom mentioned in Section~\ref{sec:approach}:
\begin{equation}\label{eq:pubinvar}
\all{p, q\in pool}{p\neq q \imp \disj{p.rep}{q.rep} \land p\notin q.rep}
\end{equation}
While it is convenient for a module to declare a public invariant, there is no subtle semantics:
a public invariant simply abbreviates a predicate that is conjoined to
the pre- and post-conditions of the module's method specs.
That invariant is typically framed by the boundary, in which case clients easily maintain the invariant (and use it in their loop invariants).

As an example spec, consider the one for \whyg{PQ}'s \whyg{insert} (Figure~\ref{fig:PQintr}).
Abbreviating the parameters as $q,v,k$, a call $\code{insert}(q,v,k)$
adds to a given queue $q$, a \whyg{Pnode} with value $v$ and key $k$. Its spec is
\[
\begin{array}{lcl}
q\neq\NULL\land q\in pool%\land k\geq0
&\;\specSym\;&
\neg\code{isEmpty}(q)\land res\in q.rep\land res.val=v\land res.key=k
\\
&& [ \rw{\sing{q}\Img\allfields,q.rep\Img\allfields,\lloc} ]
\end{array}
\]
where $res$ is the return value, which references the inserted \whyg{Pnode}.  This pointer to an internal object
serves as handle for a client to increase the priority, for which
purpose it calls $\code{decreaseKey}(q,n,k)$ with spec
\[
\begin{array}{l}
 q\neq \NULL\land q\in pool \land \neg\code{isEmpty}(q)\land n\neq\NULL \land k\leq n.key \land n\in q.rep \\
\specSym \;  n.key = k \; [  \rw{\sing{q}\Img\allfields,q.rep\Img\allfields}  ]
\end{array}
\]
Clients see these pre- and postconditions conjoined with the public invariant.
\qed\end{example}

\begin{example}
\label{ex:PQframe}  
The separator function ($\indSymbol$) is used in the frame rule (\ref{eq:frameRule})
(formalized in Figure~\ref{fig:proofrulesU}).
To illustrate, consider a program with variables $p:\code{Pqueue}$ and $q:\code{Pqueue}$.
In accord with Example~\ref{ex:PQversions},
the proof rule for method call gives a judgment like this (eliding hypothesis context):
\[
n:= \code{insert}(q,v,k): \flowty{R}{S}{\rd{q,v,k};\wri{n};\rw{\sing{q}\Img\allfields,q.rep\Img\allfields,\lloc} }
\]
where $R,S$ are the
%instantiations for arguments $p,x,y$ of the
pre- and post-condition of $insert$'s spec.
Note that the call reads the arguments, and writes the result, in addition to the effects of the method spec (Figure~\ref{fig:PQintr}).

Consider the formula $p\neq q$.  It depends only on $p$ and $q$, which are not written by the displayed call to $insert$;
so the frame rule lets us infer
\[
n:= \code{insert}(q,v,k): \flowty{R\land p\neq q}{S\land p\neq q}{\rd{q,v,k};\wri{n}; \rw{\sing{q}\Img\allfields,q.rep\Img\allfields,\lloc} }
\]
To be precise, the rule requires a \emph{framing judgment} confirming that $\rd{p,q}$ covers
the footprint of formula $p\neq q$. (This is formalized in Section~\ref{sec:framing}
and used in rule \rn{Frame} which appears in Figure~\ref{fig:proofrulesU}.)
That is, $p\neq q$ is ``framed by $\rd{p,q}$''.
The rule also requires to compute a separator for the reads of the formula ($\rd{p,q}$) 
and the writes of the command,
namely
$\ind{\rd{p,q}}{\wri{\sing{q}\Img\allfields,q.rep\Img\allfields,\lloc}}$
(see Figure~\ref{fig:sepdef})
and show it follows from the precondition.
In this case the separator formula is simply $\True$; the only locations read are the variables
$p$ and $q$, and the only variable written is $\lloc$.

Now consider the formula $\code{isEmpty}(p)$.  The spec of  \whyg{isEmpty} has frame condition
$\rd{\sing{\mself}\Img size}$, so the formula $\code{isEmpty}(p)$
is framed by $\rd{p,p.size}$, which abbreviates $\rd{p},\rd{\sing{p}\Img size}$.
The \rn{Frame} rule lets us add the formula before and after the call $n:=\code{insert}(q,v,k)$:
\[
\flowty{R\land p\neq q\land \code{isEmpty}(p)}{S\land p\neq q\land \code{isEmpty}(p)}{\rd{q,v,k}, \rw{\sing{q}\Img\allfields,q.rep\Img\allfields,\lloc} }
\]
Here the separator is
$\ind{\rd{p},\rd{\sing{p}\Img size}}{\wri{\sing{q}\Img\allfields,q.rep\Img\allfields,\lloc}}$.
Unfolding the definition of $\indSymbol$,
and using that the data group, $\allfields$, covers every field including $size$,
we get the formula $\disj{\sing{p}}{\sing{q}} \land \disj{\sing{p}}{\sing{q}\Img rep}$.
Rule \rn{Frame} requires that the separator follows from the precondition.
The first conjunct, $\disj{\sing{p}}{\sing{q}}$, follows from precondition $p\neq q$.
The second conjunct follows using (\ref{eq:pubinvar}) which implies both $p\notin q.rep$ and $q\notin p.rep$.
\qed\end{example}

\paragraph{Summary.}

So far we introduced the syntax of commands, 
unary specs and unary correctness judgments.
The symbol $\equiv$ is sometimes used for equality of syntactic objects
like variable names, and especially in the case of commands and biprograms
which we identify up to the equivalences in Figure~\ref{fig:synident}.

There are also a number of meta-operators on syntax which are used pervasively and should not be confused with the syntax:
effect subtraction ($\eff\setminus\effe$),
separator ($\ind{\eff}{\effe}$),
footprint ($\ftpt(\effe)$),
converting write effects to reads ($\wTor$), etc.
There is no concrete syntax for modules; instead
there are meta-operators for the boundary $\bnd(M)$
of the module named $M$,
the import relation $\imports$ on module names, and
the module name $\mdl(m)$ associated with method $m$.

% ok but unnecessary:
%% The remaining syntactic elements of the logic are introduced later.
%% Weaving of biprograms ($\weave$) and full alignment ($\Syncbi{~}$)
%% are introduced in Sect.~\ref{sec:examples}.
%% The local equivalence  operator $\locEq$ is defined in Sect.~\ref{sec:locEq}.

Appendix Section~\ref{sec:app:guide} has a table of notations and a table of metavariables.

\subsection{Encapsulation in unary reasoning about modules and clients}\label{sec:eg:encap}

\begin{figure}[t]
\begin{lstlisting}
module UnionFind

   class Ufind {id: IntArray; part: partition; rep: rgn;}

   public pool : rgn
   boundary { pool, pool`any, pool`rep`any }

   meth Ufind(self:Ufind, k:int) : unit
   meth find(self:Ufind, k:int) : int
   meth union(self:Ufind, x:int, y:int) : unit
end.
\end{lstlisting}
  \caption{Excerpts of union-find interface, eliding private methods and specs.}
\label{fig:UnionFind}
\end{figure}

In this subsection we consider how the requirements (E1)--(E4) for encapsulation in Section~\ref{sec:modrel}, are met in the unary logic.
Figure~\ref{fig:UnionFind} shows the interface of a module that provides a class whose instances are union-find structures.
The first requirement for encapsulation, (E1), is to delimit some locations internal to the
module.  That is the purpose of the dynamic boundary, which in the logic would be written
$\rd{pool},\rd{pool\Img \allfields},\rd{pool\Img rep\Img\allfields}$ (in accord with Def.~\ref{def:framedreadsDyn})
and abbreviated as $\rd{pool, pool\Img\allfields, pool\Img rep\Img\allfields}$.
An equivalent formulation of the boundary is
$\rd{pool, (pool\union pool\Img rep)\Img \allfields}$.

In this example we follow the idiom, and even the naming convention, sketched in Sec.~\ref{sec:approach} for a module providing stacks.
Aside from $rep$, the boundary does not mention specific fields but rather uses
the data group $\allfields$ for the sake of abstraction.

Because $\rd{pool}$ is in the boundary of \whyg{UnionFind}, client programs may neither read nor write this variable.
It serves in specs to designate references to, at least, the \whyg{Ufind} instances managed by the module;
so the constructor method \whyg{Ufind}, which should be
invoked on newly allocated \whyg{Ufind} objects, adds the new object to $pool$.
The boundary includes $\rd{pool\Img\allfields}$, which says fields of these objects
may neither be read nor written by client programs.
In specs and reasoning about clients, the $rep$ field of a \whyg{Ufind} is important:
it is used to delimit the locations modified by method calls on that instance,
and a public invariant of the module says distinct \whyg{Ufind} instances have
disjoint $rep$.
This enables reasoning that performing an operation on one
\whyg{Ufind} does not affect the state of another \whyg{Ufind} ---which is locality, not encapsulation.
Fields of objects in $rep$ are encapsulated by the module, as expressed by $\rd{pool\Img rep\Img\allfields}$.
Here $pool\Img rep$ is the union of the $rep$ fields of all allocated \whyg{Ufind}s.

We consider an implementation based on the quick-find data structure~\cite{SedgewickWayne}.
%The code uses arrays, which are not formalized in the theory but are available in the tool.  % with some machinations
Math type \whyg{partition} represents a partition on a set of numbers $0\dots n-1$.
It is used in ghost code and specs, in particular the private invariant
which says each queue $p$ satisfies a predicate defined on its internal representation, which is an array
referenced by field $id$.

\begin{lstlisting}
  predicate ufInv (p: Ufind) =
    p.id <> null /\
    let n = p.id.len in
       size(p.part) = n /\  p.rep = { p.id } /\
       (forall x:int. 0 <= x < n -> 0 <= p.id[x] < n /\ p.id[p.id[x]] = p.id[x]) /\
       (forall x:int, y:int. 0 <= x < n /\ 0 <= y < n -> ( y iin pfind(x,p.part) <-> p.id[x] = p.id[y]) )

  private invariant *?$I_{qf}$?* = forall p: Ufind iin pool. ufInv(p)
\end{lstlisting}

The union-find implementation uses a representative element for each block of the partition,
with $id[x]$ being the representative of $x$,
for each $x$ in $0\ldots n-1$.
If $x$ is a representative then $id[x]=x$.
The private invariant says that for any $x$, $id[x]$ is a representative: $p.id[p.id[x]] = p.id[x]$.
The last conjunct says $x$ and $y$ have the same representative in $p.id$ just if they
are in the same block of the abstract partition.
The ghost field $rep$ has nothing to do with representatives;
as in our usual idiom it holds references to the internal representation objects, in this case just the $id$.

Requirement (E2) for encapsulation is that a private invariant depends only on locations within
the boundary.  This is formalized in the logic by a \emph{framing judgment} which in our example is written
$\models \fra{ (\rd{pool}, \rd{(pool\union pool\Img rep)\Img \allfields}) }{ I_{qf} }$.
As formalized later, its meaning is that if $I_{qf}$ holds in some state,
then it holds in any other state that agrees on the values in the locations designated
by the read effect.
Looking at its definition, $I_{qf}$ depends on only one variable, $pool$.
The heap locations on which it depends are in expressions $p.id$ and index expressions $p.id[x]$.  As we have $p.id \in p.rep$, by the invariant,
and the slots of the array are effectively fields of $id$,
these heap locations are indeed covered by $\rd{(pool\union pool\Img rep)\Img \allfields}$.
The meaning of the framing judgment can be encoded as a universally quantified formula;
this and other framing judgments in our case studies are easily validated by SMT solvers.

Here we consider the quick-find implementation, which for the \whyg{find} method is:
\begin{lstlisting}
  meth find (self: Ufind, k: int) : int
  = result := self.id[k]
\end{lstlisting}
A key postcondition of the spec of \whyg{find} is that $result \in \mathit{pfind}(k,\mself.part)$,
where $\mathit{pfind}$ is the function that returns the block of the abstract partition that contains $k$.
The postcondition holds in virtue of conditions in the private invariant, including that $id[k]$
is a representative, for any $k$, and the connection between $\mself.part$ and $\mself.id$.

\newcommand{\deltaUF}{\delta_{\mbox{\relsize{-1}uf}}}

\paragraph{Encapsulation of a client.}

As a case study we have verified Kruskal's minimum spanning tree algorithm as client, but for present purposes
we consider a very simple client.
\begin{lstlisting}
  uf:=new Ufind(100); x:=new Thing; x.f:=y; z := find(uf,1)
\end{lstlisting}
To verify the client code, its hypothesis context needs to include the module specs, in particular for
\whyg{find}.  So \whyg{UnionFind} is in scope and its boundary must be respected
by the client.
The logic enforces encapsulation of clients, i.e., requirement (E3),
using separation checks similar to those for frame based reasoning as in Example~\ref{ex:PQframe}.

To explain the checks, let us write $\deltaUF$ for the boundary of \whyg{UnionFind}.
The command $x:=\new{\code{Thing}}$ has frame $\wri{x},\rw{\lloc}$.
Respect of $\deltaUF$ by this command is formulated in terms of
the separator function, in this case $\ind{\deltaUF}{\wri{x,\lloc}}$.
Unfolding the definition (Figure~\ref{fig:sepdef}) yields the formula
$\True \land\True$.
The only variable designated by $\deltaUF$ is $pool$, and this is distinct from $x$ and from $\lloc$.
The proof obligation here also rules out client code that assigns or reads $pool$.
In general it is untenable to include $\rd{\lloc}$ in a boundary,
or even an image expression mentioning $\lloc$, because clients typically do allocation.

The command $x.f:=y$ has frame condition $\rd{x},\rd{y},\wri{\sing{x}\Img f}$.
For the write to be outside the boundary, the obligation can be written
$\ind{\deltaUF}{\wri{\sing{x}\Img f}}$.  Unfolding by definition of the separator function,
and expanding the abbreviation $\allfields$ to be all field names in scope,
we get a conjunction of $true$s (because the read and written variables are distinct)
and two nontrivial conjuncts:
$\disj{pool}{\sing{x}}$ and $\disj{pool\Img rep}{\sing{x}}$.  That is, the assigned
object must be in neither $pool$ nor any $rep$ fields of objects in $pool$.
One way this obligation can be proved is via freshness: neither $pool$ nor $rep$ have been updated since $x$ was assigned a fresh object.
A related idiom used in some method specs is a postcondition that says all fresh objects are in $\mself.rep$, which a client can use to reason that its own regions remain disjoint.
In a postcondition, the fresh references are denoted by $\lloc\setminus\old(\lloc)$.
In the formal logic state predicates only refer to a single state, so a postcondition must
be expressed in the same way that tools desugar ``old'' expressions.  That is, a fresh spec-only variable,
say $r$, is used to snapshot the initial value: the precondition includes $r=\lloc$ and the
idiomatic postcondition is now $\lloc\setminus r \subseteq \mself.rep$.

We are not finished with $x.f:=y$.  In addition to its writes, its reads must be outside the boundary, 
specifically, $x$ and $y$ must be outside $\deltaUF$.
This can be written
$\ind{\deltaUF}{\wri{x},\wri{y}}$.  Why $\wri{}$? 
Just so we can use the separator function $\ind{}{}$ unchanged from
prior work, though it is defined to separate read effects from writes. 
(The proof rule for field update uses another metafunction, $\rTow$, to convert the reads to writes.)

%% Because the proof rule
%% that enforces encapsulation uses another metafunction, $\rTow$, to convert the reads to writes.
%% This allows us to use the separator function $\ind{}{}$ unchanged from
%% prior work (\RLI, \RLII). 

%\dn{Effects of a call need not be outside the callee's module's boundary;
%$meth find (self: Ufind, k: int) : int$
%has frame condition $\rd{\mself, k, \sing{\mself}\Img any, \sing{\mself}\Img rep\Img any}, \wri{result}$}

As an example of how encapsulation checks can fail, consider a bad client of the \whyg{PQ} interface
(Figure~\ref{fig:PQintr}) that calls \whyg{insert} and assigns the returned \whyg{Pnode} to variable $nd$, and then writes the $key$ field of $nd$ ---potentially invalidating a private invariant.
The boundary of \whyg{PQ} is similar to the one for \whyg{UnionFind},
so the separator formula is
$\disj{pool}{\sing{nd}} \land \disj{pool\Img rep}{\sing{nd}}$.
This is not valid, since the value of $nd$ is in $pool\Img rep$.

So far we saw how the frame conditions of atomic commands give rise to proof obligations
that ensure the client reads and writes are to locations disjoint from the locations designated by the boundary.  Please note that the interpretation of the boundary is at the point in execution where
the atomic command has its effects.
This does not make a difference for variables, in the sense that a separator
$\ind{\rd{x}}{\wri{y}}$ is just true or false depending on whether the variable names are distinct.
It does make a difference for heap locations, designated
by expressions like $pool\Img \allfields$ and $\sing{x}\Img f$;  in this case the obligation
$\disj{pool}{\sing{x}}$ discussed above
must hold in the pre-state of the assignment command $x.f:=y$.

Loops and conditionals also incur an encapsulation obligation that their test expressions
read outside the boundary.  In our desugared syntax (Figure~\ref{fig:bnf}) these expressions are heap independent.
In the example the check is simply that variable $pool$ does not occur in a test expression,
since the other locations in the boundary are heap locations.
Here is an example where a test crosses the boundary of \whyg{PQ}.
\begin{lstlisting}
q := new Pqueue(); nd := insert(q,0,0); if nd.prev <> null then q := null fi; nd := insert(q,1,1)
\end{lstlisting}
This client works fine with the first implementation of \whyg{PQ} since  $nd.prev$ will be null.
But for the implementation with sentinels, the second call to \whyg{insert} will fault due to null
dereference.
The client is not representation independent and the read of $nd.prev$ will fail the encapsulation check.

In our prototype, \WhyRel, encapsulation checks like this are straightforward.
At points where the encapsulation check is state dependent, like $x.f:=y$, \WhyRel\ generates an assert statement that encodes the disjointness obligation (Section~\ref{sec:cases}).
In the logic, encapsulation checks are disentangled from other reasoning considerations
by the context introduction proof rules.  The modules whose boundary must be respected
are those of the methods in the hypothesis context, given using the $\mdl$ function defined in Sec.~\ref{sec:modules}.
The technical details are not conceptually important, and are explained in Section~\ref{sec:encap}.

In summary, encapsulation requirement (E3) is achieved by checking separation
from the relevant boundaries, for each part of the client command.
Separation is checked the same way as it is for the ordinary \rn{Frame} rule,
using formulas generated from the effects using the separator function ($\ind{}{}$).
For effects on variables it is true or false depending on whether the requisite variables are distinct,
but for effects on heap locations (load and store commmands, method calls) the separation checks are region disjointness formulas that must hold at the relevant points in control flow.

\paragraph{Modular linking.}

Suppose we verify the client, using the public specs, and discharge the proof obligations,
just discussed, for encapsulation.  We verify the implementation of \whyg{find}, \whyg{union}, etc
using the private invariant $I_{qf}$, i.e., assuming it as precondition and establishing it as post,
in accord with the modular linking rule sketched as (\ref{eq:mismatch}) in Section~\ref{sec:modrel}.
Having verified the client and the implementations of module methods, we would like to conclude that
the linked program is correct, i.e., satisfies the client spec
as per rule (\ref{eq:mismatch}).
The private invariant is hidden from the client, in the sense that the method bodies are verified
for specs that include it, but it is omitted from the hypotheses used to verify the client.
There is one more requirement for this to be sound, namely (E4): the client precondition
implies the private invariant of the module.
An appropriate such precondition is $pool=\emptyset$, the default value for regions,
which implies $I_{qf}$ owing to its quantification over $pool$.

The intuition that justifies (\ref{eq:mismatch}) is that, given the client's respect for the boundary, any judgment $D:\flowty{P}{Q}{\eff}$ about a client subprogram $D$
yields $D:\flowty{P\land I}{Q\land I}{\eff}$ by an application of the frame rule
(because the encapsulation obligation ensured the footprint of the private invariant $I$ is disjoint from the effects
in $\eff$).
In particular, at a point where the client has established public precondition $R$ of a method
that has been verified using precondition $R\land I$, we do in fact have $R\land I$.
For example, having proved the judgment
$\code{find}:\flowtyf{R}{S} \proves C:\flowtyf{P}{Q}$
(omitting frame condition)
together with the encapsulation obligations for client $C$,
we have
\[ \code{find}:\flowtyf{R\land I_{qf}}{S\land I_{qf}} \proves C:\flowtyf{P\land I_{qf}}{Q\land I_{qf}} \]
This is formalized as the \dt{second order frame rule},
\rn{SOF} in Figure~\ref{fig:proofrulesU}.
The modular linking rule (\ref{eq:mismatch}) is a consequence of \rn{SOF} together
with the obvious linking rule that requires the method bodies to satisfy
exactly the specs assumed by the client.
Please note that all formulas involved in the specs are first-order; the \rn{SOF} rule is called second order only in the sense that the framed formula is conjoined to specs in the hypothesis context as well as to the
consequent of the judgment.

\paragraph{On dynamic boundaries.}

In this article we repeatedly use the idiom with $pool$ and $rep$,
but this is merely one convenient way to write specs that support module-based encapsulation and per-instance local reasoning.  Ghost variables and fields can just as well be used to express hierarchical ownership
or cooperating clusters of objects as in design patterns like subject-observer.  Such examples can
be found in \RLI--III.

A key point is that the dynamic boundary is part of a module interface, and should be 
expressed in such a way that different module implementations can have different internal data structures.
Thus the same dynamic boundary may denote different locations for different implementations.
This can be achieved using ghost state, data groups, and pure methods.
In this article we only formalize a single data group, $\allfields$, and we omit pure methods (see Sect.~\ref{sec:curlim}).

To prove the disjointnesses needed for client code to be outside a boundary, one can rely on invariants 
that constrain the relevant ghost state.  
For this purpose it is convenient for a module interface to include public invariants 
such as 
(\ref{eq:pubinvar}) in Example~\ref{ex:PQversions}.
%the one we mentioned regarding stacks, in Sect.~\ref{sec:approach}.

%In client reasoning one can maintain such invariants based 
%- meth specs
%- pure meth defns  
%- public invars 

%\section{Examples of biprogram weaving and encapsulation}\label{sec:examples}
\section{Biprograms: syntax and relational reasoning}\label{sec:biprograms}

This section formalizes  biprograms (Section~\ref{sec:biprogram:syn}), 
relation formulas (Section~\ref{sec:relform:syn}),
relational specs and correctness judgments (Section~\ref{sec:relspec}).
Section~\ref{sec:eg:relverif} uses an example to illustrate how regions are used in relation formulas 
and how biprograms express convenient alignments.
Section~\ref{sec:weave} defines the weaving relation and explains its use to account for helpful alignments.
Section~\ref{sec:eg:encapR} sketches example of relational modular linking.

In this section, as in Section~\ref{sec:Usyntax}, we use the syntax of our prototype
for program code, together with the math notations of the formal logic.
We use syntax sugar and also some features that are not formalized in the logic,
namely parameters and return values (see Section~\ref{sec:curlim}), for the sake of
readable examples.
More about the prototype can be found in Section~\ref{sec:cases}.

\subsection{Biprograms}\label{sec:biprogram:syn}

Figure~\ref{fig:bnf} gives the grammar of biprograms.
A biprogram $CC$ represents a pair of commands, which are given by syntactic projections defined in Figure~\ref{fig:synProj}.
For example, the left projection 
$\Left{\splitbi{\skipc}{x:=0};\splitbi{y:=0}{z:=1}}$ is $y:=0$,
%projections (Figure~\ref{fig:synProj}).
%For example, the left projection 
%$\Left{\splitbi{\skipc}{x:=0};\splitbi{y:=0}{z:=1}}$ is $y:=0$,
taking into account that we identify
$\skipc;y:=0$ with $y:=0$ (see Figure~\ref{fig:synident}).
The symbol $|$ is used throughout the article, in program and spec syntax and also
as alternate notation for pairing in the metalanguage, when the pair represents a pair of states or similar.\footnote{A small version of the symbol is used, interchangeably, for clarity in some contexts such as grammar rules.}

\begin{wrapfigure}{R}{0.48\textwidth} % r allows float, R means exactly here
\begin{footnotesize}
\(
\begin{array}[t]{lll@{\hspace*{3em}}l}
\Left{\splitbi{C}{C'}} &  \eqdef  & C \\[.5ex]
\Left{\syncbi{A}} &  \eqdef  & A \\[.5ex]
\Left{\ifcbi{E\smallSplitSym E'}{BB}{CC}} &  \eqdef  & \ifc{E}{\Left{BB}}{\Left{CC}} \\[.5ex]
\Left{\whilecbiA{E\smallSplitSym E'}{\P\smallSplitSym \P'}{CC}} &  \eqdef  & \whilec{E}{\Left{CC}} \\[.5ex]
\Left{\seqc{BB}{CC}} &  \eqdef  & \seqc{\Left{BB}}{\, \Left{CC}} \\[.5ex]
\Left{\varblockbi{x\scol T\smallSplitSym x'\scol T'}{CC}} &  \eqdef  & \varblock{x\scol T}{\Left{CC}} \\[.5ex]
\Left{\letcombi{m}{\splitbi{C}{C'}}{CC}} &  \eqdef  & \letcom{m}{C}{\Left{CC}}
\\[1ex]
\multicolumn{3}{l}{
\mbox{Symmetrically, }
\Right{\splitbi{C}{C'}}   \eqdef   C', \;
\Right{\syncbi{A}}  \eqdef  A, \;
\mbox{etc.}
}%multicolumn
\end{array}
\)
\end{footnotesize}
%\hrule
%\vspace{-1ex}
\caption{Syntactic projections $\protect\Left{\rule{0pt}{1ex}}$ and $\protect\Right{\rule{0pt}{1ex}}$ of biprograms.}
% ALERT LaTeX error if caption includes $\Left{CC},\Right{CC}$
\label{fig:synProj}
\end{wrapfigure}

Biprograms are given small-step semantics.
The \dt{bi-com} form $\splitbi{C}{D}$ represents executions of commands $C$ and $D$ which are meant to be aligned on their initial state and, if they terminate, final state.
Their execution steps are interleaved
(i.e., dovetailed, in the terminology of automata theory),
to ensure that the traces of $\splitbi{C}{D}$ cover
all traces of $C$ and $D$ by making progress on both sides even if one diverges.
%(See Example~\ref{ex:split} in the section on biprogram semantics (Sect.~\ref{sec:bitrans}).)
The parentheses of bi-coms are obligatory and the operator binds less tightly than others:
$\splitbi{A;B}{C;D}$ is the same as $\splitbi{(A;B)}{(C;D)}$.
In Section~\ref{sec:weave}
we consider how the other biprogram forms are introduced
for a verification problem specified using a bi-com.
For now we briefly explain the other forms.

The \dt{sync} form $\syncbi{A}$ represents two executions of the atomic command $A$, aligned as a single step.
This is mainly of interest for allocations and method calls.
For a call, $\syncbi{m()}$ indicates that a relational spec should be used to reason about the two calls.
For an allocation, the form $\syncbi{x:=\new{K}}$
has a proof rule in which the two new references are considered in agreement, i.e., ``added to the refperm''.
In the grammar (Figure~\ref{fig:bnf}), the bi-var form allows different names and types but one also wants to allow
multiple variables on each side; this is implemented in our prototype.
The bi-if form, $\ifcbi{E\smallSplitSym E'}{CC}{DD}$, asserts that the
two initial states agree on the value of the test expressions $E$ and $E'$.
%To relate two conditional commands for which such agreement is not expected,
%one can use the bi-com form $\Splitbi{\ifc{E}{C}{D}}{\ifc{E'}{C'}{D'}}$.  DN done in sec 4
The bi-while form $\whilecbiA{E\smallSplitSym E'}{\P\smallSplitSym \P'}{CC}$
incorporates relation formulas $\P$ and $\P'$ which serve as
\dt{alignment guards}.
%The idea is to indicate how to align iterations of the loop,
These serve as directives to indicate how to align iterations of the loop,
catering for situations like the $sumpub$ program in (\ref{eq:sumpub}).
This is explained in more detail in Section~\ref{sec:weave};
see the aligned $sumpub$ (\ref{eq:bi-sumpub}).

Typing of biprograms can be defined in terms of syntactic projection, roughly as $\Gamma|\Gamma'\proves CC$ iff $\Gamma\proves\Left{CC}$ and $\Gamma'\proves\Right{CC}$. But the alignment guard formulas in a bi-while  should also be typechecked in $\Gamma|\Gamma'$, and are required to be free of agreement formulas, i.e., those of the form $\Agr G\Img f$ and $F\eqbi F'$;
this ensures that the formula is refperm-independent as explained later.
Although the two sides of a biprogram may have different typing contexts,
for simplicity a single class table is assumed.
It is straightforward to generalize this to allow different field declarations for a given class
(and it is implemented in our prototype).

\subsection{Relation formulas}
\label{sec:relform:syn}

%\begin{wrapfigure}{R}{0.5\textwidth}  % r allows float, R means exactly here
\begin{figure}[t]
\begin{small}
\(
\begin{array}{lll}
F\!F  ::=\!\!\!\! & \leftex{F} \gmid \rightex{F} & \mbox{Value in left (resp.\ right) state} \\[1ex]
\P \, ::=\!\!\!\! & R(\ol{F\!F}) & \mbox{Primitive $R$ in signature} \\
& \gmid F\eqbi F & \mbox{Equal expressions, mod refperm} \\
& \gmid \Agr\, LE & \mbox{Agreement mod refperm} \\
& \gmid \later\P & \mbox{Possibly (in some extended refperm)} \\
& \gmid \leftF{P} \gmid \rightF{P} & \mbox{In the left (resp.\ right) state} \\
& \multicolumn{2}{l}{
\gmid \P\land\P \gmid \P\imp\P % \gmid \P\lorbi\P
\gmid \all{x\scol T\smallSplitSym x\scol T}{\P}
}
 \end{array}
\)

\begin{tabular}{l}
Syntax sugar:
\(\begin{array}[t]{l}
  \Both{P} \eqdef \leftF{P} \land \rightF{P} \\
  \always\P \eqdef \neg\later\neg\P \\
  \False  \eqdef \Both{\False} 
     \qquad \True \eqdef \Both{\True} \\
  \Agr x.f \eqdef \Agr \sing{x}\Img f \\
% \Agr E \eqdef E\eqbi E \\
%  \Agr (\rd{G\Img f}) \eqdef \Agr G\Img f , \quad
  \Agr (\rd{LE}) \eqdef \Agr\, LE 
      \qquad  \Agr(\wri{\ldots}) \eqdef \True \qquad
  \Agr(\eff,\effe) \eqdef \Agr(\eff)\land\Agr(\effe) 
\end{array}
\)
\\

Precedence: 
(tightest) $\Agr$, $\later$, $\eqbi$, $\land$, $\imp$ (loosest).
\end{tabular}
\end{small}
\caption{Relation formulas.
See Figure~\ref{fig:preds} for unary formulas $P$ and (\ref{eq:effects}) for left-expressions $LE$.}
\label{fig:relFormulas}
\end{figure}

Relation formulas are interpreted over a pair of states,
meant to be at aligned points in two executions.
What is important is to express not only conditions relating integers and other mathematical values, but also conditions relating structures between the two heaps.
There are many ways to formalize such formulas; it is only in the treatment of heap relations that the design choices made here have significant impact on the later development.

The relation formulas are defined in Figure~\ref{fig:relFormulas}.
Quantifiers range over allocated references; the relational form binds a variable on each side. The form $\leftF{P}$ (resp.\ $\rightF{P}$) says unary predicate $P$ holds in the left state (resp.\ right).
Left and right embedded expressions are written $\leftex{F}$ and $\rightex{F}$
and have nothing to do with left-expressions $LE$.
They may be used as arguments to atomic predicates in the ambient mathematical theories:
$\leftex{F}$ (resp.\ $\rightex{F}$) evaluates $F$ in the left (resp.\ right) state.\footnote{Written
    $\langle 1\rangle F$ and $\langle 2\rangle F$ in works following Benton~\cite{Benton:popl04}.
    Our notations $\leftex{F}$ and $\leftF{P}$ are meant to point leftward.
} %footnote

The forms $\Agr\, LE$ and $F\eqbi F'$ are called \dt{agreement formulas}.
For $E$ and $E'$ of some reference type $K$, the form  $E\eqbi E'$
(pronounced ``$E$ bi-equals $E'$'')  says the value of $E$ in the left is the same as $E'$ on the right, modulo refperm in the case of reference values.  Similarly with $G\eqbi G'$ for regions.
The form $\Agr G\Img f$
says for each reference $o\in G$, with corresponding value $o'$ in the other state,
the value of $o.f$ is the same as the value of $o'.f$,
modulo refperm if the value is of reference type.
For example, $\Agr r\Img rep\Img val$ means the $val$ fields agree, for all objects in the $rep$ field of all objects in $r$.

The form $\Agr x$ is equivalent to $x\eqbi x$.
But the form $\Agr G\Img f$ is not equivalent to $G\Img f\eqbi G\Img f$.
The former means pointwise field agreement (modulo refperm)
and the latter means equal values (modulo refperm), the two values being reference sets.

\begin{wrapfigure}{R}{0.30\textwidth}  % r allows float, R means exactly here
\begin{small}
%\( \begin{array}{l@{\hspace{1.5ex}}l@{\hspace{1.5ex}}l@{\hspace{6em}}l@{\hspace{1.5ex}}l@{\hspace{1.5ex}}l}
\( \begin{array}{l@{\hspace{1.5ex}}l@{\hspace{1.5ex}}l}
\Left{\leftF{P}} & \eqdef & P \\
\Left{\rightF{P}} & \eqdef & \True \\
\Left{\later\P} & \eqdef & \Left{\P} \\
\Left{F\eqbi F'} & \eqdef & (F=F) \\
\Left{\Agr LE}  & \eqdef & (LE = LE) \\
\Left{\all{x\scol T \smallSplitSym x' \scol T'}{\P}}  & \eqdef & \all{x:T}{\Left{\P}} \\
\Left{R(\ol{F\!F})} & \eqdef & \True \\
\\[1ex]
\Left{\rflowty{\P}{\Q}{\eff|\eff'}} &  \eqdef &
\flowty{\Left{\P}}{\Left{\Q}}{\eff} \\
\end{array}
\)
%\hfill
\end{small}
\caption{Syntactic projection $\protect\Left{\rule{0pt}{1ex}}$
of relation formulas and specs;
right projection $\protect\Right{\rule{0pt}{1ex}}$ is symmetric.
}
% ALERT LaTeX error if caption has $\Left{\P},\Right{\P}$.
\label{fig:synProjFmla}
\end{wrapfigure}

The modal form $\later\P$, read \dt{possibly} $\P$ (for lack of a better word), says $\P$ holds in a refperm possibly extended from the current one.  More on these points later.

Relation formulas and relational correctness judgments are typed in a context of the form
$\Gamma|\Gamma'$ comprises contexts $\Gamma$ and $\Gamma'$ for the left and right sides.\footnote{This
   enables reasoning about two versions of a program acting on the same variables, by contrast with other works where related
   programs are assumed to have been renamed to have no identifiers in common. Logics should account for renaming.}
Leaving aside left/right embedded expressions,
typing can be reduced to typing of unary formulas:
$\Gamma|\Gamma'\proves \P \;$ iff $\; \Gamma\proves\Left{\P}\mbox{ and }\Gamma'\proves\Right{\P}$.
This refers to syntactic projections defined in Figure~\ref{fig:synProjFmla}.
This does not work for left/right embedded expressions; we gloss over those for clarity,
in the following sections as well, but handle them in our prototype.

In accord with the definition of projections,
we have the formula typing
$\Gamma|\Gamma'\proves \Agr x$ just if $x\in\dom(\Gamma)\intersect\dom(\Gamma')$.
We have $\Gamma|\Gamma'\proves \Agr G\Img f$
just if $\Gamma\proves G:\Region$ and $\Gamma'\proves G:\Region$, with $f$ of any type.
Similarly,
$\Gamma|\Gamma'\proves F\eqbi F'$ provided $\Gamma\proves F:T$ and $\Gamma'\proves F':T$.
Also $\Gamma|\Gamma'\proves \leftF{P}$ if $\Gamma\proves P$
and $\Gamma|\Gamma'\proves \rightF{P}$ if $\Gamma'\proves P$.

% tricky junk
%% \dn{Something's amiss, there's no syntax sugar in Fig.~\ref{fig:relFormulas} for $\Agr F$;
%% we used to have sugar for $\Agr E$ but that's pointless now.
%% I guess at this point we want to say $F\eqbi F$ can be abbreviated as $\Agr F$ unless ambiguity arises,
%% and then come back to this when we give the examples with drawings.
%% }
%% Apropos the syntax sugar in Fig.~\ref{fig:relFormulas},
%% one should not write $\Agr G$ as a general abbreviation for $G\eqbi G$ ---it would be ambiguous!
%% In case $G$ has the form $H\Img f$, the meaning of $\Agr H\Img f$ is different from the meaning of $H\Img f \eqbi H\Img f$.
%% \dn{revisit once picture in place}
%% However, if $G$ is not an image expression then $\Agr G$ can be used unambiguously to abbreviate $G\eqbi G$.
%% Also, there is no conflict between $\Agr x$ and $x\eqbi x$
%% because the two have the same semantics.

\subsection{Relational specifications and correctness judgment}
\label{sec:relspec}

A \dt{relational spec} \ghostbox{$\rflowty{\P}{\Q}{\eff|\eff'}$} has relational pre- and post-conditions and a pair of frame conditions.
We write \ghostbox{$\rflowty{\P}{\Q}{\eff}$} to abbreviate the frame condition $[\eff|\eff]$.
A spec $\rflowty{\P}{\Q}{\eff|\eff'}$ is wf in $\Gamma|\Gamma'$
provided $\Left{\rflowty{\P}{\Q}{\eff|\eff'}} $ is wf in $\Gamma$ (resp.\
$\Right{\rflowty{\P}{\Q}{\eff|\eff'}}$ in $\Gamma'$), as per
Def.~\ref{def:wfspec}.  See Figure~\ref{fig:synProjFmla} for syntactic
projections.
The precondition $\P$ of a wf relational spec has spec-only variables only as snapshot equations in top level conjuncts of $\P$ (inside the left and right embedding operators $\leftF{-}$, $\rightF{-}$).
Any spec-only variables in postcondition $\Q$ must occur in $\P$.

Recall from Section~\ref{sec:modrel} that one important relational property is local equivalence.
Later we define a general construction, $\locEq$, that applies to a unary
spec $\flowty{P}{Q}{\eff}$ and yields a relational spec (Example~\ref{ex:PQagree} and Section~\ref{sec:locEq}).
The general form takes into account that encapsulated locations are not expected to be in agreement; that is formalized by means of effect subtraction.

For local equivalence and other purposes, we often want postconditions that assert agreements on fresh locations.
These agreements are modulo refperm, so a relational correctness judgment should
say there is some refperm for which the final states are related.
This can be expressed using the $\later$ modality.
Many specs of interest have the form  $\rflowty{\P}{\later\Q}{\effe|\effe'}$ where $\P,\Q$ are $\later$-free.
Such specs are said to be in \dt{standard form}.
We gloss over this in some examples.
In our prototype, the encoding maintains a ``current refperm'' in ghost state to interpret agreement formulas, and does not use the $\later$ modality explicitly in specs.
The dual, $\always$, is used in a couple of proof rules.

A \dt{relational hypothesis context} for $\Gamma|\Gamma'$ is
a triple $\Phi = (\Phi_0,\Phi_1,\Phi_2)$ comprising
unary hypothesis contexts $\Phi_0$ for $\Gamma$ and $\Phi_1$ for $\Gamma'$,
together with a mapping $\Phi_2$ of method names to relational specs that are wf.

\begin{definition}[\textbf{wf relational hypothesis context}]
\label{def:wfhypRel}
A relational hypothesis context for $\Gamma|\Gamma'$ is
wf in $\Gamma|\Gamma'$
provided that $\Phi_0,\Phi_1,\Phi_2$ specify the same methods,\footnote{One can allow
   different   methods in context, provided that left (resp.\ right, resp.\ sync'd) context
   calls have left (resp.\ right, resp.\ relational) spec's,
   and this is implemented in our prototype.}
$\Phi_0$ and $\Left{\Phi_2}$ are wf in $\Gamma$,
$\Phi_1$ and $\Right{\Phi_2}$ are wf in $\Gamma'$,
the specs in $\Phi_2$ are wf in $\Gamma|\Gamma'$,
and the distinct methods have distinct spec-only variables in $\Phi_2$ (just as in $\Phi_0$ and $\Phi_1$).
Moreover, for every $m$, the formula
\[ %begin{equation}\label{eq:pre-wf}
\preOfSpec(\Phi_2(m))\imp\leftF{\preOfSpec(\Phi_0(m))}\land\rightF{\preOfSpec(\Phi_1(m))}
\]
is valid (where metafunction $\preOfSpec$\index{$\preOfSpec$} extracts the precondition),
and the effects of $\Phi_2(m)$ project to those of $\Phi_0(m)$ and $\Phi_1(m)$.\footnote{In detail:
  Suppose $\Phi_2(m)$ is $\rflowty{\R}{\S}{\effe|\effe'}$,
  and the unary specs $\Phi_0(m)$ and $\Phi_1(m)$ are $\flowty{R_0}{S_0}{\effe_0}$ and $\flowty{R_1}{S_1}{\effe_1}$ respectively.
  Then $\effe = \effe_0$ and $\effe' = \effe_1$.
}
\end{definition}
The constraint on preconditions
ensures a compatibility condition needed to connect relational
with unary context models, see Def.~\ref{def:ctxinterpRel}.
Def.~\ref{def:wfhypRel} allows left and right to have different global variables.
It also allows that some spec-only variables on the left may also occur on the right.
However, well formedness is in the context of a single module structure (module names and their association with methods and dynamic boundaries; import relation). 

\begin{definition}%[wf relational correctness judgment]
\label{def:wfjudgeRel}
A \dt{relational correctness judgment}
has the form
\ghostbox{\( \Phi\proves^{\Gamma|\Gamma'}_M CC: \rflowty{\P}{\Q}{\eff|\eff'} \)}.
It is wf \index{wf} provided
\begin{itemize}
\item $\Phi$ is wf in $\Gamma|\Gamma'$ (see above).
\item No spec-only variables, nor $\lloc$, occur in $CC$.
Moreover, alignment guard assertions in bi-whiles contain no agreement formulas.  %\footnote{This ensures they are refperm-independent.}
\item No methods occur in $\Gamma|\Gamma'$, and $CC$ is wf in the typing context that extends $\Gamma|\Gamma'$ to declare the methods in $\Phi$.
\item $\bnd(N)$ is wf in $\Gamma$ and wf in $\Gamma'$,
for all $N$ with $N\in\Phi$ or $N=M$.
\item $\rflowty{\P}{\Q}{\eff|\eff'}$ is wf in $\Gamma|\Gamma'$,
and its spec-only variables are distinct from those in $\Phi$.
\end{itemize}
\end{definition}
%\dn{Can't wf of $CC$ be reduced to unary? - yes but may not be more clear.

\begin{example}[coupling and local equivalence for PQ]
\label{ex:PQagree}
The coupling relation expresses
that for any two corresponding queues in the left and right states' $pool$,
all the \whyg{Pnode}s in their $rep$s are in the refperm.
The sentinel is in $pool$, not in a $rep$,
and each pair of
corresponding \whyg{Pnode}s have the same value and priority. Moreover, $\NULL$
appears in the left state where the sentinel appears in the right. As a relation
formula:
\[ \begin{array}{lcl}
\multicolumn{3}{l}{
  \forall q:\code{Pqueue}\in pool \mid q:\code{Pqueue}\in pool}  \\
%  \forall q:Pqueue\mid q:Pqueue . \Both q\in pool \imp } \\ % official syntax but that's not used below
  \quad \Agr q & \imp & (\Agr(q.head) \lor (\leftF{q.head = \NULL} \land \rightF{q.head = q.sntnl})) \\
       &      & \land\: q.rep/\code{Pnode} \eqbi q.rep/\code{Pnode} \\
       &      & \land\: \forall \; n\scol\code{Pnode}\in q.rep \mid n\scol\code{Pnode} \in q.rep \: .\\
&& \qquad \Agr n\imp
         \begin{array}[t]{l}  \Agr(n.val) \land  \Agr(n.key) \\
           \land (\Agr(n.sibling) \lor (\leftF{n.sibling = \NULL} \land \rightF{n.sibling = q.sntnl}))  \\
           \land (\Agr(n.child) \lor (\leftF{n.child = \NULL} \land \rightF{n.child = q.sntnl}))  \\
           \land (\Agr(n.prev) \lor (\leftF{n.prev = \NULL} \land \rightF{n.prev = q.sntnl}))
%            \land \bigwedge_{f\in\{sibling,child,prev\}}
%            (\Agr(n.f) \lor (\leftF{n.f = \NULL} \land \rightF{n.f = q.sntnl}))
       \end{array}
   \end{array}
 \]
Here we use syntax sugar $\Agr n.val$ for $\Agr\sing{n}\Img val$.
Also, the pattern
$\forall q\scol K\in r \mid q\scol K\in r \ldots$ is sugar for
$\forall q\scol K \mid q\scol K.
\leftF{q\in r}\land\rightF{q\in r} \imp \ldots$.
Note the type restriction expressions in the agreement
$q.rep/\code{Pnode} \eqbi q.rep/\code{Pnode}$. Let $\M_{PQ}$ be the above formula, conjoined with $\leftF{I}\land\rightF{I'}$ where $I,I'$ are the private invariants.

%% \Agr(\{n\}\Img sibling \lor \leftF{n.sibling = \NULL} \land \rightF{n.sibling = q.sntnl}) \land
%%   (n.child\eqbi n.child \lor \leftF{n.child = \NULL} \land \rightF{n.child = q.sntnl}) \land
%%   (n.prev\eqbi n.prev \lor \leftF{n.prev = \NULL} \land \rightF{n.prev = q.sntnl}) }.

The relational spec for \whyg{insert} obtained by applying $\locEq$
looks like this:
\begin{equation}\label{eq:locEqInsert}
\rflowty{\Agr q \land\Agr k \land\Both P}
           {\later(\Agr(res.val)\land\Agr(res.key)\land\ldots\land\Both Q)}
           {\rw{\sing{q}\Img\allfields,q.rep\Img\allfields,\lloc}}
\end{equation}
where $P$ and $Q$ are the unary pre- and post- conditions for
\whyg{insert}, including the public invariant of \whyg{PQ}.
We elide some postconditions like $\Agr((pool\setminus(pool\union pool\Img rep))\Img head)$
which arise by subtracting the boundary from writes in the spec (and expanding $\allfields$ to
all field names).
This one can obviously be simplified to $\Agr\emptyset\Img head$ which is equivalent to $\True$.
The meta-function $\locEq$ need not perform such simplifications, as the reasoning can safely be left to the SMT solver or 
to the logic's relational consequence rule.

To verify the two implementations of \whyg{insert},
we conjoin $\M_{PQ}$ to both the pre and postcondition of the relational spec above.
The resulting precondition is
$\Agr q \land\Agr k \land\Both P \land \M_{PQ}$
and the postcondition is
\( \later(\Agr(res.val)\land\Agr(res.key)\land\ldots\land\Both Q \land \M_{PQ}) \).
Later we introduce a notation $\conjInv\M_{PQ}$ for this.
\qed\end{example}

% In this setting, the relational spec for \whyg{insert}, in essence, is
% $\rflowty{\Agr(q)\land\Agr(k)\land\Both{(q\neq\NULL\land q\in pool\land k\geq0)}}{
%  \Agr(res.val)\land
%  \Agr(res.key)\land
%  \Both{(\neg\code{isEmpty}(q)\land res\in q.rep\land res.val=v\land res.key=k)}
% }{\rw{\sing{q}\Img\allfields,q.rep\Img\allfields,}}
% $.  In particular, note that the left and right implementations of this method have
% the same frame conditions.

%\hrulefill{}

% Relational spec:
% forall q:PQ in pool | q:PQ in pool. q =:= q ==>
% forall n:Node in q.rep | n:Node in q.rep.
% n =:= n ==>
% (<| n.sibling = null /\ |> n.sibling = q.sntnl \/
%  n.sibling =:= n.sibling) /\
% (<| n.parent = null /\ |> n.parent = q.sntnl \/
%  n.parent =:= n.parent) /\
% (<| n.child = null /\ |> n.child = q.sntnl \/
%  n.child =:= n.child)

\subsection{Relational verification with biprograms} \label{sec:eg:relverif}

We consider an example of relational verification which is modular in the sense of
using relational method specs, but no information hiding.
We highlight how regions are used in relational specs, and how biprograms are used
to represent convenient alignments.

\begin{figure}[t]
\begin{subfigure}{.3\textwidth}
\begin{lstlisting}
meth tabulate (n:int) : List =
  var t: List, i: int, p: Node;
  t := new List;
  i := 0;
  while i < n do
    i := i + 1;
    p := new Node;
    p.val := mf(i);
    p.nxt := t.head;
    t.head := p;
    t.nds := t.nds ** {p};
  od;
  result := t;   /* return value */
\end{lstlisting}
\caption{Left version, $tabu$}
\end{subfigure}
\begin{subfigure}{.3\textwidth}
\begin{lstlisting}
meth tabulate (n:int) : List =
  var t: List, i: int, p: Node;
  t := new List;
  i := 1;
  while i <= n do
    p := new Node;
    p.val := mf(i);
    p.nxt := t.head;
    t.head := p;
    t.nds := t.nds ** {p};
    i := i + 1;
  od;
  result := t;
\end{lstlisting}
\caption{Right version, $tabu'$}
\end{subfigure}
\begin{subfigure}{.3\textwidth}
\begin{lstlisting}

/* Agr n */
|_ t := new List _|; connect t;
(i := 0 | i := 1);
while (i < n) | (i <= n) do
   (i := i + 1 | skip);  /* i =:= i */
   |_ p := new Node _|; connect p;
   |_ p.val := mf(i) _|;  /* Agr p.val */
   |_ p.nxt := t.head _|;
   |_ t.nds := t.nds ** {p} _|;
   (skip | i := i + 1);
od;
|_ result := t _|;
\end{lstlisting}
\caption{Biprogram $CC_{tabu}$}
\end{subfigure}
  \caption{Two implementations of tabulate, and a biprogram weaving them together.}
  \label{fig:tabulate}
\end{figure}

\paragraph{List tabulation: illustrating procedure-modular reasoning.}

Consider the two programs in Figure~\ref{fig:tabulate}, which both tabulate a linked list of the values of some
method \whyg{mf} that computes a function, applied to the numbers $n$ down to $1$.
Objects of class \whyg{List} have two fields: $head: \code{Node}$ references the head of a linked list
and $nds:\Region$ is ghost state, to which we return soon.
The goal in this example is to prove the programs are equivalent.
We reason about executions of the two programs in close alignment, in order
to exploit their similarities and make use of a relational spec for \whyg{mf}.
The example also serves to show the use of regions to describe heap structure
and in particular to express the equivalence of the lists returned.
The example illustrates two aspects of modular reasoning: procedural abstraction and local reasoning; the third aspect, data abstraction, is considered in Section~\ref{sec:eg:encapR}.

Both versions of the program use field $nds$ to hold references to the nodes reached from $head$.
It is initially empty (the default value), and in each iteration the newly allocated node is added to the list's $nds$.
An invariant of the loop, in both programs, is $t.nds\Img next \subseteq t.nds$.
Here $t.nds$ is set of references. The image expression $t.nds\Img next$ denotes the set of values in the next fields
of objects in $t.nds$ (a direct image, thinking of the field as a relation).
The containment $t.nds\Img next \subseteq t.nds$ says for any object reference in $t.nds$,
the value of the object's $next$ field is in $t.nds$.
There are no recursive definitions involved.
The containment, together with invariant $t.head \in t.nds$,
implies that everything reachable from $t.head$ is in $t.nds$.  It does not say that $t.nds$ is exactly the reachable set, though it will be; we do not need that stronger fact.
%\dn{Might we depending on usage context? maybe, but could suffice to know some separation.}

Method \whyg{mf} has an integer parameter $x$ and returns an integer result.
Its unary spec is $\flowty{true}{true}{\emptyeff}$,
which says very little but the empty frame condition says it has no effect on the heap
or global variables.  In particular, it does no allocation, since otherwise its frame condition would have to include $\rw{\lloc}$.
Implicitly it is allowed to read its parameter $x$ and write its $result$,
as we saw in Example~\ref{ex:PQframe}.
As relational spec we use $\rflowty{\Agr x}{\Agr result}{\emptyeff}$  which expresses 
determinacy as self-equivalence in a way that is local: it refers only to locations that may be read or written.
It is this relational spec, and nothing more, that we wish to use for \whyg{mf} in relational reasoning about \whyg{tabulate}.

For \whyg{tabulate}, the frame condition is $[\rw{\lloc}]$.
It allocates, which implicitly updates the special variable $\lloc$ by adding the newly allocated reference; the new value of $\lloc$ depends on its old value, so the frame condition
says $\lloc$ may be both read and written.
Like method \whyg{mf}, method \whyg{tabulate} reads its parameter and writes its result, but neither reads nor writes any other preexisting locations.

Although we aim to prove equivalence of the two versions of \whyg{tabulate} without recourse to a precise functional spec,
we do include a postcondition that constrains $nds$, as this plays a role in specifying equivalence.
The postcondition says $nds$ contains $head$ and is closed under $next$; formally:
$result.nds\Img next \subseteq result.nds$ and $result.head \in result.nds$.

To express equivalence of the two versions, the (relational) precondition is
agreement on what is readable, namely the parameter $n$.
The agreement formula $\Agr n$, or equivalently $n\eqbi n$,
simply means the two initial states have the same value for $n$.
We do not assume agreement on $\lloc$; we want the equivalence to encompass initial states without
constraint on allocated but irrelevant objects.

For the postcondition we want agreement on what is writable (aside from $\lloc$), thus $\Agr result$.
We also specify that the unary postcondition holds in both final states:
\begin{equation}\label{eq:tabu:postB}
\Both{ ( result.nds\Img next \subseteq result.nds \land result.head \in result.nds) }
\end{equation}
But $result$ is just a reference to newly allocated list structure.
To express that the two result lists have the same content we need more than $\Agr result$. A first guess is
the agreement formula $\Agr result.nds\Img val$.
The formula uses syntax sugar, to abbreviate $\Agr \sing{result}\Img nds \Img val$.
Agreement formulas, as mentioned in Section~\ref{sec:rrl}, are interpreted with respect to a refperm, 
that is, a type-respecting partial bijection on references of the two states.
Whereas $\Agr n$ means identical values for integer $n$,
the formula $\Agr result$ means equivalent reference values, i.e., connected via the bijection.
The formula $\Agr result.nds\Img val$ says that for pairs $o,o'$ of references connected by the bijection,
with $o\in result.nds$, the fields $o.val$ and $o'.val$ have equal contents; equal because the type is integer.

To fully constrain the lists to have the same structure we use this postcondition:
\begin{equation}\label{eq:tabu:postR}
\later (\Agr result \land \Agr result.nds \land \Agr result.nds\Img next \land \Agr result.nds\Img val )
\end{equation}
Here $\later$ says there exists some refperm.
The formula $\Agr result.nds$
abbreviates $\Agr\sing{result}\Img nds$ and
says the refperm cuts down to a (total) bijection between the regions
$result.nds$ in the two states.
The condition $\Agr result.nds`next$ says that bijection is compatible with the linked list structure.

% List object
\tikzset{%
  pics/listobj/.style args={#1/#2/#3}{%
    code={%
      \node[draw,rounded corners=0.25em,minimum height=1.9cm,minimum width=2.5cm] (-main) {};
      \node[above of=-main,xshift=-1cm,yshift=0.25cm] (-name) {#1};
      \node[draw,minimum height=0.5cm,minimum width=1.7cm,xshift=0.3cm,yshift=0.2cm] (-nds) {#2};
      \node[draw,minimum height=0.5cm,minimum width=1.7cm,below of=-nds,yshift=0.25cm] (-head) {#3};
      \node[left of=-nds,xshift=-0.2cm] (ndsTxt) {\footnotesize$nds$};
      \node[left of=-head,xshift=-0.2cm] (headTxt) {\footnotesize$head$};
      \node[above of=ndsTxt,yshift=-0.45cm] (listText) {$\mathsf{List}$};
    }%
  }%
}%

% Node object
\tikzset{%
  pics/listnode/.style args={#1/#2/#3}{%
    code={%
      \node[draw,rounded corners=.25em,minimum height=1.5cm,minimum width=1.75cm] (-main) {};
      \node[above of=-main,xshift=-0.75cm] (-name) {#1};
      \node[draw,minimum height=0.5cm,minimum width=1cm,xshift=0.3cm,yshift=0.35cm] (-val) {#2};
      \node[draw,below of=-val,minimum height=0.5cm,minimum width=1cm,yshift=0.3cm] (-nxt) {#3};
      \node[left of=-val,xshift=0.15cm] (valText) {\footnotesize$val$};
      \node[left of=-nxt,xshift=0.15cm] (nxtText) {\footnotesize$nxt$};
      \coordinate (-anchor-point) at ([yshift=0.5cm] -main.west);
    }%
  }%
}%

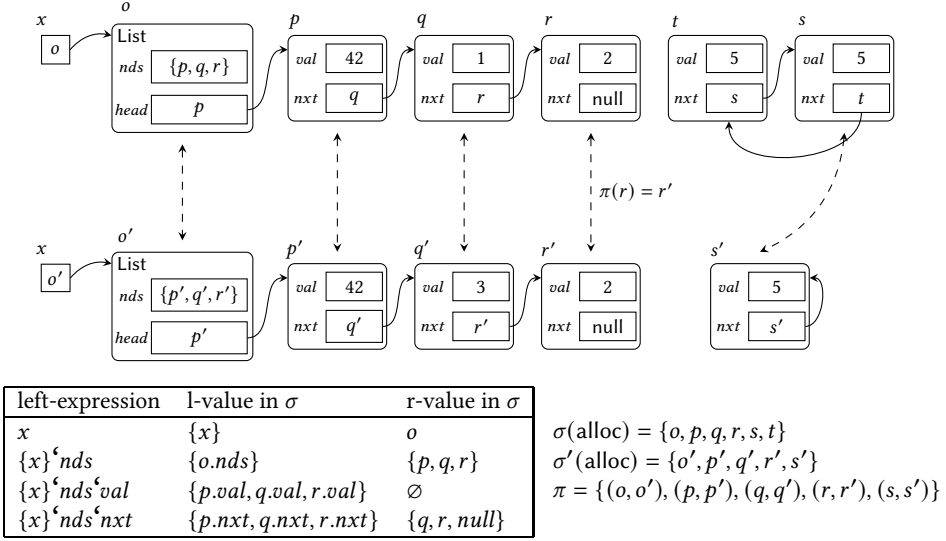
\begin{figure}[t]
  \begin{center}
  \begin{tikzpicture}[scale=0.75,transform shape,>=stealth]
    % left variable x
    \node[draw,rectangle,minimum width=0.5cm,minimum height=0.5cm] at (-1.5,0.5) (left-o) {$o$};
    \node[above of=left-o,yshift=-0.5cm,xshift=-0.25cm] (left-x) {$x$};

    % right variable x
    \node[draw,rectangle,minimum width=0.5cm,minimum height=0.4cm] at (-1.5,-3.5) (right-o) {$o'$};
    \node[above of=right-o,yshift=-0.5cm,xshift=-0.25cm] (right-x) {$x$};

    % draw left linked list -- distance between nodes: 2.25
    \pic (left-list) at (0.75,0) {listobj={$o$/$\{p,q,r\}$/$p$}};
    \pic (left-p) at (3.5,0) {listnode={$p$/$42$/$q$}};
    \pic (left-q) at (5.75,0) {listnode={$q$/$1$/$r$}};
    \pic (left-r) at (8,0) {listnode={$r$/$2$/$\NULL$}};
    \pic (left-t) at (10.25,0) {listnode={$t$/$5$/$s$}};
    \pic (left-s) at (12.5,0) {listnode={$s$/$5$/$t$}};

    % draw right linked list, y-shift by -4; could use `below of' to not have to
    % manually specify coordinates.
    \pic (right-list) at (0.75,-4) {listobj={$o'$/$\{p',q',r'\}$/$p'$}};
    \pic (right-p) at (3.5,-4) {listnode={$p'$/$42$/$q'$}};
    \pic (right-q) at (5.75,-4) {listnode={$q'$/$3$/$r'$}};
    \pic (right-r) at (8,-4) {listnode={$r'$/$2$/$\NULL$}};
    \pic (right-s) at (11,-4) {listnode={$s'$/$5$/$s'$}};

    % draw arrows in left linked list.
    \draw[->] (left-list-head) to[out=0,in=180] (left-p-anchor-point);
    \draw[->] (left-p-nxt) to[out=0,in=180] (left-q-anchor-point);
    \draw[->] (left-q-nxt) to[out=0,in=180] (left-r-anchor-point);
    \draw[->] (left-t-nxt) to[out=0,in=180] (left-s-anchor-point);
    \draw[->] (left-s-nxt) to[out=270,in=270] ([xshift=0.2cm]left-t-main.south);

    % draw arrows in right linked list.
    \draw[->] (right-list-head) to[out=0,in=180] (right-p-anchor-point);
    \draw[->] (right-p-nxt) to[out=0,in=180] (right-q-anchor-point);
    \draw[->] (right-q-nxt) to[out=0,in=180] (right-r-anchor-point);
    \draw[->] (right-s-nxt) to[out=0,in=0] ([yshift=0.5cm]right-s-main.east);

    % draw dashed arrows between the two linked lists to indicate that certain
    % objects are related by the refperm.
    \draw[<->,dashed] ([yshift=-0.2cm]left-list-main.south) -- ([yshift=0.2cm]right-list-main.north);
    \draw[<->,dashed] ([yshift=-0.2cm]left-p-main.south) -- ([yshift=0.2cm]right-p-main.north);
    \draw[<->,dashed] ([yshift=-0.2cm]left-q-main.south) -- ([yshift=0.2cm]right-q-main.north);
    \draw[<->,dashed] ([yshift=-0.2cm]left-r-main.south) -- ([yshift=0.2cm]right-r-main.north);
    \draw[<->,dashed] ([yshift=-0.2cm]left-s-main.south) to[bend left] ([yshift=0.2cm]right-s-main.north);

    % draw arrow from x to o
    \draw[->] (left-o.east) to[bend left] ([yshift=0.75cm]left-list-main.west);

    % draw arrow from x' to o'
    \draw[->] (right-o.east) to[bend left] ([yshift=0.75cm]right-list-main.west);

    % testing how a formula can be placed at a certain point.
    \node[below of=left-r-main,yshift=-1cm,xshift=0.8cm] (pi-r) {$\pi(r) = r'$};
  \end{tikzpicture}
  \end{center}

\medskip

\begin{small}

\begin{list}{}{} % trying to override centering
\item
\begin{tabular}{|lll|}
\hline
left-expression & l-value in $\sigma$ & r-value in $\sigma$ \\\hline
$x$                        & $\{x\}$                 & $o$ \\
$\sing{x}\Img nds$         & $\{ o.nds \}$           & $\{p,q,r\}$ \\
$\sing{x}\Img nds\Img val$ & $\{p.val,q.val,r.val\}$ & $\emptyset$ \\
$\sing{x}\Img nds\Img nxt$ & $\{p.nxt,q.nxt,r.nxt\}$ & $\{q,r,\semNull\}$\\
\hline
\end{tabular}
$\begin{array}{l}
\sigma(\lloc) = \{o,p,q,r,s,t\} \\
\sigma'(\lloc) = \{o',p',q',r',s'\} \\
\pi = \{(o,o'),(p,p'),(q,q'),(r,r'),(s,s')\}
\end{array}$

\medskip

%% \item
%% $\sigma(\lloc) = \{o,p,q,r,s,t\}
%% \quad \sigma'(\lloc) = \{o',p',q',r',s'\}
%% \quad \pi = \{(o,o'),(p,p'),(q,q'),(r,r'),(s,s')\}$

\item
$\sigma|\sigma'\models_\pi   \Agr x $ is true because $\rprel{o}{o'}$

\item
$\sigma|\sigma'\models_\pi \Agr \sing{x}\Img nds$ is true because
$\rprel{o}{o'}$ and
$\rprel{ \{p,q,r\} }{ \{p',q',r'\} }$

\item
$\sigma|\sigma'\models_\pi \Agr \sing{x}\Img nds\Img nxt $ is true;
note
$\rprel{p.nxt}{p'.nxt}$,
$\rprel{q.nxt}{q'.nxt}$, and
$\rprel{r.nxt}{r'.nxt}$

\item
$\sigma|\sigma'\models_\pi \Agr \sing{x}\Img nds\Img val$ is false because $\sigma(q.val)=1\neq 3=\sigma'(q'.val)$

\item
$\sigma|\sigma'\models_\pi \sing{x}\Img nds\eqbi \sing{x}\Img nds$
is true because $\rprel{ \{p,q,r\} }{ \{p',q',r'\} }$,
regardless of whether $(o,o')$ is in $\pi$
%it would be true even without $(o,o')\in\pi$

\end{list}
\end{small}
\caption{Refperm $\pi$ and relations between two states, $\sigma,\sigma'$ with variable $x$ (see Example~\ref{ex:box-and-arrows}).}\label{fig:box-and-arrows}
\end{figure}

The semantics of relation formulas is formalized in Sec.~\ref{sec:relform}.
It is a little subtle: $\sing{x}\Img f \eqbi \sing{x}\Img f$ is different from
$\Agr \sing{x}\Img f$, unless guarded by $\Agr x$ (as a conjunct or antecedent).
We invariably use such guarded formulas, e.g., conjuncts in (\ref{eq:tabu:postR}) and
antecedents in the coupling of Example~\ref{ex:PQagree}.

\begin{example}\label{ex:box-and-arrows}
To illustrate the meaning of agreement formulas like those in (\ref{eq:tabu:postR}),
Figure~\ref{fig:box-and-arrows} shows 
an example of two states with a single variable $x:\code{List}$,
and using $\sing{x}\Img nds$ rather than its sugared form $x.nds$.
The semantic notations are defined in Section~\ref{sec:relform} but the picture is meant to be understandable now.
The values of some left-expressions are given; we consider the l-value of any left-expression to be a set of locations, such as the single location $x$ (a variable name) and $p.val$ (a heap location).
\qed
\end{example}

Taken together, (\ref{eq:tabu:postB}) and (\ref{eq:tabu:postR})  say the results from \whyg{tabulate} are lists for which the nodes can be put in bijective correspondence that is compatible with the $nxt$
pointers and for which corresponding elements have the same value.
They serve as postcondition, with precondition $\Agr n$, to specify equivalence for \whyg{tabulate}.
What else would we mean by equivalence of the programs?  We do not want to say they have literally identical values, because we want equivalence to be local: It should not involve what else may have been allocated, so we do not assume agreement on $\lloc$. Hence the resulting lists may not have identical reference values.  What matters is that the heap data produced by the two implementations has the same structure.

\paragraph{On the modality $\later$.}

The modal operator $\later$ is needed for the relational postcondition
(\ref{eq:tabu:postR}) and in any spec where allocation is possible.
We gloss over it in some examples, but specs of interest usually have
this standard form: $\rflowty{\R}{\later\S}{\eff}$ where $\later$ does not occur in $\R$ or $\S$.
The \whyg{tabulate} spec can be put in standard form, because (\ref{eq:tabu:postB}) expresses unary conditions, with no dependence on refperm, so that formula can be put inside the $\later$ in (\ref{eq:tabu:postR}).

While SMT solvers typically provide some heuristic support for quantifiers,
existential quantifiers are problematic and we cannot expect a solver
to find witnesses for the existential expressed by $\later$.
In the \WhyRel\ prototype, specs do not include $\later$ explicitly.  Instead,
a refperm is maintained in ghost state, thus witnessing the existential.
A ghost instruction, \whyg{connect - with -}, can be used to designate which references the user
wants to be considered as corresponding.
For example, the biprogram Figure~\ref{fig:tabulate}(c) uses \whyg{connect p},
which abbreviates \whyg{connect p with p}, to add newly allocated \whyg{Node} references to the refperm, thereby establishing $p\eqbi p$.
The general form of \whyg{connect} caters for programs using different variables.

\paragraph{Alignment for \whyg{tabulate}.}

Recall that  (\ref{eq:tabu:postB}) and (\ref{eq:tabu:postR}) are meant to comprise the postcondition of a spec to relate the bodies, $tabu$ and $tabu'$, of the two implementations of \whyg{tabulate} in Figure~\ref{fig:tabulate}(a) and~(b). 
To say that they satisfy the relational spec we use a judgment like this:
%\[ \Phi \proves \splitbi{tabu}{tabu'} : \rflowty{\Agr n}{\ldots}{\rw{\lloc}}\]
%eliding the postcondition (\ref{eq:tabu:postB})$\land$(\ref{eq:tabu:postR}).
\[ \Phi \proves \splitbi{tabu}{tabu'} : \rflowty{\Agr n}{\R}{\rw{\lloc}}
\quad\mbox{where $\R$ is (\ref{eq:tabu:postB})$\land$(\ref{eq:tabu:postR})} 
\]
The hypothesis context specifies \whyg{mf};
$\Phi$ is a triple, with $\Phi_2(\code{mf})$  being the relational spec
$\rflowtyf{\Agr x}{\Agr result}$.
%$\rflowty{\Agr x}{\Agr result}{\rd{x},\wri{result}}$ \rmn{Omit frame?}.
The unary specs $\Phi_0(\code{mf})$ and $\Phi_1(\code{mf})$ are not relevant to this example.
% but are required by our formalization.

We derive the judgment for $\splitbi{tabu}{tabu'}$ from a judgment with the same spec for  the more conveniently aligned biprogram
$CC_{tabu}$ in Figure~\ref{fig:tabulate}(c), in a way that will be justified in Section~\ref{sec:weave}.
Several features of $CC_{tabu}$ are important.
First, its left and right syntactic projections are the two commands, $tabu$ and $tabu'$, to be related;
semantically it represents pairs of their executions, aligned in a particular way.
Second, the calls to \whyg{mf} are in the sync'd form, which signals that reasoning is to be done using the relational spec of \whyg{mf}.
A comment in the biprogram indicates that we get agreement on $p.val$ following the calls to $\code{mf}(i)$,
in virtue of that spec.
Similarly, the two allocations are also in the sync'd form and followed by the \whyg{connect} ghost operation,
achieving agreement on the allocated references.
In the proof system, there is a rule for sync'd allocations, with postcondition that yields for example  $\later \Agr p$ for the \whyg{Node} allocation.
Using this rule (or the connect ghost operation) is a good choice in the present example,
but in general it is not necessary to connect allocations, even if they happen to be aligned; this is important when relating programs that are not building the same heap structure, or when proving noninterference and reasoning about branches with tests that depend on secrets. Finally, the bi-while
in $CC_{tabu}$ signals that we reason in terms of lockstep alignment of the loop iterations.
This enables us to reason that the two executions are building isomorphic pointer structures, using a relational
invariant similar to the postcondition of the relational spec (\ref{eq:tabu:postR}), conjoined with a simple relation between the counter variables:
\[ i-1 \eqbi i \land \Agr n \land \Agr t \land \Agr t.nds \land \Agr t.nds\Img nxt \land \Agr t.nds\Img val \]
The biprogram provides a convenient alignment but incurs an additional proof obligation:
the invariant must imply that the loop tests agree, as otherwise it would be unsound to assume the iterations
can be considered to be aligned in lockstep.  Indeed, the implication is valid:
$\Agr n$ and $i-1 \eqbi i$ implies $i<n \eqbi i \leq n$.

%\dn{[[At this point we'd like to talk about rLinkX, procedure-modular reasoning, without hiding:
%relate versions of tabulate when liked with correct implementation of mf - we have no such rule]]}

In summary, this example shows biprograms express alignment of the programs under consideration in order to facilitate procedure-modular reasoning using relational specs
and to facilitate the use of simpler relational invariants for loops.
In passing we introduced ways to express relations on pointer structures, abstracting from specific addresses
(as appropriate for Java- and ML-like languages) and making it possible to specify relations where some parts of the heap are meant to have isomorphic structure while other parts may be entirely different.
There are at least two important use cases for such differences:
encapsulated data structures, when relating implementations of a module interface,
and structure manipulated by ``secret'' computations, when proving information flow properties.

The example happens to work well with close alignment of the program structure and agreement on all the data involved.
The logic must handle aligned allocation in a loop, as in this example.
It must also handle differing allocations, for example to relate programs using different encapsulated
data representations.
Differing allocations also arise when proving noninterference,
in cases where allocation occurs under high branch conditions.

The proof rules used to derive a relational modular linking rule like (\ref{eq:mismatchR}) make use of a general
form of local equivalence specification, derived from the frame condition of a unary spec (and defined in Section~\ref{sec:locEq}).
But it is also possible to express local equivalence notions suited to specific situations, as in
the example, and it is possible to work with differing program structures as illustrated in some case studies (e.g., Figure~\ref{fig:woven-insert} and Section~\ref{sec:eg:encapR}).

\subsection{Defining and using biprogram weaving for alignment}\label{sec:weave}

In this subsection we define the weaving relation on biprograms.
The purpose of the weaving relation is to connect 
a bi-com $\splitbi{C}{C'}$, that expresses a relational verification problem,
with a more tightly aligned version that facilitates reasoning. 
If $\splitbi{C}{C'}$ weaves to $DD$, written $\splitbi{C}{C'}\weave DD$, then the syntactic projections of $DD$ are $C$ and $C'$, so $DD$ models executions of the two commands.
The weaving relation $\weave$ is used in a proof rule that realizes the product principle: any judgment that holds for $DD$ also holds for $\splitbi{C}{C'}$,
given $\splitbi{C}{C'}\weave DD$.
In general, weaving brings together similarly structured subprograms,
introducing additional alignment points while preserving syntactic projections. 
In addition to defining the relation $\weave$, the rest of this section gives examples of its use, and sketches the semantic considerations that justify the proof rule and
explain the orientation of the relation.

\begin{figure}[t]
\begin{small}
\[\begin{array}{l}
\splitbi{A}{A} \weave \syncbi{A} \\[1ex]
\Splitbi{C;D}{C';D'} \weave \splitbi{C}{C'};\splitbi{D}{D'} \\[1ex]
\Splitbi{ \ifc{E}{C}{D} }{ \ifc{E'}{C'}{D'} }
   \weave \ifcbi{E\smallSplitSym E'}{\splitbi{C}{C'} }{ \splitbi{D}{D'} }
\\[1ex]
\Splitbi{ \whilec{E}{C} }{ \whilec{E'}{C'} }
   \weave \whilecbiA{E\smallSplitSym E'}{\P\smallSplitSym \P'}{\splitbi{C}{C'}}
\\[1ex]
   \Splitbi{ \letcom{m}{B}{C} }{ \letcom{m}{B'}{C'} }
    \weave \letcombi{m}{\splitbi{B}{B'}}{ \splitbi{C}{C'} }

\\[1ex]
  \Splitbi{ \varblock{x\scol T}{C} }{ \varblock{x'\scol T'}{C'} }
  \weave \varblockbi{x\scol T\smallSplitSym x'\scol T'}{\splitbi{C}{C'}}

\\[2ex]
\inferrule{ BB\weave CC }
{
BB;DD \weave CC;DD\\
DD;BB \weave DD;CC\\
\ifcbi{E\smallSplitSym E'}{BB}{DD} \weave \ifcbi{E\smallSplitSym E'}{CC}{DD} \\
\ifcbi{E\smallSplitSym E'}{DD}{BB} \weave \ifcbi{E\smallSplitSym E'}{DD}{CC} \\
\whilecbiA{E\smallSplitSym E'}{\P\smallSplitSym \P'}{BB} \weave
  \whilecbiA{E\smallSplitSym E'}{\P\smallSplitSym \P'}{CC} \\
\letcombi{m}{\splitbi{B}{B'}}{ BB } \weave \letcombi{m}{\splitbi{B}{B'}}{ CC } \\
\varblockbi{x\scol T \smallSplitSym x'\scol T'}{BB} \weave \varblock{x\scol T\smallSplitSym x'\scol T'}{CC}
}
 \end{array}\]
\end{small}
%\hrule
\caption{Axioms and congruence rules that define the weaving relation \ghostbox{$\weave$}.
Recall $A$ ranges over atomic commands (Figure~\ref{fig:bnf}).
}
\label{fig:weave}
\end{figure}

%Recall from Sect.~\ref{sec:approach} that to specify a relation between commands $C$ and $C'$ we use a judgment of the form $\splitbi{C}{C'}: \rflowty{\R}{\S}{\eff|\effe'}$.
%This designates that the initial and final states are related.
%To prove the judgment, one may use other biprogram forms to designate additional alignments within $C$ and $C'$.
The \dt{weaving relation} \ghostbox{$\weave$}
is defined inductively by axioms and congruence rules in Figure~\ref{fig:weave}.
The axioms replace a bi-com by another biprogram form including those that can assert 
agreements (bi-if and bi-while).
The congruence rules, displayed as one rule with multiple conclusions,
allow weaving in all contexts except the procedure bodies in bi-let.
Apropos congruence for bi-let, note that bi-let does not bind general biprograms but only pairs of commands
despite the appearance of the concrete syntax (see Figure~\ref{fig:bnf}).

The weaving that introduces bi-while allows the introduction
of so-called alignment guards.
The biprogram $CC_{tabu}$ omits them (Figure~\ref{fig:tabulate}(c)), 
which is syntax sugar taking them to be $\False$.
As an example of their use, later in this subsection we follow up on the example program (\ref{eq:sumpub})
discussed in Section~\ref{sec:modrel}, sketching the three-premise relational loop
rule that enables verification of the example using a simple invariant.

%% This simple example, mentioned in Section~\ref{sec:modrel}, illustrates weaving
%%   of alternate implementations of \whyg{cset(c,v)}, namely
%% \whyg{(c.val:=v | c.f:=-v); (return c.val | return -c.f)}.

\begin{example}
The sequence weaving axiom (second line of Figure~\ref{fig:weave})
can be used for an example mentioned in Section~\ref{sec:rrl},
namely \whyg{(c.val:= v | c.f:= -v); (return c.val | return -c.f)}.
For the bi-com $\Splitbi{a;b;c}{d;e;f}$
(temporarily using lower case letters for atomic commands),
there are four different alignments that can be obtained
by a single application of sequence weaving:
%a single application of sequence weaving can give us these four rearrangements:
\footnote{Keep in mind
the syntactic equivalences in Figure~\ref{fig:synident}, which enable these
different weavings.}
\begin{equation}\label{eq:weaveseq}
\begin{array}{l}
\splitbi{a;b;c}{d;e;f} \weave \splitbi{a;b}{d} ; \splitbi{c}{e;f} \\
\splitbi{a;b;c}{d;e;f} \weave \splitbi{a}{d;e} ; \splitbi{b;c}{f} \\
\splitbi{a;b;c}{d;e;f} \weave \splitbi{a;b;c}{\skipc} ; \splitbi{\skipc}{d;e;f} \\
\splitbi{a;b;c}{d;e;f} \weave \splitbi{\skipc}{d;e;f} ; \splitbi{a;b;c}{\skipc}
\end{array}
\end{equation}
These weavings introduce a semicolon at the biprogram level, which makes it possible to assert a relation at that point.
Different weavings of the same biprogram serve to align different intermediate points.
\qed\end{example}
Using the sequence axiom and congruence, we have
$\splitbi{a;b;c}{d;e;f} \weave \splitbi{a}{d};\splitbi{b;c}{e;f} \weave \splitbi{a}{d};\splitbi{b}{e};\splitbi{c}{f}$ which illustrates how fine grained alignment can be achieved when desired.
We also have $\splitbi{tabu}{tabu'}\weave^* CC_{tabu}$ which connects $tabu,tabu'$ to the particular alignment we choose for reasoning about them.

%%%%%%%%%%%%%%%%%%%%%%%%%%%%%%%%%%%%%%%%%%%%%%%%%%%%%%%%%%%%%%%%%%%%%%%%%%
\begin{figure}[t]
  \centering
  \begin{minipage}[t]{.4\textwidth}
\begin{lstlisting}
result := new Pnode(val, key);
result.sibling := self.sntnl;
result.child := self.sntnl;
result.prev := self.sntnl;
self.rep := self.rep ** {result};
if (self.head = self.sntnl) then
    self.head := result;
else
    self.head := link(self,self.head,result);
fi;
self.size := self.size + 1;
\end{lstlisting}
  \end{minipage}\hfill
  \begin{minipage}[t]{.45\textwidth}
\begin{lstlisting}
|_ result := new Pnode(val, key) _|;
( skip
| result.sibling := self.sntnl;
  result.child := self.sntnl;
  result.prev := self.sntnl );
|_ self.rep := self.rep ** {result} _|;
if (self.head = null | self.head = self.sntnl) then
   |_ self.head := result _|;
else
   |_ self.head := link(self,self.head,result) _|;
fi;
|_ self.size := self.size + 1 _|;
\end{lstlisting}
  \end{minipage}
  \caption{Body of alternative implementation of \whyg{PQ}'s \whyg{insert} (left) and woven biprogram (right).}
  \label{fig:woven-insert}
\end{figure}
%%%%%%%%%%%%%%%%%%%%%%%%%%%%%%%%%%%%%%%%%%%%%%%%%%%%%%%%%%%%%%%%%%%%%%%%%%

As noted earlier, the bi-if and bi-while forms are meant to designate reasoning in which
it will be shown that the test conditions are in agreement.
Technically, we define small step semantics for biprograms, in which these forms
can have a fault ---dubbed \dt{alignment fault}--- if the tests are not in agreement.
This can be seen as a kind of assertion failure.
As an example, recall the implementation of \whyg{insert} in the \whyg{PQ} module
in Figure~\ref{fig:PQueue1}.  Part of the alternate implementation
using sentinels (mentioned in Example~\ref{ex:PQversions})
is shown in Figure~\ref{fig:woven-insert}.
We weave the two conditionals using a bi-if,
which introduces the possibility of alignment fault.
We can use this weaving because our coupling relation will ensure that
$\mself.head = \NULL$ in the left state just when $\mself.head = \mself.sntnl$ on the right.

Use of bi-if or bi-while incurs additional proof obligations that ensure the absence
of alignment fault, which in turn implies that the designated
alignment covers all pairs of executions of the underlying programs.
The weaving transformations can introduce the bi-if and bi-while forms but not eliminate them; nor can they eliminate any other faults.
For example,
$\Splitbi{ \ifc{x>0}{y.f:=x}{\skipc} }{ \ifc{x>0}{y.f:=x}{\skipc} }$ weaves to
$\ifcbi{x>0|x>0}{\Splitbi{y.f:=x}{y.f:=x}}{\syncbi{\skipc}}$,
noting that $\splitbi{\skipc}{\skipc} \equiv \syncbi{\skipc}$.
Both biprograms can fault due to null dereference, but the second also faults
in a pair of states where $x>0$ on one side but $x\leq 0$ on the other.

Suppose $DD$ can be obtained from $CC$ by a sequence of weavings, 
i.e., $CC\weave^* DD$. 
The relation $\weave$ can introduce the possibility of additional alignment faults, but it cannot eliminate such possibility.  In this sense, $\weave$ is oriented (and not symmetric).    
A consequence is the following:
if, under some precondition, $DD$ has no faults, then under that precondition the executions of $DD$ cover all those of $CC$.
This is % formalized in Lemma~\ref{lem:weave-sync}
the gist of the argument for soundness of the following proof rule: 
\begin{equation}\label{eq:weaveRule}
\begin{array}{l}
\mbox{from} \quad BB: \rflowty{\R}{\S}{\eff} \quad \mbox{infer} \quad  \splitbi{C}{C'}: \rflowty{\R}{\S}{\eff} \quad \mbox{provided}\quad \splitbi{C}{C'}\weave^* BB
\end{array}
\end{equation}
%that says the judgment $\splitbi{C}{C'}: \rflowty{\R}{\S}{\eff}$ follows from
%$BB: \rflowty{\R}{\S}{\eff}$ provided
%$\splitbi{C}{C'}\weave^* BB$.
(See rule \rn{rWeave} in Figure~\ref{fig:proofrulesR}.)
It is this rule that yields a relational judgment for $\splitbi{tabu}{tabu'}$
from the same judgment for $CC_{tabu}$ (Figure~\ref{fig:tabulate}).

\begin{wrapfigure}{r}{0.48\textwidth}  % r allows float, R means exactly here
\begin{footnotesize}
\(
\begin{array}[t]{l@{\hspace*{.6em}}c@{\hspace*{.6em}}l}
\Syncbi{A} &\eqdef& \syncbi{A} \qquad \mbox{ (atomic commands)}\\
\Syncbi{C;D} &\eqdef& \Syncbi{C};\Syncbi{D}\\
\Syncbi{\ifc{E}{C}{D}} &\eqdef& \ifcbi{E\smallSplitSym E}{\Syncbi{C}}{\Syncbi{D}} \\
\Syncbi{\whilec{E}{C}} &\eqdef& \whilecbiA{E\smallSplitSym E}{\False\smallSplitSym \False}{\Syncbi{C}}
\\
\Syncbi{\letcom{m}{B}{C}} &\eqdef& \letcombi{m}{\splitbi{B}{B}}{ \Syncbi{C}} \\
\Syncbi{\varblock{x\scol T}{C}} &\eqdef& \varblockbi{x\scol T | x\scol T}{ \Syncbi{C} }
\end{array}
\)
\end{footnotesize}
\caption{Full alignment.}\label{fig:fullAlign}
\end{wrapfigure}

In general a biprogram may admit several possible weavings.
For the form 
$\splitbi{C}{C}$ relating $C$ to itself there is a biprogram that is maximal in 
the sense that it allows to reason about two executions aligned in lockstep.
We write $\Syncbi{C}$  for the \dt{full alignment} defined in Figure~\ref{fig:fullAlign}.
Apropos linking, we have
$\Splitbi{ \letcom{m}{B}{C} }{ \letcom{m}{B'}{C} }
\weave^* \letcombi{m}{\splitbi{B}{B'}}{ \Syncbi{C} }$.
Full alignment plays a key role in deriving the relational modular linking rule that was sketched as (\ref{eq:mismatchR})
and is formalized in Figure~\ref{fig:derivedmismatch}.
\begin{restatable}{lemma}{lemBiprojections}
  \label{lem:biprojections}
  \upshape
  $\splitbi{\Left{CC}}{\Right{CC}}\weave^* CC$ for any $CC$.
\end{restatable}
As a corollary, we have  $\splitbi{C}{C}\weave^* \Syncbi{C}$ for any $C$, because
$\Left{\Syncbi{C}} \equiv \Right{\Syncbi{C}} \equiv C$.

\paragraph{Sumpub: illustrating conditionally aligned loops.}

%[say something about not wanting to say 'synchronous/asyn'] conditionally-aligned loop?

For the \whyg{tabulate} example it is effective to reason
by aligning all iterations of the two loops in lockstep.
This is not the case for program (\ref{eq:sumpub}) in Section~\ref{sec:modrel},
recalled here.
\[ %begin{equation}\label{eq:sumpub}
sumpub: \qquad \mbox{\lstinline{s:=0;  p:=head;  while p <> null do if p.pub then s:=s+p.val fi;  p:=p.nxt od}}
\]
It sums the elements of a list that are flagged public.
It has an information flow property: the output, in variable $s$, depends only on the public elements of the input list.
(This can be viewed as a declassification or as a value-dependent classification~\cite{AmtoftB07}.)
Typically such properties are expressed using a precondition of agreement on some expression which in this case
should denote ``the public elements of the input list''.

As a pointer structure, the list can have cycles, so care needs to be taken in
defining predicates and functions.
In the \whyg{tabulate} example we choose specs that do not involve inductively defined predicates or relations.
Here, we inductively define a predicate $listpub(p,ls)$
that says $ls$ is the list of values of the public elements in a null-terminated list from $p$.
\[ \begin{array}{lcl}
p = null & \imp &  listpub(p, []) \\
p \neq null \land \neg p.pub \land listpub(p.nxt, ls)  & \imp & listpub(p, ls) \\
p \neq  null \land p.pub \land p.val = h \land listpub(p.nxt, ls) & \imp &  listpub(p, h::ls)
\end{array}
\]
We consider the following relational spec, eliding the frame condition for clarity.
The bound variables, $ls,ls'$ are of the math type \whyg{int list}.
\[ \rflowtyf{
   \some{ls: \code{int}~\code{list} \mid ls':\code{int}~\code{list} }{
     \leftF{listpub(head,ls)}\land \rightF{listpub(head,ls')}\land ls\eqbi ls'}}{\Agr s}
\]
The syntax of quantifiers in relation formulas explicitly designates left- and right-side variables, which is important in case of reference or region type (since the values must be allocated in the respective states).
There is no need to use distinct names here, so we can use a more succinct precondition
for the spec:
\( \some{ls|ls}{ \Both{(listpub(head,ls))} \land \Agr ls} \).

%Spec-only variables are scoped over the pre- and post-condition, so a logically equivalent spec has this form, where $ls$ is spec-only:
% \( \rflowtyf{\Both{(listpub(head,ls))} \land \Agr ls}{\Agr s} \).
%BUT WE CAN'T since spec-only allowed only in snapshots

We want to prove that $\splitbi{sumpub}{sumpub}$ satisfies the relational spec.
One way is to first prove unary judgment
$sumpub: \flowtyf{listpub(p,ls)}{s=sum(ls)}$, again treating $ls$ as spec-only, and thus universally quantified over the spec.
A simple embedding rule (\rn{rEmb} in Figure~\ref{fig:proofrulesR}) lifts this to
$\splitbi{sumpub}{sumpub}: \rflowtyf{\Both (listpub(p,ls))}{\Both(s=sum(ls))}$.
The relational frame rule lets us conjoin agreement on $ls$, to get
\[\splitbi{sumpub}{sumpub}: \rflowtyf{\Both(listpub(p,ls))\land\Agr ls}{\Both( s=sum(ls))\land \Agr ls} \]
The postcondition implies $\Agr s$, so we complete the proof using the relational consequence rule.

Lifting unary judgments is an important pattern of reasoning and is satisfactory for reasoning about assignment commands including those in the \whyg{tabulate} example.
But $sumpub$ has a loop, so this argument comes at the cost of proving functional correctness,
i.e., the judgment $sumpub: \flowtyf{listpub(p,ls)}{s=sum(ls)}$.
Finding a loop invariant is not difficult in this case, but it would be if sum is replaced by a 
sufficiently complex computation.

There is an alternative proof of the relational spec that avoids functional correctness,
using for the loops a simple relational invariant:
\begin{equation}\label{eq:sumpub:inv}
\some{xs|xs}{ \Both{(listpub(p,xs))} \land \Agr xs \land \Agr s}
\end{equation}
We verified the example using \WhyRel, and instead of asking the solvers to handle the existential we used the standard technique: $xs$ on each side is a ghost variable,
initialized based on the precondition and explicitly updated as appropriate.

The point of this example is that this simple invariant only suffices if we align the iterations
judiciously.
In case $p.pub$ holds on both left and right, we take a lockstep iteration, i.e., both sides
execute the loop body, and it is straightforward to show the invariant holds afterwards
using the last clause in the definition of $listpub$ and the fact that $\Agr xs$, i.e., equality of the mathematical lists, implies agreement on their tails.
If $pub$ is true on one side but not the other, lockstep iteration does not preserve (\ref{eq:sumpub:inv}).
However, if $p.pub$ is false on the left, $listpub(p,xs)$ implies $listpub(p.nxt,xs)$, and
executing the body just on the left maintains the relation (\ref{eq:sumpub:inv}).
Notice (\ref{eq:sumpub:inv}) does not include agreement on $p$; indeed the precondition requires no agreement on references.
\emph{Mutatis mutandis on the right side.}
To express this reasoning, we weave $\splitbi{sumpub}{sumpub}$ to this biprogram:
%\begin{lstlisting}
%       ( s := 0; p := head | s := 0; p := head );
%       while (p <> null | p <> null) . *<| ~ p.pub *<] | [> ~ p.pub |> do
%          ( if p.pub then s := s+p.val fi; p := p.nxt | if p.pub then s := s+p.val fi; p := p.nxt ) od
%\end{lstlisting}
\begin{equation}\label{eq:bi-sumpub}
\begin{array}{l}
\Splitbi{ s := 0; p := head }{ s := 0; p := head }; \\
\bWHILE\ p\neq \NULL \mid p\neq \NULL \ \,.\, \ \leftF{\neg p.pub} \mid \rightF{\neg p.pub} ~ \bDO \\
\qquad ( ~ \bIF\ p.pub\ \bTHEN\ s := s+p.val\ \bFI; p := p.nxt \\
\qquad | ~ \bIF\ p.pub\ \bTHEN\ s := s+p.val\ \bFI; p := p.nxt ~ ) ~ \bOD
\end{array}
\end{equation}
Although the program is being related to itself, we do not bother to fully align the initialization or loop body:
these do not involve allocation or method calls, so reasoning about those parts of the code is straightforward.
For this reason, some uses of sync in Figure~\ref{fig:tabulate}(c) could as well be bi-coms.
What is important is to use a bi-while.
For loop alignment guards we choose the relation formulas
$\leftF{\neg p.pub}$ and $\rightF{\neg p.pub}$.
The alignment guards are used in the proof rule for bi-while,
which has the following form.
\begin{equation}\label{eq:rWhileSimp}
\inferrule
{
\proves CC: \rflowtyf{\Q\land\neg\P\land\neg\P'\land\leftF{E}\land\rightF{E'}}{\Q}
\\
\proves \splitbi{\Left{CC}}{\skipc} :
 \rflowtyf{\Q\land\P\land\leftF{E}}{\Q}
\\
\proves \splitbi{\skipc}{\Right{CC}} :
\rflowtyf{\Q\land\P'\land\rightF{E'}}{\Q}
\\
\Q\imp E\eqbi E' \lorbi (\P\land\leftF{E}) \lorbi (\P'\land\rightF{E'})
}{
\proves \whilecbiA{E\smallSplitSym E'}{\P\smallSplitSym \P'}{CC} :
\rflowtyf{\Q}{\Q\land\leftF{\neg E}\land\rightF{\neg E'}}
}
\end{equation}
\emph{This rule has omissions!}
For clarity we omit details not relevant to the current discussion:
frame conditions, hypothesis context, and side conditions that enforce encapsulation and immunity.
The encapsulation condition is discussed later and is lifted from the unary logic,
as is \emph{immunity}, a technical condition needed for stateful frame conditions
(adapted unchanged from \RLI).

In the rule, $\Q$ is the relational loop invariant, like (\ref{eq:sumpub:inv}) in the example.
The three premises cover a lockstep iteration, a left-side iteration, and a right-side iteration.
The one-sided iterations are expressed using the syntactic projection metafunctions (Figure~\ref{fig:synProj})
to obtain unary commands.
In the example the two projections of the loop body are the same, namely
\whyg{if p.pub then s := s+p.val; fi; p := p.nxt}.
In each premise the invariant must be preserved, but each has a strengthened precondition
based on the alignment guards.
For the example, the first premise applies when both sides are at a public element.
The second (resp.\ third) premise applies when the element on the left (resp.\ right) is not public.
Besides alignment guards, the premises include the loop tests in the usual way, as does the conclusion of the rule.

The side condition, $\Q\imp E\eqbi E' \lorbi (\P\land\leftF{E}) \lorbi (\P'\land\rightF{E'})$,
ensures that for any initial states satisfying $\Q$, at least one of the three premises is applicable.
The reader can confirm that the side condition holds in the example, and thus the rule can be
used to carry out the proof as described.

As another example, for \whyg{tabulate} in Figure~\ref{fig:tabulate}(c) we use false alignment guards, so the one-sided premises hold trivially and the side condition simplifies to the implication mentioned earlier: the invariant implies agreement on loop tests.
That is, $i-1 \eqbi i \land \Agr n \imp i<n \eqbi i \leq n$.

The biprogram syntax allows $\P$ and $\P'$ to be relation formulas, but it happens that in the example
$\leftF{\neg p.pub}$  only constrains the left state and the other alignment guard constrains the right state.
As stated in Section~\ref{sec:progtype}, $\P$ and $\P'$ are not allowed to have agreement formulas;
it is not evident what refperm would be used to interpret agreements in such a context.

\subsection{Relational reasoning with hiding and encapsulation}\label{sec:eg:encapR}

% find's rel spec is in examples/Kruskal/ufrel.rl but we're not going to include it

Having illustrated general relational reasoning (Sects.~\ref{sec:eg:relverif} and~\ref{sec:weave})
and the use of dynamic framing for encapsulation in unary reasoning (Section~\ref{sec:eg:encap}),
we now illustrate encapsulation in relational reasoning.
In doing so we sketch how requirements (E1)--(E4) adapt to the relational setting.
%pertain to coupling relations as well as to private invariants. 

In Section~\ref{sec:eg:encap} we considered the verification of a client linked with
a quick-find implementation of \whyg{UnionFind}, hiding the private invariant.
Here we consider two implementations of that interface and consider a more interesting client: an implementation, $MST$, of Kruskal's minimum spanning tree algorithm.
For a second implementation of \whyg{UnionFind} we consider the quick-union data structure~\cite{SedgewickWayne}.

The goal is to prove a relational property:
equivalence of the two programs made by linking $MST$ with the two module implementations.
To do so we use relational modular linking, as sketched in the rule (\ref{eq:mismatchR}),
hiding a coupling relation between the two implementations which includes their private invariants.
To use the rule we do the following.
\begin{list}{}{}
\item[(i)] Prove a unary judgement for $MST$, with the \whyg{UnionFind} specs in context.
As explained in Section~\ref{sec:eg:encap}, this ensures that $MST$ respects the boundary of \whyg{UnionFind}, as per requirement (E3).

\item[(ii)] Define a coupling relation $\M_{uf}$ to connect the encapsulated data structures of the two implementations of
\whyg{UnionFind}.  Show that it is framed by the dynamic boundary, as per requirement (E2),
and follows from the $MST$ precondition, as per (E4).

\item[(iii)] For the two bodies $B,B'$ that provide alternate implementations of \whyg{find},
prove a relational judgment for $\splitbi{B}{B'}$ (and likewise for the implementations of \whyg{union}). 
The specification should express local equivalence, but with $\M_{uf}$ conjoined to the pre- and postcondition.
\end{list}
It then follows that the two linkages satisfy a local equivalence property,
specifically a relational spec that is derived by a general construction from the unary spec of $MST$.
Similar to the relational spec of \whyg{tabulate} in Section~\ref{sec:eg:relverif},
it requires agreement on inputs and ensures agreement on outputs.
But encapsulation must be taken into account: the two linkages will be equivalent in terms of client-visible inputs and outputs, but the encapsulated data structures are different.
More on this later.

For item (i), we choose $MST$ for the sake of a nontrivial example, but we do not use a functional correctness spec, i.e.,
we do not specify that it produces a minimum spanning tree.
All we need is a precondition under which $MST$ does not fault, and a frame condition.
The global variables of $MST$ are $g$ of type \whyg{Graph} and $es$ of type \whyg{List}.
For simplicity, $g$ is an abstract mathematical graph; $es$ references a list like that used in Section~\ref{sec:eg:relverif}.
The graph interface provides an enumeration of edges and $MST$ produces, in $es$, a list of edge numbers for edges in the spanning tree.
\begin{equation}\label{eq:MSTspec}
 \flowty{\code{numVerts}(g) > 0 \land pool = \emptyset}
          {\True}
          {\rd{g};  \rw{es, \lloc, pool, (pool\union pool\Img rep)\Img \allfields } }
\end{equation}
Note that the effects here include effects produced by call to \whyg{UnionFind} methods.
We verify the judgment
$\Phi_{uf} \proves_\emptymod MST: spec$ where $spec$ is (\ref{eq:MSTspec})
and $\Phi_{uf}$ has the public specs of \whyg{find} and \whyg{union}, i.e., without the private invariants.
The current module is $\emptymod$, the default module with empty boundary.

The local equivalence spec for the two linked programs is derived,
by a general construction called $\locEq$,  based on the frame condition of a unary spec,
and the dynamic boundaries of the modules in scope.  In the example there is just one module with a nontrivial boundary, \whyg{UnionFind}; math modules like \whyg{Graph} have empty boundaries.
Agreements in the precondition are derived directly from the read effects and boundary,
using the effect subtraction operator that excludes from agreement the encapsulated locations.
In this example, the relational precondition is
\[ \Both{(\code{numVerts}(g) > 0 \land pool = \emptyset)} \land
   \Both{(s_{\lloc}=\lloc)} \land
   \Agr es
\]
The conjunct $\Both{(s_{\lloc}=\lloc)}$ introduces snapshot variable $s_{\lloc}$ to be used in the postcondition to express freshness.
The agreement $\Agr es$ is in simplified form.  The general construction
takes the read effect, $\rd{es, \lloc, pool, (pool\union pool\Img rep)\Img \allfields}$
and subtracts the boundary $\rd{pool, (pool\union pool\Img rep)\Img \allfields}$
and $\lloc$, which results in the effect
$\rd{es, ((pool\union pool\Img rep)\setminus(pool\union pool\Img rep))\Img \allfields}$
which trivially simplifies to $\rd{es, \emptyset \Img \allfields}$ and then to $\rd{es}$.

What about agreements for a postcondition? In general a command may write preexisting locations and allocate new ones.  In this case the only preexisting locations that are writable are the variables $es$ and $\lloc$,
so the postcondition includes $\Agr es$.
(In general, to handle writable heap locations the general definition of $\locEq$ uses snapshots of the relevant expressions in write effects; for details see Section~\ref{sec:locEq}.)
To handle fresh locations, $\locEq$ uses the snapshot $s_{\lloc}$ in the way described in Section~\ref{sec:eg:encap}: the fresh references are $\lloc\setminus s_{\lloc}$ so the fresh locations
are $(\lloc\setminus s_{\lloc})\Img\allfields$.  Again, effect subtraction is used to exclude $\lloc$ and the boundary.  The resulting agreement is
$\Agr ( (\lloc\setminus s_{\lloc}) \setminus (pool\union pool\Img rep))\Img \allfields$.

In summary the local equivalence spec that we get from (\ref{eq:MSTspec}) for $MST$ is
\begin{equation}\label{eq:locEqMST}
\begin{array}{l}
\Both{(\code{numVerts}(g) > 0 \land pool = \emptyset}) \land \Both{(s_{\lloc}=\lloc)} \land \Agr es  \\
\rspecSym
\later ( \Both(\True) \land \Agr es \land
\Agr ( (\lloc\setminus s_{\lloc}) \setminus (pool\union pool\Img rep))\Img \allfields  ) \; [\ldots]
\end{array}
\end{equation}
If one simply wants to know that the new and old versions of the program are the same, aside from encapsulated state, this is enough.  By construction, the $\locEq$ spec requires agreement on what
the program can read and ensures agreement on its results.

In this particular case, to obtain a more explicit postcondition that refers to the list constructed,
we can do as follows.  First, strengthen the unary postcondition from $\True$ to
something like $es.head\in es.nds \land es.nds\Img next\subseteq es.nds \land (\sing{es}\union es.nds)\subseteq(\lloc\setminus s_{\lloc})$ which expresses the closure of 
$nds$ and the freshness of the list
(see Section~\ref{sec:eg:relverif}).  
The relational spec (\ref{eq:locEqMST}) then changes to have these conditions in place of $\True$.
Then using the rule of consequence and reasoning about sets,
we get $\Agr es.nds\Img next$ and $\Agr es.nds\Img val$ much like in the \whyg{tabulate} example.

For item (ii), as expected since Hoare`72, the coupling relation $\M_{uf}$ conjoins
a relational formula that connects the two implementations, together
with the two private invariants.
In particular, $\M_{uf}$ is $\leftF{I_{qf}} \land \rightF{I_{qu}}\land\ldots$,
where $I_{qf}$ is the invariant discussed in Sec.~\ref{sec:eg:encap},
and $I_{qu}$ is the private invariant of the quick-union implementation.
The two implementations have similar internal data structure, in the sense that both use an array to represent
an up-pointing tree, but quick-find and quick-union manipulate the tree quite differently.
To specify the connection between the two data structures,
the third conjunct of $\M_{uf}$ is this formula:
%% \begin{lstlisting}
%%   Agr pool /\
%%   forall u:Ufind in pool|u:Ufind in pool. Agr u ==> eqPartition(u.part | u.part)
%% \end{lstlisting}
\begin{equation}\label{eq:UFagree}
\Agr pool \land \all{u:\code{Ufind}\in pool |  u:\code{Ufind}\in pool }{
                     \Agr u \imp eqPartition(\leftex{u.part},\rightex{u.part}) }
\end{equation}
This says the two pools are in agreement, and for corresponding elements $u$ in the pool,
the abstract partition $u.part$ on the left side is an equivalent partition to the one on the right.
This means they have the same blocks.  % and in addition the same representatives.
This coupling uses a common idiom. The coupling relation is defined using a mathematical abstraction:
the two data structures are related if they have the same abstraction.
This idiom is especially suitable if the two data structures are very different.
By contrast, in our two implementations of \whyg{PQ}
we consider two similar pointer structures and for their coupling we use agreement formulas to describe fine-grained correspondence between the two pointer structures; see Example~\ref{ex:PQagree}.

To show that $\M_{uf}$ is framed by the boundary, the technique is essentially the same as for unary framing of an invariant
(Section~\ref{sec:eg:encap}).  The difference is that here we consider a pair of states that satisfy $\M_{uf}$,
and a second pair where the two left (resp.\ right) states agree on locations within the boundary, to show the second pair satisfies $\M_{uf}$.  Given a suitable representation of states, as in our prototype, the implication is easily checked by SMT solvers.

The last part of item (ii) is that $\M_{uf}$ is implied by the precondition of the client spec, in this case
(\ref{eq:MSTspec}).  To be precise, it is an implication at the level of relations:
\( \Both{(numVertices(g) > 0 \land pool = \emptyset)} \imp \M_{uf} \).
It holds owing to $pool=\emptyset$.

For item (iii), for each method we verify the local equivalence spec derived from the method's unary spec,
with $\M_{uf}$ conjoined to pre- and postcondition.
For example, the frame condition of \whyg{union} is
$[\rw{(\sing{\mself}\union\mself.rep)\Img\allfields}]$,
and its parameters are $\mself,x,y$.
Based on this, $\locEq$ uses a precondition based on the agreement
$\Agr \mself \land \Agr x \land \Agr y  \land \Agr (\sing{\mself}\union\mself.rep)\Img\allfields$.
A snapshot variable $s$ is used in precondition
$\Both{s = \sing{\mself}\union\mself.rep}$ so the postcondition can express
agreement on writables by $\Agr s\Img\allfields$, in addition to agreement on fresh locations as described for $MST$.  Recall that $\locEq$ then subtracts locations within the boundary; it is not agreement that we want for those locations, but rather the connection expressed by $\M_{uf}$.

The implementations of \whyg{union} and \whyg{find} are fairly different.
For quick-find, the union operation eagerly updates ``parents'' so find takes
constant time. % (as implemented in Sect.~\ref{sec:eg:encap}).
For quick-union, find has to traverse multiple parents to reach the representative element.
To prove the relational judgments for the method bodies,
we use biprograms that are not tightly woven.
The corresponding implementations are not very similar and are not making external calls or doing allocation,
so there is little motivation for close alignment the way there is for the \whyg{tabulate} example.

More details about the $MST$ verification can be found in Section~\ref{sec:cases}.
For now we review why relational modular linking ---shown in (\ref{eq:mismatchR}) and formalized in rule \rn{rMLink} in Figure~\ref{fig:derivedmismatch}--- is sound.
In other words, why do (i)--(iii) suffice to prove equivalence of the linkages?
Intuitively, the coupling is preserved by client steps owing to encapsulation, just like private invariants in the unary case.
This is formalized by a relational version of the \rn{SOF} rule, called \rn{rSOF}.
For that rule to be sound, the client needs to be aligned so that context calls can be sync'd (like the call to $mf$ in the \whyg{tabulate} example) so a relational spec can be used ---namely a local equivalence spec conjoined with the coupling relation.
So rule \rn{rSOF} applies to the full alignment of some command,  and its premise is that this fully aligned biprogram satisfies
a local equivalence spec.
This we obtain from the unary judgment of (i), by a rule which lifts a unary judgment
to a relational one for the local equivalence derived from the unary spec
(rule \rn{rLocEq} in Figure~\ref{fig:proofrulesR}).
It relates the command to itself, expressing the dependency property of its read effect
as a relational judgment.

\paragraph{Notations to conjoin couplings.}

To conclude this section, we define a metafunction that conjoins a relation to a relational spec; this is used
to formulate \rn{rSOF} and the modular linking rule.
It is based on a similar metafunction, \ghostbox{$\conjInv$}, which applies to a unary spec and a unary invariant $I$:
\begin{equation}\label{eq:conjInv}
(\flowty{R}{S}{\effe})\conjInv I \; \eqdef \; \flowty{R\land I}{S \land I}{\effe}
\end{equation}
This lifts to an operation on unary contexts, written $\Phi\conjInv I$,
by mapping $\conjInv I$ over the specs in $\Phi$.

For relation formula $\M$, the operation $\conjInv\M$ conjoins $\M$ to a relational spec.
The operation only applies to relational specs in the \dt{standard form},
meaning that $\later$ occurs only outermost on the postcondition, or not at all.

\begin{definition}[\textbf{conjoin coupling} \ghostbox{$\conjInv \M$}]\label{def:conjInv}
If $\R$ and $\S$ are $\later$-free then
\[ \begin{array}{l}
(\rflowty{\R}{\later\S}{\effe})\conjInv \M
\; \eqdef \;
\rflowty{\R\land \M}{\later(\S\land \M)}{\effe} \\
(\rflowty{\R}{\S}{\effe})\conjInv \M
\; \eqdef \;
\rflowty{\R\land \M}{\S\land \M}{\effe}
\end{array}
\]
For context $\Phi$, let $\Phi\conjInv \M$  conjoin $\M$ to the specs in $\Phi_2$ and
for the unary specs give $\Phi_0\conjInv \Left{\M}$ and $\Phi_1\conjInv \Right{\M}$.
In other words, $(\Phi_0,\Phi_1,\Phi_2)\conjInv\M$ is
$(\Phi_0\conjInv\Left{\M}, \, \Phi_1\conjInv\Right{\M}, \, \Phi_2\conjInv\M)$.
\end{definition}
Note that $\Phi\conjInv \M$ is only defined if the specs in $\Phi_2$ are in standard form,
and then so is the result.

\section{Semantics of programs and unary correctness}\label{sec:unarySem}

For a correctness judgment $\Phi\HPflowtr{\Gamma}{M}{P}{C}{Q}{\eff}$,
an informal sketch of the semantics is given preceding Def.~\ref{def:wfjudge}.
To make it precise we use transition semantics, so we can formulate the semantics of encapsulation
in terms of the module in which a given step is taken, initially module $M$.
To express modular correctness with respect to assumed specs,
a context call makes a single step to the result of the call, given by a \dt{context model} $\phi$
which provides denotations that satisfy the specifications of the hypothesis context $\Phi$.
Transitions go to \dt{fault}, $\Fault$, in case of runtime failure (null dereference).
Fault is also used to represent precondition violation in context calls.\footnote{One could distinguish between these two kinds of faults using different tokens, as done in RLII.  Here we would need a third kind, for alignment fault.  But the correctness judgments disallow all three kinds, so for simplicity we conflate them.}

A \dt{pre-model} provides method denotations that do not necessarily satisfy specs;
the transition relation $\trans{\phi}$ is defined for any pre-model $\phi$.

For readers familiar with O'Hearn et al~\cite{OHearnYangReynoldsInfoToplas} or \RLII, 
we note that unlike those works here we cannot use a single ``most nondeterministic'' denotation.
We need context models to be quasi-deterministic, in accord with the $\forall\forall$-interpretation
of relational correctness for deterministic programs.

This section spells out the details, which are somewhat intricate.
The most important and novel part is the semantics of encapsulation,
a condition called Encap in the semantics of correctness judgments (Def.~\ref{def:valid}).
Some readers may wish to skip to Section~\ref{sec:unaryLog},
after skimming Sects.~\ref{sec:states} and~\ref{sec:effect}.

\subsection{States, expressions, method environments and configurations}
\label{sec:states}

Assume given an infinite set $\Refset$ of references, disjoint from the integers, with distinguished element $\semNull$.
A \dt{$\Gamma$-state} comprises a finite heap
and a type-respecting assignment of values to the variables in $\Gamma$.
We confine attention to contexts $\Gamma$ that include the special variable $\lloc$.
We write $\sigma(x)$ to look up the value of $x$ in state $\sigma$.
In particular, $\sigma(\lloc)$ is the finite set of allocated references.
Any reference $o\in\sigma(\lloc)$ has a class $K$, which we write as $\type(o,\sigma)$.\index{$\type$}

A \dt{location} is either a variable $x$ or a
\dt{heap location} $o.f$, where we write $o.f$ for the pair
$(o,f)$ of a non-null reference $o$ and field name $f$.
For any state $\sigma$, define the set of its locations by
\[ \locations(\sigma) \eqdef \Vars{\sigma} \union \{o.f \mid o\in\sigma(\lloc)\land f\in\fields(\type(o,\sigma)) \} \]
The heap provides a type-respecting assignment of values to heap locations.
We write $\sigma(o.f)$ for the value of field $f$ of allocated reference $o$.
Type-respecting means that if $\type(o,\sigma)$ is $K$ and $f:T$ is in $\fields(K)$
then $\sigma(o.f)$ is in $\means{T}\sigma$.
We write $\means{T}\sigma$ for the values of type $T$ in state $\sigma$.
In the case of a reference type $K$,  define $\means{K}\sigma$ by 
\[ \means{K}\sigma \eqdef \{\semNull\} \union \{ o\in\sigma(\lloc) \mid \type(o,\sigma)=K\} \]
% the set $o\in\lloc$ with $\type(o,\sigma)=K$, together with $\{\semNull\}$.
Define $\means{\Region}\sigma$ to be
$\powerset(\sigma(\lloc) \union \{\semNull\})$.
We write $\means{\Gamma}$ for the set of $\Gamma$-states.

The transition semantics of a command typed in $\Gamma$ may introduce additional variables for local blocks, so it is convenient to define $\Vars{\sigma}$ to be the variables of the state.
%\footnote{Written $\mathconst{Dom}$ in previous RL papers.}
We write $\extend{\sigma}{x}{v}$
to extend the state with additional variable $x$ with value $v$,
and $\update{\sigma}{x}{v}$
to override the value of $x$ that is already in $\Vars{\sigma}$.
We write $\drop{\sigma}{x}$ to remove $x$ from the domain of $\sigma$.

\begin{figure}[t!]
\begin{small}
\(
\begin{array}{llll}
\sigma(E_1 \otimes E_2) &\eqdef& \sigma(E_1) \otimes  \sigma(E_2)
       \quad\mbox{where $\otimes$ is in $\{=,\leq,+,\dots\}$}
\\[.2ex]
\sigma(\sing{E}) & \eqdef & \{ \sigma(E) \}  \\[.2ex]
\sigma(\Emp)    &\eqdef& \emptyset \\[.2ex]
\sigma(G_1 \otimes G_2) &\eqdef&  \sigma(G_1) \otimes \sigma(G_2)
  \quad\mbox{where $\otimes$ is in $\{\cup,\cap,\setminus\}$}
\\[.2ex]
\sigma(G/K) & \eqdef &  \{ o \mid o\in \sigma(G) \land o\neq\semNull \land \type(o,\sigma) = K \}
\\[.2ex]
\sigma(G\Img f) & \eqdef & \emptyset \quad\mbox{if $f\scol\INT$ (or any primitive type)} \\[.2ex]
& \eqdef &  \{\sigma(o.f) \mid o\in \sigma(G) \land o\neq\semNull \land \type(o,\sigma) = \Class(f)  \}
\quad\mbox{if $f\scol K$ for some $K$} \\[.2ex]
& \eqdef &   \bigcup \{\sigma(o.f) \mid  o\in \sigma(G) \land o\neq\semNull\land\type(o,\sigma) = \Class(f) \}
\quad\mbox{if $f\scol \Region$}
\end{array}
\)
\end{small}
%\hrule
%\vspace{-1ex}
\caption{Semantics $\sigma(F)$ of selected program and region expressions (r-values), for state $\sigma$.}
\label{fig:rexpsem}
\end{figure}

We write $\sigma(F)$ for the value of expression $F$.
The semantics of program expressions $E$ and region expressions $G$ is in Figure~\ref{fig:rexpsem}.
(To be very precise, the semantics of expressions is defined on a typing $\Gamma\proves F:T$,
such that $\sigma(F)$ is in $\means{T}\sigma$.)
The syntax is designed to avoid undefinedness.
We are not formalizing arithmetic operators that can fail, there are no dangling pointers,
and program expressions $E$ do not depend on the heap.  Region expressions can depend on the heap, in the case of images $G\Img f$, and they are defined in any state.
If $f\scol K$ for some $K$,
then $\sigma(G\Img f)$ is the set of values of the $f$ fields of objects in $\sigma(G)$.
If $f\scol \INT$ then $\sigma(G\Img f)$ is empty.
Finally, for $f\scol\Region$,
$\sigma(G\Img f)$ is the union of the regions $\sigma(o.f)$ for $o$ in $\sigma(G)$.

Transitions relate configurations of the form $\configm{C}{\sigma}{\mu}$.
The \dt{environment} $\mu$ maps method names to commands.
The empty environment is written $\_$.
In a configuration, the command $C$ may include the pseudo-commands:
$\Endcall(m)$ ends the code of a call to method $m$,
$\Endvar(x)$ ends the scope of a local variable,
and $\Endlet(\ol{m})$ ends the scope of some methods $\ol{m}$
(arising from simultaneous binding $\letcom{\ol{m}}{\ol{B}}{C}$).
The pseudo-commands do not occur in source programs. % and they have their own transition rules.
The code of a configuration thus takes a form that represents the execution stack for environment calls:
\[
 C_n;\Endcall(m_n);\ldots;C_1;\Endcall(m_1);C_0
\quad \mbox{where $n\geq 0$ and each $C_i$ is $\Endcall$-free.}
\]
So the leftmost command $C_n$ is on top of the stack and $m_n$ is the leftmost environment call.
We write $\Active(C)$ for the \dt{active command}\index{$\Active$}
(which one might call the redex),
i.e., the unique sub-command that gets rewritten by the applicable transition rule.\footnote{We
identify  sequentially composed commands up to associativity (Figure~\ref{fig:synident})
so $\Active(C)$ can be defined as the leftmost non-sequence command of a sequence.}
For example, $\Active(x:=0;y:=1)$ is $x:=0$.

To formalize the semantics of encapsulation we need to refer to the module of the active command:
it must stay outside the boundary of every module except its own.
So we define the \dt{top module} $\topm(C,M)$ to be $N$ where $N=\mdl(m_n)$ and $m_n$ is the leftmost environment call (see above), or $M$ if $C$ has no $\Endcall$ (i.e., $n=0$).
This is used in Def.~\ref{def:valid} where the argument $M$
is from the judgment under consideration.
In Def.~\ref{def:valid} we also write $N\in(\Phi,\mu)$,
for hypothesis context $\Phi$ and method environment $\mu$,
to mean there is $m\in\dom(\Phi)\union\dom(\mu)$ with $\mdl(m)=N$.

For an empty method context, the transition relation is standard
(Figure~\ref{fig:trans}). % and unchanged from \RLII.
For non-empty contexts the transition relation depends on a pre-model,
which is defined in terms of the semantics of specs, to which we proceed.

\subsection{Semantics of state predicate formulas and effects}
\label{sec:effect}

Satisfaction of formula $P$ in state $\sigma$ is written $\sigma\models P$.
The semantics of formulas is standard and two-valued.  % (\RLI).
The points-to relation $x.f=E$ is defined by
\( \sigma\models x.f=E \mbox{ iff }
\sigma(x)\neq\semNull \mbox{ and } \sigma(\sigma(x).f)=\sigma(E)
\).
The type predicate is defined by $\sigma\models\Type(G,\ol{K})$ iff $\type(o,\sigma)\in\ol{K}$ for all $o \in\sigma(G)$.
Quantifiers for reference types range over allocated (thus non-null) references:
$\sigma\models \all{x:K}{P}$ iff $\extend{\sigma}{x}{o}\models P$ for all $o\in\sigma(\lloc)$
with $\type(o,\sigma)=K$.

\begin{lemma}[unique snapshots]
\label{lem:uniquespeconly}
\upshape
If $P,\Gamma,\hat{\Gamma}$ satisfy the condition for precondition $P$ in
Def.~\ref{def:wfspec} then for all $\Gamma$-states $\sigma$ there is at most one
$(\Gamma,\hat{\Gamma})$-state $\hat{\sigma}$
that extends $\sigma$ such that $\hat{\sigma}\models P$.
\end{lemma}

In contexts where we consider a precondition $P$ and suitable state $\sigma$,
we adopt the \dt{hat convention} of writing $\hat{\sigma}$ for the
extension of $\sigma$ uniquely determined by $\sigma$ and $P$
as in Lemma~\ref{lem:uniquespeconly}.
% not needed?:
%Furthermore, in contexts where the names $\ol{s}$ and values $\ol{v}$ of spec-only variables are not important, we abuse notation and write $\sigma\models P$
%to abbreviate ``there exist $\ol{v}$ such that $\extend{\sigma}{\ol{s}}{\ol{v}}\models P$''.
%Hence  we write $\sigma\not\models P$ to say no such values exist.

For an effect $\eff$ in a given state $\sigma$, its read effects designate a set $\rlocs(\sigma,\eff)$ of locations.  Specifically, it is the set of l-values of the left-expressions in its read effects:
\[ \rlocs(\sigma,\eff) \eqdef
\begin{array}[t]{l}
  \{ x \mid \mbox{$\eff$ contains $\rd{x}$} \} \; \union \\
  \{ o.f \mid \mbox{$\eff$ contains some $\rd{G\Img f}$ with $o\in \sigma(G)$, $o\neq\semNull$,
         $f\in\fields(\type(o,\sigma))$ } \}
\end{array}
\]
\index{$\wlocs$}\index{$\rlocs$}%
Define \ghostbox{$\wlocs(\sigma,\eff)$} the same way but for the l-values in write effects.
Note that for an effect of the form $\rd{G\Img f}$ the definition of $\rlocs$ uses
the r-value $\sigma(G)$ (Figure~\ref{fig:rexpsem}) where $G$ may itself involve images.
These functions are used in the key lemma about effect subtraction
(see (\ref{eq:effsubtract})).

\begin{restatable}[subtraction]{lemma}{lemeffectSubtract}
\label{lem:effectSubtract}
\upshape
%For any effects $\eff,\effe$ we have
\(
\rlocs(\sigma, \eff\setminus\effe) = \rlocs(\sigma, \eff) \setminus \rlocs(\sigma,\effe)
\)
and the same for $\wlocs$.
\end{restatable}

For use in the semantics of write effects, define
the locations of $\sigma$ that have been changed in $\tau$ as
$\index{\written}$
\[ \written(\sigma,\tau) \eqdef
\{ x \mid x\in\Vars{\sigma}\intersect\Vars{\tau} \land \sigma(x)\neq\tau(x) \}
     \union \{ o.f \mid o.f\in\locations(\sigma) \land \sigma(o.f)\neq\tau(o.f) \}
\]
This captures the variables still in scope that have been changed, together with changed heap locations.\footnote{The definitions are formulated to be applicable to intermediate states in the scope of local blocks, which introduce variables not present in the typing context of the initial command.}
Say $\tau$ \dt{can succeed} $\sigma$, written \ghostbox{$\sigma\successorTo\tau$}, provided
$\sigma(\lloc)\subseteq\tau(\lloc)$ and
$\type(o, \sigma) = \type(o, \tau)$ for all $o\in\sigma(\lloc)$.
Say $\eff$ \dt{allows change} from $\sigma$ \emph{to} $\tau$,
in symbols \ghostbox{$ \sigma \allowTo \tau \models \eff $},
iff $\sigma\successorTo\tau$ and $\written(\sigma,\tau)\subseteq \wlocs(\sigma,\eff)$.
The locations of $\tau$ not present in $\sigma$ are designated by $\freshLocs(\sigma,\tau)$.
Define $\freshRefs(\sigma,\tau) \eqdef \tau(\lloc)\setminus \sigma(\lloc)$ and
\index{$\freshRefs$}\index{$\freshLocs$}
\[ \begin{array}{l}
\freshLocs(\sigma,\tau)
\eqdef \{ p.f \mid p\in \freshRefs(\sigma,\tau)
           \land
           f\in\fields(\type(p,\tau)) \}
\union \Vars{\tau}\setminus\Vars{\sigma}
\end{array}
\]

\paragraph{Read effects and refperms.}

Read effects constrain the locations on which the outcome of a computation can depend.
Dependency is expressed by considering two initial states that agree on the values in the locations deemed readable, though the states may differ
on the values in other locations.
Agreement between a pair of states needs to take into account variation in allocation, as the relevant pointer structure in the two states may be isomorphic but involve differently chosen references.
Such variation must also be taken into account in relation formulas, as in Example~\ref{ex:PQagree}.
For use with both read effects and relation formulas, agreements are formalized using refperms,
as mentioned in Section~\ref{sec:rrl}.

Let $\pi$ range over \dt{partial bijections} on $\Refset\setminus\{\semNull\}$,
i.e., injective partial functions.
Write $\pi(p)=p'$ to express that $\pi$ is defined on $p$ and has value $p'$.
A \dt{refperm from $\sigma$ to $\sigma'$} is a partial bijection $\pi$ such that
$dom(\pi)\subseteq \sigma(\lloc)$,
$\rng(\pi)\subseteq \sigma'(\lloc)$, and
$\pi(p)=p'$ implies $\type(p,\sigma)=\type(p',\sigma')$.
Define $\rprel{p}{p'}$ to mean $\pi(p)=p'$ or $p=\semNull=p'$.
Extend
$\stackrel{\pi}{\sim}$
to a relation on integers by $\rprel{i}{j}$ iff $i=j$.
For reference sets $X,Y$, define
$\rprel{X}{Y}$ to mean that $\pi\union\{(\semNull,\semNull)\}$ restricts to a total bijection between $X$ and $Y$.
% formatting - no break here
The image of $\pi$ on location set $W$ is written $\pi(W)$ and defined for variables and heap locations by two conditions:
$ x\in \pi(W)$ iff $x \in W$,
and $o.f\in \pi(W)$ iff $(\pi^{-1}(o)).f \in W $.
In words: variables map to themselves, and a heap location $p.f$ is transformed by applying $\pi$ to the reference $p$.

Next we define notations for agreement between states.
Agreement is formalized in terms of a condition  which applies to two states together with a refperm
and a subset $W$ of the locations of $\sigma$.
The location agreement $\Lagree(\sigma,\sigma',\pi,W)$ holds just if $W$ is a set of locations of $\sigma$ and for each of these locations, the contents in $\sigma$ is the same as the contents of the location that corresponds according to $\pi$. Of course ``same as'' is modulo $\pi$, for reference values.

\begin{definition}[\textbf{agreement on a location set}, $\Lagree$]\label{def:locagreement}
\index{$\Lagree$}
For $W$ a set of locations in $\sigma$, and $\pi$ a refperm from $\sigma$ to $\sigma'$,
define %$\Lagree$ by
\[
\mbox{\ghostbox{$\Lagree(\sigma,\sigma',\pi,W)$}}
\mbox{ iff }
\all{x\in W}{ \rprel{\sigma(x)}{\sigma'(x)} } \; \land
\all{(o.f)\in W}{ o\in dom(\pi) \land \rprel{ \sigma(o.f) }{ \sigma'(\pi(o).f) }}
\]
\end{definition}
This is defined for any $W\subseteq\locations(\sigma)$.
%It is not required for $W$ to contain the value $\sigma(x)$ or $\sigma(o.f)$. %Silly comment to help a confused referee; the value's aren't locations
Agreement is monotonic in the refperm, in the sense that
\begin{equation}\label{eq:LagreeMono}
\Lagree(\sigma,\sigma',\pi,W) \mbox{ and } \pi\subseteq\rho \mbox{ implies } \Lagree(\sigma,\sigma',\rho,W)
\end{equation}
%Another direct consequence of the definition is that agreement is $\supseteq$-monotonic in the location set:
%$\Lagree(\sigma,\sigma',\pi,W)$ and $W\supseteq X$ implies $\Lagree(\sigma,\sigma',\pi,X)$.

\begin{definition}[\textbf{agreement on read effects}, $\agree$]\label{def:agreeX}
\index{$\agree$}
Let $\eff$ be an effect that is wf in $\Gamma$.
Consider $\Gamma$-states $\sigma,\sigma'$.
Let $\pi$ be a refperm.
Say that $\sigma$ and $\sigma'$ \dt{agree on $\eff$ modulo $\pi$},
\index{$\agree$}
written \ghostbox{$\agree(\sigma, \sigma', \pi, \eff)$},
iff
$\Lagree(\sigma,\sigma',\pi,\rlocs(\sigma,\eff))$.
Let $\agree(\sigma,\sigma',\eff) \eqdef \agree(\sigma,\sigma',\pi,\eff)$ where
$\pi$ is the identity on $\sigma(\lloc)\intersect\sigma'(\lloc)$.
\end{definition}
Often we use $\agree(\sigma,\tau,\eff)$ where $\sigma\successorTo\tau$, in
which case $\sigma(\lloc)\intersect\tau(\lloc)=\sigma(\lloc)$.

Agreement on location sets enjoys a kind of symmetry:
\begin{equation}\label{eq:LagreeSym}
\Lagree(\sigma,\sigma',\pi,W) \mbox{ implies } \Lagree(\sigma',\sigma,\pi^{-1},\pi(W))
\mbox{ for all $\sigma,\sigma',\pi,W$}
\end{equation}
By contrast, Def.~\ref{def:agreeX} of agreement on read effects is left-skewed, in the sense that it refers to the locations denoted by effects interpreted in the left state.
The asymmetry makes working with agreement somewhat delicate.
For example, agreement on $\rd{G\Img f}$ (modulo $\pi$) implies that $\sigma(G)\subseteq\dom(\pi)$
(by Def.~\ref{def:locagreement}),
but it does not imply $\rprel{\sigma(G)}{\sigma'(G)}$.
At a higher level there will be symmetry, for reasons explained in due course.

%this is addressed in Lemma~\ref{lem:selfframingonreads}.

\subsection{Pre-models and program semantics}\label{sec:progSem}

The transition relation depends on a pre-model $\phi$, defined below, and is written $\trans{\phi}$.
%\footnote{In \RLIII\ the term ``interpretation'' is used, but model is shorter.}
The pre-model provides semantics for context calls
and represents denotations of method bodies.
Transitions act on configurations where the environment $\mu$ has procedures
distinct\footnote{This representation takes advantage of the hygiene condition that variable and method names are never re-used in nested declarations.}
from those of $\phi$.

\begin{definition}[\textbf{state isomorphism $\RprelT{\pi}{}{}$, outcome equivalence $\RprelTS{\pi}{}{}$}]
\label{def:state-iso}
For $\Gamma$-states $\sigma,\sigma'$,
define \ghostbox{$\RprelT{\pi}{\sigma}{\sigma'}$} (read: \dt{isomorphic mod $\pi$}) to mean that refperm $\pi$ is a total bijection from $\sigma(\lloc)$ to $\sigma'(\lloc)$ and the states agree mod $\pi$ on all variables and all fields of all objects.
That is, $\Lagree(\sigma,\sigma',\pi,\locations(\sigma))$.\footnote{Which is equivalent to $\Lagree(\sigma',\sigma,\pi^{-1},\locations(\sigma'))$, in this context where
$\rprel{\sigma(\lloc)}{\sigma'(\lloc)}$.}
For $S,S'\in \powerset(\means{\Gamma}\union\{\Fault\})$,
define \ghostbox{$\RprelTS{\pi}{S}{S'}$} (read \dt{equivalent mod $\pi$}) to mean that
(i) $\Fault\in S$ iff $\Fault\in S'$;
(ii) for all states $\sigma\in S$ and $\sigma'\in S'$
there is $\rho\supseteq\pi$ such that $\RprelT{\rho}{\sigma}{\sigma'}$; and
(iii) $S = \emptyset$ iff $S' = \emptyset$.
\end{definition}
Note that item (ii) involves extensions of $\pi$,
whereas the relations $\Rprel{\pi}{}{}$ and
$\RprelT{\pi}{}{}$ involve only $\pi$ itself.

\begin{restatable}{lemma}{leminsensi}
\label{lem:insensi}
\upshape
Suppose $\RprelT{\pi}{\sigma}{\sigma'}$.
Then $\Rprel{\pi}{ \sigma(F) }{ \sigma'(F) }$, and
$\sigma\models P$ iff $\sigma'\models P$.
\end{restatable}

\begin{definition}\label{def:preinterp}
%A \dt{pre-model} for $\Gamma$ and set $X$ of method names
%is a mapping $\phi$ from $X$ such that for $m\in X$,
A \dt{pre-model} for $\Gamma$ is a mapping from some set of method names,
such that for $m\in\dom(\phi)$, $\phi(m)$ is a function of type
\( \means{\Gamma} \to \powerset(\means{\Gamma}\union\{\Fault\}) \)
such that $\sigma\successorTo\tau$ for all $\sigma,\tau$ with $\tau\in\phi(m)(\sigma)$,
and
\begin{list}{}{}
\item[\quad(\dt{fault determinacy})]
$\Fault \in \phi(m)(\sigma)$ implies $\phi(m)(\sigma)= \{\Fault\}$
\item[\quad(\dt{state determinacy})]
$\RprelT{\pi}{\sigma}{\sigma'}$ implies
$\RprelTS{\pi}{ \phi(m)(\sigma) }{  \phi(m)(\sigma') }$
\end{list}
For $\Phi$ wf in $\Gamma$, a pre-model of $\Phi$
is a pre-model for $\Gamma$ and $\dom(\Phi)$.
\end{definition}
We say pre-models are \dt{quasi-deterministic}, because
from a given initial state, these three outcomes are mutually exclusive: fault, non-empty set of states, empty set.
Moreover, instantiating $\sigma':=\sigma$
and setting $\pi$ to the identity on $\sigma(\lloc)$ in the condition (state determinacy)
yields that all results from a given initial state are isomorphic.\footnote{In light of these
  definitions and the results to follow, we could as well replace the codomain of a
  pre-model, \emph{i.e.}, $\powerset(\means{\Gamma}\union\{\Fault\})$, by
  the disjoint sum of $\powerset(\means{\Gamma})$ and $\{\Fault\}$.
  The chosen formulation helps streamline a few things later.
} %footnote

\begin{figure}[t!]
\begin{small}
\begin{mathpar}
   \inferrule[uCall]
   { \tau\in\phi(m)(\sigma) }
   { \configm{m()}{\sigma}{\mu} \trans{\phi} \configm{\skipc}{\tau}{\mu} }

   \inferrule[uCallX]
   { \Fault\in\phi(m)(\sigma) }
   { \configm{m()}{\sigma}{\mu} \trans{\phi} \Fault}

   \inferrule[uCall0]
   { \phi(m)(\sigma) = \emptyset }
   { \configm{m()}{\sigma}{\mu} \trans{\phi} \configm{m()}{\sigma}{\mu} }

  \inferrule[uCallE]{ \mu(m) = C }
            { \configm{ m() }{\sigma}{\mu} \trans{\phi} \configm{C;\Endcall(m)}{\sigma}{\mu}  }

   \inferrule[uECall]{}{
   \configm{ \Endcall(m) }{\sigma}{\mu} \trans{\phi} \configm{ \skipc }{\sigma}{\mu}
   }

% ok but single meth
% \inferrule[uLet]
%    {}
%    {  \configm{ \letcom{m()}{B}{C} }{\sigma}{\mu}
%       \trans{\phi}
%       \configm{C;\Endlet(m)\,}{\sigma}{\extend{\mu}{m}{B}} }

%  \inferrule[uElet]
%     {}
%     {  \configm{ \Endlet(m) }{\sigma}{\mu} \trans{\phi} \configm{ \skipc}{\sigma}{\drop{\mu}{m}} }

  \inferrule[uLet]
     {}
     {  \configm{ \letcom{\ol{m}}{\ol{B}}{C} }{\sigma}{\mu}
        \trans{\phi}
        \configm{C;\Endlet(\ol{m})\,}{\sigma}{\extend{\mu}{\ol{m}}{\ol{B}}} }

  \inferrule[uElet]
     {}
     {  \configm{ \Endlet(\ol{m}) }{\sigma}{\mu} \trans{\phi} \configm{ \skipc}{\sigma}{\drop{\mu}{\ol{m}}} }

\end{mathpar}
\end{small}
%\vspace{-3ex}
\caption{Selected transition rules, for pre-model $\phi$.  The others are in appendix Figure~\ref{fig:trans}.}
\label{fig:transSel}
\end{figure}

The transition relation is defined in Figure~\ref{fig:transSel}.
A \dt{trace} via pre-model $\phi$ is a non-empty finite sequence of configurations
that are consecutive for the transition relation $\trans{\phi}$.
For example, this sequence is a trace (for any $\phi$):
\[ \configm{x:=1;y:=2}{[x\scol 0,y\scol 0]}{\_}
\configm{y:=2}{[x\scol 1,y\scol 0]}{\_}
\configm{\skipc}{[x\scol 1,y\scol 2]}{\_} \]
Recall that we identify $(\skipc;C)$ with $C$ (Figure~\ref{fig:synident}).
By definition, a trace does not contain $\Fault$.

\subsection{Context models and program correctness}
\label{sec:ctxmodel}

For syntactic substitution we use the notation $\subst{P}{x}{F}$.
Substitution notations are mainly used with spec-only variables.
In addition, for clarity we also use substitution notation for values, even references---although
the syntax does not include reference literals.
%This is only done in certain contexts, for which we define the following abbreviations.

\begin{definition}[\textbf{substitution notation}]
\label{def:subst}
If $\Gamma,x\scol T\proves P$ and  $\sigma\in\means{\Gamma}$ and $v$ is a value in $\means{T}\sigma$, we write
$\sigma\models^\Gamma \subst{P}{x}{v} $
to abbreviate
$\extend{\sigma}{x}{v}\models^{\Gamma,x:T} P$.
\end{definition}

%To define correct pre-model, we use substitution notation like $\subst{P}{x}{v}$.

%With these ingredients we define what it means for a possible method denotation to satisfy a specification.
%(Notation: $\subst{P}{x}{v}$ for substitution, Def.~\ref{def:subst}.)

A context model, or $\Phi$-model when we refer to a specific context $\Phi$,
is a pre-model that satisfies its specs.

\begin{definition}[\textbf{context model}]\label{def:ctxinterp}
\index{context model}
Let $\Phi$ be wf in $\Gamma$ and let $\phi$ be a pre-model.
Say $\phi$ is a \dt{$\Phi$-model} iff
$\dom(\phi)=\dom(\Phi)$ and
for each $m$ in $\dom(\Phi)$  with $\Phi(m)= \flowty{R}{S}{\effe}$
and for any $\sigma$ and $\sigma'$ in $\means{\Gamma}$,
\begin{list}{}{}
\item[(a)] $\Fault\in\phi(m)(\sigma)$ iff
there are no values $\ol{v}$ with $\sigma\models \subst{R}{\ol{s}}{\ol{v}}$
where $\ol{s}$ are the spec-only variables.

\item[(b)] For all $\tau \in \phi(m)(\sigma)$,
and all $\ol{v}$,
if $\sigma\models \subst{R}{\ol{s}}{\ol{v}}$
then $\tau\models \subst{S}{\ol{s}}{\ol{v}}$
and $\sigma\allowTo\tau\models \effe$.
\item[(c)] For all $\tau \in \phi(m)(\sigma)$ and
all $N$ with $\mdl(m)\imports N$, $\rlocs(\sigma,\bnd(N))\subseteq\rlocs(\tau,\bnd(N))$.
\item[(d)] For all $\pi$, if
$\Lagree(\sigma,\sigma',\pi,\rlocs(\sigma,\effe)\setminus\{\lloc\})$
then
\begin{list}{}{}
\item[(i)] $\phi(m)(\sigma)=\emptyset$ iff $\phi(m)(\sigma')=\emptyset$,
and
\item[(ii)] if  $\tau\in\phi(m)(\sigma)$ and $\tau'\in\phi(m)(\sigma')$
then there is $\rho\supseteq\pi$ with
$\rho(\freshLocs(\sigma,\tau))\subseteq \freshLocs(\sigma',\tau')$ and
$\Lagree(\tau,\tau',\rho,
  (\freshLocs(\sigma,\tau)\union\written(\sigma,\tau))\setminus\{\lloc\})$.
\end{list}
\end{list}
\end{definition}
Condition (a) says $\phi(m)$ faults just on states outside the precondition of $m$,
(b) says the postcondition holds and write effect is respected,
(c) is a technical condition we call boundary monotonicity,
and (d) is the dependency condition of the read effect.

The snapshot values $\ol{v}$ in (a) and (b) are uniquely determined by $\sigma$
(Lemma~\ref{lem:uniquespeconly}).
So (a) can be rephrased:  $\Fault\in\phi(m)(\sigma)$ iff
$\sigma \not\models \subst{R}{\ol{s}}{\ol{v}}$
where $\ol{v}$ are the values uniquely determined by $R$ in $\sigma$.
Similarly for (b), which treats spec-only variables as being quantified over the pre- and post-condition.

Finally we can give the semantics of correctness judgments, which embodies
encapsulation for dynamic boundaries.
In the definition to follow we write \ghostbox{$\delta^\oplus$} to abbreviate $\delta,\rd{\lloc}$.
Apropos Def.~\ref{def:ctxinterp}(d), note that $\{\lloc\} = \rlocs(\sigma,\rd{\lloc}) = \rlocs(\sigma,\emptyeff^\oplus)$.

The conditions for a valid correctness judgment include that there are no faulting executions,
terminated executions satisfy the postcondition and write effect, and boundary monotonicity.
These conditions are like (a)--(c) above for context model.
The absence of fault means more than no null dereference; it means there are no
method calls outside the method's precondition---because otherwise the call would fault,
by condition (a) for context models.
An additional condition for correctness is that the read effects of the judgment should subsume the read effects in the specs of methods in context calls; this is called r-safety.
Finally, the Encap condition says that each step reads and writes outside the boundaries
of any module the step is not within.  The Encap condition is formulated using the read effects
of the judgments and implies the expected end-to-end read effect as will be explained later.
Reading is meant in the extensional sense of a two-run dependency property, similar to condition (d) for context model.

The Encap condition applies to every reachable step, and refers to the initial state, so we use the following
schema to designate identifiers for the elements of a
step reached from command $C$ and state $\sigma$:
\[ \configm{C}{\sigma}{\_} \tranStar{\phi} \configm{B}{\tau}{\mu}
                           \trans{\phi} \configm{D}{\upsilon}{\nu} \]
The step is taken by the active command of $B$, from state $\tau$ to state $\upsilon$.
For such a step, we need to refer to the locations encapsulated by all modules except
the current module, $M$, of the correctness judgment.
To this end, the \dt{collective boundary} is an effect $\delta$ defined by cases:
  %% \begin{itemize}
  %% \item if $\Active(B)$ is not a context call, then
  %%     $\delta \eqdef \unioneff{N\in(\Phi,\mu),N\neq \topm(B,M)}{\bnd(N)}$
  %% \item if $\Active(B)$ is a context call of some $m$, then
  %%     $\delta \eqdef \unioneff{N\in(\Phi,\mu),\mdl(m)\not\imports N}{\bnd(N)}$
  %% \end{itemize}
\begin{equation}\label{eq:collectiveB}
\begin{array}{lcll}
\delta & \eqdef & \unioneff{N\in(\Phi,\mu),N\neq \topm(B,M)}{\bnd(N)}
                & \mbox{if $\Active(B)$ is not a context call} \\
       & \eqdef & \unioneff{N\in(\Phi,\mu),\mdl(m)\not\imports N}{\bnd(N)}
                & \mbox{if $\Active(B)$ is a context call of $m$}
\end{array}
\end{equation}

\begin{definition}[\textbf{valid judgment}]\label{def:valid}
A wf %correctness
judgment $\Phi\HPflowtr{\Gamma}{M}{P}{C}{Q}{\eff}$
is \dt{valid}
iff the following %conditions
hold for all $\Phi$-models $\phi$,
all values $\ol{v}$ for the spec-only variables $\ol{s}$ in $P$,
and all states $\sigma$ such that  $\sigma\models^{\Gamma} \subst{P}{\ol{s}}{\ol{v}}$.
\begin{list}{}{}
\item[(\dt{Safety})] It is not the case that
$\configm{C}{\sigma}{\_} \tranStar{\phi} \,\Fault$.
\item[(\dt{Post})]
$\tau \models \subst{Q}{\ol{s}}{\ol{v}}$
for every $\tau$ with $\configm{C}{\sigma}{\_} \tranStar{\phi} \configm{\skipc}{\tau}{\_}$.
\item[(\dt{Write})]
$\sigma\allowTo\tau\models \eff$
for every $\tau$ with $\configm{C}{\sigma}{\_} \tranStar{\phi} \configm{\skipc}{\tau}{\_}$.
\item[(\dt{R-safe})]
Every reachable configuration
\( \configm{C}{\sigma}{\_} \tranStar{\phi} \configm{B}{\tau}{\mu} \)
satisfies the \dt{r-safe condition for $(\Phi,\eff,\sigma)$}:
If $\Active(B)$ is a context call to  $m$ with $\Phi(m) \equiv m:\flowty{R}{S}{\effe}$, then
$\rlocs(\tau,\effe)\subseteq\freshLocs(\sigma,\tau)\union\rlocs(\sigma, \eff)$.

\item[(\dt{Encap})]
Every reachable step
\( \configm{C}{\sigma}{\_} \tranStar{\phi} \configm{B}{\tau}{\mu}
                           \trans{\phi} \configm{D}{\upsilon}{\nu} \)
\dt{respects $(\Phi,M,\phi,\eff,\sigma)$}, i.e.,
%\begin{list}{}{}
%\item[(w-respect)]
\begin{itemize}
\item
For every $N$ with $N\in (\Phi,\mu)$ and $N\neq \topm(B,M)$,
the step % \( \configm{B}{\tau}{\mu} \trans{\phi} \configm{D}{\upsilon}{\nu} \)
\dt{w-respects} $N$, which means:
either $\Active(B)$ is a call to some $m$ with $\mdl(m)\imports N$
or $\agree(\tau,\upsilon,\bnd(N))$.

%\item[(r-respect)]
\item
 For $\delta$ the collective boundary given by (\ref{eq:collectiveB})
for $B,\tau,\mu$,
%% where the \dt{collective boundary} $\delta$ is defined by % cases:
%%   \begin{itemize}
%%   \item if $\Active(B)$ is not a context call, then
%%       $\delta \eqdef \unioneff{N\in(\Phi,\mu),N\neq \topm(B,M)}{\bnd(N)}$
%%   \item if $\Active(B)$ is a context call of some $m$, then
%%       $\delta \eqdef \unioneff{N\in(\Phi,\mu),\mdl(m)\not\imports N}{\bnd(N)}$
%%   \end{itemize}
the step \dt{r-respects $\delta$ for $(\phi,\eff,\sigma)$},
which means: for any\footnote{\label{fn:r-respect}To be precise:
       such that $\tau'$ has the same variables as $\tau$---there may be local variables in addition to those declared by $\Gamma$. }
 $\pi,\tau',\upsilon',D'$
\vspace*{-1ex} % ALERT formatting hacks
\begin{equation}\label{eq:rrespectAnte}
\vspace*{-1ex}
\begin{array}{l}
  \mbox{if }
  \configm{B}{\tau'}{\mu} \trans{\phi} \configm{D'}{\upsilon'}{\nu} \mbox{ and }
  \agree(\tau',\upsilon',\delta) \mbox{ and }
\\
  \Lagree(\tau,\tau',\pi,
(\freshLocs(\sigma,\tau)\union\rlocs(\sigma,\eff))\setminus\rlocs(\tau,\delta^\oplus))
\end{array}
\end{equation}
then $D'\equiv D$ and there is $\rho$ with $\rho\supseteq\pi$ such that
\vspace*{-1ex} % ALERT formatting hacks
\begin{equation}\label{eq:rrespect}
\vspace*{-1ex}
\begin{array}{l}
\Lagree(\upsilon,\upsilon',\rho,
(\freshLocs(\tau,\upsilon)\union\written(\tau,\upsilon))\setminus\rlocs(\upsilon,\delta^\oplus)) \mbox{ and } \\
\rho(\freshLocs(\tau,\upsilon)\setminus\rlocs(\upsilon,\delta))\subseteq
\freshLocs(\tau',\upsilon')\setminus\rlocs(\upsilon',\delta)
\end{array}
\end{equation}

%\item[(boundary monotonicity)]
\item
For every $N$ with $N\in\Phi$ or $N=M$, the step satisfies \dt{boundary monotonicity}:\\
  $\rlocs(\tau,\bnd(N)) \subseteq \rlocs(\upsilon,\bnd(N))$.
\end{itemize}
\end{list}
\qed
\end{definition}
In addition to the terms introduced above to refer to parts of the definition,
we also use the following derived notions:
A trace from $\configm{C}{\sigma}{\_}$
\dt{respects $(\Phi,M,\phi,\eff,\sigma)$} just if each step of the trace does,
and it is \dt{r-safe for $(\Phi,\eff,\sigma)$} just if each configuration is.
A step is called \dt{r-safe} if its starting configuration is r-safe.

While w-respect can be defined one module at a time, this is not the case for r-respect,
because dependency properties do not compose in a simple way.\footnote{For readers familiar
with \RLII, the w-respect condition is the same except that, here, to support r-respect
we add w-respect of modules in the environment (in addition to those in context).
} % footnote
The absence of dependency needs to be expressed in terms of the collective boundary $\delta$ with which a given step must not interfere.  As with w-respect, this depends on whether the step is a context call.
If not, then the current module's boundary is exempt (see condition $N\neq\topm(B,M)$
in (\ref{eq:collectiveB})).
If so,
the step is exempt from the boundary of the callee's module together with modules
into which its  implemenation may call (second condition in (\ref{eq:collectiveB})).
Dependency is expressed as usual by an implication from initial agreement (\ref{eq:rrespectAnte}) on reads
to final agreement (\ref{eq:rrespect}) on writes---subtracting the encapsulated locations.
The read effects in $\eff$ are interpreted in the pre-state $\sigma$,
as are the write effects (which cover the written locations according to the condition labelled Write).
The collective boundary $\delta$ is interpreted at intermediate states.

In case the module boundaries are all empty,
in Def.~\ref{def:valid}, two parts of the Encap condition become vacuous,
namely w-respect and boundary monotonicity.
And r-respect reduces to the property that the dependency
of each step is within the readable locations of the given frame condition.
This implies an end-to-end read effect condition given in the following lemma.\footnote{The
   condition is much like the semantics of effects in \RLIII, with a small
   difference concerning the treatment of variable $\lloc$.  (See Def.~5.2 in \RLIII.)
} % foot
The lemma is used to prove soundness of the linking rule; in that proof
we derive a pre-model from the denotation of the method body, and the lemma
is used to show it is a context model.

\begin{restatable}[read effect]{lemma}{lemreadeff}
\label{lem:readeff}
\upshape
Suppose $\Phi\HVflowtr{\Gamma}{M}{P}{C}{Q}{\eff}$ and $\phi$ is a $\Phi$-model.
Suppose $\sigma\models P$ and $\sigma'\models P$.
Suppose $\Lagree(\sigma,\sigma',\pi,\rlocs(\sigma,\eff)\setminus\{\lloc\})$.
Then $\configm{C}{\sigma}{\_}$ diverges iff $\configm{C}{\sigma'}{\_}$ diverges.
And for any $\tau,\tau'$,
if 
$\configm{C}{\sigma}{\_}\tranStar{\phi}\configm{\skipc}{\tau}{\_}$ and
$\configm{C}{\sigma'}{\_}\tranStar{\phi}\configm{\skipc}{\tau'}{\_}$ then 
\[
\some{\rho\supseteq\pi}{
  \begin{array}[t]{l}
    \Lagree(\tau,\tau',\rho,
     (\freshLocs(\sigma,\tau)\union\written(\sigma,\tau))\setminus\{\lloc\}) \;\mbox{ and} \\
    \rho(\freshLocs(\sigma,\tau))\subseteq \freshLocs(\sigma',\tau')
  \end{array} }
\]
\end{restatable}

\section{Unary logic}\label{sec:unaryLog}

Correctness judgments of the unary logic play a crucial role in the relational logic.
They are premises in relational rules such as local equivalence.
Framing and encapsulation are handled at the unary level,
separate from the concerns of alignment and relation formulas.
%in proof rules at the relational level.

The unary proof rules use two subsidiary judgments, for subeffects and framing of formulas.
These can be presented by inference rules (as shown in \RLI).
In this article we present them semantically, in Section~\ref{sec:framing},
as the semantics is amenable to direct checking by SMT solver.
Informal descriptions are given, but for the detailed definitions in Section~\ref{sec:framing}
the reader needs to be familiar with the definitions in Sects.~\ref{sec:states} and~\ref{sec:effect}.
Aside from that, Section~\ref{sec:unaryLog} can be read without being familiar with Section~\ref{sec:unarySem}.

\subsection{Framing and subeffects}\label{sec:framing}

The \dt{subeffect judgment}, written \ghostbox{$P\models \eff \leq \effe$},
says that in states satisfying $P$, the readable or writable locations designated by $\eff$
are contained in those designated by $\effe$.
It is defined as follows:
\begin{equation}\label{eq:subeffect}
P\models \eff \leq \effe
\mbox{ iff }
\rlocs(\sigma,\eff)\subseteq\rlocs(\sigma,\eta)
\mbox{ and }
\wlocs(\sigma,\eff)\subseteq\wlocs(\sigma,\eta) \mbox{ for all $\sigma$ with $\sigma \models P$}
\end{equation}

The \dt{framing judgment} for formulas, written \ghostbox{$P \models \fra{\effe}{Q}$},
can loosely be understood to say the read effects in $\effe$ cover the footprint of $Q$.
It is used in the frame rule and also second order frame rule where we need
framing of the module invariant by the dynamic boundary.
To be precise, the judgment says of states $\sigma$ and $\tau$ that if $\sigma$ satisfies $P\land Q$
and $\tau$ agrees with $\sigma$ on the contents of locations designated by the read effects of $\effe$,
then $\tau$ satisfies $Q$.  Here $\effe$ is interpreted in state $\sigma$, which only matters if its
effect expressions mention mutable variables.
The judgment is defined as follows:
\begin{equation}\label{eq:frmAgree}
P \models \fra{\effe}{Q} \mbox { iff for all }\sigma, \tau,
\mbox { if } \agree(\sigma, \tau, \effe)
\mbox{ and }
\sigma \models P \land Q \mbox{ then }
\tau \models Q
\end{equation}
For example, we have
$x\in r \models \fra{ \rd{x},\rd{r\Img f} }{x.f=0}$. % ok though not read framed
The $\ftpt$ function,  defined in Figure~\ref{fig:foot}, provides framing for atomic formulas.
The basic lemmas about $\ftpt$ are that
$\models \fra{\ftpt(P)}{P}$, for atomic $P$, and
\begin{equation}\label{eq:footprintAgreement}
\agree(\sigma, \sigma', \pi,\ftpt(F)) \mbox{ implies }
\rprel{\sigma(F)}{\sigma'(F)}
\end{equation}
The framing judgment is used, in the \rn{Frame} rule, in combination
with a separator formula (Figure~\ref{fig:sepdef}).
A key property of separators
% eq:sepdisj
is that a formula obtained as $\ind{\effe}{\eff}$
holds in  $\sigma$ iff $\rlocs(\sigma,\effe)\intersect \wlocs(\sigma,\eff)=\emptyset$.
From this it follows that
\begin{equation}\label{eq:sepagree} % lem:sepagree in readRL
\sigma \allowTo \tau \models \eff
\mbox{ and }
\sigma\models \ind{\effe}{\eff}
\mbox{ implies }
\agree(\sigma, \tau, \effe)
\end{equation}

Separator formulas are also used in the notion of immunity, which amounts to framing for frame conditions.
Immunity is only needed for the sequence and loop rules, which we relegate to the appendix
as there is no interesting change from \RLI.
Framing and immunity are about preserving the value of an expression or formula from one control point to a later one.
For preservation of agreements, framed reads (Def.~\ref{def:framedreadsDyn}) are crucial;
e.g., in proving the lockstep alignment Lemma~\ref{lem:rloceq}.

\subsection{Proof rules}

Selected proof rules are in Figure~\ref{fig:proofrulesU}.
They are to be instantiated only with wf premises and conclusions.
In the rest of the section we comment briefly about some rules
and derive the modular linking rule.
Then Section~\ref{sec:encap} discusses how the rules work together to enforce encapsulation.

\begin{figure}[t]
\begin{footnotesize}
\begin{mathpar} % no space before first rule
\ruleConseq

\inferrule*[left=Frame]
{ \Phi\HPflowtr{}{M}{P}{C}{Q}{\eff} \\
  P \models \fra{\effe}{R} \\
  P\land R \imp \ind{\effe}{\eff}
}{
\Phi\HPflowtr{}{M}{P\land R}{C}{Q\land R}{\eff}
}

\ruleSOF

\ruleCtxIntroInOne

\ruleCtxIntro

\ruleCall

\inferrule*[left=FieldUpd]
           {}
{ \HPflowtr{}{\emptymod}{ x \neq \NULL }{x.f := y}{x.f = y}{\wri{x.f},\rd{x},\rd{y}}}

\ruleLink

\inferrule*[left=Alloc]
{ \fields(K) = \ol{f}:\ol{T} \\ \mbox{spec-only}(r) }
{ \HPflowtr{}{\emptymod}
             {r=\lloc}
             { x:=\new{K}}
             { x\notin r \land \lloc = r\union\{x\} \land x.\ol{f}=\Default{\ol{T}}}
             {\wri{x}, \rw{\lloc}}}

\ruleIf

\end{mathpar}
\end{footnotesize}
\vspace*{-3ex}
\caption{Selected unary proof rules. For others see appendix~Figs.~\ref{fig:proofrules}
and~\ref{fig:proofrulesStruct}.
}
\label{fig:proofrulesU}
\end{figure}

The proof rules for assignment, like \rn{FieldUpd} and \rn{Alloc},
are ``small axioms''~\cite{OHearnRY01} that have empty context, are in the default module,
and have precise frame conditions.
The \rn{Conseq} rule
can be used to subsume a frame condition like $\wri{\sing{x}\Img f}$ by a more general one like $\wri{r\Img f}$, given precondition $x\in r$
and using subeffect judgment $x\in r \models \wri{\sing{x}\Img f} \leq \wri{r\Img f}$.
Rule \rn{Alloc} can be used with the \rn{Frame} rule
to express freshness in several ways.\footnote{Shown in detail in \RLIII~(Section~7.1).}
These and the method call rule have the minimum needed hypothesis context.
Extending the context is done by rules discussed in Section~\ref{sec:encap}.

The gist of the second order frame rule, \rn{SOF}, is to conjoin a formula not only to the spec in the conclusion, like rule \rn{Frame}, but also conjoin it to the specs in the hypothesis context.
The rule distils a property of program semantics; its practical role
is to derive the modular linking rule.

In rule \rn{SOF}, the conditions $N\in \Theta$ and $N\neq M$ ensure that the command $C$
respects the encapsulation of $\bnd(N)$,
in accord with the semantic condition Encap of Def.~\ref{def:valid}.
Together with the framing judgment  $\models \fra{\bnd(N)}{I}$,
this ensures that $C$ does not falsify $I$.
The condition  \dt{$C$ binds no $N$-method} means $C$ contains no let-binding of a method $m$ with $\mdl(m)=N$.
This and the condition $\all{m\in\Phi}{\mdl(m)\not\imports N}$
ensure that all of $N$'s method specs are in $\Theta$ and have the invariant added simultaneously.
Such conditions are the price we pay for not cluttering the logic with explicit syntax and judgments for a module calculus.  Rule \rn{Link} has analogous conditions.

\begin{figure}[t!]
\begin{footnotesize}
\(
\inferrule*[right=Conseq]{
\inferrule*[right=Link]{
        \inferrule*[right=SOF]{
		\Phi\HPflowtr{}{{\emptymod}}{P}{C}{Q}{\eff}
	}{
            \Phi \conjInv I  \proves_{\emptymod}  C : (\flowty{P}{Q}{\eff}) \conjInv I
	}
	\Phi\conjInv I \proves_M B : \Phi(m)\conjInv I
	}{
	\proves_{\emptymod}
	\letcom{m}{B}{C} : (\flowty{P}{Q}{\eff}) \conjInv I
	}
	}{
	\proves_{\emptymod}
	\letcom{m}{B}{C} : \flowty{P}{Q}{\eff}
	}
\)
\end{footnotesize}
\vspace{-1ex}
\caption{Derivation of \rn{MLink}, with side conditions
$\mdl(m)=M$,
$\models \fra{ \bnd(M) }{ I }$, % (for \rn{SOF})
and $P \imp I$. %  (for \rn{Conseq}).
}
\label{fig:unaryMismatch}
\end{figure}

In rule \rn{Link}, $\letcom{\ol{m}}{\ol{B}}{C}$ means the simultaneous linking of $m_i$ with $B_i$ for $i$ in some range.
This version of \rn{Link} supports simultaneous linking of multiple methods that may be defined in different modules.
Note that $\Theta$ is in the hypotheses for $B_i$ because some methods in $\Theta$ may call others in $\Theta$, and for recursion.
Condition $\all{N\in\Phi,L\in\Theta}{N\not\imports L}$ precludes dependency of the ambient modules on the ones being linked.
Condition $\all{N,L}{ N\in\Theta \land N\irimports L \imp L\in(\Phi,\Theta)}$ expresses import closure,
which is needed to ensure that all relevant boundaries are considered in the Encap condition of the premises.

Recall the modular linking rule (\ref{eq:mismatch}) sketched in Section~\ref{sec:modrel}.
It can now be made precise as follows.
\[
\inferrule*[left=MLink]{
  \Phi\HPflowtr{}{{\emptymod}}{P}{C}{Q}{\eff} \\
  \Phi\conjInv I \proves_M B : \Phi(m)\conjInv I  \\
  \mdl(m)=M \\
  \models\fra{\bnd(M)}{I} \\
  P\imp I
}{
  \proves_{\emptymod} \letcom{m}{B}{C} : \flowty{P}{Q}{\eff}
}
\]
In Section~\ref{sec:modrel} we mention requirements for soundness of (\ref{eq:mismatch}),
in vague terms which can now be made precise.
Requirement (E1) is to delimit some internal locations, which is expressed as a dynamic boundary $\bnd(M)$.
Requirement (E2) is that the module invariant $I$ depends only on encapsulated locations,
which we express by a framing judgment $\models\fra{\bnd(M)}{I}$.
Requirement (E3) says the client stays outside boundaries, a part of the meaning of
the correctness judgment for $C$; more on this in Section~\ref{sec:encap}.
Finally, (E4) requires that the invariant holds initially;
we simply require that $I$ follows from the main program's precondition ($P\imp I$).
Rule \rn{MLink} is derived in Figure~\ref{fig:unaryMismatch}.
The side conditions $\models \fra{ \bnd(M) }{ I }$, % (for \rn{SOF})
and $P \imp I$ are the responsibility of the module developer.
The idea is that precondition $P$ expresses initial conditions for the linked program,
e.g., that globals have default values (null for class types, $\emptyset$ for $\Region$).
In our examples, the invariant quantifies over elements of the global variable $pool$ and holds when $pool$ is empty.
For a more sophisticated language, we would have module initialization code to establish the module invariant.

\begin{restatable}[soundness of unary logic]{theorem}{unarysoundness}
\label{thm:unarysoundness}
All the unary proof rules are sound
(Figure~\ref{fig:proofrulesU} and appendix
Figs.~\ref{fig:proofrules} and~ \ref{fig:proofrulesStruct}).
\end{restatable}

\subsection{How the proof rules ensure encapsulation}\label{sec:encap}

The proof rules for commands must enforce requirement (E3), i.e., a command respects the
boundaries of modules in context other than the current module.
In part this is done by what we call \dt{context introduction} rules.
One may expect a weakening rule that allows additional specs to be added to the context,
and indeed there is such a rule (\rn{CtxIntroIn1}) for the case that the method's module is already
in context. If the method's module is not already in context, adding its spec actually
strengthens
the property expressed by the judgment, namely respect of the added module's boundary.
For this we have a rule \rn{CtxIntro} that extends the context by adding a spec for method $m$
and has side conditions (using separator formulas generated by $\ind{}{}$)
that ensure both the read and write effects of atomic command $A$ are separate
from the boundary of $m$'s module.
Two other variations are needed to handle method calls and adding a spec for the current module;
these are relegated to the appendix.
(A more elegant treatment may be possible using an explicit calculus of modules and their
correctness, but that would have its own intricacies.)

As an example, consider this code which acts on variables \whyg{s: Stack} and \whyg{c,d: Cell}.
\begin{lstlisting}
d.val:=0; push(s,d); d:=new Cell; d.val:=1; push(s,d)
\end{lstlisting}
Using variable $r:\Region$ and idiomatic precondition
 $d\in r \land \disj{r}{(pool\union pool\Img rep)}$, this code has frame condition
$\rw{d,r,\lloc, r\Img val}$.
(Here we use the spec idiom depicted in Figure~\ref{fig:pool-and-rep}.)
The small axiom for the store command $d.val:=0$ says it reads $d$ and writes $d.val$.
To add the Stack module to this command's context, rule \rn{CtxIntro} requires
the precondition to imply a separator 
which when simplified is
$\disj{ \sing{d} }{ pool } \land
\disj{ \sing{d} }{ pool\Img rep }$.
This says $d$ is neither in $pool$ nor in any $rep$ unless $d$ is null.

There is also a rule to change the current module from the default module used in,
e.g., rules \rn{Call}, \rn{FieldUpd}, and \rn{Alloc}.
In a proof these and the context introduction rules are used at the ``leaves'' of the proof, i.e., for atomic commands, in order to introduce the intended modules.  This organization is the same as
used previously in \RLII.
However, here the notion of encapsulation is stronger. To enforce 
that reads do not transgress boundaries (r-respect in Def.~\ref{def:valid}),
the proof rules for \rn{If} and \rn{While} also have side conditions 
to ensure the conditional expressions are separate from boundaries.
For test expression $E$, the condition is
$\ind{\unioneff{N\in\Phi,N\neq M}{\bnd(N)}}{\rTow(\ftpt(E)) }$.
This separator formula simplifies to true or false depending on whether
any variable in $E$ occurs in any of the boundaries of modules $N$ in scope other than the current module $M$.
Although the details are different from \RLII, the general idea is the same
so we relegate most of these rules to the appendix (see Figure~\ref{fig:proofrules} and Remark~\ref{rem:ctxIntro}).
Relevant examples can be found in Section~8 of \RLII.

\section{Biprograms: semantics and correctness}\label{sec:biprog}

This section defines (in Section~\ref{sec:rel-premodel})
the relational analog of the pre-models used in unary program semantics of Section~\ref{sec:progSem}.
This is used (in Section~\ref{sec:bitrans}) to define the transition semantics of biprograms.
Some details are intricate, as needed to ensure quasi-determinacy and
to ensure that a biprogram execution faithfully represents a pair of unary executions.
On this basis, the semantics of relational judgments is defined
and shown to entail the expected relational property of unary executions
(Section~\ref{sec:rel-corr}).
The first step is to define the semantics of relation formulas
(Section~\ref{sec:relform}).

\subsection{Relation formulas}\label{sec:relform}

\begin{figure}[t]
\begin{footnotesize}
\(
% HACK alert
\begin{array}{l@{\hspace*{1ex}}l@{\hspace*{1ex}}l}
\sigma|\sigma'\models_\pi\leftF{P}
  & \mbox{iff} & \sigma\models P \\

\sigma|\sigma'\models_\pi F\eqbi F'
  & \mbox{iff} & \rprel{\sigma(F)}{\sigma'(F')} \\

\sigma|\sigma'\models_\pi \Agr G\Img f
  & \mbox{iff} & \agree(\sigma,\sigma',\pi,\rd{G\Img f}) \mbox{ and }
\agree(\sigma',\sigma,\pi^{-1},\rd{G\Img  f})\\

\sigma|\sigma'\models_\pi \Agr x
  & \mbox{iff} & \rprel{\sigma(x)}{\sigma'(x)} \\

\sigma|\sigma'\models_\pi \later \P
%  & \mbox{iff} & \some{\rho}{\rho\supseteq\pi \mbox{ and } \sigma|\sigma'\models_\rho \P } \\
  & \mbox{iff} & \sigma|\sigma'\models_\rho \P \mbox{ for some } \rho\supseteq\pi \\

\sigma|\sigma'\models_\pi  \P\imp\Q
  & \mbox{iff} & \sigma|\sigma'\models_\pi  \P \mbox{ implies } \sigma|\sigma'\models_\pi \Q
\\[1ex]

\sigma|\sigma'\models \P &\mbox{iff}& \sigma|\sigma'\models_\pi \P \mbox{ for all $\pi$}
\\ % used in semantics of whilebi
\models\P & \mbox{iff} & \sigma|\sigma'\models \P \mbox{ for all $\sigma,\sigma'$}

\end{array}\)
\end{footnotesize}
\vspace{-2ex}
\caption{Relation formula semantics \ghostbox{$\sigma|\sigma'\models^{\Gamma|\Gamma'}_\pi\P$} (selected).
See appendix Figure~\ref{fig:relFmlaSemA} for other cases.
}
\label{fig:relFmlaSem}
\end{figure}

Refperms and agreement, the basis for semantics of read effects,
are also used for semantics of agreement formulas.
For relation formulas, satisfaction $\sigma|\sigma'\models_\pi \P$ says state $\sigma$ relates to $\sigma'$ according to $\P$ and refperm $\pi$ (see Figure~\ref{fig:relFmlaSem}).
The propositional connectives have classical semantics.
Formula $\P$ is called \dt{valid} if $\models\P$.
%The semantics is classical, so we can define $\neg\P \eqdef \P\imp\False$ and $\P\lor\Q\eqdef \neg\P\imp\Q$.  (Disj is in appendix)

Recall that semantic agreement ($\Lagree,\agree$) is skewed in the sense that region expressions are evaluated in the left state, as noted following (\ref{eq:LagreeSym}).
The semantics of $\Agr G\Img f$ uses agreement via refperm $\pi$
and agreement via $\pi^{-1}$ for the swapped pair of states.
As a result, $\sigma|\sigma'\models_\pi\Agr G\Img f$
implies not only $\sigma(G)\subseteq dom(\pi)$ but also
$\sigma'(G)\subseteq rng(\pi)$.
However, $\Agr G\Img f$ does not imply $G\eqbi G$ in general.
%\footnote{\label{fn:agree}Consider states $\sigma,\sigma'$
% where there are exactly three references $o,p,q$ in both $\sigma$ and $\sigma'$
% and let $\pi = \{(o,o),(p,p),(q,q)\}$.
% Suppose $\sigma(o.f)=\sigma'(o.f)$, $\sigma(p.f)=\sigma'(p.f)$, and $\sigma(q.f)=\sigma'(q.f)$.
% Let $r$ be a variable with $\sigma(r)=\{o,p\}$ and $\sigma'(r)=\{p,q\}$.
% We have $\sigma|\sigma'\models_\pi \Agr r\Img f$ by semantics.
% But it is not the case that $\sigma|\sigma'\models_\pi r\eqbi r$.
%\label{fn:neqbi}}
So the form $G\eqbi G\land \Agr G\Img f$ is often used,
e.g., formula (\ref{eq:tabu:postR});
in particular it appears in the agreements from a read framed effect.

The formulas $\Agr G\Img f$ and $G\Img f\eqbi G\Img f$
have different meaning and in general are incomparable.
In case $f:\INT$, the region $G\Img f$ is empty
in which case $\Agr G\Img f$ implies $G\Img f\eqbi G\Img f$ trivially.
Using a diagram like in Figure~\ref{fig:box-and-arrows},
Figure~\ref{fig:eg-agree} shows two states and a refperm such that
$\Agr \sing{x}\Img f$ holds
(noting that $(q,q')\in\pi$ and $(r,r')\in\pi$).
But  $\sing{x}\Img f\eqbi \sing{x}\Img f$ does not;
we have  $\sigma(\sing{x}\Img f)=\{q\}$ and
$\sigma'(\sing{x}\Img f)=\{r'\}$ but $(q,r')\notin \pi$.
Also $\sing{x}\eqbi\sing{x}$ is false because $(o,p')\notin\pi$.

\tikzset{
  pics/genobj/.style args={#1/#2}{
    code={
      \node[draw,rounded corners,minimum height=1cm,minimum width=1.6cm] (-main) {};
      \node[above of=-main,xshift=-0.75cm,yshift=-0.3cm] (-name) {#1};
      \node[draw,minimum height=0.55cm,minimum width=0.6cm,xshift=0.3cm] (-f) {#2};
      \node[left of=-f,xshift=0.4cm] (fTxt) {$f$};
    }
  }
}

\tikzset{
  pics/varblock/.style args={#1/#2}{
    code={
      \node[draw,rectangle,minimum width=0.5cm,minimum height=0.6cm] (-value) {#2};
      \node[above of=-value,yshift=-0.5cm,xshift=-0.25cm] (-name) {#1};
    }
  }
}

\begin{figure}[t]
  \centering
  \begin{tikzpicture}[scale=0.75,transform shape,>=stealth]
    % draw left heap structure
    \pic (left-x) at (0,0) {varblock={$x$/$o$}};
    \pic (left-o) at (2,0) {genobj={$o$/$q$}};
    \pic (left-p) at (4.5,0) {genobj={$p$/$r$}};
    \pic (left-q) at (7,0) {genobj={$q$/$\NULL$}};
    \pic (left-r) at (9.5,0) {genobj={$r$/$\NULL$}};

    % draw left arrows
    \draw[->] (left-x-value.east) to [bend left=20] ([yshift=0.3cm]left-o-main.west);
    \draw[->] (left-o-f.east) to [out=50,in=140] (left-q-main.north);
    \draw[->] (left-p-f.south) to [out=-30,in=-150] (left-r-main.south);

    % draw right heap structure
    \pic (right-x) at (0,-3) {varblock={$x$/$p'$}};
    \pic (right-o) at (2,-3) {genobj={$o'$/$q'$}};
    \pic (right-p) at (4.5,-3) {genobj={$p'$/$r'$}};
    \pic (right-q) at (7,-3) {genobj={$q'$/$\NULL$}};
    \pic (right-r) at (9.5,-3) {genobj={$r'$/$\NULL$}};

    % draw right arrows
    \draw[->] (right-x-value.north) to [bend left=40] (right-p-main.north);
    \draw[->] (right-o-f.south) to [bend right=30] (right-q-main.south);
    \draw[->] (right-p-f.north) to [bend left=35] (right-r-main.north);

    % draw refperm arrows
    \draw[<->,dashed] ([yshift=-0.2cm]left-o-main.south) -- ([yshift=0.2cm]right-o-main.north);
    \draw[<->,dashed] ([yshift=-0.2cm]left-p-main.south) -- ([yshift=0.2cm]right-p-main.north);
    \draw[<->,dashed] ([yshift=-0.2cm]left-q-main.south) -- ([yshift=0.2cm]right-q-main.north);
    \draw[<->,dashed] ([yshift=-0.2cm]left-r-main.south) -- ([yshift=0.2cm]right-r-main.north);

    % trying to add pi - ALERT hard coded unit of length
%    \node[below of=left-r-main,yshift=-1cm,xshift=0.8cm] (pi-r) {$\pi(r) = r'$};
    \node[below of=left-r-main,yshift=-.4cm,xshift=0.8cm] (pi-r) {$\pi(r) = r'$};
  \end{tikzpicture}

%  \caption{Refperm $\pi$ and states $\sigma,\sigma'$
%           where $\Agr\sing{x}\Img f$ but neither $\sing{x}\eqbi \sing{x}$ nor
%           $\sing{x}\Img f \eqbi\sing{x}\Img f$ hold.}
  \caption{Refperm $\pi$ and states $\sigma,\sigma'$
    that satisfy $\Agr\sing{x}\Img f$ but neither $\sing{x}\eqbi \sing{x}$ nor
           $\sing{x}\Img f \eqbi\sing{x}\Img f$.}
  \label{fig:eg-agree}
\end{figure}

Here are some valid schemas:
$\P\imp\later\P$,
$\later\later\P\imp\later\P$,
%$\always\P\imp\P$,
and
$\later(\P\land\Q) \imp \later\P\land\later\Q$.
Another validity is $(\lloc\eqbi\lloc) \land \later\P \imp \P$,
in which $\lloc\eqbi\lloc$ says the refperm is a total bijection on allocated references.
The strong condition $\lloc\eqbi\lloc$ is not local, and is not a useful requirement for most purposes.

Validity of $\P\imp\always\P$ is equivalent to $\P$ being \dt{refperm monotonic},
i.e., not falsified by extension of the refperm.
Agreement formulas are refperm monotonic, as a consequence of (\ref{eq:LagreeMono}).
A key fact is:
\begin{equation}\label{eq:mono-distrib}
\mbox{If } \Q\imp\always\Q
\mbox{ is valid then so is }
\later\P\land\Q \imp \later(\P\land\Q) 
%\later\P\land\Q \imp \later(\P\land\Q) \mbox{ is valid if } \Q \mbox{ is refperm monotonic.}
\end{equation}
Validity of $\later\P\imp\P$ expresses that $\P$ is \dt{refperm-independent},
i.e., $\sigma|\sigma'\models_\pi \P$ iff $\sigma|\sigma'\models_\rho \P$,
for all $\sigma,\sigma',\pi,\rho$.
If $\P$ contains no agreement formula then it is refperm-independent
(even if $\later$ occurs in $\P$).
For such formulas the condition in (\ref{eq:mono-distrib}) can be strengthened:
\begin{equation}\label{eq:refp-ind-dist}
\mbox{If } \later\Q\imp\Q
\mbox{ is valid then so is } \later\P\land\Q \iff \later(\P\land\Q)
\end{equation}

Syntactic projection is weakening:
$\P\imp \leftF{P}\land \rightF{P'}$ where $P$ is
$\Left{\P}$ and $P'$ is $\Right{\P}$.
The implication is strict, in general, because projection discards agreements (Figure~\ref{fig:synProjFmla}).
Syntactic projection is not $\imp$-monotonic: for boolean variable $x$, the formula
$x\eqbi x \land \rightF{x>0} \imp \leftF{x>0}$ is valid,
but $\Left{x\eqbi x \land \rightF{x>0}} \equiv true\land true$
and $\Left{\leftF{x>0}} \equiv x>0$.
The example also shows that agreements can have unary consequences.
As another example, this is valid:
$\later(x\eqbi x' \land x\eqbi y') \imp \rightF{x'=y'}$.
The antecedent holds if the refperm relates the value of $x$ to both the values of $x'$ and $y'$,
or can be extended to do so.  Neither is possible if the value of $x'$ is different from the value of $y'$.

The framing judgment generalizes the unary version (\ref{eq:frmAgree}).
\begin{definition}[\textbf{framing judgment}]\index{framing judgment}
\label{def:frmAgreeRel}
Let \ghostbox{$\P\models \fra{\effe|\effe'}{\Q}$}
iff for all $\pi, \sigma, \sigma', \tau, \tau'$,
if $\agree(\sigma, \tau, \effe)$,
$\agree(\sigma', \tau', \effe')$,
%$\sigma\successorTo\tau$,
%$\sigma'\successorTo\tau'$,
and $\sigma|\sigma' \models_\pi \P \land \Q$ then $\tau|\tau' \models_\pi \Q$.
\end{definition}
For example,
$G\eqbi G \models\fra{\eta|\eta}{\Agr G\Img f }$
where $\eta$ is $\ftpt(G),\rd{G\Img f}$ (Lemma~\ref{lem:frameEqbi}).
Apropos relations of the form $\R \eqdef G\eqbi G \land \Agr G\Img f$,
we have $\models\fra{\delta|\delta}{\R}$ where
$\delta$ is $\ftpt(G),\rd{G\Img f}$.
If $P\models \fra{\effe}{Q}$ then $\leftF{P}\models \fra{\effe|\emptyeff}{\leftF{Q}}$ (and same on the right).
Also, $\models\fra{\ftpt(F) | \ftpt(F')}{F\eqbi F'}$, which can be shown using the footprint agreement lemma (\ref{eq:footprintAgreement}).

% ok
%Using the same refperm in the conclusion makes sense because $\Q$ only depends on
%some preexisting locations.

The subeffect judgment
\ghostbox{$ \P\models (\eff|\eff')\leq(\effe|\effe')$}
is also a direct generalization of the unary version: the inclusions of
(\ref{eq:subeffect}) hold on both sides, for $\sigma,\sigma',\pi$ with  $\sigma|\sigma'\models_\pi\P$.

\begin{definition}[\textbf{substitution notation}]
%\label{def:relsubst}
If $\Gamma,x\scol T|\Gamma',x'\scol T'\proves \P$,
$\sigma\in\means{\Gamma}$, $v\in \means{T}\sigma$,
$\sigma'\in\means{\Gamma'}$, and $v'\in \means{T'}\sigma'$, we write
$\sigma|\sigma'\models^{\Gamma|\Gamma'} \subst{\P}{x|x'}{v|v'} $
to abbreviate
$\extend{\sigma}{x}{v}|\extend{\sigma'}{x'}{v'} \models^{\Gamma,x:T|\Gamma',x':T'} \P$.
\end{definition}

%Conjecture: the framing rules in Part I are sound for the generalization that allows agreement modulo a refperm, along the lines of Lemma~\ref{lem:subeff} in this document.

\subsection{Relational pre-models}\label{sec:rel-premodel}

A relational pre-model involves two unary pre-models (Def.~\ref{def:preinterp})
together with a function on state pairs
as appropriate for the denotation of a biprogram.
This function is subject to similar conditions as for unary pre-models,
and must also be compatible with its two unary pre-models.

\begin{definition}[\textbf{state pair iso} 
\ghostbox{$\RprelT{\pi\smallSplitSym \pi'}{}{}$}, 
\ghostbox{$\RprelTS{\pi\smallSplitSym \pi'}{}{}$}]
\label{def:state-iso-rel}
Building on Def.~\ref{def:state-iso},
we define isomorphism of state pairs modulo refperms:
\( \RprelT{\pi\smallSplitSym\pi'}{(\sigma|\sigma')}{(\tau|\tau')}
\mbox{ iff }
   \RprelT{\pi}{\sigma}{\tau} \mbox{ and }
   \RprelT{\pi'}{\sigma'}{\tau'}  \)
.
For relational outcome sets $S$ and $S'$, \emph{i.e.}, $S$ and $S'$ are in
$ \powerset((\means{\Gamma}\times\means{\Gamma'})\union\{\Fault\})$,
define $\RprelTS{\pi\smallSplitSym \pi'}{S}{S'}$ (read \dt{equivalence mod $\pi,\pi'$}) to mean that
(i) $\Fault\in S$ iff $\Fault\in S'$;
(ii) for all state pairs $(\sigma|\sigma')\in S$ and $(\tau|\tau')\in S'$
there are $\rho,\rho'$ with
$\rho\supseteq\pi$ and
$\rho'\supseteq\pi'$, such that
$\RprelT{\rho|\rho'}{(\sigma|\sigma')}{(\tau|\tau')}$; and
(iii) $S\setminus\{\Fault\} = \emptyset$ iff $S'\setminus\{\Fault\} = \emptyset$.
\end{definition}

\begin{definition}
\label{def:interpRel}
A \dt{relational pre-model} for $\Gamma|\Gamma'$ is a triple $\phi = (\phi_0,\phi_1,\phi_2)$
with $\dom(\phi_0)=\dom(\phi_1)=\dom(\phi_2)$, such that
$\phi_0$ (resp.\ $\phi_1$) is a unary pre-model for $\Gamma$ (resp.\ $\Gamma'$)
(Def.~\ref{def:preinterp}),
and for each $m$,
%% A \dt{relational pre-model} for $\Gamma|\Gamma'$ and set $X$
%% of method names is a triple $\phi = (\phi_0,\phi_1,\phi_2)$ where
%% $\phi_0$ (resp.\ $\phi_1$) is a unary pre-model for $\Gamma$ (resp.\ $\Gamma'$) and $X$
%% (Def.~\ref{def:preinterp}),
%% and for each $m\in X$,
the \dt{bi-model} $\phi_2(m)$ is a function
\( \phi_2(m) \ : \
\means{\Gamma}\times\means{\Gamma'} \to \powerset (\means{\Gamma}\times\means{\Gamma'} \:\union\: \{\Fault\})
\)
such that
% can omit successor conditions which follow from unary compat
%if $(\tau,\tau')\in \phi_2(m)(\sigma,\sigma')$ then $\sigma\successorTo\tau$ and $\sigma'\successorTo\tau'$
\begin{list}{}{}
\item[(\dt{fault determinacy})]
$\Fault \in \phi_2(m)(\sigma|\sigma')$ implies $\phi_2(m)(\sigma|\sigma')= \{\Fault\}$
\item[(\dt{state determinacy})]
$\RprelT{\pi|\pi'}{(\sigma|\sigma')}{(\tau|\tau')}$
implies
$\RprelTS{\pi|\pi'}{ \phi_2(m)(\sigma|\sigma') }{  \phi_2(m)(\tau|\tau') }$
%(see Def.~\ref{def:state-iso-rel})
\item[(\dt{divergence determinacy})]
$\RprelT{\pi|\pi'}{(\sigma|\sigma')}{(\tau|\tau')}$
implies that
$\phi_2(m)(\sigma|\sigma') = \emptyset$ iff $\phi_2(m)(\tau|\tau') = \emptyset$.
\end{list}
Moreover $\phi_0,\phi_1,\phi_2$ must be compatible in the following sense:
%for all $\sigma,\sigma',\tau,\tau'$: DN is trying to learn to omit needless quantifiers
\begin{list}{}{}
\item[(\dt{unary compatibility})]
$\tau|\tau' \in \phi_2(m)(\sigma|\sigma') \imp \tau\in\phi_0(m)(\sigma) \land
\tau'\in\phi_1(m)(\sigma')$
\item[(\dt{relational compatibility})]
$\tau\in\phi_0(m)(\sigma) \land \tau'\in\phi_1(m)(\sigma') \imp
\tau|\tau' \in \phi_2(m)(\sigma|\sigma') \lor \Fault \in \phi_2(m)(\sigma|\sigma')$
\item[(\dt{fault compatibility})]
$\Fault\in\phi_0(m)(\sigma) \lor \Fault\in\phi_1(m)(\sigma') \imp
\Fault \in \phi_2(m)(\sigma|\sigma')$
\end{list}
\end{definition}
% omit
%The term bi-model refers to the part of a pre-model that acts on state pairs, whereas ``relational pre-model '' refers to the triple comprised of a bi-model together with two unary pre-models.

We do not require $\Fault\in\phi_2(m)(\sigma|\sigma')$ to imply $\Fault\in\phi_0(m)(\sigma)$ or $\Fault\in\phi_1(m)(\sigma')$.
The bi-model denoted by a biprogram may fault due to
relational precondition, or alignment conditions, even though the underlying commands do not fault.
%DN seems unnecessary:
%\footnote{The motivation can be seen in Lemma~\ref{lem:denotBiprog}
%          which is the relational counterpart to Lemma~\ref{lem:denotComm}.}

\begin{lemma}[empty outcome sets]\label{lem:emptyOutcomes}
\upshape
For any relational pre-model $\phi$,
 %begin{equation}\label{eq:emptyOutcomes}
$\phi_2(m)(\sigma|\sigma') = \emptyset $
implies that $\phi_0(m)(\sigma)=\emptyset$ or $\phi_1(m)(\sigma')=\emptyset$.
\end{lemma}
\begin{proof}
If either  $\phi_0(m)(\sigma)$ or $\phi_1(m)(\sigma')$ contains fault then so does
$\phi_2(m)(\sigma|\sigma')$, by fault compatibility;
and if both $\phi_0(m)(\sigma)$ and $\phi_1(m)(\sigma')$ contain states,
say $\tau\in\phi_0(m)(\sigma)$ and $\tau'\in\phi_1(m)(\sigma')$,
then by relational compatibility $\phi_2(m)(\sigma|\sigma')$ contains either $(\tau|\tau')$ or $\Fault$.
\end{proof}

In a relational pre-model, the bi-model outcome sets are convex in this sense:
\[ \tau|\tau'\in \phi_2(m)(\sigma|\sigma')
\mbox{ and }
\upsilon|\upsilon'\in \phi_2(m)(\sigma|\sigma')
\mbox{ imply }
\tau|\upsilon'\in \phi_2(m)(\sigma|\sigma')
\mbox{ and } \upsilon|\tau'\in \phi_2(m)(\sigma|\sigma')
\]
This is a consequence of unary compatibility, relational compatibility, and fault determinacy.
But it is not a consequence of the three conditions imposed on bi-models alone.

\subsection{Biprogram transition relation}\label{sec:bitrans}

\begin{figure}[t!]
\begin{footnotesize}

\begin{mathpar}
\inferrule*[left=bSync]{
\mbox{$A$ not a method call}\\
   \configm{A}{\sigma}{\mu} \trans{\phi_0} \configm{\skipc}{\tau}{\nu}\\
   \configm{A}{\sigma'}{\mu'} \trans{\phi_1} \configm{\skipc}{\tau'}{\nu'}
}{
   \configr{\syncbi{A}}{\sigma}{\sigma'}{\mu}{\mu'}
   \biTrans{\phi}
   \configr{\syncbi{\skipc}}{\tau}{\tau'}{\nu}{\nu'}
}

\inferrule*[left=bSyncX]{
\mbox{$A$ not a method call}\\
   \configm{A}{\sigma}{\mu} \trans{\phi_0} \Fault
\quad\mbox{or}
\quad
   \configm{A}{\sigma'}{\mu'} \trans{\phi_1} \Fault
}{
   \configr{\syncbi{A}}{\sigma}{\sigma'}{\mu}{\mu'}
   \biTrans{\phi}
   \Fault
}

  \inferrule*[left=bCallS]{
    (\tau|\tau') \in \phi_2(m)(\sigma|\sigma')
  }{
    \configr{\syncbi{m()}}{\sigma}{\sigma'}{\mu}{\mu'}
      \biTrans{\phi} \configr{\syncbi{\skipc}}{\tau}{\tau'}{\mu}{\mu'}
   }

 \inferrule*[left=bCallX]{
     \Fault \in \phi_2(m)(\sigma|\sigma')
  }{
    \configr{\syncbi{m()}}{\sigma}{\sigma'}{\mu}{\mu'}
      \biTrans{\phi} \Fault
   }

 \inferrule*[left=bCall0]{
     \phi_2(m)(\sigma|\sigma') = \emptyset
  }{
    \configr{\syncbi{m()}}{\sigma}{\sigma'}{\mu}{\mu'}
      \biTrans{\phi}
    \configr{\syncbi{m()}}{\sigma}{\sigma'}{\mu}{\mu'}
   }

  \inferrule*[left=bCallE]{
    \mu(m) = B \\ \mu'(m) = B'
  }{
    \configr{\syncbi{m()}}{\sigma}{\sigma'}{\mu}{\mu'}
      \biTrans{\phi} \configr{\splitbi{B}{B'};\syncbi{\Endcall(m)}}{\sigma}{\sigma'}{\mu}{\mu'}
   }

\inferrule*[left=bComL]{ % bSplitL
   \configm{C}{\sigma}{\mu} \trans{\phi_0} \configm{D}{\tau}{\nu}
\\
DD = (\mifthenelse{ \splitbir{D}{C'} }{ (C'\nequiv\skipc) }{ \splitbi{D}{\skipc} })
}{
   \configr{\splitbi{C}{C'}}{\sigma}{\sigma'}{\mu}{\mu'}
   \biTrans{\phi}
   \configr{DD}{\tau}{\sigma'}{\nu}{\mu'}
}

\inferrule*[left=bComR]{ % bSplitR
   \configm{C'}{\sigma'}{\mu'} \trans{\phi_1} \configm{D'}{\tau'}{\nu'}
}{
   \configr{\splitbir{C}{C'}}{\sigma}{\sigma'}{\mu}{\mu'}
   \biTrans{\phi}
   \configr{\splitbi{C}{D'}}{\sigma}{\tau'}{\mu}{\nu'}
}

% old
%% \inferrule*[left=bSplitR]{
%%    \configm{C'}{\sigma'}{\mu'} \trans{\phi_1} \configm{D'}{\tau'}{\nu'}
%% \\
%% DD = (\mifthenelse{ \splitbi{C}{D'} }{ (C\nequiv\skipc) }{ \splitbir{\skipc}{D'}) }
%% }{
%%    \configr{\splitbir{C}{C'}}{\sigma}{\sigma'}{\mu}{\mu'}
%%    \biTrans{\phi}
%%    \configr{DD}{\sigma}{\tau'}{\mu}{\nu'}
%% }

\inferrule*[left=bComR0]{ % bComR0
   \configm{C'}{\sigma'}{\mu'} \trans{\phi_1} \configm{D'}{\tau'}{\nu'}
}{
   \configr{\splitbi{\skipc}{C'}}{\sigma}{\sigma'}{\mu}{\mu'}
   \biTrans{\phi}
   \configr{\splitbi{\skipc}{D'}}{\sigma}{\tau'}{\mu}{\nu'}
}

%% \inferrule*[left=bSplitRX]{
%%    \configm{C'}{\sigma'}{\mu'} \trans{\phi_1} \Fault
%% }{
%%    \configr{\splitbi{\skipc}{C'}}{\sigma}{\sigma'}{\mu}{\mu'}
%%    \biTrans{\phi}
%%    \Fault
%% }

\inferrule*[left=bComLX]{ % bSplitLX
   \configm{C}{\sigma}{\mu} \trans{\phi_0} \Fault
}{
   \configr{\splitbi{C}{C'}}{\sigma}{\sigma'}{\mu}{\mu'}
   \biTrans{\phi}
   \Fault
}

\inferrule*[left=bComRX]{ % bSplitRX
   \configm{C'}{\sigma'}{\mu'} \trans{\phi_1} \Fault\\
   BB \mbox{ is } \splitbir{C}{C'} \mbox{ or } \splitbi{\skipc}{C'}
}{
   \configr{BB}{\sigma}{\sigma'}{\mu}{\mu'}
   \biTrans{\phi}
   \Fault
}

\inferrule*[left=bLet]{
%m\notin dom(\mu) \\
\nu = \extend{\mu}{m}{C}\\
\nu' = \extend{\mu'}{m}{C'}
}{
  \configr{ \letcombi{m}{\splitbi{C}{C'}}{DD} }{\sigma}{\sigma'}{\mu}{\mu'}
  \biTrans{\phi}
   \configr{DD;\syncbi{\Endlet(m)}}{\sigma}{\sigma'}{\nu}{\nu'}
}

\inferrule*[left=bIfTT]{
\sigma(E)= \True = \sigma'(E')
}{
   \configr{\ifcbi{E|E'}{CC}{DD}}{\sigma}{\sigma'}{\mu}{\mu'}
   \biTrans{\phi}
   \configr{CC}{\sigma}{\sigma'}{\mu}{\mu'}
}

\inferrule*[left=bIfFF]{
\sigma(E)= \False = \sigma'(E')
}{
   \configr{\ifcbi{E|E'}{CC}{DD}}{\sigma}{\sigma'}{\mu}{\mu'}
   \biTrans{\phi}
   \configr{DD}{\sigma}{\sigma'}{\mu}{\mu'}
}

\inferrule*[left=bIfX]{
   \sigma(E)\neq \sigma'(E') % \\ CC\nequiv DD
}{
   \configr{\ifcbi{E|E'}{CC}{DD}}{\sigma}{\sigma'}{\mu}{\mu'}
   \biTrans{\phi}
   \Fault
}

\inferrule*[left=bVar]{
w = \varfresh(\sigma)\\
w' = \varfresh(\sigma')\\
\tau = \extend{\sigma}{w}{\Default{T}} \\
\tau' = \extend{\sigma'}{w'}{\Default{T'}} \\
DD = (\syncbi{\Endvar(w)} \mbox{ if } w\equiv w' \mbox{ else } \splitbi{\Endvar(w)}{\Endvar(w')})
}{
  \configr{ \varblockbi{x\scol T|x'\scol T'}{CC} }{\sigma}{\sigma'}{\mu}{\mu'}
  \biTrans{\phi}
%n\configr{\subst{CC}{x,x'}{w,w'};\splitbi{\Endvar(w)}{\Endvar(w')}}
\configr{\subst{CC}{x,x'}{w,w'};DD}
{\tau}{\tau'}{\mu}{\mu'}
}

\inferrule*[left=bSeq]{
  \configr{ BB }{\sigma}{\sigma'}{\mu}{\mu'}
   \biTrans{\phi}
  \configr{ CC }{\tau}{\tau'}{\nu}{\nu'}
 }{
  \configr{ BB ; DD }{\sigma}{\sigma'}{\mu}{\mu'}
   \biTrans{\phi}
  \configr{ CC ; DD }{\tau}{\tau'}{\nu}{\nu'}
}

\inferrule*[left=bSeqX]{
  \configr{ BB }{\sigma}{\sigma'}{\mu}{\mu'}
   \biTrans{\phi}  \Fault
}{
  \configr{ BB ; DD }{\sigma}{\sigma'}{\mu}{\mu'}
   \biTrans{\phi}  \Fault
}
\end{mathpar}
\end{footnotesize}
%\hrule
\vspace*{-2ex}
\caption{Transition rules for biprograms, except bi-while (for which see Figure~\ref{fig:biprogTransU}).}
\label{fig:biprogTrans}
\end{figure}

%%%%%%%%%%%%%%%%%%%%%%%%%%%%%%%%%%%%%%%%%%%%%%%%%%%%%%%%%%%%%%

\begin{figure}[t!]
\begin{small}
\begin{mathpar}
\inferrule*[left=bWhL]{
  \sigma(E)=\True \\
  \sigma|\sigma'\models\P
}{
   \configr{CC}{\sigma}{\sigma'}{\mu}{\mu'}
   \biTrans{\phi}
   \configr{\splitbi{\Left{BB}}{\skipc};CC}{\sigma}{\sigma'}{\mu}{\mu'}
}

\inferrule*[left=bWhR]{
  \sigma'(E')=\True \\
  \sigma|\sigma'\models\P' \\
  (\sigma(E)=\False \mbox{ or } \sigma|\sigma'\not\models\P)
}{
   \configr{CC}{\sigma}{\sigma'}{\mu}{\mu'}
   \biTrans{\phi}
   \configr{\splitbi{\skipc}{\Right{BB}};CC}{\sigma}{\sigma'}{\mu}{\mu'}
}

\inferrule*[left=bWhTT]{
  \sigma|\sigma'\not\models\P \\
  \sigma|\sigma'\not\models\P'\\
  \sigma(E)= \True = \sigma'(E')
}{
   \configr{CC}{\sigma}{\sigma'}{\mu}{\mu'}
   \biTrans{\phi}
   \configr{BB;CC}{\sigma}{\sigma'}{\mu}{\mu'}
}

\inferrule*[left=bWhFF]{
  \sigma(E) =  \False = \sigma'(E')
}{
   \configr{CC}{\sigma}{\sigma'}{\mu}{\mu'}
   \biTrans{\phi}
   \configr{\syncbi{\skipc}}{\sigma}{\sigma'}{\mu}{\mu'}
}

\inferrule*[left=bWhX]{
  \mbox{($\sigma(E)=\True$ and $\sigma'(E')=\False$ and   $\sigma|\sigma'\not\models\P$)}
\\
  \mbox{or ($\sigma(E)=\False$ and $\sigma'(E')=\True$ and $\sigma|\sigma'\not\models\P'$)}
}{
   \configr{CC}{\sigma}{\sigma'}{\mu}{\mu'}
   \biTrans{\phi}
   \Fault
}
\end{mathpar}
\end{small}
\vspace*{-2ex}
\caption{Transition rules for bi-while,  in which we
abbreviate
$CC\;\equiv\; \whilecbiA{E|E'}{\P|\P'}{BB}$.
}
\label{fig:biprogTransU}
\end{figure}

Biprograms are given transition semantics by relation \ghostbox{$\biTrans{\phi}$}
on configurations, 
defined in Figs.~\ref{fig:biprogTrans} and~\ref{fig:biprogTransU}
for any (relational) pre-model $\phi$.
Configurations have the form
$\configm{CC}{\sigma|\sigma'}{\mu|\mu'}$ which represents an aligned pair of unary configurations.
These have projections
$\Left{\configr{CC}{\sigma}{\sigma'}{\mu}{\mu'}}   \eqdef   \configm{\Left{CC}}{\sigma}{\mu}$
and
$\Right{\configr{CC}{\sigma}{\sigma'}{\mu}{\mu'}}   \eqdef   \configm{\Right{CC}}{\sigma'}{\mu'} $.
Environments are unchanged from unary semantics: $\mu$ and $\mu'$ map procedure names to commands, not biprograms.\footnote{This simplification streamlines the development but is revisited in section \ref{sec:nestedX}.}
The rules are designed to ensure quasi-determinacy (see Lemma~\ref{lem:Rdeterminacy}).
%for any pre-model and any initial configuration these outcomes are mutually exclusive: fault, normal termination, and divergence
% A corollary of Lemmas in te~\ref{lem:Rdeterminacy} and~\ref{lem:bi-to-unary}

The bi-com $\splitbi{C}{C'}$ represents a pair of programs for which the only alignment of interest is the initial states and the final states (if any).
Its steps are dovetailed, unless one side has terminated, so that divergence on one side cannot prevent progress on the other side. It make direct use of the unary transition relation.
The exact order of dovetailing does not matter; what matters is that one-sided divergence is not possible.
Here are the details of the specific formulation we have chosen.
The bi-com $\splitbi{C}{C'}$ takes a step on the left (rule \rn{bComL} in Figure~\ref{fig:biprogTrans}),
leaving the right side unchanged.
It transitions to the \dt{r-bi-com} form \ghostbox{$\splitbir{C}{C'}$} 
which does not occur in source programs, and which takes a right step (\rn{bComR}).
In configurations, identifier $CC$ ranges over biprograms that may include
endmarkers from the unary semantics and also the r-bi-com.\footnote{The left and right
    projections of $\splitbir{-}{-}$ are as with $\splitbi{-}{-}$.
} %foot
Rule \rn{bComR0} is needed to handle biprograms of the form $\splitbi{\skipc}{D}$.
The rules ensure that $\splitbir{\skipc}{D}$ never occurs
for $D\nequiv\skipc$, and we identify $\splitbir{\skipc}{\skipc}\equiv\syncbi{\skipc}$.

Rules \rn{bSeq} and \rn{bSeqX} simply close the transitions under command sequencing.
Recall that we identify some biprograms,
e.g., $\splitbi{\skipc}{\skipc} \equiv \syncbi{\skipc}$,
to avoid the need for bureaucratic transitions (see Figure~\ref{fig:synident}).
A \dt{trace} $T$ via $\phi$ is a finite sequence of configurations that is consecutive under $\biTrans{\phi}$.
%We sometimes treat $T$ as a map defined on an initial segment of the naturals, so $\dom(T)$
% is the set $\{0,\ldots,len(T)-1\}$.
The projection lemma (Lemma~\ref{lem:bi-to-unary}) confirms that $T$ gives rise to unary trace $U$ on the left via $\trans{\phi_0}$ and
$V$ on the right via $\trans{\phi_1}$.

\begin{figure}[t]
\begin{small}
  \begin{tikzpicture}[
    a1/.style={thin, dash pattern=on 2pt off 2pt}
    ]
% Node names refer to traces U,T,V with subscripts
    \matrix[column sep=5mm]{

      \node (U0) {$\configc{a;b;c}$};
      & \node (T0) {$\configc{\splitbi{a;b;c}{d;e;f;g}}$};
      & \node (V0) {$\configc{d;e;f;g}$};       \\[-1ex]

      \node (U1) {$\configc{b;c}$};
      & \node (T1) {$\configc{\splitbir{b;c}{d;e;f;g}}$}; \\[-1ex]

      & \node (T2) {$\configc{\splitbi{b;c}{e;f;g}}$};
      & \node (V1) {$\configc{e;f;g}$};  \\[-1ex]

      \node (U2) {$\configc{c}$};
      & \node (T3) {$\configc{\splitbir{c}{e;f;g}}$}; \\[-1ex]

      & \node (T4) {$\configc{\splitbi{c}{f;g}}$};
      & \node (V2) {$\configc{f;g}$};  \\[-1ex]

      \node (U3) {$\configc{\skipc}$};
      & \node (T5) {$\configc{\splitbir{\skipc}{f;g}}$}; \\[-1ex]

      & \node (T6) {$\configc{\splitbi{\skipc}{g}}$};
      & \node (V3) {$\configc{g}$};  \\[-1ex]

      & \node (T7) {$\configc{\syncbi{\skipc}}$};
      & \node (V4) {$\configc{\skipc}$};  \\[-1ex]
};
    \draw[-] (U0) edge[a1] (T0);
                  \draw[-] (T0) edge[a1] (V0);

    \draw[-] (U1) edge[a1] (T1);
                  \draw[-] (T1) edge[a1, bend right=5] (V0);

    \draw[-] (U1) edge[a1, bend right=8] (T2);
                  \draw[-] (T2) edge[a1] (V1);

    \draw[-] (U2) edge[a1] (T3);
                  \draw[-] (T3) edge[a1, bend right=5] (V1);

    \draw[-] (U2) edge[a1, bend right=10] (T4);
                  \draw[-] (T4) edge[a1] (V2);

    \draw[-] (U3) edge[a1] (T5);
                  \draw[-] (T5) edge[a1, bend right=5] (V2);

    \draw[-] (U3) edge[a1, bend right=8] (T6);
                  \draw[-] (T6) edge[a1] (V3);

    \draw[-] (U3) edge[a1, bend right=10] (T7);
                  \draw[-] (T7) edge[a1] (V4);

  \end{tikzpicture}
% ALERT no blank line here -- want side by side diagrams
\quad
  \begin{tikzpicture}[
    a1/.style={thin, dash pattern=on 2pt off 2pt}
    ]
% Node names refer to traces U,T,V with subscripts
    \matrix[column sep=5mm]{

\node (U0) {$\configc{ a;b;c }$};
& \node (T0) {$\configc{ \splitbi{a}{d;e} ; \splitbi{b;c}{f}  }$};
& \node (V0) {$\configc{ d;e;f }$};
\\[-1ex]

\node (U1) {$\configc{ b;c }$};
& \node (T1) {$\configc{ \splitbir{\skipc}{d;e} ; \splitbi{b;c}{f}  }$};
\\[-1ex]

& \node (T2) {$\configc{ \splitbi{\skipc}{e} ; \splitbi{b;c}{f}  }$};
& \node (V1) {$\configc{ e;f }$};
\\[-1ex]

& \node (T3) {$\configc{ \splitbi{b;c}{f} }$};
& \node (V2) {$\configc{ f }$};
\\[-1ex]

\node (U2) {$\configc{ c }$};
& \node (T4) {$\configc{ \splitbir{c}{f} }$};
\\[-1ex]

& \node (T5) {$\configc{   \splitbi{c}{\skipc} }$};
& \node (V3) {$\configc{ \skipc }$};
\\[-1ex]

\node (U3) {$\configc{ \skipc  }$};
& \node (T6) {$\configc{   \syncbi{\skipc}  }$};
\\[-1ex]
};
    \draw[-] (U0) edge[a1] (T0);
                  \draw[-] (T0) edge[a1] (V0);

    \draw[-] (U1) edge[a1] (T1);
                  \draw[-] (T1) edge[a1, bend right=8] (V0);

    \draw[-] (U1) edge[a1, bend right=8] (T2);
                  \draw[-] (T2) edge[a1] (V1);

    \draw[-] (U1) edge[a1, bend right=15] (T3);
                  \draw[-] (T3) edge[a1] (V2);

    \draw[-] (U2) edge[a1] (T4);
                  \draw[-] (T4) edge[a1, bend right=8] (V2);

    \draw[-] (U2) edge[a1, bend right=10] (T5);
                  \draw[-] (T5) edge[a1] (V3);

    \draw[-] (U3) edge[a1] (T6);
                  \draw[-] (T6) edge[a1, bend right=8] (V3);

  \end{tikzpicture}
\end{small}
\vspace*{-2ex}
\caption{Two example biprogram traces, with alignments, omitting states and environments.}
\label{fig:alignedTraces}
\end{figure}

\begin{example}\label{ex:split}
To illustrate the dovetailed execution of bi-coms,
we show a trace for the bi-com $\splitbi{a;b;c}{d;e;f;g}$ of some atomic commands,
omitting states and environments from the configurations.
%% % ok but redundant:
%% \[
%% \configc{\splitbi{a;b;c}{d;e;f;g}}
%% \configc{\splitbir{b;c}{d;e;f;g}}
%% \configc{\splitbi{b;c}{e;f;g}}
%% \configc{\splitbir{c}{e;f;g}}
%% \configc{\splitbi{c}{f;g}}
%% \configc{\splitbir{\skipc}{f;g}}
%% \configc{\splitbir{\skipc}{g}}
%% \configc{\syncbi{\skipc}}
%% \]
The trace is displayed vertically on the left side of Figure~\ref{fig:alignedTraces},
between the two corresponding unary traces.
Thus $\splitbi{a;b;c}{d;e;f;g}$
executes the commands in the order $a,d,b,e,c,f,g$. 
Dashed lines in the figure show the
correspondence between unary and biprogram configurations.
In this example, the right side takes additional steps after the left has terminated.
The opposite can also happen, as in % this trace:
\(
\configc{\splitbi{a;b;c}{d}}
\configc{\splitbir{b;c}{d}}
\configc{\splitbi{b;c}{\skipc}}
\configc{\splitbi{c}{\skipc}}
\configc{\syncbi{\skipc}}
\)
which executes $a,d,b,c$.

%\qed\end{example}
%\begin{example}

The right side of Figure~\ref{fig:alignedTraces}
shows a trace for the second of the weavings in (\ref{eq:weaveseq}).
\qed\end{example}

The sync atomic command $\syncbi{A}$ steps $A$ by unary transition on both sides,
unless $A$ is a context call in which case the context bi-model is used.
Endmarkers are considered to be atomic commands, e.g., $\syncbi{\Endlet(m)}$
transitions via rule \rn{bSync} and removes $m$ from the environment on both sides.

A bi-if, $\ifcbi{E\smallSplitSym E'}{CC}{DD}$,
faults from initial states that do not agree on the tests $E,E'$,
which we call an \dt{alignment fault} (rule \rn{biIfX}).
%\dt{alignment fault} is like a failed assertion of $E\eqbi E'$.
A bi-while, $\whilecbiA{E\smallSplitSym E'}{\P\smallSplitSym \P'}{CC}$,
executes the left part of the body, $\Left{CC}$, if $E$ and the left alignment guard $\P$ both hold,
and \emph{mutatis mutandis} for the right.  If neither alignment guard holds,
the loop faults unless the tests $E,E'$ agree (\rn{bWhX}).

The transition relation $\biTrans{\phi}$ uses the unary models
$\phi_0$ and $\phi_1$ for method calls in the bi-com form, e.g.,
$\splitbi{m()}{\skipc}$ goes via $\phi_0$ according to \rn{bComL}.
A sync'd call $\syncbi{m()}$
in the body  of a loop that has non-false left or right alignment guards
may give rise to steps where the active biprogram has the form
$\splitbi{m();C}{D}$ or $\splitbi{\skipc}{m();C}$ (rules \rn{bWhL}, \rn{bWhR}).
The \dt{active biprogram}, like the active command in a unary configuration,
is the unique sub-biprogram that gets rewritten by the applicable transition rule.
As with unary programs, we define $\Active(CC)$\index{$\Active(CC)$} to be the unique $BB$ such that $CC\equiv BB;DD$ for some $DD$ and $BB$ is not a sequence; it is what gets rewritten by the applicable transition rule. % in a configuration with code $CC$.

Projecting from a biprogram trace does not simply mean mapping the syntactic projections over the trace,
because that would result in stuttering steps that do not arise in the unary semantics
(where stuttering only happens for context calls and only if the model returns an empty set).
In the preceding diagrams, some unary configurations correspond with more than one biprogram configuration;
one may say the unary program is idling while a step is taken on the other side.

The alignment of biprogram traces with unary ones is formalized as follows.
Here we treat a trace $T$ as a map defined on an initial segment of the naturals, so $\dom(T)$
is the set $\{0,\ldots,len(T)-1\}$.

\begin{definition}[\textbf{schedule, alignment}, $\Align(l,r,T,U,V)$]
\index{$\Align$}
\label{def:schedAlign}
Let $T$ be a biprogram trace and $U,V$ unary traces.
A \dt{schedule of $U,V$ for $T$} is a pair $l,r$ with $l:(\dom(T))\to(\dom(U))$ and $r:(\dom(T))\to(\dom(V))$, each surjective and monotonic.
A schedule $l,r$ is an \dt{alignment}
of $U,V$ for $T$, written $\Align(l,r,T, U, V)$,
\index{$\Align(l,r,T, U, V)$}
iff % $l,r$ is a schedule of $U,V$ for $T$ and
$U_{l(i)} = \Left{T_i}$ and $V_{r(i)} = \Right{T_i}$
for all $i$ in $\dom(T)$.
\end{definition}

The dashed lines in Figure~\ref{fig:alignedTraces} represent the $l$ and $r$ index mappings of a schedule.
For Example~\ref{ex:split}, left side of the figure, the mapping is $r(0)=0$, $r(1)=0$, $r(2)=1$, etc.

The following result makes precise that every biprogram trace represents a pair of unary traces.
It is phrased carefully to take into account the possibility of stuttering transitions at the unary level.
% said earlier:
%The only situation in which a unary configuration transitions to itself is when the active command is a context call and the pre-model returns no outcome for that state (transition rule \rn{uCall0}).

\begin{restatable}[trace projection]{lemma}{lembitounary}
\label{lem:bi-to-unary}
\upshape
Suppose $\phi$ is a pre-model.
Then the following hold.
(a) For any step
$\configr{BB}{\sigma}{\sigma'}{\mu}{\mu'}\biTrans{\phi} \configr{CC}{\tau}{\tau'}{\nu}{\nu'}$,
either
\begin{itemize}
\item
$\configm{\Left{BB}}{\sigma}{\mu}\trans{\phi_0}\configm{\Left{CC}}{\tau}{\nu}$  and
$\configm{\Right{BB}}{\sigma'}{\mu'}\trans{\phi_1}\configm{\Right{CC}}{\tau'}{\nu'}$, or
\item $\configm{\Left{BB}}{\sigma}{\mu} = \configm{\Left{CC}}{\tau}{\nu}$ and
$\configm{\Right{BB}}{\sigma'}{\mu'}\trans{\phi_1}\configm{\Right{CC}}{\tau'}{\nu'}$, or
\item
$\configm{\Left{BB}}{\sigma}{\mu}\trans{\phi_0}\configm{\Left{CC}}{\tau}{\nu}$ and
$\configm{\Right{BB}}{\sigma'}{\mu'} = \configm{\Right{CC}}{\tau'}{\nu'}$.
\end{itemize}
(b) For any trace $T$
via $\biTrans{\phi}$,
there are unique traces $U$ via $\trans{\phi_0}$ and
$V$ via $\trans{\phi_1}$,
and schedule $l,r$, such that $\Align(l,r,T,U,V)$.
\\
(c)  If $\Active(BB)\equiv \Syncbi{B}$ for some $B$, then
$\configm{\Left{BB}}{\sigma}{\mu}\trans{\phi_0}\configm{\Left{CC}}{\tau}{\nu}$  and
$\configm{\Right{BB}}{\sigma'}{\mu'}\trans{\phi_1}\configm{\Right{CC}}{\tau'}{\nu'}$.
\end{restatable}

\subsection{Relational context models, biprogram correctness and adequacy}
\label{sec:rel-corr}

Owing to careful design of Defs.~\ref{def:ctxinterp}, \ref{def:valid}, and~\ref{def:interpRel},
the following notions are mostly about relational aspects.
Relational context models are pre-models that satisfy some specs.
They play the same role in the semantics of relational judgments
as unary context models play in unary correctness.

\begin{definition}[\textbf{context model of relational spec, $\Phi$-model}]
\label{def:ctxinterpRel}\index{context model}
% ok but omit
% Let $\Phi$ be unary specs $\Phi_0$ for $\Gamma$,
% unary specs $\Phi_1$ for $\Gamma'$,
% and relational specs $\Phi_2$ wf in $\Gamma|\Gamma'$.
A pre-model $\phi$ is a \dt{$\Phi$-model} provided that
$\phi_0,\phi_1$ are $\Phi_0, \Phi_1$-models, and
for each $m$, with $\Phi_2(m) = \rflowty{\R}{\S}{\effe|\effe'}$,
the bi-model $\phi_2(m)$ satisfies the following,
for all $\sigma,\sigma'$
\begin{list}{}{}
\item[(a)] $\Fault \in \phi_2(m)(\sigma,\sigma')$ iff
there are no $\pi,\ol{v},\ol{v}'$ such that $\sigma|\sigma' \models_\pi  \subst{\R}{\ol{s},\ol{s}'}{\ol{v},\ol{v}'} $ \\
where $\ol{s},\ol{s}'$ are the spec-only variables on left and right.
\item[(b)] for all  $(\tau,\tau')$ in $\phi_2(m)(\sigma,\sigma')$,
and all $\pi,\ol{v},\ol{v}'$ such that $\sigma|\sigma' \models_\pi \subst{\R}{\ol{s},\ol{s}'}{\ol{v},\ol{v}'}$ we have
$\tau|\tau'\models_\pi \subst{\S}{\ol{s},\ol{s}'}{\ol{v},\ol{v}'} $
and
$\sigma\allowTo\tau\models \effe$ and
$\sigma'\allowTo\tau'\models \effe'$
\end{list}
\end{definition}

A direct consequence of Def.~\ref{def:ctxinterpRel},
together with unary compatibility of pre-models
and condition (c) of Def.~\ref{def:ctxinterp}, is that
for all $N$ with $\mdl(m)\imports N$, letting $\delta\eqdef\bnd(N)$ we have
\[ %\label{eq:boundMonoR}
(\tau|\tau') \in \phi_2(m)(\sigma|\sigma') \mbox{ implies }
\rlocs(\sigma,\delta)\subseteq\rlocs(\tau,\delta) \mbox{ and }
\rlocs(\sigma',\delta)\subseteq\rlocs(\tau',\delta)
\]
and there is also a direct consequence of condition (d) of Def.~\ref{def:ctxinterp}.

% Obsolete - partly folded into def of wf rel hyp ctx
%% \begin{definition}[\bf \SCompat]
%% \index{$\SCompat$}
%% Define $\SCompat(\Theta,m)$ iff $\dom(\Theta)=\{m\}$
%% and the following condition holds:
%% Suppose $\Theta_2(m)$ is $\rflowty{\R}{\S}{\effe|\effe'}$,
%% and the unary specs
%% $\Theta_0(m)$ and $\Theta_1(m)$ are
%% $\flowty{R_0}{S_0}{\effe_0}$ and
%% $\flowty{R_1}{S_1}{\effe_1}$ respectively.
%% Then $\effe\equiv\effe_0$, \, $\effe'\equiv\effe_1$,
%% and these formulas are valid:
%% $\R\imp\leftF{R_0}\land\rightF{R_1}$ and
%% $\S\imp\leftF{S_0}\land\rightF{S_1}$.
%% \end{definition}

% ctx model moved to app

The projections of Lemma~\ref{lem:bi-to-unary} are used in the following definition of relational correctness.

\begin{definition}[\textbf{valid relational judgment}
\ghostbox{$\;\Phi\rHVflowtr{}{M}{\P}{CC}{\Q}{\eff|\eff'}\;$}]
\label{def:validR} % formerly def:validR4
The judgment is \dt{valid} iff the following conditions hold for all
states $\sigma$ and $\sigma'$,
$\Phi$-models $\phi$,
% $\Gamma$-environments $\mu$, $\Gamma'$-environments $\mu'$,
refperms $\pi$, and values $\ol{v},\ol{v}'$ such that
$\sigma|\sigma'\models_\pi \subst{\P}{\ol{s},\ol{s}'}{\ol{v},\ol{v}'}$
(where $\ol{s},\ol{s}'$ are the spec-only variables)
\begin{list}{}{}
\item[\quad(\dt{Safety})] It is not the case that
$\configr{CC}{\sigma}{\sigma'}{\_\,}{\,\_} \biTranStar{\phi} \,\Fault$.
\item[\quad(\dt{Post})] $\tau|\tau' \models_\pi \subst{\Q}{\ol{s},\ol{s}'}{\ol{v},\ol{v}'}$
\quad for every $\tau,\tau'$ with
$\configr{CC}{\sigma}{\sigma'}{\_}{\_} \biTranStar{\phi}
\configr{\syncbi{\skipc}}{\tau}{\tau'}{\_}{\_}$
\item[\quad(\dt{Write})]
$\sigma\allowTo\tau\models \eff$ and
$\sigma'\allowTo\tau'\models \eff'$
\quad for every $\tau,\tau'$ with
$\configr{CC}{\sigma}{\sigma'}{\_}{\_} \biTranStar{\phi}
\configr{\syncbi{\skipc}}{\tau}{\tau'}{\_}{\_}$

\item[\quad (\dt{R-safe})]
For every trace $T$ from $\configr{CC}{\sigma}{\sigma'}{\_}{\_}$,
let $U,V$ be the projections of $T$;
then every configuration of $U$ (resp.\ $V$) satisfies r-safe for $(\Phi_0,\eff,\sigma)$ (resp.\ $(\Phi_1,\eff',\sigma'$)).

\item[\quad (\dt{Encap})]
For every trace $T$ from $\configr{CC}{\sigma}{\sigma'}{\_}{\_}$,
let $U,V$ be the projections of $T$;
then every step of $U$ (resp.\ $V$) satisfies
respect for $(\Phi_0,M,\phi_0,\eff,\sigma)$ (resp.\ $(\Phi_1,M,\phi_1,\eff',\sigma')$).

\end{list}
\end{definition}
%R-safe and Encap refer to the projections given by Lemma~\ref{lem:bi-to-unary}.
The values of spec-only variables are uniquely determined by the pre-states, just like in unary specs.
In virtue of the universal quantification over refperms $\pi$, for a
spec in standard form $\rflowtyf{\P}{\later\Q}$, the judgment says for any $\pi$ that supports the agreements in $\P$ there exists an extension $\rho\supseteq\pi$ that supports the agreements in $\Q$.

The following result confirms that the relational judgment is about unary executions.
In particular, a judgment about a bi-com $\splitbi{C}{C'}$ implies the expected
property relating executions of $C$ and $C'$.
The proof uses the embedding Lemma~\ref{lem:unary-to-bi} which says a biprogram's traces cover all the executions of its unary projections, unless it faults.

\begin{restatable}[adequacy]{theorem}{thmbiprogramsoundness}
\label{thm:biprogram-soundness}
Consider a valid judgment  $\Phi\rHVflowtr{}{M}{\P}{CC}{\Q}{\eff|\eff'}$.
%Suppose  $\Phi\rHVflowtr{}{M}{\P}{CC}{\Q}{\eff|\eff'}$ is valid.
Consider any $\Phi$-model $\phi$
and any $\sigma,\sigma',\pi$ with $\sigma|\sigma'\models_\pi\P$.
If $\configm{\Left{CC}}{\sigma}{\_}\tranStar{\phi_0}\configm{\skipc}{\tau}{\_}$ and
$\configm{\Right{CC}}{\sigma'}{\_}\tranStar{\phi_1}\configm{\skipc}{\tau'}{\_}$
then $\tau|\tau'\models_\pi\Q$.
Moreover, all executions from
$\configm{\Left{CC}}{\sigma}{\_}$ and from
$\configm{\Right{CC}}{\sigma'}{\_}$
satisfy Safety, Write, R-safe, and Encap in Def.~\ref{def:valid}.
\end{restatable}

\begin{remark}
\upshape
It is not straightforward to formalize a converse to this result.
The judgment about $CC$ says not only that the underlying unary executions are related as in the conclusion of the theorem, but in addition certain intermediate states are in agreement according to the alignment designated by the bi-ifs and bi-whiles in $CC$.
\qed\end{remark}

\section{Relational logic}\label{sec:rellog}

This section presents the rules for proving relational correctness judgments.
Section~\ref{sec:locEq} defines how local equivalence specs are derived from unary specs.
Section~\ref{sec:MLink} gives the proof rules and discusses them,
including the derivation of the modular linking rule \rn{rMLink}, sketched as (\ref{eq:mismatchR}) in Section~\ref{sec:modrel}.
Section~\ref{sec:refpmono} considers derived rules involving framing and the $\later$ modality.
Section~\ref{sec:lockstep} states and explains the lockstep alignment lemma, which
is the key to proving soundness of rules \rn{rLocEq}, \rn{rSOF}, and \rn{rLink}
from which \rn{rMLink} is derived.
Section~\ref{sec:nestedX} considers nested linking
and Section~\ref{sec:uequiv} addressess unconditional equivalences.
For Section~\ref{sec:lockstep} readers need to be familiar with the semantic definitions in Section~\ref{sec:biprog}.

\begin{restatable}[soundness of relational logic]{theorem}{thmsound}
\label{thm:sound}
All the relational proof rules are sound (Figure~\ref{fig:proofrulesR} and appendix Figure~\ref{fig:proofrulesRapp}).
\end{restatable}

\subsection{Local equivalence}\label{sec:locEq}

In Section~\ref{sec:modrel} we introduced the notion of local equivalence.
There is a relational proof rule, \rn{rLocEq}, which lifts a unary judgment to a
relational one.  The unary read effect, which has an extensional semantics that is relational (Def.~\ref{def:valid}) gets lifted to an explicit relational property, a local equivalence 
relating a command to itself.
As basis for the proof rule, we now formalize a construction, $\locEq$,
that applies to a unary spec and makes a relational spec---like the spec (\ref{eq:locEqInsert}) in
Example~\ref{ex:PQagree},
and others in Section~\ref{sec:eg:encapR}---that expresses equivalence in terms of the given frame condition and takes into account encapsulation boundaries.

Both unary and relational proof rules have conditions to enforce encapsulation with respect to the boundaries of modules in scope.  For unary this is discussed in Section~\ref{sec:encap}.
The semantic condition Encap, in Def.~\ref{def:valid}, refers to a collective boundary.
This is an effect formed as a union of the relevant boundaries,
for example in the expression $\unioneff{N\in\Phi,N\neq M}{\bnd(N)}$ where $M$ is the current module
and $\Phi$ is the hypothesis context.
For brevity, several relational proof rules are expressed using $\delta$ to name the collective boundary;
in particular rule \rn{rLocEq} which introduces the $\locEq$ spec we now define.

Given a boundary $\delta$ and unary spec $\flowty{P}{Q}{\eff}$,
the desired pre-relation expresses agreement on the readable locations.
Absent a boundary, this can be written $\Agr\eff$,
taking advantage of our abbreviations which say that $\Agr\eff$
abbreviates $\Agr\reads(\eff)$ which in turn abbreviates a conjunction of agreement formulas
(Figure~\ref{fig:relFormulas}).
But we should avoid requiring agreement on variable $\lloc$, as we want to allow entirely different data structures within boundaries.
The requisite agreement can be expressed,
using effect subtraction, as $\Agr(\eff\setminus\delta^\oplus)$,
where $\delta$ is the collective boundary of the modules to be respected.
Note that  $\delta^\oplus$ abbreviates $\delta,\rd{\lloc}$
(as in Def.~\ref{def:ctxinterp}).

A first guess for the post-relation  would use agreement on the writable locations, but
that cannot be written as $\Agr\wTor(\eff)$ because any state-dependent region expressions in
write effects of $\eff$ should be interpreted in the pre-state.  % according to Def.~\ref{def:valid}.
This is why the concluding agreements in the definition of r-respect are expressed in terms of the fresh and written locations. %  (Def.~\ref{def:valid}).
So this is what we need to express in a spec.
The solution is to use snapshot variables.
If we use fresh variable $s_{\lloc}$ in precondition $s_{\lloc}=\lloc$, the fresh
references can be described in post-states as $\lloc\setminus s_{\lloc}$ and
agreement on fresh locations can be expressed as
$\Agr(\lloc\setminus s_{\lloc})\Img\allfields$.
For written (pre-existing) locations, we can obtain the requisite agreements
in terms of initial snapshots of the locations deemed writable by $\eff$.
%The interpretation of $\eff$, i.e., Write condition (Def.~\ref{def:valid}), will ensure that these cover all written locations. % DN unnecessary dependence on semantics 
For an example, see (\ref{eq:locEqMST}) in Section~\ref{sec:eg:encapR}.

For each $\wri{G\Img f}$ in $\eff$ we add a snapshot equation $s_{G,f} = G$
to the precondition, or rather $\Both{(s_{G,f} = G)}$.
The desired post-relation is then $\Agr s_{G,f}\Img f$.
Please note that $s_{G,f}$ is just a fresh identifier, written in a way to keep track of its use in connection with $G\Img f$.
The snapshots and agreements are given by functions $\snap$ and $\Asnap$ defined next.
The following definitions make use of effects like $\rd{s_{G,f}\Img f}$ in which spec-only variables occur.  These are used to define agreement formulas used in postconditions---they are not used in frame conditions, where spec-only variables are disallowed.

\begin{definition}[\textbf{write snapshots}]
\label{def:Gsnapshot}
For any effect $\eff$ we define functions
$\snap$ \index{$\snap$} from effects to unary formulas
and $\Asnap$ \index{$\Asnap$} from effects to read effects.
\[
\begin{array}{lcllcl}
\snap(\eff,\effe) & \eqdef & \snap(\eff) \land \snap(\effe)
\quad
  &\Asnap(\eff,\effe) & \eqdef  & \Asnap(\eff),\ \Asnap(\effe)\\
\snap(\wri{x}) & \eqdef & \True
  &\Asnap(\wri{x}) & \eqdef  &
  \mifthenelse{ \rd{x} \; }{x\nequiv\lloc}{ \emptyeff }  \\
\snap(\wri{G\Img f}) & \eqdef & s_{G,f} = G
  &\Asnap(\wri{G\Img f}) & \eqdef  & \rd{s_{G,f}\Img f} \\
\snap(\wri{G\Img \allfields}) & \eqdef & s_{G,\allfields} = G
  &\Asnap(\wri{G\Img \allfields}) & \eqdef  & \rd{s_{G,\allfields}\Img f}, \rd{s_{G,\allfields}\Img g}, \dots\\
\snap(\ldots) & \eqdef & \True
  &\Asnap(\ldots) & \eqdef  & \emptyeff
\end{array}
\]
\end{definition}

Notice that $\Asnap$ omits $\lloc$ and uses the snapshot variables introduced by $\snap$.\footnote{\label{fn:msnapshot}The snapshot variables used should be distinct from each other,
distinct from the ones used in the original spec,
and also globally unique so that the local equivalence specs of different methods use different variables.
In the definition of $\LocEq$, where multiple method specs are considered,
we adopt the convention of naming snapshots for method $m$ as $s_{G,f}^m$
(and $\snap^m$, $\Asnap^m$ for short), to distinguish them from each other and from the snapshots
used in the conclusion of a  judgment.} Notice also that in the case $\Asnap(\wri{G\Img \allfields})$ a single snapshot variable $s_{G,\allfields}$ is used, but the image expression in $G\Img \allfields$ gets expanded to the constituent fields ($f, g, \dots$).

The following result confirms that $\Asnap$ serves the purpose of designating the writable locations from the perspective of the post-state.
It uses semantic notions from Sects.~\ref{sec:states} and~\ref{sec:effect}.
\begin{restatable}{lemma}{lemsnap}
\label{lem:snap}
\upshape
If $\tau\models\snap(\eff)$ and $\tau\allowTo\upsilon\models\eff$ then
$\wlocs(\tau,\eff)\setminus\rlocs(\upsilon,\delta^\oplus)=
\rlocs(\upsilon,\Asnap(\eff)\setminus\delta)$.
\end{restatable}

The following definition of $\locEq$ uses effect subtraction to avoid asserting agreement inside the given boundary,
in both pre and post.
For example, if $\eff$ includes $\wri{x},\wri{G\Img f}$
we convert to read effects and use the snapshot variable:
$\rd{x},\rd{s_{G,f}\Img f}$.  Then
$(\rd{x},\rd{s_{G,f}\Img f})\setminus \delta$ will remove $x$ if $\rd{x}$ is in $\delta$,
and result in $\rd{(s_{G,f}\setminus H)\Img f}$ if $\rd{H\Img f}$ is in $\delta$.

\begin{definition}[\textbf{local equivalence}]
\label{def:loceq}
For spec $\flowty{P}{Q}{\eff}$ and boundary $\delta$,
define relational spec
\(
\begin{array}[t]{lcl}
\ghostbox{$\locEq_\delta(\flowty{P}{Q}{\eff})$}\index{$\locEq$}\index{$\LocEq$} 
& \!\!\!\eqdef\!\!\! &
\rflowty{\Both{P}\land\Agr\eff^\leftarrow_\delta\land
          \Both{(s_{\lloc}=\lloc \land \snap(\eff))}}
{\later(\Both{Q}\land\Agr\eff^\rightarrow_\delta)}{\eff}
\\[1ex]
&& \mbox{where }
  \ghostbox{$\eff^\leftarrow_\delta$}\eqdef\reads(\eff)\setminus\delta^\oplus
%  \quad\mbox{and}\quad
  \mbox{ and }
  \ghostbox{$\eff^\rightarrow_\delta$}\eqdef (\rd{(\lloc\setminus s_{\lloc})\Img\allfields},\Asnap(\eff))\setminus\delta
\end{array}
\)
\\
For unary context $\Phi$,
define \ghostbox{$\LocEq_\delta(\Phi)$} $\eqdef (\Phi,\Phi,\Phi_2)$ where
$\Phi_2(m)$ is $\locEq_\delta(\Phi(m))$ for each $m\in\Phi$.
\end{definition}
If $\flowty{P}{Q}{\eff}$ and $\delta$ are wf in $\Gamma$ then
$\locEq_\delta(\flowty{P}{Q}{\eff})$ is wf in $\Gamma|\Gamma$ and has the same spec-only variables on both sides.

\begin{sloppypar}
Recall from Section~\ref{sec:encap} the Stack client  with precondition $P \eqdef c\in r\land \disj{r}{(pool\union pool\Img rep)}$
and frame $\eff \eqdef \rw{c,r,\lloc, r\Img val}$,
where the boundary $\delta$ is $\rd{pool, pool\Img \allfields, pool\Img rep\Img\allfields}$.
For the precondition, the reads are
$\rd{c},\rd{r},\rd{\lloc},\rd{r\Img val}$.
Subtracting $\delta^\oplus$ leaves the variables $c,r$ and is more interesting for $r\Img val$.
Expanding abbreviation $\allfields$ and discarding empty regions,
we are left with $\rd{ (r\setminus(pool\union pool\Img rep))\Img val }$.
So the precondition $\Agr\eff^\leftarrow_\delta$ is
$\Agr c \land \Agr r \land \Agr (r\setminus(pool\union pool\Img rep))\Img val$.
(In conjunction with $\Both P$, the formula $\Agr(r\setminus(pool\union pool\Img rep))\Img val$ is equivalent to $\Agr r \Img val$.)
There is a snapshot variable in precondition $s_{r,val} = r$, due to $\wri{r\Img val}$.
It is used in this conjunct of the $\Asnap$ part of the postcondition:
$\Agr (s_{r,val}\setminus(pool\union pool\Img rep))\Img val$.
\end{sloppypar}

\subsection{Relational proof rules and derivation of \rn{rMLink}}\label{sec:MLink}

\begin{figure}[t!]
\begin{footnotesize}
\begin{mathpar}

\rulerLink

\inferrule*[left=rWeave]{
  \Phi \rHPflowtr{}{}{\P}{DD}{\Q}{\eff|\eff'} \\
  CC \weave^* DD
}{  \Phi \rHPflowtr{}{}{\P}{CC}{\Q}{\eff|\eff'}  }

\rulerCall

\inferrule*[left=rAlloc]{ \ind{ \unioneff{L\in(\Phi),L\neq M}{\bnd(L)} }{\wri{x},\wri{\lloc}}
}{
\Phi\proves_M \syncbi{x:=\new{K}} : \rflowty{\True}{\later(x\eqbi x)}{\wri{x},\rw{\lloc}}
}

\inferrule*[left=rEmpPre]{}{
   \Phi\rHPflowtr{}{}{\False}{CC}{\Q}{\eff|\eff'}  }

\rulerLocEq

\inferrule*[left=rEmb]{
   \Phi_0\HPflowtr{}{}{P}{C}{Q}{\eff}\\
   \Phi_1\HPflowtr{}{}{P'}{C'}{Q'}{\eff'}\\
}{ \Phi\rHPflowtr{}{}{\leftF{P}\land\rightF{P'}}{\splitbi{C}{C'}}{\leftF{Q}\land\rightF{Q'}}{\eff|\eff'}
}

\inferrule*[left=rPoss]
{\Phi\rHPflowtr{}{}{\P}{CC}{\Q}{\eff|\eff'}
}
{ \Phi\rHPflowtr{}{}{\later\P}{CC}{\later\Q}{\eff|\eff'}
}

\rulerSOF

% same as copy in fig with its derivation
%% \inferrule*[left=rMLink]{
%% 	\Phi\HPflowtr{}{\emptymod}{P}{C}{Q}{\eff}\\
%% 	\Phi\conjInv \M \proves_M \splitbi{B}{B'} : \locEq_\delta(\Phi)(m)\conjInv\M \\
%%         \delta=\bnd(M) \\
%% 	\Phi_0\conjInv \Left{\M} \proves_M B : \Phi_0(m)\land\Left{\M} \\
%% 	\Phi_1\conjInv \Right{\M} \proves_M B' : \Phi_1(m)\land\Right{\M} \\
%%         M = \mdl(m) \\
%% 	P\models \wTor(\eff)\leq\reads(\eff)\\
%% 	\models \fra{ \delta| \delta }{ \M } \\
%% 	\ACompat(\Phi, \flowty{P}{Q}{\eff}, \M,\delta)\\
%%         \mbox{$C$ is let-free} \\
%% 	pre(\locEq_\delta(\flowty{P}{Q}{\eff})) \imp \M
%% }{
%% 	\proves_{\emptymod}
%% 	\letcombi{m}{\splitbi{B}{B'}}{\Syncbi{C}} : \locEq_\delta (\flowty{P}{Q}{\eff})
%% }

\inferrule*[left=rFrame]
{ \Phi \rHPflowtr{}{}{\P}{CC}{\Q}{\eff|\eff'} \\
  \P\models \fra{\effe|\effe'}{\R} \\
  \P\land\R \imp \leftF{\ind{\effe}{\eff}} \land \rightF{\ind{\effe'}{\eff'}}
}{
\Phi \rHPflowtr{}{}{\P\land \R}{CC}{\Q\land \R}{\eff|\eff'}
}

\rulerConseq

\inferrule*[left=rDisj]{
\Phi \rHPflowtr{}{}{\P_0}{CC}{\Q}{\eff|\eff'} \\
\Phi \rHPflowtr{}{}{\P_1}{CC}{\Q}{\eff|\eff'}
}{
\Phi \rHPflowtr{}{}{\P_0\lor\P_1}{CC}{\Q}{\eff|\eff'}
}

\inferrule*[left=rConj]{
\Phi \rHPflowtr{}{}{\P}{CC}{\Q_0}{\eff|\eff'} \\
\Phi \rHPflowtr{}{}{\P}{CC}{\Q_1}{\eff|\eff'}
}{
\Phi \rHPflowtr{}{}{\P}{CC}{\Q_0\land\Q_1}{\eff|\eff'}
}

\end{mathpar}
\end{footnotesize}
%\hrule
\vspace{-3ex}
\caption{Selected relational proof rules
(for others see appendix~Figure~\ref{fig:proofrulesRapp}).
The typing context $\Gamma|\Gamma'$ is unchanged thoughout, so omitted.
The current module is omitted in rules where it is the same in all the judgments and unconstrained.
}
\label{fig:proofrulesR}
\end{figure}

Selected proof rules are in Figure~\ref{fig:proofrulesR}.
For relational judgments, the validity conditions (Def.~\ref{def:validR}) have been carefully formulated to leverage the unary ones (Def.~\ref{def:valid}).
This obviates the need for rules like  \rn{CtxIntro} at the relational level.
Rule \rn{rCall}, for aligned calls using a relational spec, relies on unary premises to enforce the requisite encapsulation conditions.
The relational rules for bi-if and bi-while
%Rules \rn{rIf} and \rn{rWhile} DN don't occur in body of paper
have separator conditions to enforce encapsulation, taken straight from their unary rules (e.g., \rn{If} in Figure~\ref{fig:proofrulesU}).
The relational rules for bi-while and sequence
%Rules \rn{rSeq} and \rn{rWhile}
include an immunity condition for framing of their effects, again taken straight from the unary rules.

The linking rule, \rn{rLink}, relates a client command $C$ to itself using relations that imply
its executions can be aligned lockstep.  It can be instantiated with local equivalence specs
but also with more general specs that include hidden invariants and coupling on encapsulated state.
To allow this generality in a sound way, rule \rn{rLink} uses the following notion.

\begin{definition}[\textbf{covariant spec implication} \ghostbox{$\prePostImply$}]
\label{def:prePostImply}
\index{covariant spec implication}
Define
$(\rflowty{\R_0}{\S_0}{\eff_0|\eff'_0}) \prePostImply (\rflowty{\R_1}{\S_1}{\eff_1|\eff'_1}) $
iff $\R_0\imp\R_1$ and $\S_0\imp\S_1$ are valid and the effects are the same:
$\eff_0=\eff_1$ and $\eff'_0=\eff'_1$.
For contexts $\Phi$ and $\Psi$, define $\Phi\prePostImply\Psi$ to  mean
they have the same methods and $\prePostImply$ holds for the relational spec of each method.
\end{definition}
For example we have $\locEq_\delta(spec)\conjInv \M \prePostImply \locEq_\delta(spec)$ for any $\delta,spec,\M$.

In \rn{rLink}, side conditions constrain module imports, exactly as in unary \rn{Link},
as part of the enforcement of encapsulation.
As with \rn{Link}, some of the conditions merely express module structure.
The soundness proof for \rn{rLink} goes by induction on biprogram traces, similar to the soundness proof for unary \rn{Link}; the relational hypothesis can be used because the relevant context calls are aligned (see appendix~\ref{sec:app:link} and~\ref{sec:app:rLink}).

Rule \rn{rEmb} lifts unary judgments to a relational one.  It applies to arbitrary commands.
For example, it can be applied to the $sumpub$ program of (\ref{eq:sumpub}),
to prove the judgment about $\splitbi{sumpub}{sumpub}$ by lifting a unary spec
as described in Section~\ref{sec:weave}.
It is also needed to obtain relational judgments about assignments,
and it enables the use of unary specs in one-sided method calls.

For allocation, there needs to be a way to indicate when a pair of allocations are meant to be aligned;
this is the purpose of \rn{rAlloc}.
Using \rn{rConj}, \rn{rEmb}, the unary rule \rn{Alloc}, and the frame rules, one can add postconditions like $\Agr\sing{x}\Img f$ and freshness of $x$.  (Detailed derivations for freshness can be found in \RLIII~(Section~7.1)).
Like \rn{rCall}, rule \rn{rAlloc} does not have the minimal hypothesis context but rather allows an arbitrary one; this is needed because we do not have context introduction rules at the relational level.
To enforce encapsulation, \rn{rAlloc} has a side condition which simply says neither $x$ nor $\lloc$ occur in the boundaries of any models other than the current one.
% DN said elsewhere
%In our prototype, refperms are manipulated as ghost state,
%and ghost code is used to express intended agreement on fresh objects.

% NOTE: following is now said slightly differently in sec:refpmono
%% Rule \rn{rPoss} is used to derive a sequence rule
%% where the intermediate relation has $\later$.
%% To sequence $\rflowtrf{\P}{CC}{\later\Q}$ with $\rflowtrf{\Q}{DD}{\later\R}$,
%% use  \rn{rPoss} to get $\rflowtrf{\later\Q}{DD}{\later\later\R}$, and
%% then by \rn{rConseq} get $\rflowtrf{\later\Q}{DD}{\later\R}$
%% (as $\later\later\R\imp\later\R$ is valid).
%% Then by \rn{rSeq} we get $\rflowtrf{\P}{CC;DD}{\later\R}$.

%% Along similar lines: from $\rflowtrf{\P}{CC}{\later\Q}$
%% and $\rflowtrf{\P}{CC}{\later\R}$ we get
%% $\rflowtrf{\P}{CC}{\later\Q\land\later\R}$ by \rn{rConj},
%% then \rn{rConseq} can be used to obtain
%% $\rflowtrf{\P}{CC}{\later(\Q\land\R)}$ provided that this instance of agreement compatibility is valid:
%% $\later\Q \land \later \R \imp \later(\Q\land \R)$.

Rule \rn{rLocEq} has a side condition about the unary judgment's frame condition:
the writes must be subsumed by the reads (subeffect judgment $P\models \wTor(\eff)\leq\reads(\eff)$).
This ensures that the precondition of the relational conclusion has agreement for writable locations.
The requirement that $C$ is let-free is needed in accord with Lemma~\ref{lem:rloceq}.
% but is revisited in Sect.~\ref{sec:nested}.

\begin{example}[how framing is used with \rn{rLocEq}]
%\paragraph{Example showing how framing is used with \rn{rLocEq}}

Just as the unary axioms for assignments are ``small'' in the sense that they only describe the locations relevant to the command's behavior, we are interested in program equivalence described in terms of the relevant locations.
As an example, without methods, consider this valid judgment (omitting the module, which is irrelevant):
\[
\proves (x:=y.f; z:=w): y\neq 0 \leadsto true [ \eff ]
\]
where $\eff \eqdef \wri{x,z},\rd{w,y,y.f}$.
It should entail this relational one:
\[
\proves \Syncbi{x:=y.f; z:=w}:
\Both{(y\neq 0)} \land \Agr(y,w,\sing{y}\Img f)
\rspecSym \Both{\True}
\land \Agr(x,z) [\eff]
\]
Desugared, the precondition agreement is
$\Agr y \land \Agr w \land \Agr\sing{y}\Img f$.
The precondition only requires agreement on locations that are read.
The postcondition tells about the variables that are written.
In fact $w$ and $y$ are unchanged, and we can strengthen the postcondition to
\[
\proves \Syncbi{x:=y.f; z:=w}:
\Both{(y\neq 0)} \land \Agr(y,w,\sing{y}\Img f)
\rspecSym \Both{\True}
\land \Agr(x,z,y,w) [\eff]
\]
using the \rn{rFrame} rule, because $\Agr(y,w)$ is separate from the writes.
Rule \rn{rConseq} allows to strengthen the precondition by adding
the agreements $\Agr(u,\sing{y}\Img g)$:
\[
\proves \Syncbi{x:=y.f; z:=w}:
\Both{(y\neq 0)} \land \Agr(y,w,\sing{y}\Img f,u,\sing{y}\Img g)
\rspecSym \Both{\True}
\land \Agr(x,z,y,w) [\eff]
\]
Now rule \rn{rFrame} allows to carry these agreements over the command, because the locations $u$ and $y.g$ are separate from the write effects.
\[
\proves \Syncbi{x:=y.f; z:=w}:
\Both{(y\neq 0)} \land \Agr(y,w,\sing{y}\Img f,u,\sing{y}\Img g)
\rspecSym \Both{\True}
\land \Agr(x,z,y,w,u,\sing{y}\Img g) [\eff]
\]
In summary, the local equivalence spec expresses a program relation
in terms of only the locations readable and writable by the command.
Such equivalence can be extended to arbitrary other locations not touched by the command.
\qed\end{example}

Rule \rn{rSOF} follows the pattern of the unary \rn{SOF} in its use of $\conjInv\M$ from Def.~\ref{def:conjInv}.  It can only be instantiated with specs in standard form,
so that $\conjInv\M$ is defined.
It requires refperm  monotonicity of the coupling, i.e., $\N\imp\always\N$;
more on this in Section~\ref{sec:refpmono}.

\begin{figure}[t]
\begin{footnotesize}
\(
\inferrule*[left=rMLink]{
	\Phi\HPflowtr{}{\emptymod}{P}{C}{Q}{\eff}\\
	\Phi\conjInv \M \proves_M \splitbi{B}{B'} : \locEq_\delta(\Phi)(m)\conjInv\M \\
        \delta=\bnd(M) \\
	\Phi\conjInv \Left{\M} \proves_M B : \Phi(m)\conjInv\Left{\M} \\
	\Phi\conjInv \Right{\M} \proves_M B' : \Phi(m)\conjInv\Right{\M} \\
        M = \mdl(m) \\
	P\models \wTor(\eff)\leq\reads(\eff)\\
	\models \fra{ \delta| \delta }{ \M } \\
        \M\imp\always\M \\
        \mbox{$C$ is let-free} \\
	pre(\locEq_\delta(\flowty{P}{Q}{\eff})) \imp \M
}{
	\proves_{\emptymod}
        \Splitbi{ \letcom{m}{B}{C} }{ \letcom{m}{B'}{C} } : \locEq_\delta (\flowty{P}{Q}{\eff})
}
\)

\vspace*{4ex}

\(
\inferrule*[right={\tiny rWeave}]{ %rWeave
\inferrule*[right={\tiny rConseq}]{ %rConseq
  \inferrule*[right={\tiny rLink}]{
     \inferrule*[right={\tiny rSOF}]{
        \inferrule*[right={\tiny rLocEq}]{
            \Phi\HPflowtr{}{\emptymod}{P}{C}{Q}{\eff}
        }{
            \LocEq_\delta(\Phi) \proves_{\emptymod}
            \Syncbi{C} : \locEq_\delta(\flowty{P}{Q}{\eff})
        }}{
            \Psi \proves_{\emptymod}
                   \Syncbi{C} : \locEq_\delta (\flowty{P}{Q}{\eff}) \conjInv \M
        }
        \Psi\proves_M \splitbi{B}{B'} : \locEq_\delta(\Phi(m))\conjInv\M  \\ \vdots
        }{
        \proves_{\emptymod}
        \letcombi{m}{\splitbi{B}{B'}}{\Syncbi{C}} : \locEq_\delta (\flowty{P}{Q}{\eff}) \conjInv \M
        }
        }{
        \proves_{\emptymod}
        \letcombi{m}{\splitbi{B}{B'}}{\Syncbi{C}} : \locEq_\delta (\flowty{P}{Q}{\eff})
        }
}
{ \proves_{\emptymod}
  \Splitbi{ \letcom{m}{B}{C} }{ \letcom{m}{B'}{C} } : \locEq_\delta (\flowty{P}{Q}{\eff})
}
\)

\end{footnotesize}
\vspace*{-1ex}
\caption{\rn{rMLink} and its derivation,
where $\Psi$ abbreviates $\LocEq_\delta(\Phi)\conjInv\M$,
$\Phi$ specifies $m$, $\delta=\bnd(M)$, and $M=\mdl(m)$.
See text for details.
}
\label{fig:derivedmismatch}
\end{figure}

Figure~\ref{fig:derivedmismatch} presents the relational modular linking rule, \rn{rMLink},
and its derivation.
(Here specialized to a single method, i.e., $\dom(\Phi)=\{m\}$, for clarity).
The side conditions are
$P\models \wTor(\eff)\leq\reads(\eff)$ (for \rn{rLocEq});
$\models \fra{ \delta | \delta }{ \M }$ and
$\M\imp\always\M$
(for \rn{rSOF});
$\dom(\Phi)=\{m\}$ (for \rn{rLink});
and % added for better line break
$pre(\locEq_\delta(\flowty{P}{Q}{\eff})) \imp \M$ (for \rn{rConseq},
to drop $\land\M$ from the precondition; of course $\land\M$ is also
dropped from postcondition).
For \rn{rWeave} we use the fact that
$\Splitbi{ \letcom{m}{B}{C} }{ \letcom{m}{B'}{C} }
    \weave^* \letcombi{m}{\splitbi{B}{B'}}{ \Syncbi{C} }$.
Vertical elipses in the derivation indicate that, in addition to the expected relational premise for $B$ and $B'$,
unary premises are required:
$\Phi\conjInv \Left{\M} \proves_M B : \Phi(m)\conjInv\Left{\M} $
and $\Phi\conjInv \Right{\M} \proves_M B' : \Phi(m)\conjInv\Right{\M}$.
These are required by \rn{rLink}, for technical reasons explained in its proof
(Section~\ref{sec:app:rLink}).

The implication
$pre(\locEq_\delta(\flowty{P}{Q}{\eff})) \imp \M$
refers to the precondition of local equivalence.
Typically, the implication
is valid because $P$ includes initial conditions that imply $\M$ just as in the
case of unary modular linking and module invariant.
This is the responsibility of the module developer, who defines $\M$,
shows its framing by the boundary, and shows refperm monotonicity of $\M$.

\begin{example}[Illustrating \rn{rMLink} with SSSP]\label{ex:SSSP}
% In the code it's procedure Dijkstra and DijkstraNoDec
We instantiate $M$ in the rule with \whyg{PQ} (Section~\ref{sec:Usyntax}) and
$\Phi$ with the specs of \whyg{PQ}'s public methods.
Let $\delta$ be \whyg{PQ}'s dynamic boundary 
$\rd{pool, pool\Img\allfields, pool\Img rep\Img\allfields}$.
We instantiate client $C$
with $C_{sssp}$, an implementation of Dijkstra's single-source shortest-paths
algorithm acting on global variables $gph$, $src$, and $wts$.
For simplicity, $gph$ is a variable of type ``mathematical graph''
for which we use an API supporting usual operations.
We assume the vertex set $V(gph)$ is an initial segment of naturals so the source vertex variable $src$ has type $\INT$.
Edges have positive integer weights.
%An edge is a triple comprising start and end vertices and its positive integer weight.
The integer array $wts$, of length $|V(gph)|$ and allocated by the client, is for the output:
for every vertex $v\in V(gph)$,
$C_{sssp}$ computes in $wts[v]$ the weight of the shortest path from $src$ to $v$.

The unary spec for $C_{sssp}$ is $\flowty{P}{Q}{\eff}$ where
$P \eqdef  src\in V(gph)\land pool = \emptyset$; $Q\eqdef \True$; and
$\eff \eqdef
\rd{gph, src}, \rw{wts, pool, pool\Img\allfields, pool\Img rep\Img\allfields, \lloc}
$.
The trivial postcondition does not specify functional behavior but the spec is still useful.
The local equivalence spec
$\locEq_\delta(\flowty{P}{Q}{\eff})$ is $\rflowty{\R}{\later \S}{\eff}$ where
$\R \eqdef \Both{(src \in V(gph) \land pool = \emptyset \land s_\lloc = \lloc)} \land \Agr(wts, gph, src)$;
and $\S \eqdef \Agr(wts, (\lloc \setminus (s_\lloc \cup pool \cup pool\Img rep))\Img\allfields)$,
eliding details about spec-only variables apart from $s_\lloc$.
Here $s_\lloc$ snapshots $\lloc$ so fresh locations are those in $\lloc\setminus s_\lloc$.  This spec
ensures agreement on fresh locations that are not in \whyg{PQ}'s dynamic boundary.

% $\rflowty{
%   \Both(src\in V(gph) \land
%   pool = \emptyset \land
%   s_\lloc = \lloc) \land
%   \Agr(D,
%        gph,
%        src
%        )
% }
% {
%   \later
%   \Agr(
%   D\effsep
%   (\lloc \setminus (s_\lloc \cup pool \cup pool\Img rep))\Img\allfields
%   )
% } { \eff }$.

The coupling $\M_{PQ}$ is
$\all{q\scol\code{Pqueue}\in pool | q\scol\code{Pqueue}\in pool}{
  \Agr q\imp
  \all{n\in q.rep | n\in q.rep}{ \Agr n \imp \ldots } }
$,
conjoined with the private invariants $I$ and $I'$
(eliding parts shown in Example~\ref{ex:PQagree}).
One side condition of \rn{rMLink} is
$pre(\locEq_\delta(\flowty{P}{Q}{\eff}))\imp\M_{PQ}$
which is easy to show:  expanding definitions,
the antecedent includes $\Both{(pool = \emptyset)}$
which implies the private invariants and the coupling relation.
The subeffect $P\models\wTor(\eff)\leq\reads(\eff)$
is immediate from the definition of $\eff$.
The framing judgment, $\models\fra{\delta | \delta}{\M_{PQ}}$, is easily proved by SMT,
as is refperm monotonicity of $\M_{PQ}$.
\qed\end{example}

\subsection{Refperm monotonicity, standard form, and agreement compatibility}\label{sec:refpmono}

For modular linking and most other purposes, we are concerned with specs in the standard form,
i.e., either $\rflowty{\R}{\later\S}{\effe}$ or $\rflowty{\R}{\S}{\effe}$
where $\R$ and $\S$ are $\later$-free.
In this section we consider the rules that give rise to other forms,
and related notions concerning formulas with $\later$.
It is possible to reformulate the logic to consider only standard form specs.
We choose the present formulation because some proof rules can be simpler and
more orthogonal.

For reasoning about sequential composition one wants to combine
judgments for specs $\rflowtyf{\P}{\later\Q}$ and $\rflowtyf{\Q}{\later\R}$
into a judgment for $\rflowtyf{\P}{\later\R}$ (omitting frame for clarity).
It is easy to derive a rule for specs of this form, from the more basic rule for sequence together
rules \rn{rPoss} and \rn{rConseq}.  From
$\rflowtyf{\Q}{\later\R}$ we get
$\rflowtyf{\later\Q}{\later\later\R}$ by \rn{rPoss}.
Then we get $\rflowtyf{\later\Q}{\later\R}$ by \rn{rConseq}, because $\later\later\R \iff \later\R$ is valid.
From $\rflowtyf{\P}{\later\Q}$ and $\rflowtyf{\later\Q}{\later\R}$
we get $\rflowtyf{\P}{\later\R}$ by the sequence rule.

Similarly, one can derive a relational rule for loops, with premises in standard form and
relational invariant $\Q$ that is $\later$-free.
In accord with the loop rule sketched as (\ref{eq:rWhileSimp}), we elide frame conditions, context,
and side conditions for immunity and encapsulation.
The derived rule looks like this:
\begin{equation}\label{eq:rWhileStd}
\inferrule
{
\proves CC: \rflowtyf{\Q\land\neg\P\land\neg\P'\land\leftF{E}\land\rightF{E'}}{\later\Q}
\\
\proves \splitbi{\Left{CC}}{\skipc} :
 \rflowtyf{\Q\land\P\land\leftF{E}}{\later\Q}
\\
\proves \splitbi{\skipc}{\Right{CC}} :
\rflowtyf{\Q\land\P'\land\rightF{E'}}{\later\Q}
\\
\Q\imp E\eqbi E' \lorbi (\P\land\leftF{E}) \lorbi (\P'\land\rightF{E'})
}{
\proves \whilecbiA{E\smallSplitSym E'}{\P\smallSplitSym \P'}{CC} :
\rflowtyf{\Q}{\later(\Q\land\leftF{\neg E}\land\rightF{\neg E'})}
}
\end{equation}
Given the premises, three applications of \rn{rPoss} yields
$CC: \rflowtyf{\later(\Q\land\neg\P\land\neg\P'\land\leftF{E}\land\rightF{E'})}{\later\later\Q}$,
$\splitbi{\Left{CC}}{\skipc} : \rflowtyf{\later(\Q\land\P\land\leftF{E})}{\later\later\Q}$
and
$\splitbi{\skipc}{\Right{CC}} : \rflowtyf{\later(\Q\land\P'\land\rightF{E'})}{\later\later\Q}$.
But $\later\later\Q$ is equivalent to $\later\Q$.
Furthermore,
$\leftF{E}$ and $\rightF{E'}$ are agreement-free and thus refperm independent.
Also $\P,\P'$ are refperm independent,
because they are agreement free by the wellformedness condition
mentioned at the end of Section~\ref{sec:progtype}.
So, using property (\ref{eq:refp-ind-dist}),
the precondition of the second judgment,
$\later(\Q\land\P\land\leftF{E})$
is equivalent to one where $\later$ is applied only to $\Q$, i.e.,
$\later\Q\land\P\land\leftF{E}$.
Similarly for the other two preconditions.
So by \rn{rConseq} we get
\begin{itemize}
\item $CC: \rflowtyf{\later\Q\land\neg\P\land\neg\P'\land\leftF{E}\land\rightF{E'}}{\later\Q}$
\item $\splitbi{\Left{CC}}{\skipc} :
 \rflowtyf{\later\Q\land\P\land\leftF{E}}{\later\Q}$
\item $\splitbi{\skipc}{\Right{CC}} :
\rflowtyf{\later\Q\land\P'\land\rightF{E'}}{\later\Q}$
\end{itemize}
With these we instantiate the rule (\ref{eq:rWhileSimp}) with $\later\Q$ for $\Q$,
which yields
$\whilecbiA{E\smallSplitSym E'}{\P\smallSplitSym \P'}{CC} :
\rflowtyf{\later\Q}{\later\Q\land\leftF{\neg E}\land\rightF{\neg E'}}$.
Finally, the implication $\Q\imp\later\Q$ is valid
and we can distribute refperm independent formulas under $\later$;
so using \rn{rConseq} we obtain the conclusion of (\ref{eq:rWhileStd}).

For a bi-while with false alignment guards, there is a derived rule with
a single premise $\proves CC: \rflowtyf{\Q\land\leftF{E}\land\rightF{E'}}{\later\Q}$.
It can be derived, using rule \rn{rEmpPre}.

\paragraph{Refperm monotonicity.}

Given a judgment
$\Phi \rHPflowtr{}{}{\P}{CC}{\later\Q}{\eff|\eff'}$,
rule \rn{rFrame} yields
$\Phi \rHPflowtr{}{}{\P\land \R}{CC}{\later\Q\land \R}{\eff|\eff'}$
which is not in the standard form.
But suppose $\R$ is refperm monotonic, i.e., $\R\imp\always\R$ is valid.
Then by (\ref{eq:mono-distrib}) we have
$\later\Q\land \R \imp \later(\Q\land\R)$.  So using \rn{rConseq}
we get this derived frame rule:
\[
\inferrule
{ \Phi \rHPflowtr{}{}{\P}{CC}{\later\Q}{\eff|\eff'} \\
  \P\models \fra{\effe|\effe'}{\R} \\
  \P\land\R \imp \leftF{\ind{\effe}{\eff}} \land \rightF{\ind{\effe'}{\eff'}} \\
  \R\imp\always\R
}{
\Phi \rHPflowtr{}{}{\P\land \R}{CC}{\later(\Q\land \R)}{\eff|\eff'}
}
\]
Refperm monotonicity is also a side condition
for the coupling relation in rule \rn{rSOF}.
In that rule, moving the coupling relation under $\later$ is done by the $\conjInv$ operation (Def.~\ref{def:conjInv}).

Agreement formulas are refperm monotonic, as are refperm independent formulas.
But negation does not preserve refperm monotonicity, and in particular a formula
of the form $\Agr x \imp \R$ is not refperm monotonic even if $\R$ is.
Such implications are used in our example couplings.
In particular, implication is used in the following idiomatic pattern:
\begin{equation}\label{eq:rmono}
G\eqbi G' \land
(\all{x\scol K\smallSplitSym x\scol K}{
\leftF{x\in G}\land\rightF{x\in G'}\land \Agr x \imp \R}).
\end{equation}
The second conjunct can be written in sugared form as
$\all{x\scol K \in G\smallSplitSym x\scol K\in G'}{\Agr x \imp \R}$.

\begin{restatable}[refperm monotonicity]{lemma}{lemrefpermsupp}
\label{lem:refpermsupp}
\upshape
(i) Any agreement formula is refperm monotonic and so is any refperm independent formula.
(ii) Refperm monotonicity is preserved by conjunction, disjunction,
and quantification.
(iii) Any formula of the form (\ref{eq:rmono}), with $\R$ refperm monotonic, is refperm monotonic.
\end{restatable}

%TODO integrate this as (iv) in lemma, and some comments: Refperm mono is preserved by logical equiv.  (is it pres'd by implic?)

The coupling $\M_{uf}$ in Section~\ref{sec:eg:encapR} is refperm monotonic.
The embedded invariants $\leftF{I_{qf}}$ and $\rightF{I_{qu}}$ are refperm monotonic,
by (i) in the lemma, as is the consequent $eqPartition(\leftex{u.part},\rightex{u.part})$ in
the relation (\ref{eq:UFagree}).  So refperm monotonicity of $\M_{uf}$ follows using (ii) and (iii).

The coupling $\M_{PQ}$ in Example~\ref{ex:PQagree} is refperm monotonic.  To see why,
first note that (\ref{eq:rmono}) is equivalent to
\( G/K\eqbi G'/K \land
(\all{x\scol K \in G \smallSplitSym x\scol K\in G'}{\Agr x \imp \R}) \)
because a quantified variable of type $K$ ranges over allocated (non-null) references
of type $K$. So inside the quantification, $x\in G$ is equivalent to $x\in G/K$.
The relevant subformula of $\M_{PQ}$ is $q.rep/\code{Pnode} \eqbi q.rep/\code{Pnode}$.
Now we distil the following pattern from $\M_{PQ}$,
in which we assume $f:\Region$ and assume both $\Q$ and $\R$ are refperm monotonic.
\[ G\eqbi G \land
   (\all{x\scol K\in G \smallSplitSym x\scol K\in G}{\Agr x \imp
     \Q \land
     \sing{x}\Img f \eqbi \sing{x}\Img f \land
     (\all{y\scol L \in \sing{x}\Img f \smallSplitSym y\scol L \in \sing{x}\Img f}{
       \Agr y \imp \R })})
\]
By (iii) in the lemma
the subformula
$\sing{x}\Img f \eqbi \sing{x}\Img f \land
     (\all{y\scol L \in x.f \smallSplitSym y\scol L \in x.f}{
       \Agr y \imp \R })$
is refperm monotonic.
Then by (ii) we extend that to the conjunction with $\Q$.
Then by (iii) the displayed formula is refperm monotonic.
Note that this relies on agreement of the region values,
$\sing{x}\Img f \eqbi \sing{x}\Img f$,
not pairwise agreement $\Agr \sing{x}\Img f$ on field values.

This discussion provides guidelines for writing specs, but checking refperm monotonicity can be automated.
Validity of $\R\imp\always\R$ only involves universal quantification.
Unfolding semantic definitions, it says:
for all $\pi,\rho,\sigma,\sigma'$, if
$\sigma|\sigma'\models_\pi \R$ and $\rho\supseteq\pi$ then $\sigma|\sigma'\models_\rho \R$.
A straightforward encoding of this in our prototype suffices to show refperm monotonicity
of the example couplings.

\paragraph{Agreement compatibility.}

The last rule for which $\later$ is an issue is \rn{rConj}.  With premises of the form
$\rflowtyf{\P}{\later\Q_0}$ and $\rflowtyf{\P}{\later\Q_1}$ it yields
$\rflowtyf{\P}{\later\Q_0\land\later\Q_1}$.
To obtain the standard form $\rflowtyf{\P}{\later(\Q_0\land\Q_1)}$ one can use \rn{rConseq} but only
if $\Q_0$ and $\Q_1$ are \dt{agreement compatible} which means this implication is valid:
\begin{equation}\label{eq:agrCompat}
 \later\Q_0\land\later\Q_1 \imp \later(\Q_0\land\Q_1)
\end{equation}
An easy case is where $\Q_0$ or $\Q_1$ is refperm independent,
in which case agreement compatibility holds by (\ref{eq:refp-ind-dist}).
Formulas that depend on the refperm involve agreements, and for these we do not
have an easy characterization of agreement compatibility.

In the prototype, $\later$ is not explicit in specs.
A current refperm is witnessed in ghost state,
so even when using conjunctive splitting we effectively get $\later(\Q_0\land\Q_1)$ as desired.
So agreement compatibility is not an issue in the tool.
Morever our case studies show that agreement compatibility is achievable in practical examples where it is needed.
Please note that nontrivial formulas of the form (\ref{eq:agrCompat}) are not amenable to validity checking by SMT,
owing to the existential quantifier that underlies $\later$ in the consequent.\footnote{For the record,
earlier versions of this article had a slightly different \rn{rSOF}, with agreement compatibility
as a side condition for the coupling rather than refperm monotonicity
(arXiv:1910.14560v3).
}

We end this section with some examples regarding agreement compatibility.
But it is not needed later so it is safe to skip now to Section~\ref{sec:lockstep}.

As a first example, consider the agreements $\Agr (G/\code{List})\Img head$
and $\Agr (G/\code{Cell})\Img val$, where class \whyg{List} has field $head:\code{Node}$
and class \whyg{Cell} has field $val:\INT$.
The truth value of $\Agr (G/\code{List})\Img head$ depends only on references of type \whyg{List} and \whyg{Node}.
The truth value of $\Agr (G/\code{Cell})\Img val$ depends only on references of type \whyg{Cell}.
Refperms respect types, so extensions of a refperm to witness
$\later\Agr (G/\code{List})\Img head$
and
$\later \Agr (G/\code{Cell})\Img val$
can be combined to witness
$\later(\Agr (G/\code{List})\Img head \land \Agr (G/\code{Cell})\Img val)$.
Such considerations also apply in a case like
$\Both{\Type(G,\code{List})}\land \Agr G\Img head$ and
$\Both{\Type(H,\code{Cell})}\land \Agr H\Img val$.

Agreement compatibility of $\Q_0$ and $\Q_1$ may fail even if both formulas are
$\Q$ and $\R$ are refperm monotonic.
For example, the formula
\( \later (x\eqbi y) \land \later(x\eqbi z\land \rightF{z\neq y}) \)
is satisfiable
but $\later (x\eqbi y \land x\eqbi z\land \rightF{z\neq y})$ is not.
This example may give the impression that disequalities are the culprit but they are not.
Consider these two formulas:
$\later(x\eqbi x' \land y\eqbi y')$ and
$\later(x\eqbi y' \land y\eqbi x')$ (for distinct variables $x,x',y,y'$).
Both are satisfiable.
In fact their combination,
$\later(x\eqbi x' \land y\eqbi y' \land x\eqbi y' \land y\eqbi x')$,
is also satisfiable: it can hold when $\leftF{x=y}\land\rightF{x'=y'}$.
But the agreement-compatibility implication is not valid.
Consider $\sigma,\sigma',\pi$ where $x,y,x',y'$ have four distinct values, none of which are in the domain or range of $\pi$.
Then both $\later(x\eqbi x' \land y\eqbi y')$ and
$\later(x\eqbi y' \land y\eqbi x')$ are true
but $\later(x\eqbi x' \land y\eqbi y' \land x\eqbi y' \land y\eqbi x')$ is false.

One might guess $\Agr G\Img f$ is agreement compatible with $\Agr H\Img g$ where $f,g$ are distinct field names.
But consider $\Agr\sing{x}\Img f$ and $\Agr\sing{x}\Img g$ for distinct fields $f,g$ of some reference type.
Suppose $\sigma|\sigma'\models_\pi x\eqbi x$, so $\pi(\sigma(x))=\sigma'(x)$.
Suppose $\sigma(x.f)$ and $\sigma(x.g)$ are non-null values not in $\dom(\pi)$,
and likewise $\sigma'(x.f)$ and $\sigma'(x.g)$ are non-null values not in $\rng(\pi)$.
Then we have $\sigma|\sigma'\models_\pi \later\Agr\sing{x}\Img f \land \later\Agr\sing{x}\Img g$,
because $\pi$ can be extended to link $\sigma(x.f)$ with $\sigma'(x.f)$
and \emph{mut.\ mut.} for $g$.
However, if $\sigma(x.f) = \sigma(x.g)$ and $\sigma'(x.f) \neq \sigma'(x.g)$
then there is no single extension of $\pi$ that satisfies
$\Agr\sing{x}\Img f \land \Agr\sing{x}\Img g$.

Region disjointness $\disj{G}{H}$ does not entail agreement compatiblity of
$\Agr G\Img f$ with $\Agr H\Img f$.
Consider $\Agr\sing{x}\Img f$ and $\Agr\sing{y}\Img g$.
Suppose $\sigma|\sigma'\models_\pi x\eqbi x \land y\eqbi y \land \Both{(x\neq y)}$.
Similar to the preceding example, if $\sigma(x.f)=\sigma(y.g)$
and $\sigma'(x.f)\neq\sigma'(y.g)$ and none of the field values are in $\pi$, then
we have $\sigma|\sigma'\models_\pi \later\Agr\sing{x}\Img f \land \later\Agr\sing{y}\Img g$
but again there is no extension of $\pi$ that satisfies $\Agr\sing{x}\Img f \land \Agr\sing{y}\Img g$.

\subsection{Lockstep alignment lemma}\label{sec:lockstep}

The lockstep alignment lemma brings together the semantics of encapsulation
in the unary logic (Def.~\ref{def:valid}),
in which dependency is expressed in terms of two runs under a single unary context model,
with the biprogram semantics which involves two possibly different unary context models
as needed for linking with two module implementations.
The lemma says that, from states that agree on what may be read, a fully-aligned biprogram remains fully aligned through its execution, and maintains agreements sufficient to establish the postcondition of local equivalence---for any of its traces that satisfy the r-safe and respect conditions of Def.~\ref{def:valid}.
In light of trace projection (Lemma~\ref{lem:bi-to-unary}), it says a pair of unary executions can be aligned lockstep, with strong agreements asserted at each aligned pair of configurations.
The result does not rely on validity of a judgment---rather, we use this result to prove soundness of rules \rn{rLocEq}, \rn{rSOF}, and \rn{rLink}.

A number of subtleties in the unary semantics of encapsulation,
in the biprogram semantics, and in the definition of $\locEq$
are all motivated by difficulties in obtaining a result that is sufficiently strong
to support the soundness proofs for the three rules from which the modular relational
linking rule is derived (\rn{rLocEq}, \rn{rSOF}, and \rn{rLink}).

\begin{restatable}[lockstep alignment]{lemma}{lemrloceq}
\label{lem:rloceq}
\upshape
Suppose
\begin{list}{}{}
\item[(i)] $\Phi\prePostImply\LocEq_\delta(\Psi)$ and $\phi$ is a $\Phi$-model,
where $\delta = \unioneff{N\in\Psi,N\neq M}{\bnd(N)}$.
\item[(ii)] $\sigma|\sigma'\models_\pi pre(\locEq_\delta( \flowty{P}{Q}{\eff} ))$.
\item[(iii)] $T$ is a trace
$\configr{\Syncbi{C}}{\sigma}{\sigma'}{\_}{\_}
\biTranStar{\phi}
\configr{BB}{\tau}{\tau'}{\mu}{\mu'}$
and $C$ is let-free.
\item[(iv)]
Let $U,V$ be the projections of $T$.
Then $U$ (resp.\ $V$) is r-safe for $(\Phi_0,\eff,\sigma)$ (resp.\ for $(\Phi_1,\eff,\sigma')$)
and respects  $(\Phi_0,M,\phi_0,\eff,\sigma)$
(resp.\ $(\Phi_1,M,\phi_1,\eff,\sigma')$).
\end{list}
Then there are $B,\rho$ with
\begin{list}{}{}
\item[(v)] $BB\equiv\Syncbi{B}$, $\rho\supseteq\pi$, and $\mu=\mu'$,
\item[(vi)] $\Lagree(\tau,\tau',\rho, (\freshLocs(\sigma,\tau) \union \rlocs(\sigma,\eff)\union
\written(\sigma,\tau))\setminus\rlocs(\tau,\delta^\oplus))$, and
\item[(vii)]
$\Lagree(\tau',\tau,\rho^{-1},(\freshLocs(\sigma',\tau') \union \rlocs(\sigma',\eff)\union
\written(\sigma',\tau'))\setminus\rlocs(\tau',\delta^\oplus))$.
\end{list}
\end{restatable}
In words, the Lemma says that if we have fully aligned code, unary encapsulation (iv), initial agreement (ii),
and relational specs that imply the local equivalence spec (but may be strengthened to include hidden invariants and coupling) (i),
then the code remains fully aligned at every step, and agreements outside encapsulated state are preserved.
Condition (v) can be strengthened to say $\mu$ and $\mu'$ are empty, which holds owing to the assumption that $C$ is let-free.  We keep this formulation because it suffices
and shows what we expect for the extensions discussed in Section~\ref{sec:nestedX}.

The lemma is proved by induction on steps, maintaining (v)--(vii),
using several technical lemmas for preservation of agreement (in appendix Section~\ref{sec:app:lemrloceq}).

Lemma~\ref{lem:rloceq} resembles Lemma~\ref{lem:readeff} but has significant differences.
Lemma~\ref{lem:rloceq} is for client code outside boundaries, in a setting where there are different implementations of methods.
Lemma~\ref{lem:readeff} is for code potentially inside boundaries, but relating two runs of exactly the same program.
In the proofs of both results, r-safety helps ensure that the small-step dependency embodied
by r-respect implies an end-to-end dependency condition.

\subsection{Nested linking}\label{sec:nestedX}

The unary and relational linking rules allow simultaneous linking of multiple modules,
for example linking $MST$ with the \whyg{PQ} and \whyg{Graph} modules.
In \RLII\ (Section~9), a modular linking rule is derived for simultaneous linking of two modules with mutually recursive methods, each respecting the other's boundary.
That can be done with
both the unary and relational rules in this article: the judgments for correctness of the bodies are extended with the other module's invariant or coupling (using \rn{SOF} or \rn{rSOF}) and then
linked (using \rn{Link} or \rn{rLink}).
In \RLII\ and the unary logic in this article, it is also possible for linking to be nested
(shown by examples in Section~2.4 and~8.4 of \RLII).
However, there is a limitation of the relational rules with nested use of bi-let.

To set the stage, we carry out the derivation of modular linking as in Figure~\ref{fig:unaryMismatch} but with
a second module in context, to which we then apply modular linking.
Methods of $\Phi$ may be used in both the client $C$ and the implementation $B$.
The implementation of $\Phi$ has its own internal state with invariant $J$.

% ALERT need blank line here or line spacing messed up 
\begin{footnotesize}
\[ %begin{equation} \label{eq:nestA}
\inferrule*%[right=Link]
{
  \inferrule*%[right=SOF]
  {
    \inferrule*%[right=Link]
    {
      % first premise of Link
      \inferrule*%[right=SOF]
      {
          \Phi,\Theta \proves_{\emptymod} C: \flowty{P}{Q}{\eff}
      }{%SOF
          \Phi,(\Theta \conjInv I)  \proves_{\emptymod}  C : (\flowty{P}{Q}{\eff}) \conjInv I
      }%SOF
    \\ %Link
      \Phi,(\Theta\conjInv I) \proves_M B : \Theta(m)\conjInv I
    }{ %Link
      \Phi  \proves_{\emptymod} \letcom{m}{B}{C} :  (\flowty{P}{Q}{\eff}) \conjInv I
    } %Link
  }{%SOF
      \Phi\conjInv J  \proves_{\emptymod} \letcom{m}{B}{C} :  (\flowty{P}{Q}{\eff}) \conjInv I \conjInv J
  }%SOF
  \\%Link
      \Phi\conjInv J \proves_N D : \Phi(n)\conjInv J
  }{%Link
      \proves_{\emptymod} \letcom{n}{D}{\letcom{m}{B}{C}} :  (\flowty{P}{Q}{\eff}) \conjInv I \conjInv J
  }%Link
\]
\end{footnotesize}
% omit
%The second use of \rn{SOF} requires $\fra{\bnd(N)}{J}$ and $N\in\Phi$.
%(We omit the use of \rn{Conseq} at the end to eliminate the invariants, as it is unproblematic.)
% OK but omit to save space:
%% Instead of the left sub-derivation leading to (\ref{eq:nestA}),
%% here is another derivation which uses an additional instance of \rn{SOF} for the body, $B$, of $m$.
%% \begin{footnotesize}
%% \[
%%   \inferrule%*[right=Link]
%%   {
%%     \inferrule%*[right=SOF]
%%     {
%%        \inferrule%*[right=SOF]
%%        {
%%          \Phi,\Theta \proves_{\emptymod} C: \flowty{P}{Q}{\eff}
%%        }{
%%          \Phi,\Theta \conjInv I  \proves_{\emptymod}  C : (\flowty{P}{Q}{\eff}) \conjInv I
%%        }
%%     }{
%%          \Phi\conjInv J,\Theta \conjInv I\conjInv J  \proves_{\emptymod}  C : (\flowty{P}{Q}{\eff}) \conjInv I \conjInv J
%%     }
%%     \\
%%       \inferrule%*[right=SOF]
%%       {
%%          \Phi,\Theta\conjInv I \proves_M B : \Theta(m)\conjInv I
%%       }{
%%          \Phi \conjInv J, \Theta\conjInv I\conjInv J \proves_M B : \Theta(m)\conjInv I \conjInv J
%%       }
%%   }{ %  Link inner
%%      \Phi\conjInv J  \proves_{\emptymod} \letcom{m}{B}{C} :  (\flowty{P}{Q}{\eff}) \conjInv I \conjInv J
%%   } % Link inner
%% \]
%% \end{footnotesize}
%% We would like a relational analog of at least one of these derivations,
We would like the relational analog of this derivation,
so that with coupling $\M$ for module $M$ and coupling $\N$ for $N$
one could obtain the judgment
\[ \proves_{\emptymod}
 \letcombi{n}{\splitbi{D}{D'}}{
    \letcombi{m}{\splitbi{B}{B'}}{\Syncbi{C}}} : \locEq_\delta (\flowty{P}{Q}{\eff}) \conjInv \M \conjInv \N
\]
Following the pattern of the derivation above, one would like to apply \rn{rSOF} for $\N$ to the judgment
$\LocEq_\delta(\Phi)\proves_{\emptymod}
        \letcombi{m}{\splitbi{B}{B'}}{\Syncbi{C}} : \locEq_\delta (\flowtyC{P}{Q}{\eff}) \conjInv \M$,
where $\delta=\bnd(M),\bnd(N)$.
However, the current  \rn{rSOF} and \rn{rLink} are only for fully aligned client code,
and the ``client'' body $\letcombi{m}{\splitbi{B}{B'}}{\Syncbi{C}}$ of the outer let is not
in that form.
Soundness of \rn{rSOF} hinges on the calls being sync'd---but in the program
$\letcombi{m}{\splitbi{B}{B'}}{\Syncbi{C}}$, calls to $n$ (the method of $\Phi$)
from $B$ or $B'$ are not sync'd,
because $m()$ steps to $\splitbi{B}{B'}$ which has no sync'd calls.
The restriction of bi-let to separate unary commands simplifies the technical development considerably.
But we would like to generalize the bi-let form to allow
$\letcombi{m}{BB}{CC}$ where $BB$ is sufficiently woven that all its calls are sync'd,
and $CC$ is a nest of such bi-lets enclosing a fully aligned client.
% $\letcombi{m}{BB}{\Syncbi{C}}$ where $BB$ is sufficiently woven that all its calls are sync'd.
This requires Lemma~\ref{lem:rloceq} to be generalized to account for such biprogram computations.
The Lemma relies on agreements derived from unary Encap, but this is no longer sufficient
to handle computations with sub-computations that are not fully aligned.
The premises of \rn{rSOF} and \rn{rLink} entail that such computations can
make sync'd calls, but this fact is not retained in the semantics of relational judgments.
Details of our solution %of this problem
are beyond the scope of this article.

%% \dn{Just say it's beyond scope of paper, like unconditional equivalences}
%% Our solution of this problem is to annotate each bi-let-bound biprogram with the spec used for its linkage,
%% to which the definition of valid judgment can refer.
%% This has been sketched elsewhere.\footnote{See Appendix F of \url{https://arxiv.org/abs/1910.14560v2}.
%% The present document supersedes that work but what is described in Appendix F
%% remains applicable.}

\subsection{Unconditional equivalence transformations}\label{sec:uequiv}

An important feature of relational logic which is introduced in Banerjee et al.~\cite{BanerjeeNN16} (long version) is unconditional rewrites.
These are correctness-preserving transformations of control structure in commands that
enable the use of the bi-if and bi-while forms for programs with differing control structure.
An example is the equivalence
$\whilec{E}{C} \uequiv
\whilec{E}{( C; \whilec{E\land E0}{C})}$.
Banerjee et al.~use this and another loop unrolling equivalence to prove correctness of a loop tiling optimization.  In that proof the loop iterations are aligned lockstep, i.e.,
rule \rn{rWhile} and a bi-while with false alignment guards.

In the cited work, it suffices to define $\uequiv$ as a safety-preserving trace equivalence.
These sort of transformations do not alter the series of states reached
and which atomic commands are executed.
From the same initial state and environment, the computations proceed almost in step-by-step correspondence, the exceptions being different manipulation of the control state in some cases,
which leaves the (data) state and method environment unchanged.
As a result, correctness is preserved in the sense that
if  $C\uequiv D $ then
$\Phi\HVflowtr{}{}{P}{C}{Q}{\eff}$ implies
$\Phi\HVflowtr{}{}{P}{D}{Q}{\eff}$.
Moreover
$\Phi\rHVflowtr{}{}{\P}{\splitbi{C}{C'}}{\Q}{\eff|\eff'}$
implies
$\Phi\rHVflowtr{}{}{\P}{\splitbi{D}{C'}}{\Q}{\eff|\eff'}$
(and the same on the right side).
However, to cater for the stronger conditions of valid unary and relational judgments in the present work (Defs.~\ref{def:valid} and~\ref{def:validR}), a stronger notion is needed because those conditions
refer to the control.

As an example, suppose we have a valid correctness judgment
$\Phi\proves_M \whilec{E}{C}: \flowty{P}{Q}{\eff}$ and consider
the form $\whilec{E}{( C; \whilec{E\land E0}{C})}$.
If $E0$ reads some variable that is encapsulated by a module, different from $M$, in $\Phi$,
it may violate the Encap condition of Def.~\ref{def:valid}
and invalidate the judgment
$\Phi\proves_M \whilec{E}{( C; \whilec{E\land E0}{C})} : \flowty{P}{Q}{\eff}$.
For the equivalences considered here, which involve rearranging control structure,
branch conditions turn out to be the main complication.
Details of our formalization of $\uequiv$ and its rules are beyond the scope of this article.

\section{Remarks on case studies}\label{sec:cases}

\WhyRel\ is a proof-of-principle prototype relational verifier which we developed and used
to investigate the applicability of the logic and its amenability to automation.
%, focusing on: (i) SMT encodings and viability of the logic's verification conditions (VCs)  including those for encapsulation, and (ii) adequacy of lightweight specifications to support equivalence verification.
%The latter centers on frame conditions, so usually we choose the weakest specs for which reasonably tight frame conditions are provable.
%We typically have to strengthen loop invariants and specs but not all the way to full functional specs.
The tool supports general relational verification and includes support for relational modular linking.
It has been used to specify and verify a number of examples.
This includes examples discussed in earlier sections:
Kruskal's $MST$ as client of two implementations of union-find;
Dijkstra's shortest-path algorithm as client of two implementations of \whyg{PQ};
and the $tabulate$ and $sumpub$ examples.
We have done other examples taken from recent literature on relational verification,
including information flow, other relational properties, and equivalence for program transformations.
A current version of the prototype and examples are available open source.\footnote{\url{https://github.com/dnaumann/RelRL}}
In addition to the following highlights and the documentation in the software distribution,
further information is available in the thesis of Nikouei~\cite{Nikouei19}
(but note it describes a previous implementation of \WhyRel).

The \WhyRel\ prototype is based on the Why3 platform.\footnote{\url{why3.lri.fr}}.
Why3 serves as an intermediate verification language to which \WhyRel\ translates specs and programs.
Why3 generates verification conditions for pre-post specs and programs in a first-order fragment of ML (WhyML) without shared references, and discharges those conditions by orchestrating calls to automated provers and proof assistants.
Like Why3, \WhyRel\ is ``auto-active''~\cite{LeinoM10}, requiring some user interaction while leveraging
automated provers especially SMT solvers.
Our translation involves substantial encoding, because Why3 does not support shared mutable objects, dynamic frames,
or hiding of invariants.
In this section we describe the encoding, the user interaction needed, and our experience with the case studies.

The language supported by \WhyRel\ extends the language of Figure~\ref{fig:bnf} and Section~\ref{sec:modules}
with arrays, parameters/results, and mathematical data types (defined in Why3 theories).
Module interfaces are separate from module implementations and class fields can have module scope.
The spec language is like that of the article (with usual keywords \whyg{requires}, \whyg{ensures}, etc.),
extended with ``old'' expressions, assertions, loop invariants, assumptions, and explicit ghost declarations.
\WhyRel\ effectively works with relational specs in standard form:
the possibility modal ($\later$) is not used and instead a ghost refperm is updated by
the \whyg{connect-with} ghost operation described in Section~\ref{sec:eg:relverif}.

\WhyRel\ has three main capabilities: unary verification, relational verification, and relational verification with modular linking.
The user provides module interfaces (class declarations, method specs, and boundaries which may be empty)
and unary module implementations which can import Why3 theories providing mathematical types (like lists, graphs, and partitions used in our case studies).
These theories can include lemmas, which get proved by Why3.
The user can also state lemmas in our source language, e.g.,
useful consequences of public invariants.
For relational verification, the user provides a module with biprograms,
which we call a bimodule.  Each bimodule relates two unary modules.
\WhyRel\ checks, for each bimethod in a bimodule, that its
unary projections conform to the (unary) programs being related.
This ensures the biprogram can be constructed by weaving those
unary programs (Lemma~\ref{lem:biprojections}).  Thus, verification of the
biprogram implies a relation between the unary programs, as per the weaving rule (\ref{eq:weaveRule}).
% This can be a standalone bimodule, for general relational verification
% like the $tabulate$ and $sumpub$ examples (Sections~\ref{sec:eg:relverif} and~\ref{sec:weave}).

For relational modular linking of a client program and two versions of
a module the client imports,
\WhyRel\ can generate the local equivalence specs for the module methods.
The user can edit the specs to add the chosen coupling relation,
and use these in a bimodule for relating the module methods.
\WhyRel\ also generates the side conditions of rule \rn{rMLink}
which include framing of invariants/coupling by the boundary and refperm monotonicity of the coupling.

The user provides specs and also loop invariants and loop frame conditions; for hiding, the user provides
boundaries, private invariants, and coupling relations.
Once \WhyRel\ has translated the specs and programs/biprograms to WhyML,
Why3 generates verification conditions.
The user guides Why3 to prove these, by applying tactics (called transformations) like splitting conjunctions.
To complete a verification the user typically has to assert intermediate facts
and sometimes state and prove lemmas (expressed in our source language).
In our case studies, the SMT-solvers Alt-Ergo, Z3, and CVC4 discharge all obligations automatically.

\paragraph{Translation to Why3.}

We encode methods and specs as Why3 functions which have specs.
Why3 is procedure-modular: it verifies each function assuming the specs of the ones it imports,
which corresponds to a hypothesis context in our logic.
Why3 provides ghost annotations and checks that ghost code terminates and does not interfere with the underlying program.
We use this feature to mark the allocation map, which is part of our heap model,
and translate source code ghost state to Why3 ghost state.
Why3 is sound under idealizations also made in our logic: unbounded integers and unbounded maps (which we used to model unbounded heap).

The Why3 language (including WhyML) does not include shared mutable objects.
So we use mutable records and maps to explicitly model the heap
using the standard field-as-array representation,
with references as an uninterpreted type and an extra field, \whyg{alloct}, for allocation
to model the $\lloc$ variable and typing of references.
WhyML has ML-style references constrained by static analysis that precludes aliasing;
we use those to encode local variables.
Invariants of source language semantics, like the absence of dangling pointers,
are encoded using Why3's invariant feature for the data type of states.
(States have the heap and global variables.)
Common elements of translation are included in a \WhyRel\ standard library that
includes lemmas about operations on regions, which aids automated proving.
Why3 specs include coarse grained \whyg{reads} and \whyg{writes} clauses
enforced by simple syntactic analysis, which is not suited to our purposes.
To encode the stateful frame conditions of our logic, \WhyRel\ expresses
write effects semantically, in universally quantified postconditions using ``old'' expressions.
In accord with Def.~\ref{def:valid}, read effects are checked together with the encapsulation checks,
discussed below.

\begin{figure}[t]
  \begin{subfigure}{0.48\textwidth}
\begin{lstlisting}[basicstyle=\linespread{0.5}\footnotesize\sffamily]
meth sum (self:List | self:List) : (int | int)
  requires { both self <> null }
  requires { exists ls:int list | ls:int list.
                      both listpub(self.head,ls) /\ ls =:= ls }
  ensures  { agree result }
= var ghost xs : int list | ghost xs : int list in
   /* Initial values of math type variables are havoc'd;
      assume they witness the existential
      in the precondition */
   assume { both listpub(self.head,xs) };
   /* Initial value of result:int is 0 */
   var p : Node | p : Node in
   |_ p := self.head _|;
   while (p <> null) | (p <> null) . *<| not p.pub *<] | [> not p.pub |>
     invariant { both listpub(p,xs) /\ agree xs /\ agree result }
     ( if p.pub then
         result := result + p.value; xs := tl(xs);
       fi; p := p.nxt
     | if p.pub then
         result := result + p.value; xs := tl(xs);
       fi; p := p.nxt )
   od;
\end{lstlisting}
  \end{subfigure}
  \begin{subfigure}{0.48\textwidth}
\begin{lstlisting}[language=whythree,basicstyle=\linespread{0.5}\footnotesize\sffamily,captionpos=b]
let sum (msigma_l msigma_r: state) (mpi: refperm) 
           (self_l self_r: reference) : (int, int)
  requires { self_l <> null /\ msigma_l.alloct[self_l] = List }
  requires { self_r <> null /\ msigma_r.alloct[self_r] = List }
  requires { exists ls_l, ls_r: int list.
                    listpub msigma_l msigma_l.heap.head[self_l] ls_l
                /\ listpub msigma_r msigma_r.heap.head[self_r] ls_r
                /\ ls_l = ls_r }
  ensures { fst result = snd result }
= let ref result_l = 0 in (* default value for int *)
   let ref result_r = 0 in
   (* variables of math type initialized using any *)
   let ghost ref xs_l = any (int list) in
   let ghost ref xs_r = any (int list) in
   assume { listpub msigma_l msigma_l.heap.head[self_l] xs_l
               /\ listpub msigma_r msigma_r.heap.head[self_r] xs_r }
   let ref p_l = msigma_l.heap.head[self_l] in
   let ref p_r = msigma_r.heap.head[self_r] in
   while (p_l <> null) || (p_r <> null) do
     invariant { listpub msigma_l p_l xs_l /\ listpub msigma_r p_r xs_r }
     invariant { xs_l = xs_r /\ result_l = result_r }
     invariant { (* generated using alignment guards *)
                        p_l <> null /\ not msigma_l.heap.pub[p_l]
                  \/ p_r <> null /\ not msigma_r.heap.pub[p_r]
                  \/ p_l <> null /\ p_r <> null
                  \/ p_l = null  /\ p_r = null }
     if (p_l <> null and not msigma_l.heap.pub[p_l]) then (* left *)
         p_l <- msigma_l.heap.nxt[p_l]
     else begin
         if (p_r <> null and not msigma_r.heap.pub[p_r]) then (* right *)
             p_r <- msigma_r.heap.nxt[p_r]
         else begin (* lockstep *)
             result_l <- result_l + msigma_l.heap.value[p_l];
             xs_l <- tl xs_l;
             p_l <-  msigma_l.heap.nxt[p_l];
             result_r <- result_r + msigma_r.heap.value[p_r];
             xs_r <- tl xs_r;
             p_r <- msigma_r.heap.nxt[p_r]
         end;
     end;
   done; (result_l, result_r)
\end{lstlisting}
  \end{subfigure}
\caption{\WhyRel\ source biprogram for $sumpub$ and translated WhyML (eliding frame conditions).}\label{fig:sumpub-whyrel}
\end{figure}

\WhyRel\ translates a biprogram to a WhyML function acting on a pair of states
together with the current refperm.  Relational pre- and post-conditions
are translated to WhyML requires/ensures.
\WhyRel\ represents a refperm by a pair of maps subject to universally quantified formulas
that express bijectivity and are type-respecting.
As an example, Figure~\ref{fig:sumpub-whyrel} shows our source code
for $sumpub$ biprogram (\ref{eq:bi-sumpub}),
together with its translation to WhyML.
The WhyML loop body reflects the semantics of loop alignment guards.
For readability, some dead code has been removed from the actual translation.

\paragraph{Checking read effects and encapsulation.}
By contrast with the check of write effects, \WhyRel\ does not directly check the relational semantics
of read effects (r-respect in Def.~\ref{def:valid}).  Rather, it performs local checks based on the relevant conditions in the proof rules of our logic.
When used for relational modular linking of modules with nontrivial boundaries,
\WhyRel\ must also enforce encapsulation, that is, the conditions on reads of if, while, bi-if, and bi-while, as well as the conditions of the context introduction rules used for atomic commands.
These checks involve computing separator formulas,
following a preliminary step that normalizes dynamic boundaries and expands the $\allfields$ datagroup to concrete fields.
The tool immediately reports a violation when variables are required to be distinct but are not,
or are read but not included in the read effect.
For separation of heap locations, it generates disjointness formulas (in accord with Figure~\ref{fig:sepdef})
in assert statements added to the generated code where the encap checks should be made.
For reads of heap locations, it asserts an inclusion based on the reads allowed by the frame condition.
A snapshot of the initial state is used so the frame condition can be interpreted
where it should be; the asserted inclusion is at the point in the code where the read takes place,
which may follow updates to the state.

When true, the disjointness and inclusion assertions for reads and encapsulation are usually proved without any need for user interaction.  The user does see the assertions among the proof obligations enumerated by Why3.  The user does not compute separators or effect subtractions, those are done by \WhyRel.

\paragraph{Modular linking.}

In terms of the logic, Why3 verifies the premises of the standard linking rule
(\rn{Link} in Figure~\ref{fig:proofrulesU}) so the contracts assumed by a procedure's callers are the ones
for which the procedure's implementation is verified.
\WhyRel\ generates code that expresses hiding,
i.e., the premises of our modular linking rules: the implementations get to assume the private invariant (or coupling, in the relational case) and must maintain it.
For this to be sound, \WhyRel\ checks encapsulation, as described above,
and generates Why3 lemmas to encode the additional proof obligations.

\begin{figure}[t]
\begin{lstlisting}[language=whythree,basicstyle=\linespread{0.5}\footnotesize\sffamily,captionpos=b]
lemma boundary_frames_QuickFind_invariant :
  forall msigma: state, mtau: state, mpi: refperm.
    okRefperm msigma mtau mpi /\ identityRefperm mpi (domain msigma.alloct) (domain mtau.alloct) ->
    idRgn mpi msigma.pool mtau.pool -> (* msigma(pool) =:= mtau(pool) *)
    agreeAny msigma mtau mpi (union msigma.pool (imgRep msigma msigma.pool)) ->
    ufPriv msigma -> (* private invariant *?\color{dkgreen}{${I}_{uf}$}?* *)
    ufPriv mtau

lemma boundary_frames_UnionFind_coupling :
  forall msigma: state, mtau: state, msigma': state, mtau': state, mpi: refperm, mpi': refperm, mrho: refperm.
    okRefperm msigma mtau mpi /\ identityRefperm mpi (domain msigma.alloct) (domain mtau.alloct) ->
    okRefperm msigma' mtau' mpi' /\ identityRefperm mpi' (domain msigma'.alloct) (domain mtau'.alloct) ->
    okRefperm msigma msigma' mrho /\ okRefperm mtau mtau' mrho ->
    idRgn mpi msigma.pool mtau.pool -> (* msigma(pool) =:= mtau(pool) *)
    agreeAny msigma mtau mpi (union msigma.pool (imgRep msigma msigma.pool)) ->
    idRgn mpi' msigma'.pool mtau'.pool -> (* msigma'(pool) =:= mtau'(pool) *)
    agreeAny msigma' mtau' mpi' (union msigma'.pool (imgRep msigma' msigma'.pool)) ->
    ufCoupling msigma msigma' mrho -> (* coupling relation *?\color{dkgreen}{$\M_{uf}$}?* *)
    ufCoupling mtau mtau' mrho
\end{lstlisting}
\caption{Framing judgments as lemmas.}\label{fig:framelem}
\end{figure}

For unary hiding, the private invariant should be framed by the module boundary;
this obligation is generated in the form of a lemma that expresses the framing semantics (\ref{eq:frmAgree}).
At the same time, \WhyRel\ generates the obligation that the client precondition implies the private invariant.
For relational hiding, the coupling invariant should be framed, on both left and right, by the boundary
(using relational framing semantics Def.~\ref{def:frmAgreeRel}).
Example framing lemmas are in Figure~\ref{fig:framelem}.

Another obligation generated in the form of a lemma is that
the coupling should be refperm monotonic:
\begin{lstlisting}[language=whythree,basicstyle=\linespread{0.5}\small\sffamily,captionpos=b]
lemma ufCoupling_is_monotonic :
  forall msigma: state, mtau: state, mpi: refperm.
    okRefperm msigma mtau mpi -> ufCoupling msigma mtau mpi ->
    forall mrho: refperm. okRefperm msigma mtau mrho -> extends mpi mrho -> ufCoupling msigma mtau mrho
\end{lstlisting}

\WhyRel\ can generate a local equivalence spec, given boundaries and a unary spec;
it is generated as source code, which the user can include in a biprogram.
Local equivalence specs are defined in Section~\ref{sec:locEq} and examples appear in Section~\ref{sec:biprograms}.

\emph{Experience and findings.}
Despite achieving a high level of automation based on SMT solvers, auto-active tools
require user effort and intelligence to devise specs and find loop invariants.
Here, there is the additional task of writing a biprogram to express an alignment
for which straightforward invariants suffice.
(See Section~\ref{sec:related} for work on automated inference of alignments.)
Use of dynamic frames entails extensive reasoning about set expressions, set disjointness and containment.
Aided by some lemmas in the \WhyRel\ standard library, the solvers have little difficulty in this regard;
the requisite reasoning about refperms also works fine.
In most of our examples, the user needs to do a few clicks in Why3 to invoke the tactic to split conjunctions,
and sometimes introduce assertions or lemmas that aid the solvers in finding proofs.
%  \rmn{and also instantiate certain universally quantified facts in the context}.  True but omi this
Why3's assert tactic is helpful for this.  This sort of interaction is typical in ordinary use of Why3.

For $sumpub$ we provide a couple of lemmas about the $listpub$ relation, proved using the rule-induction transformation (i.e., a Why3 induction rule, dispatched to SMT).
For the SSSP biprogram we needed a number of asserts in the code (plus assert tactics);
but not many for the other examples.
Our priority has been to complete illustrative examples and a prototype that can be used by interested researchers; we have not tried to find optimal specs and minimal use of Why3 tactics.
We are not proposing the concrete syntax for use in practice, nor does the tool
provide sufficient error handling to be usable by software engineers.
Moreover, although the prototype implements some syntax sugar relative to the formal development,
the current language has desugared loads and stores, which entails the use of annoyingly many
temporary variables (sugared in examples in the article).

Finally, Why3 generates many proof obligations about the state being well formed, which is actually
guaranteed by type-checking of source programs.  
The obligations are simple to prove but it is still one more thing to do.
It should be possible to eliminate these through more sophisticated use of Why3's abstraction mechanisms.
In BoogiePL these pointless obligations could be avoided using ``free requires/ensures'',
and we could achieve the same effect using Why3 assumptions instead of type invariants;
but the latter make it easier to read the generated WhyML.

Why3 records sessions in order to replay the user's choices of provers and tactics to apply.
Replaying the sessions for our big case studies takes on the order of an hour or more of prover time,
though clock time is a little faster owing to parallelism.
The smaller examples take minutes or less.  
Less time would be needed if we used assumptions to avoid pointless checks about states being well formed.
Significantly more automation could be achieved if Why3 enabled scripting of routine choices of tactics.

In summary, the formal development in preceding sections shows that general relational reasoning with encapsulation, for first-order programs, can be carried out using only first-order assertions and relations.
The case studies carried out using \WhyRel\ demonstrate that the verification conditions
are well within what can be automated by SMT solvers.
User interaction is needed mainly to deal with specs and loop invariants involving mathematical
properties of data types and inductively defined predicates and relations.
Inductive definitions are often needed for problem-specific properties, but are not required
for encapsulation, framing, hiding or any other element of the logic.

\section{Related work}\label{sec:related}

Our main result (Theorem~\ref{thm:sound}) brings together modular reasoning techniques, relational properties, representation independence, automated verification, and their semantic foundations.

We make a rough categorization of related work as follows:
(Section~\ref{sec:relatedRL})
Directly related precursors;
(Section~\ref{sec:relatedRelver})
Algorithmic studies and implementations of automated verification for relational properties,
often lacking detailed foundational justification and support for dynamic allocation or data abstraction,
but identifying FOL fragments enabling automated inference of relational invariants and alignment;
and
(Section~\ref{sec:relatedRepind})
Semantic studies of representation independence, focused on contextual equivalence and challenging language features including dynamic allocation, higher order procedures, and concurrency, leading to the higher order relational separation logic ReLoC implemented in the Coq proof assistant.

Union-find implementations have been verified interactively using Coq~\cite{ChargueraudP19}.
Functional correctness of Kruskal has been verified in a proof assistant~\cite{Guttmann18}.
Functional correctness of C implementations of Dijkstra's, Kruskal's, and Prim's algorithms have been verified by Mohan et al~\cite{MohanLH20} using VST~\cite{CaoBGDA18}.
The point of our case studies is to achieve automated equivalence proof for clients, without recourse to functional correctness.
A purely applicative implementation of pairing heaps has been verified in Why3
(\url{http://toccata.lri.fr/gallery/}).

\subsection{Region logic and other logics with explicit footprints}\label{sec:relatedRL}

Bao et al.~\cite{BaoLE18} introduce a unified fine-grained region logic with both separating conjunction and explicit read/write effects, subsuming a fragment of separation logic.
To enable effective use of SMT solvers,
Piskac et al.~\cite{piskac2013automating,piskac2014grasshopper} encode separation logic style specifications using explicit regions. Several works implement
implicit dynamic frames~\cite{SmansJacobsPiessens,MuellerSchwerhoffSummers17}
which combines the succinctness of separation logic with the automation of SMT.
For recent work on decidable fragments of separation logic, see Echenim et al.~\cite{EchenimIP19}.
Using an extension of FOL with recursive definitions, the logic of Murali et al.~\cite{MuraliPLM20}
has an expression form for the footprint of a formula, akin to our $\ftpt$ operator but usable in formulas, avoiding the need for a separate framing judgment; this can encode a fragment of separation logic but effectiveness for automation has not been thoroughly evaluated.

The most closely related works are the \RL\ articles.
The image notation, introduced in \RLI~\cite{RegLogJrnI}, was inspired by the use of field images to express relations in the information flow logic of Amtoft et al.~\cite{AmtoftBB06}.
In \RLI\ this style of dynamic framing was shown to facilitate local reasoning about global invariants,
and this was extended to dynamic boundaries and hiding of invariants in \RLII~\cite{RegLogJrnII}.

In \RLIII~\cite{BanerjeeNN18}, pure methods are formalized with end-to-end read effects.
The end-to-end semantics of read effects is also used in the preliminary work~\cite{BanerjeeNN16}, from which we take biprograms, weaving, and bi-while alignment guards. But we change the semantics of bi-com $\splitbi{C}{C'}$ to eliminate one-sided divergences and to allow models to diverge
%(i.e., $\phi(m)(\sigma)$ may be empty;
(see rules \rn{uCall0} in Figure~\ref{fig:transSel}
and \rn{bCall0} in Figure~\ref{fig:biprogTrans}).
This validates a better weaving rule (no termination conditions) and a stronger adequacy theorem (Thm.~\ref{thm:biprogram-soundness}).
We drop their semantics of read effects, which is inadequate for our purposes (and is subsumed by r-respects in Def.~\ref{def:valid}),
but use quasi-determinacy and agreement-preservation results from \RLIII.
Neither \RLIII\ nor~\cite{BanerjeeNN16} addresses information hiding or encapsulation.
Our semantics of encapsulation (Def.~\ref{def:valid}) is a major extension of that in
\RLII, from which we take the minimalist formalization of modules;
but we change the semantics to use context models (from \RLIII\ where models are called interpretations)
and add r-respects etc.
We adapt unary rules from \RLII\ but use the term modular linking for what they call mismatch.
The case studies in \RLIII\ are implemented using Why3 with an encoding of heaps and frame conditions
similar to the one used by \WhyRel.

\subsection{Relational verification}\label{sec:relatedRelver}

Francez~\cite{Francez83,NaumannISOLA20} articulated the product principle
reducing relational verification to the inductive assertion method and
introduced a number of proof rules.
Benton~\cite{Benton:popl04} introduced the term  Relational Hoare Logic
and brought to light applications including compiler optimizations.
Yang~\cite{Yang07relsep} introduced relational separation logic, motivated by data abstraction 
although the logic does not formalize that as such.
Beringer~\cite{Beringer11} extends Benton's logic with heap (still not procedures), and provides proof rules
for non-lockstep loops, on which our \rn{rWhile} is based;
a similar rule appears in Barthe et al~\cite{BartheGHS17}.
%dissonant loops.
%Owing to the wide range of applications to k-safety properties and program relations (and relational hyperproperties~\cite{AbateEtalJourney}),
There has been a lot of work on relational logics and verification techniques~\cite{BeckertU18},
e.g., applications in security and privacy~\cite{BartheDGKSS13,NanevskiBG13,RadicekBG0Z18} %\dn{easycrypt}
and merges of software versions~\cite{SousaDL18}.
A shallow embedding of relational Hoare logic in $F^\star$ is used to interactively prove refinements
between union-find implementations~\cite{GrimmMFHMPRRSB18}.
Aguirre et al.~\cite{AguirreBGGS19} develop a logic based on relational refinement types, for
terminating higher order functional programs, and provide an extensive discussion of work on relational logics.

Automated relational verification based on product programs is implemented in several works
which address effective  alignment of control flow points and the inference of alignment points and relational assertions and procedure summaries~\cite{ZuckPGBFH05,ZaksP08,BartheCK11,BartheCK13,FelsingGKRU14,KieferKU18,BartheCK16,WoodDLE17,ChurchillP0A19}.
One line of work, centered around the SymDiff verifier~\cite{hawblitzelklr13,LahiriMSH13,LahiriHKR12}, proves properties of program differences using relational procedure summaries.
Godlin and Strichmann~\cite{GodlinS08} prove soundness of proof rules for equivalence checking taking into account similar and differing calls.
Eilers et al.~\cite{EilersMH20} implement a novel product construction for procedure-modular verification of k-safety properties of a program, maximizing use of relational specs for procedure calls.  (We follow O'Hearn et al.~\cite{OHearnYangReynoldsInfoToplas} in using ``modular'' to imply also information hiding.)
% It is not clear how the construction of Eilers et al.\ could be generalized to relations between different programs.
Girka et al.~\cite{GirkaEtalPPDP17} explore forms of alignment automata.
Shemer et al.~\cite{ShemerGSV19} provide for flexible alignments
and infer state-dependent alignment conditions,
as do Unno et al.~\cite{UnnoTerauchiKoskinen21}.
The latter works rely on constraint solving techniques which are not yet
applicable to the heap.
For the heap the state of the art for finding alignments is syntactic matching heuristics.

For $\forall\exists$ properties, product constructions appear in some recent works~\cite{BartheCK13,ClochardMP20,UnnoTerauchiKoskinen21,LamportSchneider21,AntonopoulosEtal2022}.
Pioneering work by Rinard and Marinov~\cite{RinardCredible99,RinardMarinov99} 
introduces a logic of  $\forall\exists$ simulations for correct compilation, for programs represented as control flow graphs. 

Sousa and Dillig's Cartesian Hoare Logic~\cite{SousaD16} (a generalization of Benton's logic) can be used to reason about $k$-safety properties such as secure information flow (2-safety) and transitivity (3-safety). They also develop an algorithm, based on an implicit product construction, for automatically proving $k$-safety properties; The corresponding tool, Descartes, has been used in the verification of several user-defined relational operators in Java programs. For more efficient relational verification, Pick et al.~\cite{PickFG18} introduce a new algorithm atop Descartes, which automatically detects opportunities for alignment (the synchrony phase) and detects opportunities for pruning subtasks by exploiting symmetries in program structure and relational specs.

None of the above works address hiding, and many do not fully handle the heap~\cite{LahiriMSU18}.
Our work is complementary, providing a foundation for verified toolchains implementing these algorithmic techniques.
The use of \rn{rWhile} with alignment guards,
together with the disjunction rule to split cases
and unconditional rewriting (Section~\ref{sec:uequiv}),
enables our logic to express a wide range of state-dependent alignments.

\subsection{Representation independence}\label{sec:relatedRepind}

It is difficult to account for encapsulation in semantics of languages with dynamically allocated mutable state and especially with higher order features.
Crary's tour de force proves parametricity for a large fragment of ML but excluding reference types~\cite{CraryPOPL17}. % the only effect considered is termination
Semantic studies of the problem~\cite{BanerjeeNaumann02c,AhmedDR09}
have been connected with unary~\cite{BanerjeeN13} and relational logics~\cite{DreyerNRB10}.
The latter relies on intensional atomic propositions about steps in the transition semantics. In this sense it is very different from standard (Hoare-style) program logics.

Birkedal and Yang~\cite{BirkedalY08rel} show client code proved correct
using the SOF rule of separation logic is relationally parametric,
using a semantics that does not validate the rule of conjunction which plays a key role in automated verification.
That rule is an issue in some other models as well, e.g., Iris (in part owing to its treatment of ghost updates as logical operators).

Thamsborg et al.~\cite{ThamsborgBY12} also lift separation logic to a relational interpretation, but instead of second order framing, address abstract predicates. Their goal is to give a relational interpretation of proofs. They uncover and solve a surprising problem: due to the nature of entailment in separation logic, not all uses of the rule of consequence lift to relations. Our logic does not directly lift proofs but does lift judgments from unary to relational (the \rn{rEmb} and \rn{rLocEq} rules). In general, most works on representation independence, including work on encapsulation of mutable objects, are essentially semantic developments~\cite{BanerjeeNaumann02c,BanerjeeN13}; general categorical models of Reynolds' relational parametricity~\cite{Reynolds84} which validate his abstraction theorem and identity extension lemma have been developed and are under active study by Johann et al.~\cite{SojakovaJ18}.

The state of the art for data abstraction in separation logics is abstract predicates, which are satisfactory
in many specs where some abstraction of ADT state is of interest to clients,
but less attractive for composing libraries such as runtime resource management with no client-relevant state.
Such logics have been implemented in interactive provers~\cite{NanevskiLSD14,JungKJBBD18,BeringerAppel19}.
These are unary logics with concurrency;
they do not feature second order framing
but they have been used to verify challenging concurrent programs.
As shown by the recent extension of VST with Verified Software Units~\cite{Beringer21},
higher order logics with impredicative quantification
facilitate expressive interface specifications for modular reasoning about heap based programs.

ReLoC~\cite{FruminKB18}, based on Iris~\cite{JungKJBBD18},
is a relational logic for conditional contextual refinement of higher order concurrent programs.
Iris and the works in the preceding paragraph do support hiding in the sense of abstraction:
through existential quantification and abstract predicates, and in Iris through the invariant-box modality
and the associated ``masks''.
With respect to our context and goals, we find such machinery to be overkill.  Like O'Hearn et al.~\cite{OHearnYangReynoldsInfoToplas}, we only need invariants in the sense of conditions that hold when control enters or exits the module---not conditions that hold at every step.
There is a considerable gap between this work and the properties/techniques for which automation has been developed; moreover their step-indexed semantics does not support termination reasoning or transitive composition of relations (which needs relative termination~\cite{hawblitzelklr13}); our logic is easily adapted to both.

Maillard et al.~\cite{MaillardHRM20} provide a general framework for relational program logics that can be instantiated for different computational effects represented by monads.  The paper does not address encapsulation except insofar as the system is based on dependent type theory.

\section{Conclusion}\label{sec:discuss} % dreaming of Envoi

We introduced a relational Hoare logic that accounts for strong encapsulation of data representations in object-based programs with dynamic allocation and shared mutable data structures. Consequently, changes to internal data representations of a module can be proved to
lead to equivalent observable behaviors of clients that have been proved to respect encapsulation.
The technique of simulation, articulated by Hoare~\cite{Hoare:data:rep} and formalized in theories of representation independence, is embodied directly in the logic as a proof rule (\rn{rMLink} in Figure~\ref{fig:derivedmismatch}).
The logic provides means for specifying state based encapsulation methodologies such as ownership. %and proving client conformance,
It also supports effective relational reasoning about simulation between
both similar and disparate control and data structure.
Although our exposition focuses on encapsulation and simulation,
the logic is general, encompassing a range of relational properties including conditional equivalence (including compiler optimizations), specified differencing (as in regression verification), and secure information flow with downgrading~\cite{AmtoftBB06,BNR08,BanerjeeNN16,Chudnov14}.
The rules are proved sound.

% This semantic notion of correctness developed by Hoare (and in follow on work on representation independence) is now captured directly in the logic as program equivalence. This is manifest in the \rn{rMismatch} proof rule: the client need only concern itself with local reasoning with its own state and, because of encapsulation, automatically respects the coupling relation between the data representation invariants, of which it is unaware.

The programmer's perspective articulated by Hoare is about a single module and client, distinguishing inside versus outside.  The general case, with state based encapsulation for a hierarchy of modules, requires a precise definition of the boundaries within which a given execution step lies. While we build on prior work on state based encapsulation, we find that
to support change of representation, the semantics of encapsulation needs to be formulated in terms of not only the context (hypotheses/library APIs) but also modular structure of what's already linked, via the dynamic call chain embodied by the runtime stack. This novel formulation of an extensional semantics for encapsulation against dependency is subtle (Def.~\ref{def:valid}), yet it remains amenable to simple enforcement. 
Our relational assertions and verification conditions for modules and clients are first-order. As proof of concept, we demonstrate that they can be effectively used in an auto-active SMT-based verification prototype.  

To a great extent, the three goals in Section~\ref{sec:intro} have been achieved.
Beyond this progress,  
for foundational justification one might like to machine check the soundness proofs.
% and prove correctness of WhyRel
For automation, one could explore techniques for inferring alignment conditions 
and relational invariants~\cite{ShemerGSV19,UnnoTerauchiKoskinen21}.

Apropos completeness of the logic, the ordinary notion of completeness is that valid relational judgments are provable (relative to validity of entailments).
Completeness in this sense is an immediate consequence of
completeness of the underlying unary logic together with the presence of a single
rule (like \rn{rEmb}) that lifts unary judgments to relational
ones~\cite{Francez83,BartheDArgenioRezk,BartheDR11}---provided that
unary assertions can express relations.
That proviso is easy to establish for simple imperative programs, by using renamed variables.
For pointer programs, expressing a relation as an assertion can be done using separating conjunction~\cite{BartheDArgenioRezk},
but to do so using only FO assertions requires a complicated encoding~\cite{Naumann06esorics}.
The recently introduced notion of alignment completeness~\cite{NagasamudramN21} is better  than ordinary completeness 
as a way to evaluate relational logics.
We have not yet investigated completeness for either unary or relational region logic.

\section{Envoi}
Hoare's 1972 paper articulates the fundamental notions of hiding and encapsulation with a minimum of extraneous formalization.
In seeking to formulate the ideas in a logic for first-order programs using first-order assertions, we hoped to achieve a comparably elementary and transparent account. In order to handle dynamically allocated mutable state, however, we have been unable to avoid some amount of auxiliary notions.  

Having incorporated encapsulation into a unary+relational logic that supports hiding of internal invariants, we are poised to investigate a longstanding problem: the hiding of unobservable effects for object-based programs.  This is intimately connected with encapsulation~\cite{obspureTCS,Pottier08,BentonHN14} and appears already in Hoare's work under the term benevolent side effects~\cite{Hoare:data:rep}.

%% Hoare's investigation of encapsulation brought to light the justification to hide state changes that have no observable effect (observational purity).  Our system does not provide for hiding of effects, only abstracting from them.

% Note: \cite{BentonHN14} which specifically says its not doing rep ind.

% not for submission

\begin{acks}
We thank the anonymous TOPLAS reviewers for their insightful technical feedback and stuctural suggestions which have improved the exposition. We thank Andrew Myers for his encouragement and diligent editing throughout the reviewing process. Stephen Sondheim's lyrics 
%: occasionally during the review process he reminded us that
``Perpetual anticipation is good for the soul//But it's bad for the heart'' gave us perspective as we worked through multiple review iterations.

The ideas in this article arose from discussions between Banerjee and Naumann during a long walk at PLDI 2009 in Dublin, following which, Naumann jotted down initial thoughts at a cafe. The discussions spurred a long-term research program that has produced substantial intermediate results (\RLI--\RLIII) that have culminated in this article. For arranging presentations of the work at various stages of its development, and for their comments and encouragement, we thank Nina Amla, Lennart Beringer, Lars Birkedal, Stephen Chong, Rance Cleaveland, Matthias Felleisen, Neil Immerman, Patricia Johann, Assaf Kfoury, Shriram Krishnamurthi, Cesar Kunz, Gary Leavens, David Liu, Aleks Nanevski, Minh Ngo, Noam Rinetzky, Mooly Sagiv, Don Sannella, Gordon Stewart and Jan Vitek. 

We thank the organizers and participants of the Dagstuhl Seminar 18151 on Program Equivalence. The stimulating atmosphere of the seminar and Dagstuhl's salubrious environs (which naturally inspired us to take many long walks) aided technical progress at a crucial stage. Naumann acknowledges Manuel Hermenegildo for arranging an enjoyable and fruitful stay at the IMDEA Software Institute in 2011, and Andrew Appel for arranging an engaging stay at Princeton in 2017-18. Finally, we thank our families for their continuing and steadfast support.

Nagasamudram and Nikouei were partially supported by National Science Foundation (NSF) award 1718713. Naumann was partially supported by NSF award 1718713 and Office of Naval Research (ONR) award N00014-17-1-2787. Banerjee's research was based on work supported by the NSF, while working at the Foundation; in particular, he gratefully acknowledges NSF's support of ``Long-term Professional Development'' for FY 2020. Any opinions, findings, and conclusions or recommendations expressed in this article are those of the authors and do not necessarily reflect the views of the NSF and other funding agencies.
\end{acks}

\newpage

\appendix

%\section{Appendix: Syntax (re Sect.~\ref{sec:syntax})}

\section{Appendix: Program semantics and unary correctness (re Sect.~\ref{sec:unarySem})}

\subsection{On effects, agreement, and valid correctness judgment}

\lemeffectSubtract*
\begin{proof} 
Assume w.l.o.g.\ that $\eff$ and $\effe$ are in the normal form
described as part of the definition, Eqn.~(\ref{eq:effsubtract}).
For a variable $x$ we get
$x \in \rlocs(\sigma, \eff\setminus\effe)$ iff
$x \in \rlocs(\sigma, \eff) \setminus \rlocs(\sigma,\effe)$ directly from definitions.
For a heap location, $o.f$ is in $\rlocs(\sigma,\eff)\setminus\rlocs(\sigma,\effe)$
just if there is $\rd{G\Img f}$ in $\eff$ with $o\in\sigma(G)$ 
and there is no $\rd{H\Img f}$ in $\effe$ with $o\in\sigma(H)$ (by definitions).
This can happen in two cases:
either there is no read for $f$ in $\effe$,
or there is $\rd{H\Img f}$ in $\effe$ but $o\notin\sigma(H)$.
In the first case, $\rd{G\Img f}$ is in $\eff\setminus\effe$ so 
$o\in\rlocs(\eff\setminus\effe)$.
In the second case, $\rd{(G\setminus H)\Img f}$ is in $\eff\setminus\effe$ 
and since $o\in\sigma(G\setminus H)$ we have  
$o\in\rlocs(\eff\setminus\effe)$.
\end{proof}

\leminsensi*
\begin{proof}
Straightforward, by induction on $F$ and induction on $P$.
\end{proof}

\begin{remark}\upshape
For partial correctness, all specs are satisfiable (at least by divergence).
This is manifest in Def.~\ref{def:ctxinterp}, 
which allows that $\phi(m)(\sigma)$ can be $\emptyset$ for any $\sigma$ that satisfies the precondition. 
In \RLII, a context call faults in states where the precondition does not hold.
It gets stuck if the precondition holds but there is no successor state that satisfies the postcondition.
Here (and in \RLIII, for impure methods),
the latter situation can be represented by a model that returns the empty set. Instead of letting the semantics get stuck we include a stuttering transition, \rn{uCall0}.
\qed\end{remark}

\begin{remark}
\upshape
Apropos Def.~\ref{def:valid}, one might expect r-respect to consider steps 
$\configm{B}{\tau'}{\mu} \trans{\phi} \configm{D'}{\upsilon'}{\nu'}$ 
with potentially different environment $\nu'$, and add to the consequent that $\nu'=\nu$.  
But in fact the only transitions that affect the environment are those for $\keyw{let}$ and for the $\Endlet$ command used in the semantics at the end of its scope.  The transitions for these are independent of the state, and so $B$ and $\mu$ suffice to determine $\nu$.
\qed\end{remark}

\begin{remark}\upshape
The consequent (\ref{eq:rrespect}) of r-respect express that the visible (outside boundary) writes and allocations depend only on the visible starting state.
One may wonder whether the conditions fully capture dependency, noting that they do not consider faulting.  But r-respects is used in conjunction with the (Safety) condition that rules out faults.  
\qed\end{remark}

%% % DN: uninteresting; dates from when we made more use of allowed dependency
%% Often, as in r-respect,  allowed dependency is used in conjunction with the condition
%% \[ \rho(\freshLocs(\tau,\upsilon)\setminus\rlocs(\upsilon,\delta))\subseteq
%% \freshLocs(\tau',\upsilon')\setminus\rlocs(\upsilon',\delta) \] 
%% so we considered including that as part of allowed dependency
%% (Def.~\ref{def:allowedDepEResB}).
%% But in some proofs we need to separate the conditions.

\begin{remark}\upshape
In separation logic, preconditions serve two purposes: 
in addition to the usual role as an assumption about initial states, 
the precondition also designates the ``footprint'' of the command.
This is usually seen as a frame condition:
the command must not read or write any preexisting locations outside the footprint of the precondition.
In a logic such as the one in this article, where frame conditions are distinct from preconditions, it is possible for the frame condition to designate a smaller set of locations than the footprint of the precondition.
As a simple example, consider the spec
$\flowty{x>0\land y>0}{\True}{\rw{x}}$.
In our logic, it is possible for two states to agree on the read effect but disagree on the precondition.
For example, the states $[x:1,y:0]$ and $[x:1,y:1]$ agree on $x$ but only the second satisfies $x>0\land y>0$.
Lemma~\ref{lem:readeff} describes the read effect only in terms of states that satisfy the precondition.
For a command satisfying the example spec, 
and the states $[x:1,y:1]$ and $[x:1,y:2]$
which satisfy the precondition but do not agree on $y$,
that the command 
must either diverge on both states or converge to states that agree on the value of $x$.
\qed\end{remark}

%\subsection{Additional definitions and results for unary programs and judgments}

% Note that tag can't have punctuation or numbers, e.g., selfframingonreads
\begin{restatable}[agreement symmetry]{lemma}{selfframingonreads}
\label{lem:selfframingonreads}
\upshape
Suppose $\eff$ has framed reads. 
If $\agree(\sigma,\sigma',\pi,\eff)$ then 
(a) $\rlocs(\sigma',\eff)=\pi(\rlocs(\sigma,\eff))$
and 
(b) $\agree(\sigma',\sigma,\pi^{-1},\eff)$.
\end{restatable}
\begin{proof}
(a) For variables the equality follows immediately by definition of $\rlocs$.
For heap locations the argument is by mutual inclusion. To show 
$\rlocs(\sigma',\eff) \subseteq \pi(\rlocs(\sigma,\eff))$,
let $o.f\in\rlocs(\sigma',\eff)$. By definition of $\rlocs$, there exists region $G$ such that $\eff$ contains 
$\rd{G\Img f}$ and $o\in\sigma'(G)$. Since $\eff$ has framed reads, $\eff$ 
contains $\ftpt(G)$, hence from $\agree(\sigma,\sigma',\pi,\eff)$
by Eqn~(\ref{eq:footprintAgreement})
we get $\rprel{\sigma(G)}{\sigma'(G)}$.
Thus $o\in\pi(\sigma(G))$. So, we have $o.f\in\pi(\rlocs(\sigma,\eff))$.
Proof of the reverse inclusion is similar.

(b) For variables this is straightforward.  For heap locations, consider 
any $o.f\in\rlocs(\sigma',\eff)$.  From (a), we have  $\pi^{-1}(o).f\in\rlocs(\sigma,\eff)$. From 
$\agree(\sigma,\sigma',\pi,\eff)$,
we get $\Rprel{\pi}{\sigma(\pi^{-1}(o).f)}{\sigma'(o.f)}$. Thus we have 
$\Rprel{\pi^{-1}}{\sigma'(o.f)}{\sigma(\pi^{-1}(o).f)}$.
\end{proof}

The definition of r-respect is formulated (in Def.~\ref{def:valid}) in a way to make evident that client steps are independent from locations within the boundary. 
But r-respect can be simplified, as follows, when used in conjunction with w-respects.

The following notion is used to streamline the statement of some technical results.
It is used with states 
$\sigma,\tau,\tau',\upsilon,\upsilon'$,
where $\sigma$ is an initial state from which $\tau$ and then later $\upsilon$ is reached,
and in a parallel execution $\tau'$ reaches $\upsilon'$.
Moreover, $\delta$ is a dynamic boundary.
We write $\delta^\oplus$ to abbreviate $\delta,\rd{\lloc}$.

\begin{definition}\label{def:allowedDepEResB}
Say $\eff$ \dt{allows dependence from $\tau,\tau'$ to 
$\upsilon,\upsilon'$ for $\sigma,\delta,\pi$}, written \ghostbox{$\AllowedDep{\tau}{\tau'}{\upsilon}{\upsilon'}{\eff}{\sigma}{\delta}{\pi}$}
iff  the agreement
\( \Lagree(\tau,\tau',\pi,(\freshLocs(\sigma,\tau)\union\rlocs(\sigma,\eff))\setminus\rlocs(\tau,\delta^\oplus))\)
implies there is $\rho\supseteq\pi$ with 
\(\Lagree(\upsilon,\upsilon',\rho, (\freshLocs(\tau,\upsilon)\union 
\written(\tau,\upsilon))\setminus\rlocs(\upsilon,\delta^\oplus))\). 
\end{definition}

Like Definition~\ref{def:agreeX}, this definition is left-skewed,
both because $\eff$ is interpreted in the left state $\sigma$ 
and because the fresh and written locations are determined by the left transition $\sigma$ to $\tau$.
This is tamed in case $\eff$ has framed reads (Lemma~\ref{lem:selfframingonreads}).

%\label{lem:respectalt}

Allowed dependence gives an alternate way to express part of 
the Encap condition in Def.~\ref{def:valid}. 
For a step  
\( \configm{B}{\tau}{\mu} \trans{\phi} \configm{D}{\upsilon}{\nu} \)
that r-respects $\delta$ for $(\phi,\eff,\sigma)$ 
and $\Active(B)$ is not a call, 
and alternate step (\ref{eq:rrespectAnte}),
the condition implies
$\AllowedDep{\tau}{\tau'}{\upsilon}{\upsilon'}{\eff}{\sigma}{\delta}{\pi}$
in the notation of Def.~\ref{def:allowedDepEResB}.

A critical but non-obvious consequence of framed reads
is that for a pair of states $\sigma,\sigma'$ that are in `symmetric' agreement and transition to a pair $\tau,\tau'$ forming an allowed dependence, the transitions preserve agreement on any set of locations whatsoever.
The formal statement is somewhat intricate; it generalizes \RLIII~Lemma 6.12.

\begin{restatable}[balanced symmetry]{lemma}{freshsym}
\label{lem:fresh-sym}
\upshape
Suppose $\AllowedDep{\tau}{\tau'}{\upsilon}{\upsilon'}{\eff}{\sigma}{\delta}{\pi}$ and 
$\AllowedDep{\tau'}{\tau}{\upsilon'}{\upsilon}{\eff}{\sigma'}{\delta}{\pi^{-1}}$.
Suppose 
\[
\begin{array}{l}
\Lagree(\tau,\tau',\pi,(\freshLocs(\sigma,\tau)\union\rlocs(\sigma,\eff))\setminus\rlocs(\tau,\delta^\oplus))\\
\Lagree(\tau',\tau,\pi^{-1},(\freshLocs(\sigma',\tau')\union\rlocs(\sigma',\eff))
\setminus\rlocs(\tau',\delta^\oplus))
\end{array}
\]
Let $\rho,\rho'$ be any refperms with $\rho\supseteq\pi$ and $\rho'\supseteq\pi^{-1}$
that witness the allowed dependencies, i.e., 
\begin{equation}\label{eq:freshsym3}
\begin{array}{l}
\Lagree(\upsilon,\upsilon',\rho, (\freshLocs(\tau,\upsilon)\union 
\written(\tau,\upsilon))\setminus\rlocs(\upsilon,\delta^\oplus))\\
\Lagree(\upsilon',\upsilon,\rho', (\freshLocs(\tau',\upsilon')\union 
\written(\tau',\upsilon'))\setminus\rlocs(\upsilon',\delta^\oplus))
\end{array}
\end{equation}
Furthermore suppose 
\begin{equation}\label{eq:freshsym5}
\begin{array}{l}
\rho(\freshLocs(\tau,\upsilon)\setminus\rlocs(\upsilon,\delta))\subseteq\freshLocs(\tau',\upsilon')
\setminus\rlocs(\upsilon',\delta)\\
\rho'(\freshLocs(\tau',\upsilon')\setminus\rlocs(\upsilon',\delta))\subseteq\freshLocs(\tau,\upsilon)
\setminus\rlocs(\upsilon,\delta)
\end{array}
\end{equation}
Then we also have 
\[
\begin{array}{l}
\Lagree(\upsilon',\upsilon,\rho^{-1},(\freshLocs(\tau',\upsilon')\union\written(\tau',\upsilon'))
\setminus\rlocs(\upsilon',\delta^\oplus))\\
\rho(\freshLocs(\tau,\upsilon))\setminus\rlocs(\upsilon,\delta) \:=\:
\freshLocs(\tau',\upsilon')\setminus\rlocs(\upsilon',\delta)\
\end{array}
\]
\end{restatable}
\begin{proof}
From Definition~\ref{def:locagreement}
and (\ref{eq:freshsym3}) we know that 
$\rho$ and $\rho'$ are total on $\freshLocs(\tau,\upsilon)\setminus\rlocs(\upsilon,\delta)$ and 
$\freshLocs(\tau',\upsilon')\setminus\rlocs(\upsilon',\delta)$ respectively. 
Since $\rho$ and $\rho'$ are bijections, from (\ref{eq:freshsym5}), we have equal cardinalities:
$|\freshLocs(\tau,\upsilon)\setminus\rlocs(\upsilon,\delta)| = 
|\freshLocs(\tau',\upsilon')\setminus\rlocs(\upsilon',\delta)| $. 
So we get 
$\rho(\freshLocs(\tau,\upsilon)\setminus\rlocs(\upsilon,\delta))
=\freshLocs(\tau',\upsilon')\setminus\rlocs(\upsilon',\delta)$. 
Now from (\ref{eq:freshsym3}) using 
the symmetry lemma Eqn~(\ref{eq:LagreeSym}) for $\Lagree$ we get 
\[
\Lagree(\upsilon',\upsilon,\rho^{-1},\rho(\freshLocs(\tau,\upsilon)\setminus\rlocs(\upsilon,\delta))) \]
So, we have 
$\Lagree(\upsilon',\upsilon,\rho^{-1},\freshLocs(\tau',\upsilon')\setminus\rlocs(\upsilon',\delta))$. 
On other hand, we have $\written(\tau',\upsilon')\subseteq\locations(\tau')$ and we have
\( \rho'|_{\locations(\tau')}=\pi^{-1}|_{\locations(\tau')}=\rho^{-1}|_{\locations(\tau')} \), 
using vertical bar for domain restriction.
So from (\ref{eq:freshsym3}) we get
\[ \Lagree(\upsilon',\upsilon,\pi^{-1},\written(\tau',\upsilon')\setminus\rlocs(\upsilon',\delta^\oplus)) \] 
which we can write as 
$\Lagree(\upsilon',\upsilon,\rho^{-1},\written(\tau',\upsilon')\setminus\rlocs(\upsilon',\delta^\oplus))$.
Thus we get 
\[\Lagree(\upsilon',\upsilon,\rho^{-1},(\freshLocs(\tau',\upsilon')\union\written(\tau',\upsilon'))
\setminus\rlocs(\upsilon',\delta^\oplus))\]
\end{proof}

\begin{restatable}[preservation of agreement]{lemma}{selfframingagreement}
\label{lem:selfframing-agreement2}
\upshape
Suppose $\AllowedDep{\tau}{\tau'}{\upsilon}{\upsilon'}{\eff}{\sigma}{\delta}{\pi}$ and 
$\AllowedDep{\tau'}{\tau}{\upsilon'}{\upsilon}{\eff}{\sigma'}{\delta}{\pi^{-1}}$.
Suppose 
\[ \begin{array}{l}
\Lagree(\tau,\tau',\pi,(\freshLocs(\sigma,\tau)\union\rlocs(\sigma,\eff))\setminus\rlocs(\tau,\delta^\oplus)) \quad \mbox{and} \\
\Lagree(\tau',\tau,\pi^{-1},(\freshLocs(\sigma',\tau')\union\rlocs(\sigma',\eff))\setminus\rlocs(\tau',\delta^\oplus))
\end{array}
\]
Then for any $W\subseteq\locations(\tau)$, if $\Lagree(\tau,\tau',\pi,W)$ then 
$\Lagree(\upsilon,\upsilon',\rho,W\setminus\rlocs(\upsilon,\delta^\oplus))$, for any refperm $\rho$
that witnesses  $\AllowedDep{\tau}{\tau'}{\upsilon}{\upsilon'}{\eff}{\sigma}{\delta}{\pi}$.
\end{restatable}
%\selfframingagreement*
\begin{proof}
\begin{sloppypar}
Suppose 
$\Lagree(\tau,\tau',\pi,(\freshLocs(\sigma,\tau)\union\rlocs(\sigma,\eff))\setminus\rlocs(\tau,\delta^\oplus))$ 
suppose that $\rho\supseteq\pi$ witnesses 
$\AllowedDep{\tau}{\tau'}{\upsilon}{\upsilon'}{\eff}{\sigma}{\delta}{\pi}$,
so we get
\begin{equation}\label{eq:lemmagree3}
\Lagree(\upsilon,\upsilon',\rho, (\freshLocs(\tau,\upsilon)\union 
\written(\tau,\upsilon))\setminus\rlocs(\upsilon,\delta^\oplus)))
\end{equation}
Suppose 
$\Lagree(\tau',\tau,\pi^{-1},(\freshLocs(\sigma',\tau')\union\rlocs(\sigma',\eff))\setminus\rlocs(\tau',\delta^\oplus))$
and let $\rho'\supseteq\pi^{-1}$ 
witness 
$\AllowedDep{\tau'}{\tau}{\upsilon'}{\upsilon}{\eff}{\sigma'}{\delta}{\pi^{-1}}$
so we get 
\begin{equation}\label{eq:lemagree4}
\Lagree(\upsilon',\upsilon,\rho', (\freshLocs(\tau',\upsilon')\union 
\written(\tau',\upsilon'))\setminus\rlocs(\upsilon',\delta^\oplus))
\end{equation}
Now suppose $W$ is a set of locations in $\tau$ such that $\Lagree(\tau,\tau',\pi,W)$. 
We show that \[\Lagree(\upsilon,\upsilon',\rho,W\setminus\rlocs(\upsilon,\delta^\oplus))\]
\end{sloppypar}

For $x\in W\setminus\rlocs(\upsilon,\delta^\oplus)$, 
either $x\in\written(\tau,\upsilon)$ or $\tau(x)=\upsilon(x)$.
\begin{itemize}
\item If $x\in\written(\tau,\upsilon)$ then from~(\ref{eq:lemmagree3}), we have 
$\Rprel{\rho}{\upsilon(x)}{\upsilon'(x)}$.
\item If $\tau(x)=\upsilon(x)$, we claim that $\tau'(x)=\upsilon'(x)$. 
It follows that from $\Lagree(\tau,\tau',\pi,W)$ we have 
\( \upsilon(x)=\Rprel{\pi}{\tau(x)}{\tau'(x)}=\upsilon'(x) \).

We prove the claim by contradiction. If it does not hold then 
$x\in \written(\tau',\upsilon')$. By~(\ref{eq:lemagree4}) this implies
\( \Rprel{\rho'}{\upsilon'(x)}{\upsilon(x)}=\Rprel{\pi}{\tau(x)}{\tau'(x)} \).
Then, since $\rho'\supseteq\pi^{-1}$,
we would have $\tau'(x)=\pi(\pi^{-1}(\upsilon'(x)))=\upsilon'(x)$, which is a contradiction. 
\end{itemize}
For $o.f\in W\setminus\rlocs(\upsilon,\delta^\oplus)$, 
either $o.f\in \written(\tau,\upsilon)$ or $\tau(o.f)=\upsilon(o.f)$. 
\begin{itemize}
\item If $o.f\in \written(\tau,\upsilon)$ then from~(\ref{eq:lemmagree3}), we have
$\Rprel{\rho}{\upsilon(o.f)}{\upsilon'(\rho(o).f)}$.

\item If $\tau(o.f)=\upsilon(o.f)$, we claim that $\tau'(\pi(o).f)=\upsilon'(\pi(o).f)$. 
It follows that from $\Lagree(\tau,\tau',\pi,W)$ we have 
\( \upsilon(o.f)=\Rprel{\pi}{\tau(o.f)}{\tau'(\pi(o).f)}=\upsilon'(\pi(o).f) \).

The claim $\tau'(\pi(o).f)=\upsilon'(\pi(o).f)$ is proved by contradiction. 
If it does not hold then 
$\pi(o).f\in \written(\tau',\upsilon')$. By~(\ref{eq:lemagree4}) this implies
\( 
\Rprel{\rho'}{\upsilon'(\pi(o).f)}{\upsilon(\rho'\pi(o).f)}=\upsilon(o.f)=\Rprel{\pi}{\tau(o.f)}{\tau'(\pi(o).f)}
 \).
Then, since $\rho'\supseteq\pi^{-1}$,
we would have $\tau'(\pi(o).f) 
=\pi(\pi^{-1}(\upsilon'(\pi(o).f)))
=\upsilon'(\pi(o).f)$, hence $\tau'(\pi(o).f)=\upsilon'(\pi(o).f)$, which is a contradiction. 
\end{itemize}
This completes the proof of $\Lagree(\upsilon,\upsilon',\pi,W\setminus\rlocs(\upsilon,\delta^\oplus))$ 
for heap locations.
\end{proof}

\begin{restatable}[subeffect]{lemma}{lemsubeff}
%\label{lem:subeff}
\upshape
If $P \models \eff \leq \effe$ then the following hold for all 
$\sigma,\sigma',\tau,\tau',\upsilon,\upsilon',\pi,\delta$ such that $\sigma\models P$ and $\sigma'\models P$:
(a) 
$\sigma \allowTo \tau\models \eff$ implies $\sigma \allowTo \tau\models\effe$;
(b)
$\agree(\sigma, \sigma', \pi,\effe)$ implies $\agree(\sigma,\sigma',\pi,\eff)$; and 
(c) 
$\AllowedDep{\tau}{\tau'}{\upsilon}{\upsilon'}{\eff}{\sigma}{\delta}{\pi}$
implies 
$\AllowedDep{\tau}{\tau'}{\upsilon}{\upsilon'}{\effe}{\sigma}{\delta}{\pi}$.
% $\sigma,\sigma'\allowDep\tau,\tau'\models \eff$ implies
% $\sigma,\sigma'\allowDep\tau,\tau'\models \effe$
\end{restatable}
%\lemsubeff*
\begin{proof}
Straightforward from the definitions.  
For part (c), we have $\rlocs(\sigma,\eff)\subseteq\rlocs(\sigma,\effe)$,
so $\effe$ gives a stronger antecedent in Def.~\ref{def:allowedDepEResB}
and the consequent is unchanged between $\eff$ and $\effe$.
\end{proof}

\subsection{On the transition relation}

\begin{figure}[t!]
\begin{small}
\begin{mathpar}
  \inferrule[uLoad]
            { \sigma(y) = o\\ o\neq \semNull}
            { \configm{x:=y.f}{\sigma}{\mu} \trans{\phi}
              \configm{\skipc}{\update{\sigma}{x}{\sigma(o.f)}}{\mu} }

   \inferrule[uLoadX]
             { \sigma(y) = \semNull }
             { \configm{x:=y.f}{\sigma}{\mu} \trans{\phi} \Fault }

  \inferrule[uStore]
            { \sigma(x) = o\\ o\neq \semNull}
            { \configm{x.f:=y}{\sigma}{\mu} \trans{\phi}
              \configm{\skipc}{\update{\sigma}{o.f}{\sigma(y)}}{\mu} }

   \inferrule[uStoreX]
             { \sigma(x) = \semNull }
             { \configm{x.f:=y}{\sigma}{\mu} \trans{\phi} \Fault }

   \inferrule[uAssg]{}{
   \configm{x:=F}{\sigma}{\mu} \trans{\phi}
      \configm{\skipc}{\update{\sigma}{x}{\sigma(F)}}{\mu} }

   \inferrule[uSeq]{
          \configm{C}{\sigma}{\mu} \trans{\phi} \configm{D}{\tau}{\nu} 
       }{ 
          \configm{C;B}{\sigma}{\mu} \trans{\phi} \configm{D;B}{\tau}{\nu} 
       }

   \inferrule[uSeqX]{
          \configm{C}{\sigma}{\mu} \trans{\phi} \Fault
       }{ 
          \configm{C;B}{\sigma}{\mu} \trans{\phi} \Fault
       }

  \inferrule[uNew]
            { o \in \fresh(\sigma) \\
              \fields(K) \:=\:\ol{f}:\ol{T} \\
              \sigma_1 = \mbox{``$\sigma$ with $o$ added to heap, with type $K$ and default field values''}
            }
            { \configm{x:=\new{K}}{\sigma}{\mu} \trans{\phi}
              \configm{\skipc}{\update{\sigma_1}{x}{o} }{\mu}}

   \inferrule[uVar]{ x'= \varfresh(\sigma) } %\notin \dom{\sigma} }
          { \configm{ \varblock {x\scol T}{C} }{\sigma}{\mu}
            \trans{\phi}
            \configm{ \subst{C}{x}{x'} ; \Endvar(x') 
                    }{\extend{\sigma}{x'}{\Default{T}}}{\mu} }

   \inferrule[uEVar]{}{
   \configm{ \Endvar(x) }{\sigma}{\mu} \trans{\phi} \configm{ \skipc }{\drop{\sigma}{x}}{\mu}
   }

   \inferrule[uWhT]{ \sigma(E) = \True }
            { \configm{\whilec{E}{C}}{\sigma}{\mu} \trans{\phi}
              \configm{C;\whilec{E}{C}}{\sigma}{\mu} }

   \inferrule[uWhF]{ \sigma(E) = \False }
            { \configm{\whilec{E}{C}}{\sigma}{\mu} \trans{\phi}
              \configm{\skipc}{\sigma}{\mu} }

   \inferrule[uIfT]{ \sigma(E) = \True }
            { \configm{\ifc{E}{C}{D}}{\sigma}{\mu} \trans{\phi}
              \configm{C}{\sigma}{\mu} }

   \inferrule[uIfF]{ \sigma(E) = \False }
            { \configm{\ifc{E}{C}{D}}{\sigma}{\mu} \trans{\phi}
              \configm{D}{\sigma}{\mu} }

\end{mathpar}
\end{small}
%\hrule
\vspace{-1ex}
\caption{Rules for unary transition relation $\trans{\phi}$ omitted from Fig.~\ref{fig:transSel}.}
\label{fig:trans}
\end{figure}

Fig.~\ref{fig:trans} completes the definition of the transition relation,
with respect to a given pre-model $\phi$.\footnote{\label{fn:coerce}%
To be very precise, in the transition rules for context calls (Fig.~\ref{fig:transSel}), we implicitly use a straightforward coercion:
the pre-model is applied to states which may have more variables than the ones in scope for the method context $\Phi$ for $\phi$.
Suppose $\Phi$ is wf in $\Gamma$.
For method $m$ in $\Phi$, $\phi(m)$ is defined on $\Gamma$-states.
Suppose $\sigma$ is a state for $\Gamma$ plus some additional variables $\ol{x}$ (including but not limited to spec-only variables).  
Then $\phi(m)(\sigma)$ is defined by discarding the additional variables and applying $\sigma$.
If the result is a set of states, then each of these states is extended with the additional variables mapped to their initial values.
This coercion is implicitly used in the rules context calls, i.e., 
rules \rn{uCall}, \rn{uCallX}, and \rn{uCall0} in Figure~\ref{fig:transSel}.
The coercion is also used in \RLIII\ where it is formalized in more detail.
} 
The definition is also parameterized by a function, $\fresh$\index{$\fresh$}, for which we assume
that, for any $\sigma$,
$\fresh(\sigma)$ a non-empty set of non-null references that are not 
in $\sigma(\lloc)$.

We take care  to model realistic allocators, allowing their behavior to be nondeterminisic at the level of states, to model their dependence on unobservable low-level implementation details, yet not requiring the full, unbounded allocator required by some separation logics.
However, the language is meant to be deterministic modulo allocation.  
To make that possible for local variables, we assume given a function 
$\varfresh: states \to \mathconst{LocalVar}$
\index{$\varfresh$} such that
$\varfresh(\sigma)\notin\Vars{\sigma}$.
We also assume that $\varfresh$ depends only on the domain of the state:
\begin{equation}\label{eq:varfresh}
\Vars{\sigma}\setminus SpecOnlyVars = \Vars{\sigma'}\setminus SpecOnlyVars \mbox{ implies } \varfresh(\sigma)=\varfresh(\sigma') 
\end{equation}
These technicalities are innocuous and consistent with stack allocation of locals.

A configuration $\cfg$ \dt{faults} if $\cfg\tranStar{\phi}\Fault$.
It \dt{faults next} if $\cfg\trans{\phi}\Fault$.
It \dt{terminates} if $\cfg\tranStar{\phi}\configm{\skipc}{\tau}{\_}$ for some $\tau$---
so ``terminates'' means eventual normal termination.
When applied to traces, these terms refer to the last configuration:
a trace faults if it can be extended to a trace in which the last configuration faults next.
Perhaps it goes without saying that $\cfg$ \dt{diverges} means it begins an infinite sequence of transitions; in other words, it has traces of unbounded length.

For any pre-model $\phi$, the transition relation $\trans{\phi}$
is total in the sense that, for any $\configm{C}{\sigma}{\mu}$ with $C\not\equiv\skipc$, 
there is an applicable rule and hence a 
successor---which may be another configuration or $\Fault$.
This relies on the starting configuration being well formed in the sense 
that all free methods are bound either in the model or the environment,
all free variables are bound in the state,
and the command has no occurrences of $\Endvar$ or $\Endlet$.
Moreover, $\Endvar(x)$ (resp.\ $\Endlet(m)$) only occurs in a configuration if
$x$ is in the state (resp.\ $m$ is in the environment).

Well formedness is preserved by the transition rules, and can be formalized straightforwardly
(see \RLII) but in this article we gloss over it for the sake of clarity.

The transition relation $\trans{\phi}$ is called \dt{rule-deterministic} if for every configuration $\configm{C}{\sigma}{\mu}$ there is at most one applicable transition rule.
Strictly speaking, this is a property of the definition (Figs.~\ref{fig:transSel} and~\ref{fig:trans}), not of the relation $\trans{\phi}$.

\begin{lemma}[quasi-determinacy of transitions]\label{lem:determinacy}
\upshape
For any pre-model $\phi$,
\begin{list}{}{}
\item[(a)] $\trans{\phi}$  is rule-deterministic.
\item[(b)] If $\RprelT{\pi}{\sigma}{\sigma'}$ and 
$\configm{C}{\sigma}{\mu} \trans{\phi} \configm{D}{\tau}{\nu}$ 
and 
$\configm{C}{\sigma'}{\mu} \trans{\phi} \configm{D'}{\tau'}{\nu'}$ 
then $D\equiv D'$, $\nu=\nu'$, and $\RprelT{\rho}{\tau}{\tau'}$ for some $\rho\supseteq\pi$.
\item[(c)] If $\RprelT{\pi}{\sigma}{\sigma'}$
then $\configm{C}{\sigma}{\mu} \trans{\phi} \Fault$ iff
$\configm{C}{\sigma'}{\mu} \trans{\phi} \Fault$.
\end{list}
\end{lemma}
\begin{proof}
(a) This is straightforward to check by inspection of the transition rules: 
for each command form, check that the applicable rules are mutually 
exclusive.
One subtlety is in the case of context call.
If there is $\tau\in\phi(m)(\sigma)$, and also $\Fault\in\phi(m)(\sigma)$,
then two transition rules can be used for $\configm{m()}{\sigma}{\mu}$.
This is disallowed by Def.~\ref{def:preinterp} (fault determinacy).
Also, Def.~\ref{def:preinterp} (state determinacy),
and condition (iii) in the definition of $\RprelTS{\pi}{}{}$
(Def.~\ref{def:state-iso})
distinguishes between the two transition rules for empty and non-empty $\phi(m)(\sigma)$
(see Fig.~\ref{fig:transSel}).

(b) Go by cases on $\Active(C)$.  For any command other than context call or allocation,
take $\rho=\pi$ and inspect the transition rules.  
For example, $x.f:=y$ changes the state by updating a field with values that are in agreement mod $\pi$.
For the case of $x:=E$ we need that expression evaluation respects isomorphism of states, Lemma~\ref{lem:insensi}.
For allocation, let $\rho=\{(o,o')\}\union \pi$ where $o,o'$ are the allocated objects.  
For context call we get the result by
the determinacy conditions of Def.~\ref{def:preinterp}.
The only commands that alter the environment are $\keyw{let}$ and $\Endlet$,
and we get $\nu=\nu'$ because their behavior is independent of the state.

(c) Similar to the proof of (b); using item (i) in the definition of $\RprelTS{\pi}{}{}$, for context calls.
\end{proof}
A consequence of (a) is that the transition relation is fault deterministic:
no configuration has both a fault and non-fault successor
(by inspection, no single rule yields both fault and non-fault).
We note these other corollaries: \\
(d) For all $i$, if $\RprelT{\pi}{\sigma}{\sigma'}$ and 
$\configm{C}{\sigma}{\mu} \trans{\phi}{\!\!}^i \configm{D}{\tau}{\nu}$ 
and 
$\configm{C}{\sigma'}{\mu} \trans{\phi}{\!\!}^i \configm{D'}{\tau'}{\nu'}$ 
then $D\equiv D'$, $\nu=\nu'$, and $\RprelT{\rho}{\tau}{\tau'}$ for some $\rho\supseteq\pi$
(by induction on $i$).
\\
(e) If $\RprelT{\pi}{\sigma}{\sigma'}$ and 
$\configm{C}{\sigma}{\mu} \trans{\phi} \configm{D}{\tau}{\nu}$ 
then $\configm{C}{\sigma'}{\mu} \trans{\phi} \configm{D}{\tau'}{\nu}$ 
and $\RprelT{\rho}{\tau}{\tau'}$, for some $\tau$ and some $\rho\supseteq\pi$
(because only $\skipc$ lacks a successor).\\
(f) From a given configuration $\configm{C}{\sigma}{\mu}$, exactly one of these three outcomes is possible:
normal termination, faulting termination, divergence.

\lemreadeff*
\begin{proof}
To prove the lemma we prove a stronger result.

\textbf{Claim:} Under the assumptions of Lemma~\ref{lem:readeff},
for any $i\geq 0$ and any $B,B',\mu,\mu'$ with 
\[ \configm{C}{\sigma}{\_}\trans{\phi}^i \configm{B}{\tau}{\mu} \mbox{ and }
\configm{C}{\sigma'}{\_}\trans{\phi}^i \configm{B'}{\tau'}{\mu'} \]
there is some $\rho\supseteq\pi$ such that $B\equiv B'$, $\mu=\mu'$, and 
\[ %\label{eq:readeffA}
\begin{array}{l}
\Lagree(\tau,\tau',\rho,
(\freshLocs(\sigma,\tau)\union\rlocs(\sigma,\eff)\union\written(\sigma,\tau))\setminus\{\lloc\}) 
\\
\Lagree(\tau',\tau,\rho^{-1},
(\freshLocs(\sigma',\tau')\union\rlocs(\sigma',\eff)\union\written(\sigma',\tau'))\setminus\{\lloc\}) 
\\ 
\rho(\freshLocs(\sigma,\tau))\subseteq \freshLocs(\sigma',\tau') \\
\rho^{-1}(\freshLocs(\sigma',\tau'))\subseteq \freshLocs(\sigma,\tau) \\
\end{array}
\]
This directly implies the conclusion of the Lemma.

The claim is proved by induction on $i$.  
The base case holds because the fresh and written locations are empty, and agreement on 
$\rlocs(\sigma,\eff)$ is an assumption of the Lemma.
For the induction step, suppose the above holds and consider the next steps:
\[ \configm{B}{\tau}{\mu}\trans{\phi} \configm{D}{\upsilon}{\nu} \mbox{ and }
\configm{B}{\tau'}{\mu}\trans{\phi} \configm{D'}{\upsilon'}{\nu'} \]
%% Go by cases on the transition rule for the left step from 
%% $\configm{B}{\tau}{\mu}$ to $\configm{D}{\upsilon}{\nu}$,
%% in order according to Figs.~\ref{fig:transSel} and~\ref{fig:trans}.
%% By safety from the judgment $\Phi\HVflowtr{\Gamma}{M}{P}{C}{Q}{\eff}$, there are no faulting steps.
Go by cases on whether $\Active(B)$ is a call.

\medskip

\underline{Case \Active(B) not a call.}
By judgment $\Phi\HVflowtr{\Gamma}{M}{P}{C}{Q}{\eff}$, the step from $\tau$ to $\upsilon$ respects $(\Phi,M,\phi,\eff,\sigma)$,
as does the step from $\tau'$ to $\upsilon'$.  
As this is not a call, the collective boundary is 
\[ \delta = \unioneff{N\in(\Phi,\mu),N\neq mod(B,M)}{\bnd(N)} \]
So by w-respect for each step we have 
$\agree(\tau,\upsilon,\delta)$ and $\agree(\tau',\upsilon',\delta)$.

We begin by proving the left-to-right agreement and inclusion for the induction step, i.e.,
we will find $\dot{\rho}$ such that 
$\Lagree(\upsilon,\upsilon',\dot{\rho},
(\freshLocs(\sigma,\upsilon)\union\rlocs(\sigma,\eff)\union\written(\sigma,\upsilon))\setminus\{\lloc\})$
and $\dot{\rho}(\freshLocs(\sigma,\upsilon))\subseteq \freshLocs(\sigma',\upsilon')$.

We will apply r-respect of the left step, instantiated with $\pi:=\rho$ and with the right step.  
The two antecedents in r-respect are 
$\agree(\tau',\upsilon',\delta)$, which we have, 
and 
\[ \Lagree(\tau,\tau',\rho, (\freshLocs(\sigma,\tau)\union\rlocs(\sigma,\eff))\setminus\rlocs(\tau,\delta^\oplus)) \]
which follows directly from the induction hypothesis.
So r-respect yields some $\dot{\rho}\supseteq\rho$ (and hence $\dot{\rho}\supseteq\pi$)
with $D\equiv D'$, $\nu=\nu'$,
and 
\begin{equation}\label{eq:readeffB}
\begin{array}{l}
\Lagree(\upsilon,\upsilon',\dot{\rho},
(\freshLocs(\tau,\upsilon)\union\written(\tau,\upsilon))\setminus\rlocs(\upsilon,\delta^\oplus)) 
\\
\dot{\rho}(\freshLocs(\tau,\upsilon))\subseteq\freshLocs(\tau',\upsilon')
\end{array}
\end{equation}
To conclude the left-to-right $\Lagree$ part of the induction step 
it remains to show the two conditions
\[ \begin{array}{l}
\Lagree(\upsilon,\upsilon',\dot{\rho}, \rlocs(\sigma,\eff)\setminus\{\lloc\}) \\
\Lagree(\upsilon,\upsilon',\dot{\rho},
(\freshLocs(\tau,\upsilon)\union\written(\tau,\upsilon))\intersect \rlocs(\upsilon,\delta^\oplus)) 
\end{array}\]
The latter holds because the intersection is empty, owing to 
$\agree(\tau,\upsilon,\delta)$ and $\agree(\tau',\upsilon',\delta)$
(noting that $\rlocs(\upsilon,\delta)=\rlocs(\tau,\delta)$ from those agreements and using
Eqn~(\ref{eq:footprintAgreement}) and the requirement that boundaries have framed reads).
For the same reasons, we have 
\[ \Lagree(\upsilon,\upsilon',\dot{\rho}, \rlocs(\sigma,\eff)\intersect\rlocs(\upsilon,\delta)) \]
So it remains to show 
\( \Lagree(\upsilon,\upsilon',\dot{\rho}, \rlocs(\sigma,\eff)\setminus\rlocs(\upsilon,\delta^\oplus)
) \).
This we get by applying Lemma~\ref{lem:selfframing-agreement2},
instantiated by $\pi,\rho := \rho,\dot{\rho}$ 
and $W:=\rlocs(\sigma,\eff)$
(fortunately, the other identifiers in the Lemma are just what we need here).  
The antecedents of the Lemma include allowed dependencies and agreements that
we have established above,
and also the reverse of (\ref{eq:readeffB}), for $\dot{\rho}^{-1}$,
which we get by symmetric arguments, using the reverse conditions in the induction hypothesis.
The Lemma yields exactly what we need:
$\Lagree(\upsilon,\upsilon',\dot{\rho}, \rlocs(\sigma,\eff)\setminus\rlocs(\upsilon,\delta^\oplus)$.

\begin{sloppypar}
Finally, we have 
$\dot{\rho}(\freshLocs(\sigma,\upsilon))
= \rho(\freshLocs(\sigma,\tau)) \union \dot{\rho}(\freshLocs(\tau,\upsilon)) % by defs
\subseteq \freshLocs(\sigma',\tau') \union \dot{\rho}(\freshLocs(\tau,\upsilon)) % by ind hyp
\subseteq  \freshLocs(\sigma',\tau') \union \freshLocs(\tau',\upsilon') % by (\ref{eq:readeffB})
= \freshLocs(\sigma',\upsilon')$
by definitions, (\ref{eq:readeffB}), and the induction hypothesis.
\end{sloppypar}

The reverse agreement and containment in the induction step is proved symmetrically.

\medskip

\underline{Case \Active(B) is a call.}
Let the method be $m$ and suppose $\Phi(m) = \flowty{R}{S}{\effe}$.
By R-safe from the judgment $\Phi\HVflowtr{\Gamma}{M}{P}{C}{Q}{\eff}$, 
we have $\rlocs(\tau,\effe)\subseteq\rlocs(\sigma,\eff)\union\freshLocs(\sigma,\tau)$.
So by induction hypothesis we have 
$\Lagree(\tau,\tau',\rho,\rlocs(\tau,\eta)\setminus\{\lloc\})$.
So by $\phi\models\Phi$ and Def.~\ref{def:ctxinterp}(d),
there are two possibilities:
\begin{itemize}
\item $\phi(m)(\tau)=\emptyset=\phi(m)(\tau')$ and the steps both go by \rn{uCall0}.

\item $\phi(m)(\tau)\neq\emptyset\neq\phi(m)(\tau')$ and the steps both go by \rn{uCall}.
\end{itemize}
In the first case, $D\equiv B\equiv D'$, $\nu=\mu=\nu'$, and the states are unchanged so the agreements hold and we are done.

In the second case, we have $D\equiv B\equiv D'$, $\nu=\mu=\nu'$, 
$\upsilon\in\phi(m)(\tau)$ and $\upsilon'\in\phi(m)(\tau')$.
Moreover, by Def.~\ref{def:ctxinterp}(d) there is some $\dot{\rho}\supseteq\rho$
such that
\begin{equation}\label{eq:readeffC}
\begin{array}{l}
\Lagree(\upsilon,\upsilon',\dot{\rho},
(\freshLocs(\tau,\upsilon)\union\written(\tau,\upsilon))\setminus\{\lloc\}
\\
\dot{\rho}(\freshLocs(\tau,\upsilon))\subseteq\freshLocs(\tau',\upsilon')
\end{array}
\end{equation}
We also get reverse conditions, for $\dot{\rho}^{-1}$, by instantiating
Def.~\ref{def:ctxinterp}(d) with $\rho^{-1}$ and the states reversed.
We must show 
\[
\begin{array}{l}
\Lagree(\upsilon,\upsilon',\dot{\rho},
(\freshLocs(\sigma,\upsilon)\union\rlocs(\sigma,\eff)\union\written(\sigma,\upsilon))\setminus\{\lloc\}
\\
\dot{\rho}(\freshLocs(\sigma,\upsilon))\subseteq\freshLocs(\sigma',\upsilon')
\end{array}
\]
(and the reverse, which is by a symmetric argument).  
We get $\dot{\rho}(\freshLocs(\sigma,\upsilon))\subseteq\freshLocs(\sigma',\upsilon')$ 
using the induction hypothesis and (\ref{eq:readeffC}), similar to the proof above for the non-call case.
For the $\Lagree$ condition for $\upsilon,\upsilon'$, we have it for some locations by 
(\ref{eq:readeffC}).  It remains to show
$\upsilon,\upsilon'$ agree via $\dot{\rho}$ on the locations
$\freshLocs(\sigma,\tau)$, 
$\rlocs(\sigma,\eff)\setminus\written(\tau,\upsilon)$,
and 
$\written(\sigma,\upsilon)\setminus\written(\tau,\upsilon)$.
The latter simplifies to $\written(\sigma,\tau)$ 
because $\written(\sigma,\upsilon)\subseteq \written(\sigma,\tau)\union\written(\tau,\upsilon)$.
We obtain the agreements by applying Lemma~\ref{lem:selfframing-agreement2} 
with $\delta:=\emptyeff$, $\pi:=\rho$, and 
$W:= \freshLocs(\sigma,\tau) \union \rlocs(\sigma,\eff)\setminus\written(\tau,\upsilon)
\union \written(\sigma,\tau)$. 
To that end, observe that the above arguments have established 
$\AllowedDep{\tau}{\tau'}{\upsilon}{\upsilon'}{\eff}{\sigma}{\emptyeff}{\rho}$,
and symmetric arguments establish 
$\AllowedDep{\tau'}{\tau}{\upsilon'}{\upsilon}{\eff}{\sigma'}{\emptyeff}{\rho}$.
Moreover we have the antecedent agreements and $\dot{\rho}$ as witness.
So Lemma~\ref{lem:selfframing-agreement2} yields the requisite agreements and we are done.
\end{proof}

\begin{definition}[\textbf{denotation of command}, \ghostbox{$\means{\Gamma\proves C}$}]
%\label{def:denotComm}
Suppose $C$ is wf in $\Gamma$ and $\phi$ is a pre-model  
that includes all methods called in $C$ and not bound by $\keyw{let}$ in $C$.  
Define $\means{\Gamma\proves C}_\phi$
to be the function of type
$\means{\Gamma}\to\powerset(\means{\Gamma} \union \{\Fault\})$
given by 
\[
\means{\Gamma\proves C}_\phi(\sigma)
\eqdef
\{\tau \mid \configm{C}{\sigma}{\_}\tranStar{\phi} \configm{\skipc}{\tau}{\_} \} 
\;\union\; 
(\mifthenelse{ \{ \Fault \} }{ \configm{C}{\sigma}{\_}\tranStar{\phi} \Fault }{ \emptyset }) \]
\end{definition}
The denotation of a command can be used as a pre-model (Def.~\ref{def:preinterp}),
owing to this easily-proved property of the transition semantics:
if $\configm{C}{\sigma}{\_}\tranStar{\phi} \configm{D}{\tau}{\mu}$
then $\sigma\successorTo\tau$.
We define a pre-model suited to be a context model,
by taking into account a possible precondition:
Given $C$, $\phi$, formula $R$, and method name $m$ not in $\dom(\phi)$ and not called in $C$, one can extend $\phi$ to $\dot{\phi}$
that models $m$ by
\begin{equation}\label{eq:denotComm}
\dot{\phi}(m)(\sigma) \eqdef 
(\mifthenelse{ \{ \Fault \} }{ \sigma\not\models R }{ 
               \means{\Gamma\proves C}_\phi(\sigma) })
\end{equation}
The outcome is empty in case $C$ diverges.
The conditions of Def.~\ref{def:preinterp} hold owing to Lemma~\ref{lem:determinacy}, 
see corollaries (e) and (f) mentioned following that Lemma.
(Note that $\sigma\not\models R$ means there is no extension of $\sigma$ with values for spec-only variables in $R$ that make it hold.)

\begin{lemma}[context model denoted by command]
\label{lem:denotComm}
\upshape  
Suppose $\Phi\models_M^\Gamma C:\flowty{R}{S}{\effe}$ and $M=\mdl(m)$.  
Suppose $\phi$ is a $\Phi$-model.
Let $\dot{\Phi}$ be $\Phi$ extended with 
$ m:\flowty{R}{S}{\effe}$, where $m\notin\dom(\Phi)$ and $m$ not called in $C$.
Let $\dot{\phi}$ be the extension given by (\ref{eq:denotComm}).
If $N\in\Phi$ for all $N$ with $\mdl(m)\imports N$ then
$\dot{\phi}$ is a $\dot{\Phi}$-model.
\end{lemma}
\begin{proof}
% FUTURE \dn{check $\irimports$ and needs polishing}
To check $\dot{\phi}(m)$ with respect to $\flowty{R}{S}{\effe}$,
observe that $C$ does not fault (via $\phi$) from states that satisfy $R$,
by $\Phi\models_M C:\flowty{R}{S}{\effe}$ and $\phi$ being a $\Phi$-model.  
So we get part (a) in Def.~\ref{def:ctxinterp}.
Part (b) is an immediate consequence of $\Phi\models_M C:\flowty{R}{S}{\effe}$.
Part (c) requires boundary monotonicity for every $N$ with $\mdl(m)\imports N$.
Encap for the judgment gives monotonicity for every $N\in\Phi$ and also for $M$ itself.
We're done owing to hypothesis $N\in\Phi$ for every $N$ with $M\irimports N$.
That condition is for single steps, but by simple induction on steps it implies 
$\rlocs(\sigma,\delta)\subseteq\rlocs(\tau,\delta)$ for any $\tau$ such that
$\configm{C}{\sigma}{\_}\tranStar{\phi} \configm{B}{\tau}{\mu}$ for some $B,\mu$.
Part (d) is by application of Lemma~\ref{lem:readeff}.
\end{proof}

\section{Appendix: Unary logic and its soundness (re Sect.~\ref{sec:unaryLog})} %\label{sec:subsec-proofsys}

\subsection{Additional definitions and proof rules; soundness theorem}\label{sec:app:urules}

\begin{figure}
\begin{small}
\begin{mathpar}
\inferrule*[left=FieldAcc]
{z\not\equiv x}
{ \HPflowtr{}{\emptymod}{ y \neq \NULL\land z=y }{x := y.f}{x = z.f}{\wri{x},\rd{y},\rd{y.f}}}

\inferrule*[left=Assign]
{ y\not\equiv x }
{\HPflowtr{}{\emptymod}{ x = y }{ x:=F}{x=\subst{F}{x}{y}}{\wri{x},\ftpt(F)}}

\inferrule*[left=Seq]  
{ 
\Phi\HPflowtr{}{M}{P}{C_1}{P_1}{\eff_1}\\
\Phi\HPflowtr{}{M}{P_1}{C_2}{Q}{\eff_2,\rw{H\Img\ol{f}}}\\
P_1\imp H\#r\\
\eff_2 \mbox{ is } P/\eff_1\mbox{-immune}\\
\mathconst{spec-only}(r)
}
{
\Phi\HPflowtr{}{M}{P\wedge r=\lloc}{C_1;C_2}{Q}{\eff_1,\eff_2}
}

\inferrule*[left=While]  
{ 
\Phi\HPflowtr{}{M}{P\land E}{C}{P}{\eff,\rw{H\Img\ol{f}}}\\
\eff \mbox{ is } P/(\eff,\wri{H\Img\ol{f}})\mbox{-immune}\\
P\imp H\#r\\
%\rd{x}\not\in\bnd(M) \\
\ind{\unioneff{N\in\Phi,N\neq M}{\bnd(N)}}{\rTow(\ftpt(E)) } \\
\mathconst{spec-only}(r)
} 
{
\Phi\HPflowtr{}{M}{P\land r=\lloc}{\whilec{E}{C}}{P\land \neg E}{\eff,\ftpt{E}}
}

% \ruleIf % in body 

\end{mathpar}
\end{small}
\caption{Syntax-directed proof rules not given in Fig.~\ref{fig:proofrulesU}.
}
\label{fig:proofrules}
\end{figure}

\begin{figure}
\begin{mathpar}
\ruleModIntro

\ruleCtxIntroInTwo

\ruleCtxIntroCall

\inferrule*[left=Conj]
{\Phi \HPflowtr{}{M}{ P }{ C }{ Q_0 }{\eff} \\
 \Phi \HPflowtr{}{M}{ P }{ C }{ Q_1 }{\eff} }
{\Phi \HPflowtr{}{M}{ P}{ C }{ Q_0\land Q_1 }{\eff} }

\inferrule*[left=Disj]
{\Phi \HPflowtr{}{M}{ P_0 }{ C }{ Q }{\eff} \\
 \Phi \HPflowtr{}{M}{ P_1 }{ C }{ Q }{\eff} }
{\Phi \HPflowtr{}{M}{ P_0 \lor P_1 }{ C }{ Q }{\eff} }

\inferrule*[left=Exist]
{ \Phi\HPflowtr{\Gamma,x:T}{M}{  P }{ C }{ Q }{\eff}}
{ \Phi\HPflowtr{\Gamma}{M}{ (\some{x:T}{P}) }{ C }{ Q }{\eff}  }

% superseded by single rule for x:T
%% \inferrule*[left=Exist]
%% { \Phi\HPflowtr{\Gamma,x:K}{M}{ x\in G\land P }{ C }{ Q }{\eff}}
%% { \Phi\HPflowtr{\Gamma}{M}{ (\some{x:K \in G}{P}) }{ C }{ Q }{\eff}  }

%% \inferrule*[left=ExistRegion]
%% { \Phi\HPflowtr{\Gamma,x:\Region}{M}{ x = F\land P }{ C }{ Q }{\eff}}
%% { \Phi\HPflowtr{\Gamma}{M}{\subst{P}{x}{F}}{C}{Q}{\eff}}

\end{mathpar}
\caption{Structural proof rules not given in Fig.~\ref{fig:proofrulesU}.
}
\label{fig:proofrulesStruct}
\end{figure}

Figures~\ref{fig:proofrules} and~\ref{fig:proofrulesStruct} present the proof rules omitted from Fig.~\ref{fig:proofrulesU}.
They are to be instantiated only with well-formed premises and conclusions.
To emphasize the point we make the following definitions.
A correctness judgment is \dt{derivable} iff it can be inferred 
using the proof rules instantiated with well-formed premises and conclusion.
A proof rule is \dt{sound} if for any instance 
with well-formed premises and conclusion,
the conclusion is valid if the premises are valid and the side conditions hold.

Expression $G$ is \dt{$P/\eff$-immune}
iff this is valid:
\( P\imp \ind{\ftpt(G)}{\eff}  \).
Effect  $\effe$ is \dt{$P/\eff$-immune} iff 
$G$ is $P/\eff$-immune for every $G$ with 
$\wri{G\Img f}$ or $\rd{G\Img f}$ in $\effe$ (see \RLI).
The key fact about immunity is that if $\effe$ is $P/\eff$-immune then 
\begin{equation}\label{eq:immuneE}
\sigma\models P \mbox{ and }
\sigma\allowTo\tau \models\eff
\mbox{ imply } 
\rlocs(\sigma,\effe) = \rlocs(\tau,\effe)
\mbox{ and }
\wlocs(\sigma,\effe) = \wlocs(\tau,\effe)
\end{equation}

\begin{definition}[\textbf{boundary monotonicity spec}]\label{def:bmonspec}
$BndMonSp(P,\eff,M)$ is $\flowty{P\land Bsnap_M}{Bmon_M}{\eff}$
where $Bsnap_M$ and $Bmon_M$
\index{$\Bsnap$}\index{$\Bmon$}
 are defined as follows.
Let $\delta$ be $\bnd(M)$,
normalized so that for each field $f$
for which $\rd{H\Img f}$ occurs in $\bnd(M)$ for some $H$,  
there a single region expression $G_f$ with $\rd{G_f\Img f}$ in $\delta$. 
Let $Bsnap_M$ (for ``boundary snap'') be the conjunction over fields $f$ of formulas $s_f=G_f$ where each $s_f$ is a fresh spec-only variable.
Let $Bmon_M$ be the conjunction over fields $f$ of formulas $s_f\subseteq G_f$.
\end{definition}

\begin{remark}\upshape
In case boundaries are empty, the postcondition becomes vacuously true.
As a result, the second premises in rules \rn{ModIntro} and \rn{CtxIntroCall}, for boundary monotonicity, become trivial consequences of the main premises. 
\end{remark}

\begin{remark}\upshape
The syntax directed rules in Fig.~\ref{fig:proofrules} are very similar to the unary proof rules in \RLIII.
Other than addition of modules, one noticeable 
difference is that in \RLIII\ rules \rn{Seq} and \rn{While} require the effects to be read framed.  This is not needed with the current definition of valid judgment which imposes a stronger condition for read effects (Def.~\ref{def:valid}).
\qed\end{remark}

\begin{remark}\label{rem:ctxIntro}
\upshape
Recall that rule \rn{CtxIntro} (Fig.~\ref{fig:proofrulesU}) allows the introduction of additional modules, by adding methods to the hypothesis context (see Sect.~\ref{sec:encap}).  It has side conditions which ensure encapsulation.
For method calls, \rn{CtxIntro} is useful to add context that is not imported by the method's module.
A separate rule, \rn{CtxIntroCall}, is needed to add context that is imported by the method's module (as it was in \RLII).
To add a method of the current module to the context, rule \rn{CtxIntroIn2} is used if the judgment is for a non-call; otherwise \rn{CtxIntroCall} is used.
To add a method to the context for a module already present in context, rule \rn{CtxIntroIn1} is used.
The context intro rules are not applicable to control structures, so requisite context should be introduced for their constituents before their proof rules are used.

The axioms for atomic commands (e.g., \rn{Alloc} in Fig.~\ref{fig:proofrulesU}) are for the default module $\emptymod$ and the empty context, or in the case of \rn{Call} the context with just the called method.
Rule \rn{ModIntro} changes the current module from $\emptymod$ to another one; 
this is not needed in \RLII\ because it's main significance is to enforce boundary monotonicity (Def.~\ref{def:valid}) which is not needed in \RLII.
For non-call atomic commands, the rule needs to be used before introducing methods of the current module into the context. 

Some of the rules use a second premise, the boundary monotonicity spec of Def.~\ref{def:bmonspec}, to enforce boundary monotonicity.\footnote{One can contrive a rule with  only one premise, subject to conditions that ensure it refines the second spec, but we prefer this way.}
In many cases, this judgment can be derived from the primary judgment of the rule,
by a simple use of the \rn{Frame} rule 
to get $\Bsnap$ in the postcondition, and then \rn{Conseq} to get $\Bmon$.
% DN remove comment about WhyRel; we'll make current comments where needed.
%% WhyRel implements a sort of normal form for boundaries,\dn{normal form? currently?}  and uses conservative syntactic checks for the monotonicity condition. 
%For the logic, we design a proof rule that is semantically complete.  
%Its formulation is based on concrete field names, though one can imagine a similar formulation based on data groups.
\qed\end{remark}

%% % ok example but not important 
%% For example, let $\bnd(M)$ be  $\rd{r},\rd{r\Img f}$ for variable $r$,
%% which is already in the normal form with $r$ being the expression for $f$.
%% Then $\Bsnap_M$ is $s_f=r$ and $\Bmon_M$ is $s_f\subseteq r$ for some fresh variable $s_f$.
%% Consider this instance of the \rn{FieldUpd} axiom:
%% \[ \HPflowtr{}{\emptymod}{ x \neq \NULL }{x.g := y}{x.g = y}{\wri{\{x\}\Img g},\rd{x},\rd{y}} \]
%% The boundary monotonicity spec for its precondition and effect is
%% \[
%% \flowty{x\neq\NULL \land s_f=r }{s_f\subseteq r}{\wri{\{x\}\Img g},\rd{x},\rd{y}} 
%% \]
%% From the axiom we derive, by \rn{Frame} rule
%% (since $\ind{s_f=r}{\wri{\{x\}.f}}$ is true):
%% \[ \HPflowtr{}{\emptymod}{ x \neq \NULL\land s_f=r }{x.f := y}{x.f = y\land s_f\subseteq r}{\wri{x.f},\rd{x},\rd{y}} \]
%% whence the boundary monotonicity spec follows by \rn{Conseq}.  

\unarysoundness*

The proofs comprise the following subsections~\ref{sec:app:call}--\ref{sec:app:link}.
We prove the R-safe and Encap conditions for all rules, since Encap differs from the definition in \RLII\ and R-safe is a new addition.
Otherwise, the proofs are mostly as in \RLII.
We give full proofs for the rules that have significantly changed from \RLII,\RLIII,
e.g., \rn{CtxIntro} and \rn{SOF}.

\subsection{Soundness of \rn{Call}}\label{sec:app:call}

%\[\ruleCall\]

To show soundness of the axiom
$m : \flowty{P}{Q}{\eff} \proves_{\emptymod} m(): \flowty{P}{Q}{\eff} $,
consider any $\sigma$ with $\hat{\sigma}\models P$
where $\hat{\sigma}\eqdef\extend{\sigma}{\ol{s}}{\ol{v}}$ and $\ol{s}$ are the spec-only variables of $P$.
Consider any $\phi$ that is an $(m:\flowty{P}{Q}{\eff})$-model.
Owing to $\hat{\sigma}\models P$ and Def.~\ref{def:ctxinterp} of context model,
there is no faulting transition.  
So either $\phi(m)(\sigma)$ is empty and the stuttering transition is taken (transition rule \rn{uCall0}),
or execution terminates in a single step $\configm{m()}{\sigma}{\_}\trans{\phi}\configm{\skipc}{\tau}{\_}$
with $\tau\in\phi(m)(\sigma)$ (transition rule \rn{uCall}).  
The stuttering transition repeats indefinitely, and Safety, Post, Write, R-safe, and Encap all hold because the configuration never changes.
In case execution terminates in $\configm{\skipc}{\tau}{\_}$,
Safety, Post, and Write are immediate from Def.~\ref{def:ctxinterp},
which in particular says $\hat{\tau}\models Q$ where $\tau\eqdef\extend{\tau}{\ol{s}}{\ol{v}}$.
For R-safe, there is only one configuration that is a call, the initial one, and it is r-safe 
because the frame condition in the judgment is exactly the frame condition of the method's spec.

Encap requires boundary monotonicity for the current module and every module in context.
Boundary monotonicity for module $\emptymod$ holds because $\bnd(\emptymod)=\emptymod$.
It holds for $\mdl(m)$, the one module in context, by Def.~\ref{def:ctxinterp}(c), since $\imports$ is reflexive.

Encap requires w-respect for every $N$ in context different from the current module,
which in this case means either $\mdl(m)$ or nothing, depending whether $\mdl(m)=\emptymod$.
The step w-respects $\mdl(m)$ because it is a call and $\mdl(m)\imports\mdl(m)$.

Encap considers $\sigma',\pi$ such that 
$\Lagree(\sigma,\sigma',\pi, \rlocs(\sigma,\effe)\setminus\rlocs(\sigma,\delta^\oplus))$
where collective boundary $\delta$  is the union of boundaries for $N$ in context and not imported by $\mdl(m)$; hence $\delta = \emptyeff$.  
By condition (d) in Def.~\ref{def:ctxinterp},
we have $\phi(m)(\sigma)=\emptyset$ iff $\phi(m)(\sigma')=\emptyset$,
so either both transition go via \rn{uCall0} to unchanged states, thus satisfying r-respect,
or both transition go via \rn{uCall} to states $\tau,\tau'$ with
$\tau\in\phi(m)(\sigma)$ and $\tau'\in\phi(m)(\sigma')$.
In the latter case, $\rlocs(\sigma,\emptyeff)^\oplus$ is $\{\lloc\}$ by definition of $\rlocs$, 
and the r-respect condition to be proved is exactly the 
condition (d) in Def.~\ref{def:ctxinterp}.
In a little more detail, we must show the final states agree on $\freshLocs(\sigma,\tau)\union\written(\sigma,\tau)\setminus\rlocs(\tau,\emptyeff^\oplus)$
which simplifies to
$\freshLocs(\sigma,\tau)\union\written(\sigma,\tau)\setminus\{\lloc \}$.
R-respects also requires 
a condition which simplies to  $\rho(\freshLocs(\sigma,\tau))\subseteq \freshLocs(\sigma',\tau')$
because $\rlocs(\tau,\emptyeff)=\emptyset$.

%% % Temp save but we probably don't want this 
%% \dn{We're considering to change Def.~\ref{def:ctxinterp}(d) 
%% to subtract $\delta=\unioneff{N,\mdl(m)\not\imports N}{\bnd(N)}$.  
%% Then this proof would need to be fixed, to establish 
%% $\rlocs(\sigma,\effe)\intersect \rlocs(\sigma,\delta)=\emptyset$.
%% That could be accomplished by giving \rn{Call} a new side condition:
%% $R \imp \ind{ \unioneff{N,\mdl(m)\not\imports N}{\bnd(N)} }{\rTow(\effe)}$.
%% }

\subsection{Soundness of \rn{FieldUpd}}

This is an axiom: 
$ \HPflowtr{}{\emptymod}{ x \neq \NULL }{x.f := y}{x.f = y}{\wri{x.f},\rd{x},\rd{y}}$.
The Safety, Post, and Write conditions are straightforward and proved the same way as in \RLI.  
R-safe holds because there is no method call.
For Encap, the only steps to consider are the single terminating steps from states where $x$ is not null. 
So suppose 
$\configm{x.f:=e}{\sigma}{\_}\trans{\phi}\configm{\skipc}{\upsilon}{\_}$,
where $\upsilon=\update{\sigma}{\sigma(x).f}{\sigma(y)}$.
For Encap,  boundary monotonicity: the only relevant boundary is $\bnd(\emptymod)$ which is empty,
so monotonicity holds vacuously.
For Encap, w-respect is vacuously true for the empty boundary.
For r-respect, since the command is not a call the collective boundary 
is empty.  As we are considering the initial step and the boundary is empty, the antecedent of r-respect can be written

\begin{equation}\label{eq:soundUpd}
\Lagree(\sigma,\sigma',\pi, \rlocs(\sigma,\eff) \setminus\{\lloc\}) 
\mbox{ and }
\configm{x.f:=e}{\sigma'}{\_}\trans{\phi}\configm{\skipc}{\upsilon'}{\_}
\end{equation}
Since there is no allocation, extending $\pi$ is not relevant, and the condition about fresh locations is vacuous, so it remains to show that 
\(
\Lagree(\upsilon,\upsilon',\pi,
(\written(\sigma,\upsilon))\setminus\{\lloc\}) 
\).
What is written is the location $\sigma(x).f$, so this simplifies to
\(
\Lagree(\upsilon,\upsilon',\pi,\{ \sigma(x).f \})
\).
Given that $\rd{x}$ is in the frame condition, we have $x\in\rlocs(\sigma,\eff)$
so the assumption (\ref{eq:soundUpd}) gives agreement on which location is written.
It remains to show agreement on the value written, which is $\sigma(y)$ versus $\sigma'(y)$.
From the frame condition we have $y\in\rlocs(\sigma,\eff)$, 
so by (\ref{eq:soundUpd}) we have initial agreement on it 
and we are done.

\subsection{Soundness of \rn{If}}

%\[ \ruleIf \]

Suppose the premises are valid:
$\Phi \HVflowtr{}{M}{ P \land E}{C_1}{Q}{\eff}$ and 
$\Phi \HVflowtr{}{M}{P \land \neg E}{C_2}{Q}{\eff}$.
Suppose the side condition is valid: 
$\ind{\unioneff{N\in\Phi,N\neq M}{\bnd(N)}}{\rTow(\ftpt(E)) }$.
To show $\Phi \HPflowtr{}{M}{P}{\ifc{E}{C_1}{C_2}}{Q}{\eff,\ftpt(E)} $,
we only consider R-safe and Encap, because the rest is straightforward and similar to previously published proofs.  
Consider any $\Phi$-model $\phi$, noting that the premises have the same context.
Consider and any $\sigma$ with $\sigma\models P$.
Consider the case that $\sigma(E)=true$ (the other case being symmetric).
So the first step is 
$\configm{\ifc{E}{C_1}{C_2}}{\sigma}{\_}
\trans{\phi}
\configm{C_1}{\sigma}{\_}$.
This is not a call, so the step (or rather, its starting configuration) satisfies r-safe. 
For Encap, the first step does not write, so it satisfies boundary monotonicity and w-respect.

\begin{sloppypar}
For r-respect, the requisite collective boundary is 
$\delta = \unioneff{N\in(\Phi,N\neq M}{\bnd(N)}$ because there is no $\Endcall$ and the environment is empty.
We show r-respect for the first step, i.e., instantiating r-respect with
$\tau,\upsilon := \sigma,\sigma$. 
The requisite condition for this step is that for any $\sigma'$, if 
\[ \configm{\ifc{E}{C_1}{C_2}}{\sigma'}{\_}
\trans{\phi}
\configm{D'}{\sigma'}{\_}
\]
and 
$\Lagree(\sigma,\sigma',\pi,
(\freshLocs(\sigma,\sigma)\union\rlocs(\extend{\sigma}{\ol{s}}{\ol{v}},(\eff,\ftpt(E)))\setminus\rlocs(\sigma,\delta^\oplus))
$
then $D'\equiv C_1$ and 
two agreement conditions about fresh and written locations.
(We omitted one antecedent, $\agree(\sigma',\sigma',\delta)$, which is vacuous.)
There are no fresh or written locations, so those two conditions hold.
It remains to prove $D'\equiv C_1$.
We can simplify the antecedent to 
\[ 
\Lagree(\sigma,\sigma',\pi,
(\rlocs(\sigma,(\eff,\ftpt(E)))\setminus\rlocs(\sigma,\delta^\oplus)))
\]
Because the side condition is true,
$ \ind{\unioneff{N\in\Phi,N\neq M}{\bnd(N)}}{\rTow(\ftpt(E)) }$,
we have $\rlocs(\sigma,\ftpt(E))$ disjoint from 
$\rlocs(\sigma,\delta^\oplus)$.
So 
$\Lagree(\sigma,\sigma',\pi,
(\rlocs(\sigma,(\eff,\ftpt(E)))\setminus\rlocs(\sigma,\delta^\oplus)))$
implies 
$\Lagree(\sigma,\sigma',\pi, \rlocs(\sigma,\ftpt(E)))$.
Hence  $\sigma(E)=\sigma'(E)$  by footprint agreement lemma.
By semantics, $D'\equiv C_1$ and we are done.
\end{sloppypar}

For subsequent steps in the case $\sigma(E)=true$, we can appeal to the premise for $C_1$
which applies to the trace starting from $\configm{C_1}{\sigma}{\_}$
since $\sigma\models P\land E$.  
This yields r-safe and respect (as well as the other conditions for validity).

\subsection{Soundness of \rn{Var}}

% I can't find a proof of Var in RLI, RLII, or readRL! 

Suppose the premise is valid:
\( \Phi \HVflowtr{\Gamma,x:T}{M}{P\land x=\Default{T}}{C}{P'}{\rw{x},\eff} \).
To prove the R-safe and Encap conditions for 
$ \Phi \HVflowtr{\Gamma}{M}{P}{\varblock{x\scol T}{C}}{P'}{\eff} $,
let $\phi$ be a $\Phi$-model and $\hat{\sigma}\models P$ (where $\hat{\sigma}$ extends $\sigma$
with values for the spec-only variables of $P$).
The first step is 
$\configm{ \varblock {x\scol T}{C} }{\sigma}{\mu}
            \trans{\phi}
            \configm{ \subst{C}{x}{x'} ; \Endvar(x') 
                    }{\extend{\sigma}{x'}{\Default{T}}}{\mu}$
where $x'= \varfresh(\sigma)$.  
Let $\delta=\unioneff{N\in\Phi,N\neq M}{\bnd(N)}$.
This step satisfies w-respect because the variables in $\delta$ are already in scope, so are distinct from $x'$.  (Indeed, $x'$ is a local variable and boundaries cannot contain locals.)
The first configuration satisfies r-safe because it is not a call.
To show the first step satisfies r-respect,
note first that $\rlocs(\sigma,\delta)=\rlocs(\extend{\sigma}{x'}{\Default{T}}, \delta)$, again because $x'$ is not in $\delta$.
Consider taking the first step from an alternate state $\sigma'$ satisfying the requisite agreements with $\sigma$.
Now $\sigma'$ has the same variables as $\sigma$ 
(by definition of r-respect, including footnote~\ref{fn:r-respect}), 
and by assumption (\ref{eq:varfresh}) the choice of $x'$ depends only on the domain of $\sigma$, 
so the alternate step introduces the same local $x'$
and the same command $\subst{C}{x}{x'} ; \Endvar(x')$. 
We have $\freshLocs(\sigma,\extend{\sigma}{x'}{\Default{T}}) = \{x'\}$ by definition,
and the agreements for r-respect follow directly, noting 
that $\Default{T}$ is a fixed value dependent only on the type $T$.

If execution reaches the last step, that last step satisfies r-safe and respects because it merely removes $x'$ from the state.
For any other step, the result follows straightforwardly from R-safe and Encap for the premise:
The state $\extend{\sigma}{x'}{\Default{T}})$ satisfies 
$P\land x=\Default{T}$, and a trace of $\subst{C}{x}{x'} ; \Endvar(x')$
gives rise to a trace of $C$ (by dropping $\Endvar(x')$ and renaming),
for which the premise yields r-safe, respects, and indeed Safety etc.

%% % obsolete: how it would be done if the alt r-respect needs complete trace
%% \dn{It may be worthwhile spelling this out, simply because it involves a similar but simpler structure to what happens with Let.  With Var, the premise's Encap condition is for executions of $C$, 
%% whereas to prove the conclusion we consider alternate  executions starting from 
%% $\varblock {x\scol T}{C}$, but the first step is fixed.  Moreover, 
%% one needs to connect $\rlocs(\extend{\sigma}{x'}{\Default{T}}, (\rw{x},\eff))$ with $\rlocs(\sigma,\eff)$.
%% We probably need that valid judgments are closed under renaming of variables---add as footnote to def of valid judgment?.}

\subsection{Soundness of \rn{ModIntro}}

\[\ruleModIntro\]

%% The context is unchanged, so to prove the conclusion for any $\Phi$-model $\phi$ 
%% we can use $\phi$ to appeal to the premises.
%% For Safety, Post, Write, and R-safe, we get the conclusion immediately from the first premise, because these condition do not depend on the current module.

%% For Encap (a), boundary monotonicity for $N\in\Phi$ is from the first premise, 
%% and boundary monotonicity for $N=M$ is from the second premise.

%% For Encap(b), the condition quantifies over $N$ different from the current module, the condition for the conclusion is the same as for the premise.

%% For Encap(c), go by cases whether $A$ is a call.
%% If not, then the collective boundary for the premises is $\unioneff{N, N\in\Phi,N\neq\emptymod}{\bnd(N)}$,
%% and for the conclusion it is $\unioneff{N, N\in\Phi,N\neq M}{\bnd(N)}$.
%% These are the same, because $M\notin\Phi$ by the side conddition.
%% So this is immediate from the first premise.

%% If $A$ is a call to some method $p$, the collective boundary is 
%% $\unioneff{N, N\in\Phi,\mdl(p)\not\imports N}{\bnd(N)}$.
%% This is independent of the current module, so again the conclusion is direct from the first premise.

For Encap, as $A$ is an atomic command $A$, the only reachable step is 
the single step taken in a terminating execution 
\( \configm{A}{\sigma}{\_} \trans{\phi} \configm{\skipc}{\tau}{\_} \)
or the stutter step by \rn{uCall0}, which has the form
\( \configm{A}{\sigma}{\_} \trans{\phi} \configm{A}{\sigma}{\_} \).
(A stutter step may repeat, but no other state is reached.)
In either case, there is no $\Endcall$ in the configuration, and the environment is empty.

For Encap, boundary monotonicity for $N\in\Phi$ is from the first premise, 
and boundary monotonicity for $N=M$ is from the second premise.

For Encap, the w-respect condition quantifies over $N\in(\Phi,\_)$ different from the $mod(A,M)$.
Since the environment is empty,
$N\in(\Phi,\_)$ is the same as $N\in\Phi$.
Since $A$ has no $\Endcall$, $mod(A,M)$ is $M$.
So the condition quantifies over $N\in\Phi$ with $N\neq M$.
By side condition $M\notin\Phi$, this is the same as $N\in\Phi$.
So the condition for the conclusion is the same as for the first premise, from which we obtain Encap (a).  

For Encap r-respect, go by cases whether $A$ is a method call.  
If not, then the collective boundary for the premise is 
$\unioneff{N, N\in(\Phi,\_),N\neq mod(A,\emptymod)}{\bnd(N)}$,
and for the conclusion it is $\unioneff{N, N\in(\Phi,\_),N\neq mod(A,M)}{\bnd(N)}$.
These are the same, owing to side condition $M\notin\Phi$, and simplifying as above.
So r-respect is immediate by the first premise.

If $A$ is a call to some method $p$, the collective boundary is 
$\unioneff{N, N\in(\Phi,\_),\mdl(p)\not\imports N}{\bnd(N)}$.
This is independent of the current module, so again the conclusion is direct from the first premise.

\subsection{Soundness of \rn{CtxIntro}}

\[\ruleCtxIntro\]

%Commentary: It is pointless to use this rule unless $\mdl(m)$ is neither equal to $M$ nor in $\Phi$.

\begin{proof}
Consider any $(\Phi,m\scol\flowty{R}{S}{\effe})$-model $\phi$.
By definitions, $\drop{\phi}{m}$ is a $\Phi$-model, with which we can instantiate the premise.  
The Safety, Post, Write, and R-safe conditions follow from those for the premise---it is only the Encap condition that has a different meaning for the conclusion than it does for the premise.
 
For Encap, as $A$ is an atomic command $A$, the only reachable step is 
a single step, either the terminating step
\( \configm{A}{\sigma}{\_} \trans{\phi} \configm{\skipc}{\tau}{\_} \)
given by \rn{uCall} or the stuttering step
by \rn{uCall0},
which is  \( \configm{A}{\sigma}{\_} \trans{\phi} \configm{A}{\tau}{\_} \) with $\tau=\sigma$.
%The latter satisfies all conditions of Encap, so it remains to consider a terminating execution, and the instantiation $\tau,\upsilon:=\sigma,\tau$ of Encap, where $\sigma\models P$.

For Encap, for boundary monotonicity we need $\rlocs(\sigma,\bnd(N))\subseteq\rlocs(\tau,\bnd(N))$ 
for all $N$ with $N\in (\Phi,m:\flowty{R}{S}{\effe})$ or $N=M$.
This holds for all $N\in\Phi$, and for $N=M$, by the same condition from the premise,
so it remains to consider $N=\mdl(m)$.
From the premise we have $\sigma\allowTo\tau\models \eff$.
By side condition (and $\sigma\models P$) we have $\sigma\models \ind{\bnd(N)}{\eff}$.
So we have $\agree(\sigma,\tau,\bnd(N)$ by separator property (\ref{eq:sepagree}).
Since boundaries are read framed (Def.~\ref{def:framedreadsDyn}),
we can apply footprint agreement (\ref{eq:footprintAgreement})
to get $\rlocs(\upsilon,\bnd(N))=\rlocs(\tau,\bnd(N))$. 

For Encap, we need w-respect of each $N$ with $N\in (\Phi,m:\flowty{R}{S}{\effe})$ and $N\neq mod(A,M)$. (simplified for the empty environment, as in the proof of \rn{ModIntro}).
Since $\Endcall$ does not occur in $A$, $N \neq mod(A,M)$ simplifies to $N\neq M$.
Again, we have this condition from the premise for all $N$ except $N=\mdl(m)$.
For that, in the case that $A$ is not a call to a method $m$ with $\mdl(m)\imports N$,
we must show $\agree(\sigma,\tau,\bnd(N))$; and it was shown already in the proof of (c).

For Encap, we show r-respect by cases:

\underline{Case: the step is not a call.}
Then the collective boundary 
is $\delta=\unioneff{N\in (\Phi,m:\flowty{R}{S}{\effe}), N\neq mod(A,M)}{\bnd(N)}$,
and $N\neq mod(A,M)$ is just $N\neq M$.

Let $\dot{\delta}$ be the collective boundary for the premise:
$\dot{\delta}=\unioneff{N\in \Phi, N\neq M }{\bnd(N)}$
(again, simplifying $N\neq mod(A,M)$ to $N\neq M$).
So $\delta$ is $\dot{\delta},\bnd(N)$.
If $N=M$, or $N\in\Phi$, or $\bnd(N)=\emptyeff$ then $\dot{\delta}$ is equivalent to $\delta$ 
and we get r-respect directly from the premise.
Otherwise, suppose \( \configm{A}{\sigma'}{\_} \trans{\phi} \configm{B}{\tau'}{\_} \) 
and $\agree(\sigma',\tau',\delta)$ 
and 
\begin{equation}\label{eq:CtxIntroA}
\Lagree(\sigma,\sigma',\pi, \rlocs(\sigma,\eff)\setminus\rlocs(\sigma,\delta^\oplus))
\end{equation}
(This is simplified from the general condition of r-respect, which includes fresh locations in the assumed agreement; here, because we consider the first step of computation, there are none.)
We must show 
\begin{equation}\label{eq:CtxIntro}
\begin{array}{l}
\Lagree(\tau,\tau',\rho,
(\freshLocs(\sigma,\tau)\union\written(\sigma,\tau))\setminus\rlocs(\tau,\delta^\oplus)) \\
\rho(\freshLocs(\sigma,\tau)\setminus\rlocs(\tau,\delta))\subseteq\freshLocs(\sigma',\tau')\setminus\rlocs(\tau',\delta)
\end{array}
\end{equation}
The premise gives an implication similar to (\ref{eq:CtxIntroA})$\imp$(\ref{eq:CtxIntro}) 
but for $\dot{\delta}$.
Now $\dot{\delta}$ may be a proper subeffect of $\delta$, so we only have
$\rlocs(\sigma,\dot{\delta})\subseteq\rlocs(\sigma,\delta)$ and thus
$\rlocs(\sigma,\eff)\setminus\rlocs(\sigma,\delta^\oplus)$ may be a proper subset of 
$\rlocs(\sigma,\eff)\setminus\rlocs(\sigma,\dot{\delta}^\oplus)$.
This means (\ref{eq:CtxIntroA}) does not imply the antecedent in r-respects for the premise so we 
cannot simply apply that.
% FUTURE
% \dn{Is there a way to use end-to-end read effect lemma here (together with axioms)?}
Instead, we exploit the fact that the command $A$ is one of the assignment forms:
$x := F$, $x := \new{K}$, $x := x.f$, $x.f := x$.
Each of these has a minimal set of locations on which it depends in the relevant sense.

\textbf{Claim:} for each of the atomic, non-call commands, and for each $\sigma,\sigma',\mu,\mu'$, there is a
finite number of minimal sets $X\subseteq\locations(\sigma)$ 
such that 
if $\configm{A}{\sigma}{\mu} \trans{} \configm{\skipc}{\tau}{\mu}$, 
$\configm{A}{\sigma'}{\mu} \trans{} \configm{\skipc}{\tau'}{\mu}$, and
$\Lagree(\sigma,\sigma',\pi, X)$,
then there is $\rho\supseteq\pi$ with 
\[
\Lagree(\tau,\tau',\rho, \freshLocs(\sigma,\tau)\union\written(\sigma,\tau))
\mbox{ and }
\rho(\freshLocs(\sigma,\tau)) \subseteq \freshLocs(\sigma',\tau')
\]
(Here we omit the model for $\trans{}$, which is not relevant to semantics of non-call atomics.)
In fact the minimal sets are unique in most cases, but we do not need that.\footnote{It is only assignments $x:=F$ for which non-uniqueness is possible, owing to information loss in arithmetic expressions.  For example, with the assignment
$x:=y*z$ and for $\sigma$ with $\sigma(y)=0=\sigma(z)$ then agreement on either $y$ or $z$ is enough to ensure the values written to $x$ agree.  The minimal sets are $\{y\}$ and $\{z\}$.
This also happens with conditional branches, like ``if x or y''.
}

\begin{sloppypar}
Now, consider the antecedent of r-respect for the premise:
$\Lagree(\sigma,\sigma',\pi, \rlocs(\sigma,\eff)\setminus\rlocs(\sigma,\dot{\delta}^\oplus))$.
We must have $X\subseteq \rlocs(\sigma,\eff)\setminus\rlocs(\sigma,\dot{\delta}^\oplus)$,
as otherwise, according to the Claim, r-respect would not hold for the premise.
By side condition, we have $\hat{\sigma}\models \ind{\bnd(\mdl(m))}{\rTow(\eff)}$,
hence $\rlocs(\sigma,\bnd(N))$ is disjoint from $\rlocs(\sigma,\eff)$
by the basic separator property mentioned just before (\ref{eq:sepagree}).
By set theory, from 
$X\subseteq \rlocs(\sigma,\eff)\setminus\rlocs(\sigma,\dot{\delta}^\oplus)$
and $\rlocs(\sigma,\bnd(N)) \intersect \rlocs(\sigma,\eff) = \emptyset$
we get $X\subseteq \rlocs(\sigma,\eff)\setminus\rlocs(\sigma,\delta^\oplus)$.
By monotonicity of $\Lagree$, Eqn.~(\ref{eq:LagreeMono}), the agreement (\ref{eq:CtxIntroA})
implies by $X\subseteq \rlocs(\sigma,\eff)\setminus\rlocs(\sigma,\delta^\oplus)$
the antecedent agreement in the Claim.
Whence by the Claim we get agreement on everything fresh and written, 
which implies the agreement in (\ref{eq:CtxIntro}).
As for the second line of (\ref{eq:CtxIntro}),
what the Claim gives is 
$\rho(\freshLocs(\sigma,\tau)) \subseteq \freshLocs(\sigma',\tau')$.
This implies 
$\rho(\freshLocs(\sigma,\tau)\setminus\rlocs(\tau,\delta)) \subseteq \freshLocs(\sigma',\tau')$.
From  $\agree(\sigma',\tau',\delta)$ we have  $\rlocs(\tau',\delta)=\rlocs(\sigma',\delta$ 
so there are no fresh locations in $\rlocs(\tau',\delta)$.
Hence 
$\freshLocs(\sigma',\tau') = \freshLocs(\sigma',\tau')\setminus\rlocs(\tau',\delta)$
so we have 
$\rho(\freshLocs(\sigma,\tau)\setminus\rlocs(\tau,\delta)) \subseteq \freshLocs(\sigma',\tau')\setminus\rlocs(\tau',\delta)$ and we are done.
\end{sloppypar}

The Claim is a straightforward property of the semantics.
For each of the assignment forms, one defines the evident location set (which underlies the small axioms in the proof system)
and shows that it suffices for the final agreement. 
Then by counterexamples one shows that the location set is minimal.  

\underline{Case: the step is a call.}
We show r-respect in the case that $A$ is a call to some method $p$.  
Note that $p\neq m$, because rules can only be instantiated by wf judgments and $m$ is not in scope in the premise.
The primary step has the form 
$\configm{p()}{\sigma}{\_}\trans{\phi}\configm{A_0}{\tau}{\_}$ where
either $A_0\equiv\skipc$ and $\tau\in\phi(p)(\sigma)$ or
$A_0\equiv p()$, $\tau=\sigma$, and $\phi(p)(\sigma)=\emptyset$.  
It turns out that we do not need to distinguish between these cases.
We need r-respect for 
\[ \delta =  \unioneff{N\in(\Phi,m\scol\flowty{R}{S}{\effe}), \mdl(p)\not\imports N}{\bnd(N)} \]
(as the environment is empty).
The premise gives r-respect for 
$\dot{\delta} = \unioneff{N\in\Phi, \mdl(p)\not\imports N}{\bnd(N)}$.
If $\mdl(m)\in\Phi$ or  $\mdl(p)\imports\mdl(m)$ then $\delta$ is $\dot{\delta}$ and we have r-respect
from the premise.
It remains to consider the case that $\mdl(m)\notin\Phi$ and $\mdl(p)\not\imports\mdl(m)$,
in which case $\delta=\dot{\delta},\bnd(\mdl(m))$.
Let us spell out r-respect for the premise and this step.
The r-respect from the premise says that
\begin{equation}\label{eq:ctxA}
\Lagree(\sigma,\sigma',\pi,\rlocs(\sigma,\eff)\setminus\rlocs(\sigma,\dot{\delta}^\oplus)) 
\mbox{ and }
\agree(\sigma',\tau',\delta)
\end{equation}
implies there is $\rho$ with $\rho\supseteq\pi$ such that 
$\Lagree(\tau,\tau',\rho,
(\freshLocs(\sigma,\tau)\union\written(\sigma,\tau))\setminus\rlocs(\tau,\dot{\delta}^\oplus))$
and 
$\rho(\freshLocs(\sigma,\tau)\setminus\rlocs(\tau,\dot{\delta}))\subseteq
\freshLocs(\sigma',\tau')\setminus\rlocs(\tau',\dot{\delta})$.
(The antecedent is simplified from the definition of r-respect, by omitting the set of fresh locations which is empty in the initial state.)

\begin{sloppypar}
For the conclusion, the condition is the same except with $\delta$ in place of $\dot{\delta}$.  
So suppose 
\[ \Lagree(\sigma,\sigma',\pi,\rlocs(\sigma,\eff)\setminus\rlocs(\sigma,\delta^\oplus)) \]
This implies (\ref{eq:ctxA}) because $\rlocs(\sigma,\eff)$ is disjoint from $\bnd(\mdl(m))$
owing to the condition $\ind{\bnd(\mdl(p))}{\eff}$ in the rule.
So we get some $\rho$ as above,  
and the agreement 
$\Lagree(\tau,\tau',\rho,
(\freshLocs(\sigma,\tau)\union\written(\sigma,\tau))\setminus\rlocs(\tau,\dot{\delta}^\oplus))$
implies the needed agreement for $\delta$, since $\dot{\delta}$ is a subeffect of $\delta$ which is being subtracted.  
Finally, we need to show 
$\rho(\freshLocs(\sigma,\tau)\setminus\rlocs(\tau,\delta))\subseteq
\freshLocs(\sigma',\tau')\setminus\rlocs(\tau,\delta)$.
By w-respect for the $\sigma$-to-$\tau$ step and by assumption $\agree(\sigma',\tau',\delta)$,
there are no fresh locations in $\rlocs(\tau,\delta)$ or $\rlocs(\tau',\delta)$,
so this simplifies to 
$\rho(\freshLocs(\sigma,\tau)\subseteq \freshLocs(\sigma',\tau')$,
which for the same reasons is equivalent to the inclusion 
$\rho(\freshLocs(\sigma,\tau)\setminus\rlocs(\tau,\dot{\delta}))\subseteq
\freshLocs(\sigma',\tau')\setminus\rlocs(\tau',\dot{\delta})$
from the premise.
\end{sloppypar}
\end{proof}

% Note that we do not need to prove the alternate step went via \rn{uCall} and not \rn{uCall0};
% instead we rely on r-respect from the premise

\subsection{Soundness of other context introduction rules}

In \RLII\ the rule ``CtxIntroIn'' has a disjunctive antecedent.
In the present work we need additional side conditions, so we split the rule into multiple rules.

\[\ruleCtxIntroInOne\]

\begin{proof}
Given a model $\phi$ for the conclusion, $\drop{\phi}{m}$ is a model for the hypotheses of the premise. 
Owing to $\mdl(m)\in\Phi$, we have $N\in(\Phi,m:spec)$ iff $N\in\Phi$.
As a result, all the conditions of Encap (a--c) are have identical meaning for 
the conclusion as for the premise.  The same is true for Safety, Post, Write, and R-safe.  
\end{proof}

\[\ruleCtxIntroInTwo\]

\begin{proof}
Note that $A$ is an atomic command.
Given a model $\phi$ for the conclusion, $\drop{\phi}{m}$ is an model for the hypotheses of the premise. 
Validity of the premise implies validity of the conclusion, 
for all conditions except Encap.  
Boundary monotonicity is immediate, because the premise already requires boundary monotonicity for all $N\in\Phi$ and for $N=M$.  
For w-respect,  note that $A$ is not a call and there is only a single step which has no $\Endcall$ in the configuration.  The condition exempts the current module $M$ and is a direct consequence of Encap (a) of the premise, owing to $\mdl(m)=M$.
For r-respect, the current module is not included in the collective boundary for non-call commands, so again the addition of $m$ does not change the requirement.
\end{proof}

\[\ruleCtxIntroCall\]

\begin{proof}
We get Safety, Post, Write, and R-safe from the first premise.
For Encap, we get boundary monotonicity from the first premise, 
except for $N$ in the case that $N=\mdl(m)\neq M$ and $\mdl(m)\notin\Phi$.  
Boundary monotonicity for $N$ is directly checked by the second premise.

We get w-respect, by side condition $\mdl(p)\imports\mdl(m)$, as a consequence of the first premise.  

Finally, r-respect is also a consequence of the first premise, because the collective boundary for the premise is $\unioneff{N\in\Phi,\mdl(p)\not\imports N}{\bnd(N)}$ 
and by side condition $\mdl(p)\imports\mdl(m)$ this is the same set as for the conclusion.
\end{proof}

\subsection{Soundness of \rn{SOF}} 

\newcommand{\hackyu}{^{\ol{t}}_{\ol{u}}} % ALERT - notation for a specific substitution

\[ \ruleSOF \]

Observe that, because boundaries have no spec-only variables (Def.~\ref{def:framedreadsDyn}),
and $\bnd(N)$ frames $I$, the latter does not depend on any spec-only variables.
To prove validity of the conclusion, suppose $\psi^+$ is a $(\Phi,\Theta\conjInv I)$-model.
In order to use the premise, define $\psi^-(m)$ as follows.
For $m$ in $\Phi$, let $\psi^-(m) \eqdef \psi^+(m)$.
For $m$ in $\Theta$ with $\Theta(m) = \flowty{R}{S}{\effe}$ define, for any $\tau$
\[
\psi^-(m)(\tau)\eqdef\left\{
\begin{array}{ll}
\{\Fault\}&\tau\not\models R\\
\emptyset&\tau\models R\land\neg I\\
\psi^+(m)(\tau)& \tau\models R\land I
\end{array}
\right.
\]
The precondition $R$ may have spec-only variables, in which case $\tau\models R\land I$ abbreviates that there are some values for the spec-only variables so that $R\land I$ holds.  
Because $I$ has no spec-only variables, the clauses are exhaustive and mutually disjoint.
It is straightforward to check that $\psi^-$ is a $(\Phi,\Theta)$-model according to 
Definition~\ref{def:ctxinterp}.

For the rest of the proof we consider arbitrary $\sigma$ with $\hat{\sigma}\models P\land I$,
where $\hat{\sigma}\eqdef\extend{\sigma}{\ol{s}}{\ol{v}}$
is the extension of $\sigma$ uniquely determined by $P$ and $\sigma$ according to Lemma~\ref{lem:uniquespeconly}.

To finish the proof, we need the following.

\begin{quote}
\textbf{Claim.} 
If $\configm{C}{\sigma}{\_}\tranStar{\psi^+}\configm{B}{\tau}{\mu}$
then $\tau\models I$ and 
that sequence of configurations is also a trace
$\configm{C}{\sigma}{\_} \tranStar{\psi^-}\configm{B}{\tau}{\mu}$
via $\psi^-$.
\end{quote}

We also need the following observations, to prove the Claim and to prove the rule.
For  any $B,\tau,\mu$,
(a) If $\Active(B)$ is not a call to method in $\Theta$,
then the transitions from
$\configm{B}{\tau}{\mu}$ via $\trans{\psi^+}$,
to $\Fault$ or to a configuration, are the same as those via $\psi^-$.
Because: the model is only used for calls, and the models differ only on methods of $\Theta$.
\\
(b) If $\Active(B)$ is a call to some method $m$ of $\Theta$, and $\tau\models I$, then the transitions 
from $\configm{B}{\tau}{\mu}$ via $\trans{\psi^+}$ are the same as those via $\psi^-$.
Because: For faults, fault via $\trans{\psi^+}$ is when the precondition of the original spec $\Theta(m)$ does not hold,
and that is one conjunct of the precondition for $\psi^-$, the other being $I$.
For non-fault, $\psi^-(m)(\tau)$ is defined to be $\psi^+(m)(\tau)$ when $\tau\models I$.

Before proving the Claim, we use it to prove the conditions for validity of the conclusion of \rn{SOF}.

\emph{Safety.} 
Suppose $\configm{C}{\sigma}{\_}\tranStar{\psi^+}\configm{B}{\tau}{\mu} \trans{\psi^+}\Fault$.
By the Claim, $\configm{C}{\sigma}{\_}\tranStar{\psi^-}\configm{B}{\tau}{\mu}$ and $\tau\models I$.
So by observations (a) and (b), we get a faulting step from $\configm{B}{\tau}{\mu}$ via $\psi^-$,
whence $\configm{C}{\sigma}{\_}\tranStar{\psi^-}\Fault$  which contradicts the premise of \rn{SOF}.

\emph{Post.}
For all $\tau$ such that $\configm{C}{\sigma}{\_} \tranStar{\psi^+} \configm{\skipc}{\tau}{\_}$,
we have $\tau\models I$ and
$\configm{C}{\sigma}{\_} \tranStar{\psi^-}\configm{\skipc}{\tau}{\_}$
by the Claim.
By premise of the rule, we have $\tau\models \subst{Q}{\ol{s}}{\ol{v}}$. 
So we have $\tau\models \subst{(Q\land I)}{\ol{s}}{\ol{v}}$, because $I$ has no spec-only variables.

\emph{Write.} 
Direct consequence of the premise and the Claim.

\emph{R-safe.} 
For $m$ in $\Theta$, the frame condition of $(\Theta\conjInv I)(m)$ is the same as that of $\Theta(m)$, by definition of $\conjInv$.  So this is a direct consequence of the premise and the Claim.

\emph{Encap.} 
Boundary monotonicity is a direct consequence of the Claim, using the premise.
So too the w-respects condition: the condition for the conclusion is the same as for the premise,
because $\Phi,\Theta\conjInv I$ has the same methods, thus the same modules, as $\Phi,\Theta$ has.

For r-respects, consider any reachable step
\( \configm{C}{\sigma}{\_}\tranStar{\psi^+}\configm{B}{\tau}{\mu}
   \trans{\psi^+} \configm{D}{\upsilon}{\nu} \)
and an alternate step
\( %\configm{C}{\sigma'}{\_}\tranStar{\psi^+}
   \configm{B}{\tau'}{\mu} \trans{\psi^+} \configm{D'}{\upsilon'}{\nu'} \)
where %$\sigma'\models P\land I$ (for some values of the spec-only variables $\ol{s}$ in $P$) 
$\agree(\tau',\upsilon',\delta)$
and 
$\tau'$ agrees with $\tau$ according to the r-respect condition for $\delta$,
where the collective boundary $\delta$ is determined by $\Active(B)$, $\Phi,\Theta$, and $M$, 
in the same way for the conclusion as for the premise (i.e., $\delta$ is the same for both).

If the active command of $B$ is not a call to a method in $\Theta$, 
the steps can be taken via $\psi^-$ (see (a) above) and so r-respect from the premise can be applied.
If the active command of $B$ is a call to some method $m \in \Theta$, 
then we have $\tau\models I$ and $\tau'\models I$ by definition of $\psi^+(m)$. 
So the steps can both be taken via $\psi^-$ (see (b) above).
So we can appeal to r-respect from the premise and we are done.

\emph{Proof of Claim.}
By induction on steps.

\textbf{Base case} zero steps: immediate from $\hat{\sigma}\models P\land I$.

\textbf{Induction case:}
\( \configm{C}{\sigma}{\_} \tranStar{\psi^+}\configm{B}{\tau}{\mu}
\trans{\psi^+} \configm{D}{\upsilon}{\nu}
\).
The inductive hypothesis is that 
$\configm{C}{\sigma}{\_} \tranStar{\psi^-}\configm{B}{\tau}{\mu}$, by the same intermediate configurations, and $\tau\models I$.

Case $\Active(B)$ not a call to a method of $\Theta$:  
by observation (a) above, the step to $D$ can be taken via $\psi^-$.
So we can use Encap from the premise.
In particular, we get $\agree(\tau,\upsilon,\bnd(N))$ by w-respect,
owing to side condition $N\in\Theta$ and $M\neq N$ and also the fact that if the step calls $m$ in $\Phi$ then $\mdl(m)\not\imports N$ by side condition.
Moreover we use side condition that $C$ binds no $N$-method,
so that in the definition of w-respect we have that $\topm(B,M)$ is not $N$.
So from $\models \fra{ \bnd(N) }{ I }$ and induction hypothesis $\tau\models I$,
by definition (\ref{eq:frmAgree}) of the frames judgment we get $\upsilon\models I$.

Case $\Active(B)$ is a call to some $m\in\Theta$.
Suppose $\Theta(m) = \flowty{R}{S}{\effe}$.
By induction hypothesis
$\configm{C}{\sigma}{\_} \tranStar{\psi^-} \configm{B}{\tau}{\mu}$
we have $\tau\models R\hackyu$ (with $\ol{u}$ the uniquely determined values of $R$'s spec-only variables $\ol{t}$) because otherwise there would be a fault via $\psi^-$ contrary to the premise.
Because $\tau\models R\hackyu \land I$, we have $\psi^-(m)(\tau) = \psi^+(m)(\tau)$ by definition of $\psi^-(m)$, so the step can be taken via $\psi^-$ and moreover $\upsilon\models I$ because $\psi^+$ is a $\Phi,(\Theta\conjInv I)$-model.

\subsection{Soundness of \rn{Link}}\label{sec:app:link}

\[\ruleLink\]

\begin{remark}\upshape
It is sound to generalize the rule to allow any module $M$ for $C$ and for the linkage,
provided that $\bnd(M)=\emptyeff$. 
\qed\end{remark}

For clarity, the proof is specialized to case that $\Theta$ has a single method named $m$.
We spell out the proof in considerable detail, as there are a number of subtleties. 
However, we assume there are no recursive calls in the bodies of the linked method.
There is no difficulty with recursion, it just complicates the proof:
recursion can be handled using a fixpoint construction for the 
denotational semantics (as in proof of the linking rule in Sect.~A.1 of \RLIII,
and using quasi-determinacy) and an extra induction on calling depth (as in the linking proofs in both \RLII\ and \RLIII).

We use the following from \RLII: 
For method $m$ in the environment, a trace is called \dt{$m$-truncated} provided 
that $\Endcall(m)$ does not occur in the last configuration.
This means that a call to $m$ is not in progress, though it allows that a call may happen next.
In a trace that is not $m$-truncated, an environment call has been made to $m$,
making the transition from a command of the form $m();C$ 
to $B;\Endcall(m);C$ where $B$ is the method body, and then further steps may have been taken.  
Note that in an $m$-truncated trace, it is possible that 
the active command of the last configuration is $m()$.

To prove soundness of the rule, 
suppose $\Theta(m)$ is $\flowty{R}{S}{\effe}$ and let $N \eqdef \mdl(m)$.
Assume validity of the premises for $B$ and $C$:
\begin{equation}\label{eq:linkPremX}
\Phi,\Theta \models_N  B: \flowty{R}{S}{\effe} 
\quad\mbox{and}\quad
\Phi,\Theta \models_{\emptymod} C: \flowty{P}{Q}{\eff}
\end{equation}
To prove validity of the conclusion, i.e.,
\begin{equation}\label{eq:linkConcX} 
\Phi\models_{\emptymod} \letcom{m}{B}{C} : \flowty{P}{Q}{\eff}
\end{equation}
let $\phi$ be any $\Phi$-model.
Define $\theta$ to be the singleton mapping 
$[m\scol \means{B}_\phi ]$,  
using the denotation of $B$, 
so that $\phi\union\theta$ is a $(\Phi,\Theta)$-model, by Lemma~\ref{lem:denotComm}.
(To handle recursive methods, the generalization of Lemma~\ref{lem:denotComm}
is proved by induction as in Lemma~A.10 of \RLIII.)
For brevity we write $\phi,\theta$ for $\phi\union\theta$
and $\trans{\phi\theta}$ for $\trans{\phi\union\theta}$.

For any $\sigma$, the first step is 
$\configm{\letcom{m}{B}{C}}{\sigma}{\_} \trans{\phi} \configm{C;\Endlet(m)}{\sigma}{[m\scol B]}$,
and if the computation reaches a terminal configuration then the last step is the transition for $\Endlet(m)$ which removes $m$ from the environment but does not change the state.  
So to prove (\ref{eq:linkConcX}) we use facts about traces from 
$\configm{C}{\sigma}{[m\scol B]}$.

The following result is used not only to prove (\ref{eq:linkConcX}) but also used to prove soundness of the relational linking rule.  In its statement, we rely on Lemma~\ref{lem:uniquespeconly} about spec-only variables in wf preconditions.  
% \dn{Move this to general results, once it's settled, and ref.}
\begin{lemma}\label{lem:linkClaimX}
\upshape
Suppose we have valid judgments $\Phi,\Theta \models_N B: \Theta(m)$
and $\Phi,\Theta\HVflowtr{}{\emptymod}{P}{ C }{Q}{\eff}$, and also $m\notin B$.
Let $\phi$ be a $\Phi$-model and $\theta\eqdef [m\scol \means{B}_\phi]$.
Let $\sigma$ be any state such that $\sigma\models P$.
Suppose $\configm{C}{\sigma}{[m\scol B]} \tranStar{\phi} \configm{D}{\tau}{\dot{\mu}}$ is $m$-truncated
(for some $D,\tau,\dot{\mu}$).  
Then 
\begin{itemize}
\item $\configm{C}{\sigma}{\_} \tranStar{\phi\theta} \configm{D}{\tau}{\mu}$,
where $\mu = \drop{\dot{\mu}}{m}$.
\item If $D \equiv m();D_0$ for some $D_0$ then $\tau\models R$.
\end{itemize}
(Here the abbreviations $\sigma\models P$ and $\tau\models R$ mean satisfaction by the 
states extended with the uniquely determined values for spec-only variables.)
\end{lemma}
\begin{proof}
We refrain from giving a detailed proof; it requires a somewhat intricate induction hypothesis, similar to the one for impure methods in \RLIII\ (Sect.~A.2, Claim~B)
and the one in \RLII\ (Sect.~7.6).
The main ideas are as follows.  

The combination $\phi,\theta$ is a $(\Phi,\Theta)$-model, by Lemma~\ref{lem:denotComm}.
If $\configm{C}{\sigma}{[m\scol B]} \tranStar{\phi} \configm{D}{\tau}{\dot{\mu}}$ is $m$-truncated
then we can factor it into segments alternating between code of $C$ and code of $B$ during environment calls to $m$.  The steps taken in code of $C$ can be taken via $\trans{\phi\theta}$ because the two transition relations are identical except for calls to $m$.  
A completed call to $m$ amounts to a terminated execution of $B$ (with a continuation command and environment left unchanged).  A completed call gives rise to a single step via $\trans{\phi\theta}$ with the same outcome, because $\theta(m)$ is the denotation of $B$, which is defined directly in terms of executions of $B$.\footnote{A fine point: calls of $m$ may occur in the scope of local variable blocks, so the state may have locals in addition to the variables of the context $\Gamma$ of the judgment; this is handled using the implicit conversion of context models
is discussed in Sect.~\ref{sec:progSem} footnote~\ref{fn:coerce}.
} 
Reasoning by induction on the number of completed calls, we construct a trace via $\trans{\phi\theta}$.
At each call of $m$, we appeal to the premise for $C$ to conclude that the precondition of $m$ holds, as otherwise there would be a faulting trace of $C$ via $\trans{\phi\theta}$.  
\end{proof}

%\dn{Keep in mind this is not the same sequence of configurations; the claim is only about reachability, although if need be we can rely on trace decomposition into steps of $C$ and steps of $B$.}

\textbf{Proof of \rn{Link}}.
Using Lemma~\ref{lem:linkClaimX} we prove (\ref{eq:linkConcX}), validity of 
the conclusion of rule \rn{Link}, as follows,
for any $\sigma$ such that $\hat{\sigma}\models P$
where $\hat{\sigma}$ is $\extend{\sigma}{\ol{s}}{\ol{v}}$ for the unique values $\ol{v}$ determined by $\sigma$.

\medskip
\emph{Post}.
An execution of $\configm{\letcom{m}{B}{C}}{\sigma}{\_}$ via $\phi$ that terminates in state $\tau$
gives an execution for $\configm{C}{\sigma}{[m\scol B]}$ via $\phi$ that ends in $\tau$.
It is $m$-truncated, so by Lemma~\ref{lem:linkClaimX}
we have $\configm{C}{\sigma}{\_}\tranStar{\phi\theta} \configm{\skipc}{\tau}{\_}$.
By validity of the premise for $C$, see (\ref{eq:linkPremX}), we get
$\tau\models\subst{Q}{\ol{s}}{\ol{v}}$.

\medskip
\emph{Write}.
By an argument very similar to the one for Post.

\medskip
\emph{Safety}.
By semantics of $\letcom{m}{B}{C}$ and of $\Endlet(m)$, a faulting execution
has the form 
\[ \configm{\letcom{m}{B}{C}}{\sigma}{\_}\trans{\phi}
 \configm{C;\Endlet(m)}{\sigma}{[m\scol B]} \tranStar{\phi} \configm{D;\Endlet(m)}{\tau}{\dot{\mu}} \trans{\phi} \Fault \]
for some $D,\tau,\dot{\mu}$ with $D\nequiv\skipc$.  
This yields a faulting execution 
\begin{equation}\label{eq:LetSafeX}
\configm{C}{\sigma}{[m\scol B]} \tranStar{\phi} \configm{D}{\tau}{\dot{\mu}} \trans{\phi} \Fault
\end{equation}
We show by two cases that this contradicts the premises
(\ref{eq:linkPremX}) of \rn{Link}.

\textbf{Case} The trace $\configm{C}{\sigma}{[m\scol B]} \tranStar{\phi} \configm{D}{\tau}{\dot{\mu}}$
is $m$-truncated.
Note that $\Active(D)$ is not a call to $m$, because that would be an environment call and would not fault next.  
By Lemma~\ref{lem:linkClaimX},
we get $\configm{C}{\sigma}{\_} \tranStar{\phi\theta} \configm{D}{\tau}{\mu}$
(where $\mu=\drop{\dot{\mu}}{m}$), 
and the transition from $\configm{D}{\tau}{\mu}$ to $\Fault$ can be taken via $\trans{\phi\theta}$ because it is the same relation as $\trans{\phi}$ except for calls to $m$.
But a faulting trace via $\phi,\theta$ contradicts the premise for $C$.

\textbf{Case} The trace  $\configm{C}{\sigma}{[m\scol B]} \tranStar{\phi} \configm{D}{\tau}{\dot{\mu}}$
is not $m$-truncated.
So (\ref{eq:LetSafeX}) can be factored as 
\[ 
\configm{C}{\sigma}{[m\scol B]} \tranStar{\phi} 
\configm{m();D_0}{\tau_0}{\dot{\mu}_0} \trans{\phi}
\configm{B;D_0}{\tau_0}{\dot{\mu}_0} \tranStar{\phi}
\configm{B_0;D_0}{\tau}{\dot{\mu}} \trans{\phi}
\Fault
\]
for some $D_0,B_0,\tau_0,\dot{\mu}_0$ where $D\equiv B_0;D_0$.
Applying Lemma~\ref{lem:linkClaimX} to the $m$-truncated prefix, we get
$\configm{C}{\sigma}{\_} \tranStar{\phi\,\theta} 
\configm{m();D_0}{\tau_0}{\mu_0}$ (where $\mu_0 = \drop{\dot{\mu}_0}{m}$)
and $\tau_0\models\subst{R}{\ol{t}}{\ol{u}'}$ for some $\ol{u}'$.
We also have a faulting execution of $B$ from $\tau_0$, i.e., 
$\configm{B}{\tau_0}{\mu_0} \tranStar{\phi}
\configm{B_0}{\tau}{\mu} \trans{\phi}\Fault$, 
which (because $m$ is not called in $B$) yields the same via $\phi,\theta$,
which contradict the premise for $B$      in (\ref{eq:linkPremX}).

\medskip
\emph{R-safe}.
The first step is not a call, nor is the $\Endlet$ step if reached.
Consider any other reachable configuration:
\(\configm{C}{\sigma}{[m\scol B]}\tranStar{\phi}
\configm{D}{\tau}{\dot{\mu}}
%\trans{\phi} \configm{D_0}{\upsilon}{\dot{\nu}}
\).
If $\Active(D)$ is a call to some $p$ where $\Phi(p)$ is $\flowty{R_p}{S_p}{\effe_p}$,
we must show $\rlocs(\tau,\effe_p) \subseteq \freshLocs(\sigma,\tau)\union\rlocs(\sigma,\eff)$.
%where $\hat{\tau}$ is the extension of $\tau$ uniquely determined for $R_p$.  
Depending on whether $\Active(D)$ is in code of $C$ or $B$, the conclusion follows from the premise of $C$ or $B$, similarly to the proof for Safety.
In the non-$m$-truncated case, i.e., steps of $B$, a called method $p$ is different from $m$ since we 
are assuming no recursion.
The R-safe condition refers to starting state of $B$ (which is $\tau_0$ in the Safety proof above).
The premise yields an inclusion of the $p$'s readable locations in those of $m$ in its starting state $\tau_0$.
Because the R-safe condition holds for the call of $m$ (by induction hypothesis),
its readable locations are included in $\rlocs(\sigma,\eff)$.
Moreover locations that are fresh relative to $\tau_0$ are also fresh relative to $\sigma$.
So the result follows using transitivity of inclusion.
A more detailed argument of this form can be found in the proof of Encap below.

\medskip
\emph{Encap}.
For boundary monotonicity,
we must prove, 
for every $N'$ with $N'=\emptymod$ or $N'\in\Phi$,
that every reachable step, say with states $\tau$ to $\upsilon$,
has $\rlocs(\tau,\bnd(N'))\subseteq\rlocs(\upsilon,\bnd(N'))$.
For steps of $C$ this is immediate from boundary monotonicity from the premise for $C$,
where boundary monotonicity is for all $N'\in(\Phi,\Theta)$ and $N'=\emptymod$.
For steps of $B$ and $N'\in \Phi$ this is immediate from Encap from the premise for $B$,
where boundary monotonicity is for all $N'\in(\Phi,\Theta)$ and $N'=N$.
However, the judgment for $B$ does not imply anything about the boundary of $\emptymod$ (unless $\emptymod$ happens to be in $\Phi,\Theta$).
But by wf we have $\bnd(\emptymod)=\emptyeff$, which makes boundary monotonicity for $\bnd(\emptymod)$ vacuous.

For w-respect and r-respect, we need to consider arbitrary reachable steps. 
The first step of $\letcom{m}{B}{C}$ deterministically steps to $C;\Endlet(m)$, putting $m:B$ into the environment without
changing or reading the state, so both w-respect and r-respect hold for that step.
Both conditions also hold for the step of $\Endlet(m)$ which again does not change or read the state.  
So it remains to consider reachable steps of the following form, in which 
we abbreviate $A\eqdef \Endlet(m)$.
\begin{equation}\label{eq:LetEncapAX}
\configm{\letcom{m}{B}{C}}{\sigma}{\_} \trans{\phi}
 \configm{C;A}{\sigma}{[m\scol B]} \tranStar{\phi} 
 \configm{D;A}{\tau}{\dot{\mu}} \trans{\phi}
 \configm{D_0;A}{\upsilon}{\dot{\nu}}
 \end{equation}
where $D\nequiv\skipc$.
Aside from the first step,
such traces correspond to traces of the form 
\[ %\label{eq:LetEncapBX}
\configm{C}{\sigma}{[m\scol B]}\tranStar{\phi}
\configm{D}{\tau}{\dot{\mu}}\trans{\phi}
\configm{D_0}{\upsilon}{\dot{\nu}}
\]
i.e., exactly the same sequence of configurations, but for lacking the trailing $\Endlet(m)$.

% [TODO spell more details about $\topm(D,M)$, like it's spelled out for r-respect below.]
For w-respect, 
our obligation is to prove that the step $\configm{D}{\tau}{\dot{\mu}}\trans{\phi}\configm{D_0}{\upsilon}{\dot{\nu}}$ w-respects $L$ for every $L\in(\Phi,\dot{\mu})$ and $L\neq\topm(D,\emptymod)$. In the case of an $m$-truncated trace from $C$ to $D$, we appeal to Lemma~\ref{lem:linkClaimX}. 
In the case of a non $m$-truncated trace from $C$ to $D$, the above step is one arising from an environment call to $m$ and therefore occurs in the trace from  $B$. So we use w-respects for $B$. The result follows because the condition for w-respects $L$ for $B$ is $L\in(\Phi,\Theta,\mu)$ and $L\neq\topm(D,N)$ and this is equivalent to the w-respects condition for the step from $D$, because both conditions are equivalent to $L\in (\Phi,\mu)$.
In the case of an $m$-truncated trace from $C$ to $D$, we appeal to Lemma~\ref{lem:linkClaimX}. We can use w-respects for the premise $C$. In the case where $\Active(D)$ is not a context call this condition is $L\in (\Phi,\Theta,\mu)$ and $L\neq \topm(D, \emptymod)$ which is equivalent to $L\in(\Phi,\dot{\mu})$ and $L\neq\topm(D,\emptymod)$.
In the case where $\Active(D)$ is a context call to some $p\in \Phi$, the condition to be proved is $L\in(\Phi,\dot{\mu})$ and $L\neq\topm(D,\emptymod)$ and $\mdl(p) \imports L$. We obtain this from the w-respects condition for the premise which is $L\in (\Phi,\Theta,\mu)$ and $L\neq \topm(D, \emptymod)$ and $mdl(p)\imports L$.

%% For (b), we must show for $N'\in\Phi$, $N'\neq M$, 
%% that the step from $D$ to $D_0$ w-respects $N'$. 
%% The proof is similar to the argument for (Safety),
%% appealing to the premises by using Lemma~\ref{lem:linkClaimX}.
%% In brief, if the step is in the body of $B$ for a call to $m$, we use w-respect from the premise for $B$,  
%% which is for $N'\in(\Phi,\Theta)$ such that $N'\neq N$. This is sufficient because $N$ is not in $\Phi$,
%% owing to the side condition that methods of $\Phi$ do not import modules in $\Theta$, 
%% and $N$ is in $\Theta$ because $N=\mdl(m)$.
%% Otherwise the step is in $C$ and we can appeal to the premise for $C$,
%% which has w-respect for $N'\in(\Phi,\Theta)$ such that $N'\neq M$.  
%% Now $M$ may be in $\Phi$, but owing to $\bnd(m)=\emptyeff$, w-respect for $M$ is vacuous.

For r-respect, we must show the step
\(
\configm{D}{\tau}{\dot{\mu}}\trans{\phi}
\configm{D_0}{\upsilon}{\dot{\nu}}
\)
r-respects $\delta$ for $(\phi,\eff,\sigma)$ where $\delta$ is defined by cases on $\Active(D)$:
%% \begin{itemize}
%%   \item if $\Active(D)$ is not a call, then 
%%       $\delta \eqdef \unioneff{L\in\Phi,L\neq M}{\bnd(L)}$
%%   \item if $\Active(D)$ is a call to some $m$, then 
%%       $\delta \eqdef \unioneff{L\in\Phi,\mdl(m)\not\imports L}{\bnd(L)}$
%% \end{itemize}
\begin{itemize}
  \item if $\Active(D)$ is not a call, then 
      $\delta \eqdef \unioneff{L\in(\Phi,\dot{\mu}),L\neq \topm(D,\emptymod)}{\bnd(L)}$ 
  \item if $\Active(D)$ is a call to some $m$, then 
      $\delta \eqdef \unioneff{L\in(\Phi,\dot{\mu}),\mdl(m)\not\imports L}{\bnd(L)}$
\end{itemize}
Let us spell out the r-respect conditions for the given trace (\ref{eq:LetEncapAX}).
\begin{list}{}{}
\item[(*)]
For any $\pi,\tau',\upsilon'$, if $\agree(\tau',\upsilon',\delta)$
and 
\(
%\label{eq:LetEncapCX}
\configm{D}{\tau'}{\dot{\mu}}\trans{\phi}
\configm{D'_0}{\upsilon'}{\dot{\nu}}
\)      
and 
$\Lagree(\tau,\tau',\pi, \freshLocs(\sigma,\tau)\union \rlocs(\sigma,\eff)\setminus\rlocs(\tau,\delta^\oplus))$,
then $D'_0\equiv D_0$ and there is $\rho\supseteq\pi$ such that
\[ 
\begin{array}{l}
\Lagree(\upsilon,\upsilon',\rho, \freshLocs(\tau,\upsilon)\union\written(\tau,\upsilon)\setminus\rlocs(\upsilon,\delta^\oplus)) \\
\rho(\freshLocs(\tau,\upsilon)\setminus\rlocs(\upsilon,\delta)) \subseteq \freshLocs(\tau',\upsilon')\setminus\rlocs(\upsilon',\delta) 
\end{array}
\qquad\qquad (\dagger)
\]
\end{list}
%%Note that the r-respect condition  starting from
%% $\configm{\letcom{m}{B}{C}}{\sigma'}{\_}$, but we have simplified this to the form (\ref{eq:LetEncapCX})
%% in parallel with the correspondence between (\ref{eq:LetEncapAX}) and (\ref{eq:LetEncapBX}).
%% This is justified because, as we noted at the outset, r-respects is easily shown for the initial step and also for the step of $\Endlet(m)$, so it remains to consider the remaining steps.
%% For those steps, (*) yields r-respects for (\ref{eq:LetEncapAX}) owing to the correspondence between the traces with and without $\Endlet(m)$.

To prove (*) we go by cases on whether the trace up to $D,\tau$ is $m$-truncated. 

%% Keep in mind that the Encap(b) condition of the premise for $C$ says that for any reachable step in a trace of $C$ from
%% a starting state $\sigma$, the step respects $N'$ for 
%% $((\Phi,\Theta),P,\eff,\phi,C,\sigma)$.
%% The Encap condition of the premise for $B$ says that for any reachable step in a trace of $B$ from 
%% a starting state, say $\tau_1$, the step respects $N'$ for 
%% $((\Phi,\Theta),R,\effe,\phi,B,\tau_1)$.
%% What we are proving is respect of $N'$ for $(\Phi,P,\eff,\phi,\letcom{m}{B}{C},\sigma)$.
\medskip
Suppose the antecedent of (*) holds: that is,
\[
\begin{array}{l}
\agree(\tau',\upsilon',\delta)
\mbox{ and } 
\configm{D}{\tau'}{\dot{\mu}}\trans{\phi}\configm{D'_0}{\upsilon'}{\dot{\nu}}
\mbox{ and } \\
\Lagree(\tau,\tau',\pi, (\freshLocs(\sigma,\tau)\union
 \rlocs(\sigma,\eff))\setminus\rlocs(\tau,\delta^\oplus))
\end{array}
\]

\medskip
\textbf{Case} $\configm{C}{\sigma}{[m\scol B]}\tranStar{\phi} \configm{D}{\tau}{\dot{\mu}}$ is $m$-truncated.

\smallskip
Then by Lemma~\ref{lem:linkClaimX} we have 
$\configm{C}{\sigma}{\_}\tranStar{\phi\theta} \configm{D}{\tau}{\mu}$ where $\mu=\drop{\dot{\mu}}{m}$.

If $\Active(D)$ is not a context call, the r-respect condition to be proved is for

\[
\begin{array}{lcl}
\delta&=&\unioneff{L\in(\Phi,\dot{\mu}),L\neq \topm(D,\emptymod)}{\bnd(L)}\\
      &=&\unioneff{L\in(\Phi,\mu),L\neq \topm(D,\emptymod)}{\bnd(L)},\bnd(N)
\end{array}
\]

We have the additional step $\configm{D}{\tau}{\mu}\trans{\phi\theta}\configm{D'_0}{\upsilon}{\nu}$ because in this case $\phi$ and $\phi\theta$ agree. For the same reason the
step $\configm{D}{\tau'}{\dot{\mu}}$ to $\configm{D'_0}{\upsilon'}{\dot{\nu}}$
can also be taken via $\phi\theta$, so 
$\configm{D}{\tau'}{\mu}\trans{\phi\theta}\configm{D'_0}{\upsilon'}{\nu}$, where
$\nu=\drop{\dot{\nu}}{m}$.
The Encap condition for the premise for $C$ says that
\[\configm{C}{\sigma}{\_}\tranStar{\phi\theta}
\configm{D}{\tau}{\mu} \trans{\phi\theta}
\configm{D'_0}{\upsilon}{\nu}
\]
respects $((\Phi,\Theta),\emptymod,(\phi\theta),\eff,\sigma)$.

Unpacking definitions, from r-respect we have that
the step $\configm{D}{\tau}{\mu}\trans{\phi\theta}\configm{D'_0}{\upsilon}{\nu}$ r-respects $\dot{\delta}$ for $(\phi\theta,\eff,\sigma)$ where 
\(
\begin{array}[t]{lcl}
\dot{\delta} &=& \unioneff{L\in(\Phi,\Theta,\mu), L\neq \topm(D,\emptymod)}{\bnd(L)}\\
&=& \unioneff{L\in(\Phi,\mu), L\neq \topm(D,\emptymod)}{\bnd(L)},\bnd(N)\\
&=& \delta
\end{array}
\)
\\
Now to establish $(\dagger)$ we show $\agree(\tau',\upsilon',\dot{\delta})$ and
$\Lagree(\tau,\tau',\pi, \freshLocs(\sigma,\tau)\union \rlocs(\sigma,\eff)\setminus\rlocs(\tau,\dot{\delta}^\oplus))$. Because $\dot{\delta} = \delta$, both hold by assumption. 

If $\Active(D)$ is a context call to $p\in\Phi$, the r-respect condition to be proved is for
\[
\begin{array}{lcl}
\delta &=&\unioneff{L\in(\Phi,\dot{\mu}),\mdl(p)\not\imports L}{\bnd(L)}\\
&=&\unioneff{L\in(\Phi,\mu),\mdl(p)\not\imports L}{\bnd(L)},\bnd(N)
\end{array}
\]
where the last equality follows because $\mdl(m)= N$ and $\mdl(p)\not\imports N$\ by side condition of \rn{Link}, and $\bnd(\emptymod)$ is empty.
For the premise for $C$, note that there is a step
$\configm{D}{\tau}{\mu} \trans{\phi\theta}\configm{D'_0}{\upsilon}{\nu}$ because $\phi$ and $\phi\theta$ agree on $p$. For the same reason the
step $\configm{D}{\tau'}{\dot{\mu}}$ to $\configm{D'_0}{\upsilon'}{\dot{\nu}}$
can also be taken via $\phi\theta$, so 
$\configm{D}{\tau'}{\mu}\trans{\phi\theta}\configm{D'_0}{\upsilon'}{\nu}$, where
$\nu=\drop{\dot{\nu}}{m}$. The r-respect condition for the premise is for collective boundary $\dot{\delta}$
where
\(
\begin{array}[t]{lcl}
\dot{\delta}&=&\unioneff{L\in(\Phi,\Theta,\mu), \mdl(p)\not\imports L}{\bnd(L)}\\
&=&
\unioneff{L\in(\Phi,\mu), \mdl(p)\not\imports L}{\bnd(L)},\bnd(N)\\
&=& \delta
\end{array}
\) 
\\
where the second equality follows because $\mdl(p)\not\imports N$ by the side condition of the \rn{Link} rule. From these we get an argument similar to above because $\Active(\tau',\upsilon',\delta)$ and $\Lagree(\tau,\tau',\pi, \freshLocs(\sigma,\tau)\union \rlocs(\sigma,\eff)\setminus\rlocs(\tau,\delta^\oplus))$ hold by assumption. 

%%Finally, if $\Active(D)$ is a call to $m$, then %% the r-respect condition to be proved is for
%% \[
%% \begin{array}{lcl}
%% \delta &=&\unioneff{L\in(\Phi,\dot{\mu}),\mdl(m)\not\imports L}{\bnd(L)}\\
%% &=&\unioneff{L\in(\Phi,\mu)}{\bnd(L)}
%% \end{array}
%% \]
%% the step 
%% $\configm{D}{\tau}{\dot{\mu}}\trans{\phi} \configm{D_0}{\upsilon}{\dot{\nu}}$ is an environment call, so $\upsilon=\tau, \upsilon'=\tau'$ and the fresh and written location sets of this step are empty. So we can take $\rho:=\pi$ and the requisite agreements $(\dagger)$ hold vacuously. \ab{We need not specify $\delta$ here.}

This completes the proof of (*) for $m$-truncated traces.

\medskip
\textbf{Case} $\configm{C}{\sigma}{[m\scol B]}\tranStar{\phi} \configm{D}{\tau}{\dot{\mu}}$ is not $m$-truncated.  
As in the proof of Safety, we factor out the $m$-truncated prefix for the last call to $m$.  That is, there are $B_0,D_1,\tau_1,\dot{\mu}_1$ such that 
\[
\begin{array}{lll}
\configm{C}{\sigma}{[m\scol B]}\tranStar{\phi} \configm{m();D_1}{\tau_1}{\dot{\mu}_1}
   & \trans{\phi} \configm{B;\Endcall(m);D_1}{\tau_1}{\dot{\mu}_1} & \mbox{since $\dot{\mu}_1(m)=B$} \\
   & \tranStar{\phi} \configm{B_0;\Endcall(m);D_1}{\tau}{\dot{\mu}} & \mbox{with $D\equiv B_0;\Endcall(m);D_1$} \\
   & \trans{\phi} \configm{B_1;\Endcall(m);D_1}{\upsilon}{\dot{\nu}} & \mbox{with $D_0\equiv B_1;\Endcall(m);D_1$} 
   \end{array}
\]
So for just $B$ we have
\[ \configm{B}{\tau_1}{\dot{\mu}_1} \tranStar{\phi}
   \configm{B_0}{\tau}{\dot{\mu}} \trans{\phi} 
   \configm{B_1}{\upsilon}{\dot{\nu}} 
\]
and as in the proof of Safety we have $\hat{\tau_1}\models R$ by Lemma~\ref{lem:linkClaimX}. Note that $\Active(D) = \Active(B_0)$. Moreover, $m$ does not occur in $B, B_0, B_1$ because there is no recursion. Hence $\phi$ and $\phi\theta$ agree so that
\[ \configm{B}{\tau_1}{\mu_1} \tranStar{\phi\theta}
   \configm{B_0}{\tau}{\mu} \trans{\phi\theta} 
   \configm{B_1}{\upsilon}{\nu} 
\]
By assumption, $\configm{D}{\tau'}{\dot{\mu}}\trans{\phi}\configm{D'_0}{\upsilon'}{\dot{\nu}}$. That is,
\[
\configm{B_0;\Endcall(m);D_1}{\tau'}{\dot{\mu}}\trans{\phi}\configm{B'_1;\Endcall(m);D'_1}{\upsilon'}{\dot{\nu}}
\]
where $D'_0 \eqdef B'_1;\Endcall(m);D'_1$. There are no calls to $m$ so
\[
\configm{B_0}{\tau'}{\mu}\trans{\phi\theta}\configm{B'_1}{\upsilon'}{\nu}
\]
Because $\tau$ is reached from $\sigma$ via $\tau_1$, we have $\freshLocs(\sigma,\tau) = \freshLocs(\sigma,\tau_1)\union\freshLocs(\tau_1,\tau)$, whence
$\freshLocs(\tau_1,\tau)\subseteq\freshLocs(\sigma,\tau)$. Moreover, by the validity of premise for $C$ we can use its R-safe condition for the call to $m$ to obtain $\rlocs(\tau_1,\effe)\subseteq\rlocs(\sigma,\eff)$.

\smallskip
If $\Active(D)$ is a context call to some $p\in\Phi$, the r-respect condition to be proved is for collective boundary
\(
\begin{array}[t]{lcl}
\delta&=&\unioneff{L\in(\Phi,\dot{\mu}), \mdl(p)\not\imports L}{\bnd(L)}\\
&=&
\unioneff{L\in(\Phi,\mu), \mdl(p)\not\imports L}{\bnd(L)},\bnd(N)\\
\end{array}
\)
\\
(in which we omit $L=\emptymod$ because $\bnd(\emptymod)$ is empty).
For the premise for $B$, the r-respect condition is for collective boundary
$\dot{\delta}$ where
\(
\begin{array}[t]{lcl}
\dot{\delta} &=&
\unioneff{L\in(\Phi,\Theta,\mu), \mdl(p)\not\imports L}{\bnd(L)}\\
&=&
\unioneff{L\in(\Phi,\mu), \mdl(p)\not\imports L}{\bnd(L)},\bnd(N)\\
&=& \delta
\end{array}
\)
\\
where the second equality holds by side condition $\mdl(p)\not\imports N$ of the \rn{Link} rule.

Using the antecedent of (*) and noting $\dot{\delta}=\delta$ we get
\[\Lagree(\tau, \tau', \pi, (\freshLocs(\tau_1, \tau) \union \rlocs(\tau_1, \effe)\setminus \rlocs(\tau,\delta^\oplus)))
\]
Now by the r-respect condition for the premise for $B$ (and because $\agree(\tau',\upsilon',\delta)$ holds by assumption) we obtain $\rho\supseteq\pi$ such that
\[
\begin{array}{l}
\Lagree(\upsilon, \upsilon', \rho, (\freshLocs(\tau, \upsilon) \union \written(\tau, \upsilon))\setminus\rlocs(\upsilon,\delta^\oplus) \mbox{ and} \\
\rho(\freshLocs(\tau, \upsilon)\setminus\rlocs(\upsilon,\delta)) \subseteq \freshLocs(\tau', \upsilon')\setminus\rlocs(\upsilon',\delta)
\end{array}
\]
Furthermore, $B'_1 \equiv B_1$, whence $D'_1 \equiv D_1$ because $B_1$ in the source code has a unique continuation. Thus $D'_0 \equiv D_0$. 
Thus ($\dagger$) is established.

\medskip
If $\Active(D)$ is not a context call, note that $\topm(D, \emptymod) = \topm(B_0;\Endcall(m);D_1, \emptymod)$. Hence the r-respect condition to be proved is for collective boundary
\[
\begin{array}{lcl}
\delta &=& \unioneff{L\in(\Phi,\dot{\mu}), L\neq \topm(D,\emptymod)}{\bnd(L)}
\end{array}
\]
If $B_0$ doesn't contain an $\Endcall$, then $\topm(D,\emptymod) = N$. Then
\[
\begin{array}{lcl}
\delta
&=&\unioneff{L\in(\Phi,\dot{\mu}), L\neq N}{\bnd(L)}\\
&=& \unioneff{L\in(\Phi,\mu)}{\bnd(L)}
\end{array}
\]
where the second equality follows because $mdl(m)=N$ and $m\in\dom{\dot{\mu}}$.
\\
If $B_0$ contains an outermost $\Endcall(p)$, then $p\neq m$ and $\topm(D,\emptymod) = mdl(p)$. Then
%\dn{(old note) something looks funny but I didn't check carefully yet:}
\[
\begin{array}{lcl}
\delta
&=&\unioneff{L\in(\Phi,\dot{\mu}), L\neq mdl(p)}{\bnd(L)}\\
&=& \unioneff{L\in(\Phi,\mu), L\neq \emptymod}{\bnd(L)},\bnd(mdl(p)),\bnd(N)\\
&=& \unioneff{L\in(\Phi,\mu)}{\bnd(L)},\bnd(mdl(p)),\bnd(N)
\end{array}
\]

The premise for $B$ gives r-respect for the collective boundary
\[
\begin{array}{lcl}
\dot{\delta} &=& \unioneff{L\in(\Phi,\Theta,\mu), L\neq \topm(B_0,N)}{\bnd(L)}
\end{array}
\]
If $B_0$ has no $\Endcall$s, then $\topm(B_0,N)=N$. In this case
\[
\begin{array}{lcl}
\dot{\delta} &=& \unioneff{L\in(\Phi,\Theta,\mu), L\neq N}{\bnd(L)}\\
&=& \unioneff{L\in(\Phi,\mu)}{\bnd(L)}
\end{array}
\]
If $B_0$ contains an outermost $\Endcall(p)$ as above, then $p\neq m$ and $\topm(B_0,N) = mdl(p)$. Then
\[
\begin{array}{lcl}
\dot{\delta}
&=&
\unioneff{L\in(\Phi,\Theta,\mu), L\neq mdl(p)}{\bnd(L)}\\
&=& \unioneff{L\in(\Phi,\mu)}{\bnd(L)},\bnd(mdl(p)),\bnd(N)
\end{array}
\]
In either case $\dot{\delta}=\delta$.
To obtain ($\dagger$) we must show $\agree(\tau',\upsilon',\dot{\delta})$ and
\[
\begin{array}{l}
\Lagree(\tau, \tau', \pi, (\freshLocs(\tau_1, \tau) \union \rlocs(\tau_1, \effe))\setminus\rlocs(\tau,\dot{\delta}^\oplus)
\end{array}
\]
Since $\dot{\delta}$ = $\delta$, both of these hold by assumption.

\section{Appendix: Biprogram semantics and relational correctness (re Sect.~\ref{sec:biprog})}

\subsection{On relation formulas}

\begin{figure}
\begin{small}
\(
% HACK alert
\begin{array}{l@{\hspace*{1ex}}l@{\hspace*{1ex}}l}
\sigma|\sigma'\models_\pi\rightF{P} 
  & \mbox{iff} & \sigma'\models P \\

\sigma|\sigma'\models_\pi  \P\land\Q 
  & \mbox{iff} & \sigma|\sigma'\models_\pi  \P \mbox{ and } \sigma|\sigma'\models_\pi \Q \\

\sigma|\sigma'\models_\pi  \P\lorbi\Q 
  & \mbox{iff} & \sigma|\sigma'\models_\pi  \P \mbox{ or } \sigma|\sigma'\models_\pi \Q \\

\sigma|\sigma'\models_\pi  \all{x\scol K|x'\scol K'}{\P} 
  & \mbox{iff} & 
\extend{\sigma}{x}{v}|\extend{\sigma'}{x'}{v'}\models_\pi  \P 
 \quad\mbox{for all $v\in\means{K}\sigma \setminus \{\semNull\}$ and             
                    $v'\in\means{K'}\sigma' \setminus \{\semNull\}$}
 \\

\sigma|\sigma'\models_\pi  \all{x\scol \Region|x'\scol \Region}{\P} 
  & \mbox{iff} & 
\extend{\sigma}{x}{v}|\extend{\sigma'}{x'}{v'}\models_\pi  \P 
\quad\mbox{for all $v\in\means{\Region}\sigma $ and             
                   $v'\in\means{\Region}\sigma'$}
\\

\sigma|\sigma'\models_\pi  \all{x\scol \INT|x'\scol \INT}{\P} 
  & \mbox{iff} & 
 \extend{\sigma}{x}{v}|\extend{\sigma'}{x'}{v'}\models_\pi  \P 
 \quad\mbox{for all $v\in\mathbb{Z}$ and $v'\in\mathbb{Z}$}
 \\

\sigma|\sigma'\models_\pi  R(F\!F) 
  & \mbox{iff} & \means{F\!F}(\sigma|\sigma') \in \means{R} 
\mbox{ (and similarly for list $\ol{F\!F}$)}
\\[1ex]
%\sigma|\sigma'\models \P &\mbox{iff}& \sigma|\sigma'\models_\pi \P \mbox{ for all $\pi$}
%\\ % used in semantics of whilebi
%\models\P & \mbox{iff} & \sigma|\sigma'\models \P \mbox{ for all $\sigma,\sigma'$}
\end{array}\)
\end{small}
\vspace{-1ex}
\caption{Relation formula semantics 
cases omitted from Fig.~\ref{fig:relFmlaSem}.
See Fig.~\ref{fig:relFormulas} for syntax.
}
\label{fig:relFmlaSemA}
\end{figure}

Semantics of relation formulas is given in Figs.~\ref{fig:relFmlaSem} and~\ref{fig:relFmlaSemA}.
Omitted in the figures are the left and right typing contexts for the formula.
Semantics for quantifiers is written in a way to make clear there is no built-in connection between the left and right values.  
In particular, we allow one side to bind a reference type while the other binds a variable of integer type. 
This is useful when a variable is only needed on one side
(whereas using a dummy of reference type would make the formula vacuously true in states with no allocated references on that side).
For practical purposes we find little use for quantification at type $\Region$ and on the other hand it is convenient to exclude null at reference type.

%The definition applies to states for those contexts, allowing additional variables as well---which we need to reason about intermediate steps where locals/parameters are present.}

The form $R(\ol{F\!F})$, where $\ol{F\!F}$ is a list of 2-expressions, is restricted for simplicity to heap-independent expressions of mathematical type (including integers but excluding references and  regions).  So the semantics can be defined in terms of given denotations $\means{R}$ that provide a 
fixed interpretation for atomic predicates $R$ in the signature,
as assumed already for semantics of unary formulas. 
The semantics of left and right expressions is
written using $\means{-}$ and defined as follows:
$\means{\leftex{F}}(\sigma|\sigma') = \sigma(F)$  
and $\means{\rightex{F}}(\sigma|\sigma') = \sigma'(F)$.

\begin{lemma}[unique snapshots]
\label{lem:uniquespeconlyR}
\upshape
If $\P$ is the precondition in a wf relational spec 
with spec-only variables $\ol{s}$ on the left and $\ol{s}'$ on the right,
then for all $\sigma,\sigma',\pi$ there is at most one valuation $\ol{v},\ol{v}'$ such that
$\sigma|\sigma'\models_\pi\subst{\P}{\ol{s},\ol{s}'}{\ol{v},\ol{v}'}$.
Moreover, they are independent from $\pi$, i.e., determined by $\sigma,\sigma'$ and $\Left{\P}\land\Right{\P}$.
\end{lemma}
The proof is straightforward.

\begin{restatable}[framing of region agreement]{lemma}{lemframeEqbi}
\label{lem:frameEqbi}
\upshape
$G\eqbi G\models\fra{\effe|\effe}{ \Agr G\Img f}$
where $\effe$ is $\ftpt(G),\rd{G\Img f}$.
\end{restatable}
%\lemframeEqbi*
\begin{proof}
Suppose $\sigma|\sigma'\models_\pi G\eqbi G \land \Agr G\Img f$
and $\agree(\sigma, \tau, \effe)$ and $\agree(\sigma', \tau', \effe)$.
By semantics, $\sigma|\sigma'\models_\pi \Agr G\Img f$ 
iff $\agree(\sigma,\sigma',\pi,\rd{G\Img f})$ and 
$\agree(\sigma',\sigma,\pi^{-1},\rd{G\Img f})$,
i.e., \[\Lagree(\sigma,\sigma',\pi,\rlocs(\sigma,\rd{G\Img f})) \mbox{ and }
\Lagree(\sigma',\sigma,\pi,\rlocs(\sigma',\rd{G\Img f}))\]
We must show $\Lagree(\tau,\tau',\pi,\rlocs(\tau,\rd{G\Img f}))$ and
$\Lagree(\tau',\tau,\pi^{-1},\rlocs(\tau',\rd{G\Img f}))$.

From $\agree(\sigma, \tau, \effe)$ 
we get $\sigma(G)=\tau(G)$, 
and from $\agree(\sigma', \tau', \effe)$ 
we get $\sigma'(G)=\tau'(G)$.
From $\sigma(G)=\tau(G)$ we get that $\rlocs(\sigma,\rd{G\Img f}) = \rlocs(\tau,\rd{G\Img f})$ and 
from $\sigma'(G)=\tau'(G)$ we get that $\rlocs(\sigma',\rd{G\Img f}) = \rlocs(\tau',\rd{G\Img f})$.
So it suffices to show 
\[\Lagree(\tau,\tau',\pi,\rlocs(\sigma,\rd{G\Img f})) \mbox{ and }
\Lagree(\tau',\tau,\pi^{-1},\rlocs(\sigma',\rd{G\Img f}))\]

First the left conjunct:
For any $o.f\in \rlocs(\sigma,\rd{G\Img f})$, we have from above that
$\tau(o.f) = \rprel{\sigma(o.f)}{\sigma'(\pi(o).f)}$
so it remains to show $\sigma'(\pi(o).f) = \tau'(\pi(o).f)$.
From $\sigma|\sigma'\models_\pi G\eqbi G$ we have
$\rprel{\sigma(G)}{\sigma'(G)}$, i.e., $\pi(\sigma(G)) = \sigma'(G)$.
So $\pi(o)\in\sigma'(G)$ and we get 
$\sigma'(\pi(o).f) = \tau'(\pi(o).f)$ from $\agree(\sigma', \tau', \rd{G\Img f})$.

Now the right conjunct:
For any $o.f\in \rlocs(\sigma',\rd{G\Img f})$,
$\rprel{\sigma(\pi^{-1}(o).f)}{\sigma'(o.f)} = \tau'(o.f)$
so it remains to show $\tau(\pi^{-1}(o).f) = \sigma(\pi^{-1}(o).f)$.
From $\sigma|\sigma'\models_\pi G\eqbi G$ we have
$\rprel{\sigma(G)}{\sigma'(G)}$, i.e., $\pi(\sigma(G)) = \sigma'(G)$.
So $\pi^{-1}(o)\in\sigma(G)$ and we get 
$\sigma(\pi^{-1}(o).f) = \tau(\pi^{-1}(o).f)$ from $\agree(\sigma, \tau, \rd{G\Img f})$.
\end{proof}

\begin{lemma}
\upshape
If $\RprelT{\pi,\pi'}{(\sigma|\sigma')}{(\tau|\tau')}$
then 
$\sigma|\sigma'\models_\rho \P$
implies  
$\tau|\tau'\models_{\pi^{-1};\rho;\pi'} \P$.
\end{lemma}
Here $\pi^{-1};\rho;\pi'$ denotes composition of refperms in diagrammatic order, so
$(\pi^{-1};\rho;\pi')(o)$ is 
$\pi'(\rho(\pi^{-1}(o)))$ if it is defined on $o$.
\begin{proof}
Proof by induction on $\P$. 
We consider two cases; the other cases are similar or simpler.

Consider the case of $F\eqbi F'$, where $F,F'$ are expressions of some class type $K$.  (The argument for type $\Region$ is similar and for base types $\INT$ and $\BOOL$ straightforward.)
Now suppose $\sigma|\sigma'\models_\rho F\eqbi F'$, i.e.,
$\Rprel{\rho}{\sigma(F)}{\sigma'(F')}$.
For the non-null case, this is equivalent to 
$\rho(\sigma(F)) = \sigma'(F')$.  (We leave the null case to the reader.)
We must show 
$\Rprel{\pi^{-1};\rho;\pi'}{\tau(F)}{\tau'(F)}$, 
i.e., 
$\pi'(\rho(\pi^{-1}(\tau(F)))) = \tau'(F')$.
From $\RprelT{\pi,\pi'}{(\sigma|\sigma')}{(\tau|\tau')}$
we have $\RprelT{\pi}{\sigma}{\tau}$ and $\RprelT{\pi'}{\sigma'}{\tau'}$ by definition.
By Lemma~\ref{lem:insensi}
we get $\Rprel{\pi}{\sigma(F)}{\tau(F)}$
and $\Rprel{\pi'}{\sigma'(F')}{\tau'(F')}$,
which for non-null values means 
$\pi(\sigma(F)) = \tau(F)$
and $\pi'(\sigma'(F')) = \tau'(F')$.
We conclude by using the equations to calculate
$\pi'(\rho(\pi^{-1}(\tau(F))))
= \pi'(\rho(\pi^{-1}(\pi(\sigma(F)))))
= \pi'(\rho(\sigma(F)))
= \pi'(\sigma'(F))
= \tau'(F')$.

Consider the case of $\Agr G\Img f$ where $f$ is a reference type field.
Suppose $\sigma|\sigma'\models_\rho \Agr G\Img f$.
By semantics and the definitions of $\agree$, $\rlocs$, and $\Lagree$, this is equivalent to
\begin{equation}\label{eq:hyp}
\all{o\in \sigma(G)}{ \Rprel{\rho}{\sigma(o.f)}{\sigma'(\rho(o).f)} } 
\end{equation}
In the rest of the proof we consider the non-null case, so the body can
be rephrased as $\rho(\sigma(o.f)) = \sigma'(\rho(o).f)$.
We must show 
\[ \all{p\in \tau(G)}{ \Rprel{\pi^{-1};\rho;\pi'}{
  \tau(p.f)}{\tau'(\pi'(\rho(\pi^{-1}(p))).f)} } \] 
i.e., $\pi'(\rho(\pi^{-1}(\tau(p.f)))) = \tau'(\pi'(\rho(\pi^{-1}(p))).f)$.
By $\RprelT{\pi}{\sigma}{\tau}$,
we have $p\in\tau(G)$ iff $\pi^{-1}(p)\in \sigma(G)$ so
we reformulate our obligation in terms of $\pi(o)$:
\begin{equation}\label{eq:hypA}
 \all{o\in \sigma(G)}{ 
    \pi'(\rho(\pi^{-1}(\tau(\pi(o).f)))) = \tau'(\pi'(\rho(\pi^{-1}(\pi(o)))).f) }
\end{equation}
By the isomorphisms $\Rprel{\pi}{\sigma(F)}{\tau(F)}$
and $\Rprel{\pi'}{\sigma'(F')}{\tau'(F')}$,
we have $\pi(\sigma(o.f)) = \tau(\pi(o).f)$
 and $\pi'(\sigma'(p.f)) = \tau'(\pi'(p).f)$ for any $o,p$.
We prove (\ref{eq:hypA}) by calculating for any $o\in\sigma(G)$: 
\[\begin{array}{ll}
\pi'(\rho(\pi^{-1}(\tau(\pi(o).f))))  \\
= \pi'(\rho(\pi^{-1}(\pi(\sigma(o.f))))) 
 & \mbox{by $\pi(\sigma(o.f)) = \tau(\pi(o).f)$} \\
= \pi'(\rho(\sigma(o.f))) 
 & \mbox{by $\pi$ bijective} \\
= \pi'(\sigma'(\rho(o).f))
 & \mbox{by $\rho(\sigma(o.f)) = \sigma'(\rho(o).f)$ from (\ref{eq:hyp})} \\
= \tau'(\pi'(\rho(o)).f)
 & \mbox{by $\pi'(\sigma'(p.f)) = \tau'(\pi'(p).f)$} \\ 
= \tau'(\pi'(\rho(\pi^{-1}(\pi(o)))).f)
 & \mbox{by $\pi$ bijective} 
\end{array}
\]
\end{proof}

\lemrefpermsupp*
\begin{proof}
(i)
To show $R$ is refperm monotonic we must show 
for all $\pi,\rho,\sigma,\sigma'$, if
$\sigma|\sigma'\models_\pi \R$ and $\rho\supseteq\pi$ then $\sigma|\sigma'\models_\rho \R$.
This is immediate in case $\R$ is refperm independent.

There are two general forms for agreement formulas.
For the form $F\eqbi F'$, we only need to consider $F$ (and thus $F'$) of reference or region type, as otherwise it is refperm independent.
For both reference type and region type we have
$\sigma|\sigma'\models_\pi F\eqbi F'$
iff $\rprel{\sigma(F)}{\sigma'(F')}$ (by semantics,
see Fig.~\ref{fig:relFmlaSem}).
The latter holds only if $\sigma(F)$ is in the domain of $\pi$ (for $F:K$)
or a subset of the domain (for $F:\Region$), and \emph{mut.\ mut.} for
$\sigma'(F')$ and the range of $\pi$.  
So $\sigma|\sigma'\models_\pi F\eqbi F'$ implies 
$\sigma|\sigma'\models_\rho F\eqbi F'$ for any $\rho\supseteq\pi$.

The other form of agreement formula is $\Agr LE$ where $LE$ may be a variable $x$ ---in which case the meaning is the same as $x\eqbi x$ and the above argument applies--- or $LE$ has the form $G\Img f$.  Suppose $\sigma|\sigma'\models_\pi G\Img f$.  
Unfolding the semantics, we have
$\agree(\sigma,\sigma',\pi,\rd{G\Img f})$ and 
$\agree(\sigma',\sigma,\pi^{-1},\rd{G\Img  f})$.
That is,
$\Lagree(\sigma,\sigma',\pi,\rlocs(\sigma,\rd{G\Img f}))$ and 
$\Lagree(\sigma',\sigma,\pi^{-1},\rlocs(\sigma',\rd{G\Img  f})$.
This does not entail $\rprel{\sigma(G)}{\sigma'(G)}$ 
(see % Footnote~\ref{fn:neqbi}, 
Section~\ref{sec:relform}).
But it does entail that $\sigma(G) \subseteq \dom(\pi)$
and $\sigma'(G)\subseteq\rng(\pi)$ (as already remarked in Section~\ref{sec:relform}). 
So extending $\pi$ to some $\rho\supseteq\pi$ does not affect the agreements:
we have 
$\Lagree(\sigma,\sigma',\rho,\rlocs(\sigma,\rd{G\Img f}))$ and 
$\Lagree(\sigma',\sigma,\rho^{-1},\rlocs(\sigma',\rd{G\Img  f})$,
(cf.\ Eqn.~(\ref{eq:LagreeMono})

(ii) Conjunction and disjunction are straightforward by definitions.
For quantification at a reference type, suppose $\R$ is refperm monotonic
and suppose 
$\sigma|\sigma'\models_\pi 
   \all{x\scol K \smallSplitSym x'\scol K'}{ \R }$.
Thus by definition (see Fig.~\ref{fig:relFmlaSemA}) we have
$\extend{\sigma}{x}{o}|\extend{\sigma'}{x'}{o'}\models_\pi  \R$
for all $o\in\means{K}\sigma \setminus \{\semNull\}$ and             
                    $o'\in\means{K'}\sigma' \setminus \{\semNull\}$.
Now, if $\rho\supseteq\pi$ then for any 
$o\in\means{K}\sigma \setminus \{\semNull\}$ and             
$o'\in\means{K'}\sigma' \setminus \{\semNull\}$
we have 
$\extend{\sigma}{x}{o}|\extend{\sigma'}{x'}{o'}\models_\rho  \R$
by refperm monotonicity of $\R$.
Hence 
$\sigma|\sigma'\models_\rho
   \all{x\scol K \smallSplitSym x'\scol K'}{ \R }$.
For existential quantification, and quantification at type $\INT$ and type $\Region$,
the argument is the same.

(iii) Suppose 
\( \sigma|\sigma'\models_\pi 
G\eqbi G' \land (\all{x\scol K \in G\smallSplitSym x\scol K\in G'}{\Agr x \imp \R})
\).
So $\sigma|\sigma'\models_\pi G\eqbi G'$, i.e., by semantics $\rprel{\sigma(G)}{\sigma'(G')}$.
Thus each element of $\sigma(G)$
(resp.\ $\sigma'(G')$) is in the domain (resp.\ range) of $\pi$. 
Also by semantics we have 
%$\sigma|\sigma'\models_\pi \subst{\R}{x|x}{o|o'}$, that is,
$\extend{\sigma}{x}{o}|\extend{\sigma'}{x}{o'}\models_\pi  \R$, 
for every $(o,o')\in X$ where 
$X = \{(o,o') \mid o\in\sigma(G), o'\in\sigma'(G'), \mbox{ and } (o,o')\in\pi \}$.

Now suppose $\rho\supseteq\pi$. 
We have $\sigma|\sigma'\models_\rho G\eqbi G'$ --- As already noted, agreements are refperm monotonic.
For the second conjunct, 
we need 
%$\sigma|\sigma'\models_\rho \subst{\R}{x|x}{o|o'}$
$\extend{\sigma}{x}{o}|\extend{\sigma'}{x}{o'}\models_\rho  \R$
for every $(o,o')$ in the set $Y$ where 
$Y = \{(o,o') \mid o\in\sigma(G), o'\in\sigma'(G'), \mbox{ and } (o,o')\in\rho \}$.
But $Y=X$, owing to $\rprel{\sigma(G)}{\sigma'(G')}$ hence $o\in\dom(\pi)$ and $o'\in\rng(\pi)$.
So the result follows by refperm monotonicity of $\R$.
\end{proof}

\subsection{On biprogram semantics}

\begin{example}
Bi-coms deterministically dovetail unary steps, without regard to the unary control structure. For example, traces of 
$\Splitbi{ \whilec{1}{a;b;c} }{ \whilec{1}{d} }$ look like this:\footnote{The details depend on the unary transition semantics for loops, which is a standard one that takes a step to unfold the loop body.  An alternate semantics, e.g., using a stack of continuations, would work slightly differently but the point is the same: bi-com deterministically dovetails the unary executions without regard to unary control structure. 
}
\[
\begin{array}{l}
\configc{ \Splitbi{ \whilec{1}{(a;b;c)} }{ \whilec{1}{d} } } \\
\configc{ \Splitbir{ a;b;c;\whilec{1}{(a;b;c)} }{ \whilec{1}{d} } } \\
\configc{ \Splitbi{ a;b;c;\whilec{1}{(a;b;c)} }{ d;\whilec{1}{d} } } \\ 
\configc{ \Splitbir{ b;c;\whilec{1}{(a;b;c)} }{ d;\whilec{1}{d} } } \\
\configc{ \Splitbi{ b;c;\whilec{1}{(a;b;c)} }{ \whilec{1}{d} } } \\
\configc{ \Splitbir{ c;\whilec{1}{(a;b;c)} }{ \whilec{1}{d} } } \\
\configc{ \Splitbi{ c;\whilec{1}{(a;b;c)} }{ d;\whilec{1}{d} } } \\ 
\configc{ \Splitbir{ \whilec{1}{(a;b;c)} }{ d;\whilec{1}{d} } } \\
\configc{ \Splitbi{ \whilec{1}{(a;b;c)} }{ \whilec{1}{d} } } \\
\ldots
\end{array}
\]
The right side iterated twice, the left once.  
\qed
\end{example}

\begin{example}\label{ex:weaving}
In terms of operational semantics, the respective computations of the five
biprograms in Eqn.~(\ref{eq:weaveseq}) are as follows, 
where for clarity we underline the active command for the underlying unary transition,
and abbreviate $\skipc$ as $\emptyeff$.
\[
\begin{array}{l}
\configc{ \splitbi{\underline{a};b;c}{d;e;f}  }
\configc{ \splitbir{b;c}{\underline{d};e;f}  }
\configc{ \splitbi{\underline{b};c}{e;f}  }
\configc{ \splitbir{c}{\underline{e};f}  } 
\configc{ \splitbi{\underline{c}}{f}  } 
\configc{ \splitbir{\emptyeff}{\underline{f}}  } 
\configc{ \syncbi{\emptyeff}  }
\\
\configc{ \splitbi{\underline{a};b}{d} ; \splitbi{c}{e;f} }
\configc{ \splitbir{b}{\underline{d}} ; \splitbi{c}{e;f}  } 
\configc{ \splitbi{\underline{b}}{\emptyeff} ; \splitbi{c}{e;f}  }
\configc{ \splitbi{\underline{c}}{e;f}  }
\configc{ \splitbir{\emptyeff}{\underline{e};f}  }
\configc{ \splitbi{\emptyeff}{\underline{f}}  }
\configc{ \syncbi{\emptyeff}  }
\\
\configc{ \splitbi{\underline{a}}{d;e} ; \splitbi{b;c}{f}  }
\configc{ \splitbir{\emptyeff}{\underline{d};e} ; \splitbi{b;c}{f}  }
\configc{ \splitbi{\emptyeff}{\underline{e}} ; \splitbi{b;c}{f}  }
\configc{ \splitbi{\underline{b};c}{f} }
\configc{ \splitbir{c}{\underline{f}} }
\configc{ \splitbi{\underline{c}}{\emptyeff} }
\configc{ \syncbi{\emptyeff}  }
\\
\configc{ \splitbi{\underline{a};b;c}{\emptyeff} ; \splitbi{\emptyeff}{d;e;f} }
\configc{ \splitbi{\underline{b};c}{\emptyeff} ; \splitbi{\emptyeff}{d;e;f} }
\configc{ \splitbi{\underline{c}}{\emptyeff} ; \splitbi{\emptyeff}{d;e;f} }
\configc{ \splitbi{\emptyeff}{\underline{d};e;f}  }
\configc{ \splitbi{\emptyeff}{\underline{e};f}  }
\configc{ \splitbi{\emptyeff}{\underline{f}}  }
\configc{ \syncbi{\emptyeff}  }
\\
\configc{ \splitbi{\emptyeff}{\underline{d};e;f} ; \splitbi{a;b;c}{\emptyeff}  }
\configc{ \splitbir{\emptyeff}{\underline{e};f} ; \splitbi{a;b;c}{\emptyeff}  }
\configc{ \splitbir{\emptyeff}{\underline{f}} ; \splitbi{a;b;c}{\emptyeff}  }
\configc{ \splitbi{\underline{a};b;c}{\emptyeff}  }
\configc{ \splitbi{\underline{b};c}{\emptyeff}  }
\configc{ \splitbi{\underline{c}}{\emptyeff}  }
\configc{ \syncbi{\emptyeff}  }
\end{array}
\]
Note that $d$-steps of the last two examples go by rule \rn{bComR0}.
\qed
\end{example}

\begin{example}\upshape
In the preceding, we illustrate what happens when the commands do not fault.
Now suppose that the transition for $c$ faults but none of the others do.
(I.e., the $c$-transitions above do not exist.)
Thus there are unary traces completing actions $ab$ and $def$ which can be covered by 
$(\splitbi{a}{d;e} ; \splitbi{b;c}{f})$
and by $(\splitbi{\emptyeff}{d;e;f} ; \splitbi{a;b;c}{\emptyeff})$ 
but not by  
$\splitbi{a;b;c}{d;e;f}$ or the other rearrangements.

If instead both $c$ and $e$ fault, 
then both 
$\splitbi{a;b}{d} ; \splitbi{c}{e;f}$ and 
$\splitbi{a;b;c}{\skipc} ; \splitbi{\skipc}{d;e;f}$ fault trying to execute $c$,
while the others fault trying to execute $e$.

Here is an example of the weaving axiom for conditional:
\[ \splitbi{ \ifc{E}{a;b}{c;d} }{\ifc{E'}{e;f}{g;h} } 
   \weave \ifcbi{E\smallSplitSym E'}{ \splitbi{a;b}{e;f} }{ \splitbi{c;d}{g;h} } \]
Consider a trace of the lhs, where $E$ is true in the left state and $E'$ is false on the right.
Absent faults, the trace may look as follows:
\(
\begin{array}[t]{l} 
\configc{ \splitbi{ \ifc{E}{a;b}{c;d} }{\ifc{E'}{e;f}{g;h} } } \\
\configc{ \splitbir{ a;b }{\ifc{E'}{e;f}{g;h} } } \\
\configc{ \splitbi{ a;b }{ g;h } } \\
\configc{ \splitbir{ b }{ g;h } } \\
\configc{ \splitbi{ b }{ h } } \\
\configc{ \splitbir{ \skipc }{ h } } \\
\configc{ \syncbi{\skipc} }  \\
\end{array}
\)
\\
For the rhs, a trace from the same states has only the initial configuration:
\[ \configc{ \ifcbi{E\smallSplitSym E'}{ \splitbi{a;b}{e;f} }{ \splitbi{c;d}{g;h} } } \]
It faults next, an alignment fault due to test disagreement.
\qed\end{example}

\lemBiprojections*
\begin{proof}
  We need the fact that $\weave^*$ is a congruence.  This is proved by
induction on the reflexive-transitive closure, using the congruence rules for $\weave$ (Figure~\ref{fig:weave}).

The proof of the lemma proceeds by induction on $CC$ .  It's easy to check the
lemma holds when CC is of the form $\syncbi{A}$.  For the inductive cases, we
rely on congruence and transitivity of $\weave^*$.
For example, consider the case when $CC \equiv DD;EE$.  We need to show
$\splitbi{\Left{DD;EE}}{\Right{DD;EE}} \weave^* (DD; EE)$.
We have,
     \[\begin{array}{lll}
         & \splitbi{\Left{DD; EE}}{\Right{DD; EE}} \\
  \equiv & \splitbi{\Left{DD}; \Left{EE}}{\Right{DD}; \Right{EE}} & \mbox{def of projection} \\
  \weave &  \splitbi{\Left{DD}}{\Right{DD}} ; \splitbi{\Left{EE}}{\Right{EE}}   &  \mbox{using $\weave$ axiom for sequence} \\
\weave^* &  DD ; \splitbi{\Left{EE}}{\Right{EE}} & \mbox{congruence and ind hyp $\splitbi{\Left{DD}}{\Right{DD}} \weave^* DD$} \\
\weave^* &  DD ; EE & \mbox{congruence and ind hyp $\splitbi{\Left{EE}}{\Right{EE}} \weave^* EE$}
            \end{array}\]
So $\splitbi{\Left{DD; EE}}{\Right{DD; EE}} \weave^* DD ; EE$
by transitivity.
The other cases follow the same pattern.
\end{proof}

\begin{lemma}\label{lem:activeBi}
\upshape
For any $C$ we have $\Active(\Syncbi{C}) = \Syncbi{\Active(C)}$.
\end{lemma}
The proof is by induction on $C$ using definitions.

% True but evidently we don't need to refer to it directly:
%For all $\configr{CC}{\sigma}{\sigma'}{\mu}{\mu'}$,
%if $CC\nequiv \syncbi{\skipc}$ then the configuration is not stuck, i.e., 
%it transitions to $\Fault$ or to another configuration.

\begin{restatable}[quasi-determinacy of biprogram transitions]{lemma}{lemRdeterminacy}
\label{lem:Rdeterminacy}
\upshape
Let $\phi$ be a relational pre-model.  Then
(a) $\biTrans{\phi}$ is rule-deterministic.
(b) If $\RprelT{\pi|\pi'}{(\sigma|\sigma')}{(\sigma_0|\sigma_0')}$ and 
$\configr{CC}{\sigma}{\sigma'}{\mu}{\mu'} 
\biTrans{\phi} 
\configr{DD}{\tau}{\tau'}{\nu}{\nu'}$
and 
$\configr{CC}{\sigma_0}{\sigma'_0}{\mu}{\mu'} 
\biTrans{\phi} 
\configr{DD_0}{\tau_0}{\tau'_0}{\nu_0}{\nu'_0}$
then $DD\equiv DD_0$, $\nu=\nu_0$, $\nu'=\nu'_0$,
and there are $\rho\supseteq\pi$ and $\rho'\supseteq\pi'$ such that
$\RprelT{\rho|\rho'}{(\tau|\tau')}{(\tau_0|\tau'_0)}$.
(c) If $\RprelT{\pi|\pi'}{(\sigma|\sigma')}{(\sigma_0|\sigma_0')}$ 
then 
$\configr{CC}{\sigma}{\sigma'}{\mu}{\mu'} 
\biTrans{\phi} \Fault$
iff $\configr{CC}{\sigma_0}{\sigma'_0}{\mu}{\mu'} 
\biTrans{\phi} \Fault$.
\end{restatable}

% \lemRdeterminacy*
\begin{proof}
Similar to the proof of Lemma~\ref{lem:determinacy}.  For the one-sided
biprogram transition rules like \rn{bComL}, the argument makes direct use of
Lemma~\ref{lem:determinacy}.  Explicit side conditions of rules \rn{bSync} and
\rn{bSyncX} ensure that $\syncbi{m()}$ transitions only by \rn{bCall},
\rn{bCallX}, or \rn{bCall0}.

A configuration for $\splitbi{C}{D}$ with $C\nequiv\skipc$ takes a step via 
either \rn{bComL} or \rn{bComLX} depending whether $C$ faults or steps; and these are mutually exclusive according to a result about the unary transition relation.
A configuration for $\splitbi{\skipc}{D}$ with $D\nequiv\skipc$ 
goes via either \rn{bComR0} or \rn{bComRX}, depending on whether $D$ faults or not.
A configuration for $\splitbir{C}{D}$ goes via \rn{bComR} or \rn{bComRX}.  
The slightly intricate formulation of the rules for bi-com is necessitated by the need for determinacy and liveness.

Similarly, the rules for bi-while in Fig.~\ref{fig:biprogTransU}
are formulated to be rule deterministic, e.g.,
\rn{bWhR} is only enabled if \rn{bWhL} is not.  
\end{proof}

\paragraph{Projection and embedding: between unary and biprogram traces}

It is convenient to classify the biprogram transition rules as follows.
Leaving aside \rn{bSeq} and \rn{bSeqX},
all the other biprogram rules apply to a non-sequence biprogram of some form.
Rules \rn{bComL} and \rn{bWhL} take \dt{left-only} steps, leaving the right side unchanged, whereas \rn{bComR}, \rn{bComR0}, and \rn{bWhR} take \dt{right-only} steps.
All the other rules are for \dt{both-sides} steps or faulting steps.

\lembitounary*
\begin{proof}
Part (a) is by case analysis of the biprogram transition rules.
For the rules \rn{bCallS} and \rn{bCallX},
observe that the condition (unary compatibility) ensures that the unary steps can be taken.
For rule \rn{bCall0}, the biprogram transition is a stutter,
with both $\configm{\Left{BB}}{\sigma}{\mu} = \configm{\Left{CC}}{\tau}{\nu}$
and $\configm{\Right{BB}}{\sigma}{\mu} = \configm{\Right{CC}}{\tau}{\nu}$.
Indeed, either the left or right step is in the transition relation (or both),
via the unary rule \rn{uCall0} for empty model, owing to Lemma~\ref{lem:emptyOutcomes}.

In all other cases, it is straightforward to check that the rule corresponds to a unary step on one or both sides, 
and in case it is a step on just one side the other side remains unchanged.
Note that it can happen that a step changes nothing: in the unary transition relation, this happens for empty model of a context call, e.g., biprogram step via \rn{bComL} using
unary transition \rn{uCall0}.

For part (b) the proof goes by induction on $T$ and case analysis on the rule by which the last step was taken.
Recall that traces are indexed from 0.  
The base case is $T$ comprised of a single configuration, $T_0$.
Let $U$ be $\Left{T_0}$, $V$ be $\Right{T_0}$, and let both $l$ and $r$ be the singleton mapping $\{ 0\mapsto 0 \}$.
For the induction step, suppose $T$ has length $n+1$ and let $S$ be the prefix including all but the last configuration $T_n$.  
By induction hypothesis we get $l,r,U,V$ such that $\Align(l,r,S,U,V)$.
There are three sub-cases, depending on whether the step from $T_{n-1}$ to $T_n$ is a
left-only step (rule \rn{bComL} or \rn{bWhL}), or right-only, or both sides.
In the case of left-only, let $U'$ be $U \Left{T_n}$, 
let $l'$ be $l\union \{ n\mapsto len(U) \}$,
and let $r'$ be $r\union \{ n\mapsto len(V)-1 \}$.
Then $\Align(l',r',T,U',V)$.  
The other two sub-cases are similar. 

Part (c) holds because one-sided steps are taken only by transition rules
\rn{bComL}, \rn{bComR}, \rn{bComR0}, \rn{bWhL}, and \rn{bWhR}, none of which are applicable to fully aligned programs.
\end{proof}

%% % SAVE (but predates new semantics)
%% Following Lemma~\ref{lem:determinacy} on quasi-determinacy of unary transitions, 
%% we note some corollaries including the key fact that from a given initial configuration these outcomes are mutually exclusive: fault, normal termination, and divergence.
%% This is also true for biprograms, and one proof uses the projection Lemma~\ref{lem:bi-to-unary}.  If a biprogram had both a terminating and a divergent execution from a given initial configuration, then at least one of its unary projections would have both terminating and divergent executions.
%% Exclusivity between faulting and the others can be shown using  
%% Lemma~\ref{lem:Rdeterminacy}. 

\begin{restatable}[trace embedding]{lemma}{lemunarytobi}
\label{lem:unary-to-bi}
\upshape
Suppose $\phi$ is a pre-model.
Let $\cfg$ be a biprogram configuration.
Let $U$ be a trace via $\phi_0$ from $\Left{\cfg}$, 
and $V$ via $\phi_1$ from $\Right{\cfg}$.
Then there is trace $T$ via $\phi$ from $\cfg$
and traces $W$ from $\Left{\cfg}$ and
$X$ from $\Right{\cfg}$
and $l,r$ with $\Align(l,r,T,W,X)$, such that either
\begin{list}{}{}
\item[(a)] $U\leq W$ and $V\leq X$ 
\item[(b)] $U\leq W$ and $X < V$ and $W$ faults next and so does $T$, 
\item[(c)] $V\leq X$ and $W<U$ and $X$ faults next and so does $T$,
\item[(d)] $W < U$ or $X < V$ and the last configuration of $T$ faults, 
via one of the rules \rn{bCallX}, \rn{bIfX}, or \rn{bWhX}, i.e.,
alignment fault.
%the last configuration of $T$ faults via a transition given by \rn{bCallX}, \rn{bIfX}, or \rn{bWhX}.
\end{list}
\end{restatable}

%\lemunarytobi*

%\dn FUTURE I think some of these points are explained redundantly in several places

\begin{proof}
First we make some preliminary observations about the possibilities for a single step.
Let $\cfg$ be $\configr{CC}{\sigma}{\sigma'}{\mu}{\mu'}$ such that $\cfg$
does not fault next and $CC\not\equiv \syncbi{\skipc}$ so there is a next step.
By rule determinacy (Lemma~\ref{lem:Rdeterminacy}(a)),
there is a unique applicable transition rule.
That rule may be a left-only, right-only, or both-sides step,
as per Lemma~\ref{lem:bi-to-unary}(a).  
For all but one of the biprogram transition rules, the form of the rule determines whether 
its transitions are left-, right-, or both-sides.
The one exception is \rn{bCall0}: in case of a transition by this rule, at least one of the unary parts can take a transition, owing to Lemma~\ref{lem:emptyOutcomes}, 
but whether it is left, right, or both depends on the unary models and the states.
% Just by observing the trace one cannot tell, because the step changes nothing.

For left-only transitions, the applicable rules are \rn{bComL} and \rn{bWhL}. 
In case of \rn{bWhL}, $\Left{CC}$ is a loop with test true in $\sigma$ 
and $\configm{\Left{CC}}{\sigma}{\mu}$ takes a deterministic step, unrolling the loop and leaving the state and environment unchanged.
In case of \rn{bComL}, $CC\equiv \splitbi{C}{C'}$ for some $C,C'$ with $C\not\equiv\skipc$, and $\configm{C}{\sigma}{\mu}$ can step via $\trans{\phi_0}$ 
to some $\configm{D}{\tau}{\nu}$ where $\tau$ may be nondeterministically chosen in case $C$ is an allocation or a context call.  
(If $\nu$ differs from $\mu$ it is because $C$ is a let command and 
its transition is deterministic.)
For any choice of $\tau$, rule \rn{bComL} allows 
$\configr{\splitbi{C}{C'}}{\sigma}{\sigma'}{\mu}{\mu'} 
\biTrans{\phi}
\configr{\splitbir{D}{C'}}{\tau}{\sigma'}{\nu}{\mu'}$
(or $\splitbi{D}{\skipc}$ if $C'$ is $\skipc$).
For right-only transitions, the applicable rules are \rn{bComR},
\rn{bComR0}, and \rn{bWhR}, which are similar to the left-only ones.  

The remaining transitions are both-sides. 
By cases on the many applicable both-sides rules, we find in each case that: 
(i) the left and right projections have successors under $\trans{\phi_0},\trans{\phi_1}$ and 
(ii) if
$\configm{\Left{CC}}{\sigma}{\mu} \trans{\phi_0} \configm{D}{\tau}{\nu} $ and
$\configm{\Right{CC}}{\sigma'}{\mu'} \trans{\phi_1} \configm{D'}{\tau'}{\nu'} $ 
then there is some $DD$ with $\Left{DD}\equiv D$, 
$\Right{DD}\equiv D'$,  and
$\configr{CC}{\sigma}{\sigma'}{\mu}{\mu'} \biTrans{\phi}
\configr{DD}{\tau}{\tau'}{\nu}{\nu'}$.
Note that, as in the one-sided cases, $\tau$ and/or $\tau'$ may be nondeterministically chosen (e.g., in the case of \rn{bSync}), and any such choices can also be used for the biprogram transition.
In case the active command of $\cfg$ is a sync'd conditional or loop,
the applicable rules include ones like \rn{bIfTT} that have corresponding unary transitions, but also the rules \rn{bIfX} and \rn{bWhX} in which the biprogram faults although the left and right projections can continue.  

For a both-sides step by rule \rn{bCallS} we rely on condition (relational compatibility)
in Def.~\ref{def:interpRel} of pre-model, 
to ensure that the two unary results $\tau,\tau'$ can be combined to an outcome $\tau|\tau'$ from $\phi_2(m)$---since otherwise the biprogram configuration faults via \rn{bCallX}, contrary to the hypothesis of our preliminary observation above that $\cfg$ does not fault.

To prove the lemma, we construct $T,W,X$ by iterating the preceding observations, choosing the left and right unary steps in accord with $U$ and $V$, unless and until those traces are exhausted.  If needed, $W$ (resp.\ $X$) is extended beyond $U$ (resp.\ $V$).

Let us describe the construction in more detail, as an iterative procedure in which 
$l,r,W,X,T$ are treated as mutable variables, and there is an additional variable $k$.  
Initialize $W,X,T$ to the singleton traces $\Left{\cfg}$, $\Right{\cfg}$, and $\cfg$ respectively. Initially let $k:=0$. 
Let $l$ and $r$ both be the singleton mapping $\{ 0\mapsto 0\}$.
The loop maintains this invariant: 
\[ 
\begin{array}{c}
\Align(l,r,T,W,X) \mbox{ and } (U\leq W \lor  W\leq U) 
                    \mbox{ and } (V\leq X \lor X\leq V)
\\
len(T) = k+1 \mbox{ and } len(W) = l(k)+1 \mbox{ and } len(X) = r(k)+1 
\end{array}
\]
Thus the last configurations of $T,W,X$ are indexed $k,l(k),r(k)$ respectively.

$\bullet$ While $(U\nleq W \mbox{ or } V\nleq X)$ and neither $W$, $X$, nor $T$ faults next, do the following updates, defined by cases on whether $T_k$ is left-only, right-only, or both-sides.

For left-only: update $l,r,W,T$ as follows:
\begin{itemize}
\item set $l(k+1):=l(k)+1$, $r(k+1):=r(k)$
\item if $W<U$, set $W:=W\cdot U_{l(k)}$; otherwise extend $W$ by a choosen successor of $W_{l(k)}$
\item set $T:=T\cdot \cfg'$ where $\cfg'$ is determined by 
the configuration added to $W$, in accord with the preliminary observations above.
Note in particular that $T_k$ does not fault due to failed alignment condition, i.e.,
by rules \rn{bIfX}, \rn{bCallX}, or \rn{bWhX}, because if it does the loop terminates.
\end{itemize}

For right-only: update $l,r,X,T$ as follows:
\begin{itemize}
\item set $l(k+1):=l(k)$, $r(k+1):=r(k)+1$
\item set $X:=X\cdot V_{r(k)}$ if $X<V$, otherwise extend $X$ with a choosen successor of $X_{r(k)}$
\item set $T:=T\cdot \cfg'$ where $\cfg'$ is determined by the configuration added to $X$.
\end{itemize}
 
For both-sides steps, 
set $l(k+1):=l(k)+1$, $r(k+1):=r(k)+1$, and update $W,X,T$ similarly to the preceding cases,
in accord with the preliminary observations.  

To see that the invariants hold following these updates, note that 
the invariant implies $\Left{T_k} = W_{l(k)}$ and $\Right{T_k} = X_{r(k)}$.
Then by construction 
we get a match for the new configuration: $\Left{T_{k+1}} = W_{l(k+1)}$ and $\Right{T_{k+1}} = X_{r(k+1)}$.

The loop terminates, because each iteration decreases the natural number 
\[ (2\times(len(W)\stackrel{.}{-}len(U))+(len(X)\stackrel{.}{-}len(V)) +
(\mifthenelse{1}{\mbox{``active cmd is bi-com''}}{0})
\] 
Here $n\stackrel{.}{-} m$ means subtraction but 0 if $m>n$.
The term $(\mifthenelse{1}{\mbox{``active cmd is bi-com''}}{0})$ is needed in case 
$len(W)>len(U)$ and a left-only step must be taken before the next step happens on the right.
The factor $2\times$ compensates for that term.  (Alternatively, a lexicographic order can be used.)

Now we can prove the lemma. 
If the loop terminates because condition 
$U\nleq W \lor V\nleq X$ is false then we have condition (a) of the Lemma.
If it terminates because $W$ faults next then we have (b),
using invariants $U\leq W\lor W\leq U$ and $V\leq X \lor X\leq V$,
noting that we cannot have $W<U$ if $W$ faults next,
owing to fault determinacy of unary transitions 
(a corollary mentioned following Lemma~\ref{lem:determinacy}).  
Similarly, we get (c) if it terminates because $X$ faults next.
If it terminates because $T$ faults, but the other cases do not hold, then we have (d) owing to the invariants $U\leq W \lor  W\leq U$ and $V\leq X \lor X\leq V$.
\end{proof}

%  \thmbiprogramsoundness*  % proof is currently in the body of the paper 

\begin{definition}[\textbf{denotation of biprogram} \ghostbox{$\means{\Gamma|\Gamma'\proves CC}$}]
%\label{def:denotBiprog}
Suppose $CC$ is wf in $\Gamma|\Gamma'$ and $\phi$ is a pre-model 
that includes all methods called in $C$.
Let $\means{\Gamma|\Gamma'\proves CC}_\phi$
to be the function of type
$\means{\Gamma}\times\means{\Gamma'}\to\powerset(\means{\Gamma}\times\means{\Gamma'})\union\{\Fault\}$
defined by 
\[
\begin{array}{lcl}
\means{\Gamma|\Gamma'\proves CC}_\phi(\sigma|\sigma')
  & \eqdef & \{(\tau|\tau') \mid   
  \configr{CC}{\sigma}{\sigma'}{\_}{\_} \biTranStar{\phi} \configr{\syncbi{\skipc}}{\tau}{\tau'}{\_}{\_} \} 
\\
& & \union\; 
( \{ \Fault \} \mbox{ if } \configr{CC}{\sigma}{\sigma'}{\_}{\_}\biTranStar{\phi} \Fault 
               \mbox{ else } \emptyset )
\end{array}
\]
\end{definition}

Given a pre-model $\phi$, biprogram $CC$, and relational formula $\R$, 
and method  name $m$ not called in $CC$ and not in $\dom(\phi)$,
one can extend the bi-model $\phi_2$ by
\begin{equation}\label{eq:denotBiprog}
\dot{\phi}_2(m)(\sigma|\sigma') \eqdef 
(\mifthenelse{ \{ \Fault \} }{ \neg\some{\pi}{\sigma|\sigma'\models_\pi \R} 
                           }{ \means{CC}_\phi(\sigma|\sigma') }) 
\end{equation}
To be precise, if precondition $\R$ has spec-only variables $\ol{s},\ol{s}'$ on left and right,
the condition should say there are no values for these that satisfy:
$\neg\some{\pi,\ol{v},\ol{v}'}{\sigma|\sigma'\models_\pi \subst{\R}{\ol{s},\ol{s}'}{\ol{v},\ol{v}'}}$.

\begin{lemma}[denoted relational model]
\label{lem:denotBiprog}
\upshape
(i) Suppose $\phi$ is a relational pre-model that includes all the methods in context calls in $CC$, and suppose $m$ is not in $\phi$.
Suppose $\R\imp \leftF{R}\land\rightF{R'}$ is valid.
Let $\dot{\phi}$ extend $\phi$ with
$\dot{\phi}_2(m)$ given by (\ref{eq:denotBiprog}),
$\dot{\phi}_0(m)$ given by Equation~(\ref{eq:denotComm}) for $\Left{CC}, R$, and 
$\dot{\phi}_1(m)$ given by (\ref{eq:denotComm}) for $\Right{CC}, R'$.
Then $(\dot{\phi}_0,\dot{\phi}_1,\dot{\phi}_2)$ is a pre-model.

(ii) Suppose, in addition , that $\Phi \models CC:\rflowty{\R}{\S}{\effe|\effe'}$.
Suppose $\dot{\Phi}$ extends $\Phi$ with 
$\dot{\Phi}_0(m) = \flowty{R}{S}{\effe}$,
$\dot{\Phi}_1(m) = \flowty{R'}{S'}{\effe'}$, and
$\dot{\Phi}_2(m) = \rflowty{\R}{\S}{\effe|\effe'}$
such that $\dot{\Phi}$ is wf.
If $\dot{\phi}_0(m)$ and $\dot{\phi}_1(m)$ 
are models for $\flowty{R}{S}{\effe}$  and 
$\flowty{R'}{S'}{\effe'}$ respectively, 
then $\dot{\phi}$ is a $\dot{\Phi}$-model.
\end{lemma}
\begin{proof}
(i) To show $\dot{\phi}_2(m)$ is a pre-model (Def.~\ref{def:interpRel}),
the fault, state, and divergence determinacy conditions follow from
quasi-determinacy Lemma~\ref{lem:Rdeterminacy} (cf.\ remark following 
projection Lemma~\ref{lem:bi-to-unary}).

Next we show unary compatibility, i.e.,
$\tau|\tau' \in \dot{\phi}_2(m)(\sigma|\sigma')$
implies 
$\tau \in \dot{\phi}_0(m)(\sigma)$.
and 
$\tau' \in \dot{\phi}_1(m)(\sigma')$.
Now $\tau|\tau' \in \dot{\phi}_2(m)(\sigma|\sigma')$
iff 
$\configr{CC}{\sigma}{\sigma'}{\_}{\_} \biTranStar{\phi} \configr{\syncbi{\skipc}}{\tau}{\tau'}{\_}{\_}$
and by projection Lemma~\ref{lem:bi-to-unary} that implies 
$\configm{\Left{CC}}{\sigma}{\_} 
\tranStar{\phi_0} 
\configm{\skipc}{\tau}{\_}$
whence 
$\tau\in\dot{\phi}_0(m)(\sigma)$ 
provided that $\sigma\models R$
(\emph{mut.\ mut.} for the right side).
Since $\tau|\tau' \in \dot{\phi}_2(m)(\sigma|\sigma')$,
there is some $\pi$ for which $(\sigma|\sigma')$ satisfies $\R$, 
and by validity of $\R\imp \leftF{R}\land\rightF{R'}$ 
this implies $\sigma\models R$.
Similarly for the right side.

For fault compatibility, 
suppose $\Fault\in\dot{\phi}_0(m)(\sigma)$ or $\Fault\in\dot{\phi}_1(m)(\sigma')$.
Then either $\sigma\not\models R$ or $\sigma'\not\models R'$, by definitions,
whence $\sigma|\sigma'\not\models\R$ owing to 
validity of $\R\imp \leftF{R}\land\rightF{R'}$.
So $\Fault \in \dot{\phi}_2(m)(\sigma|\sigma')$ as required.

To show relational compatibility, suppose
$\tau\in\dot{\phi}_0(m)(\sigma)$ and 
$\tau'\in\dot{\phi}_1(m)(\sigma')$.
We need $\dot{\phi}_2(m)$ to contain either
$\Fault$ or $(\tau|\tau')$.
If there is no $\pi$ with $\sigma|\sigma'\models_\pi \R$ then 
$\dot{\phi}_2(m)$ is $\{ \Fault \}$ and we are done.
Otherwise, 
from $\tau\in\dot{\phi}_0(m)(\sigma)$ and 
$\tau'\in\dot{\phi}_1(m)(\sigma')$ we have
traces 
$\configm{C}{\sigma}{\_}
\tranStar{\phi_0}
\configm{\skipc}{\tau}{\_}$
and 
$\configm{C'}{\sigma'}{\_}
\tranStar{\phi_1}
\configm{\skipc}{\tau'}{\_}$.
By embedding Lemma~\ref{lem:unary-to-bi}, 
we get that either 
$\configr{CC}{\sigma}{\sigma'}{\_}{\_} \biTranStar{\phi} \configr{\syncbi{\skipc}}{\tau}{\tau'}{\_}{\_}$ 
or else 
$\configr{CC}{\sigma}{\sigma'}{\_}{\_}$ faults due to alignment conditions.
Either way we are done showing that 
$(\dot{\phi}_0,\dot{\phi}_1,\dot{\phi}_2)$ is a pre-model.

(ii) Suppose that $\Phi \models CC:\rflowty{\R}{\S}{\effe|\effe'}$.
The conditions of Def.~\ref{def:ctxinterpRel}
for $\dot{\phi}_2(m)$ with respect to $\rflowty{\R}{\S}{\effe}$
are direct consequences of $\Phi \models CC:\rflowty{\R}{\S}{\effe|\effe'}$
and (\ref{eq:denotBiprog}).
\end{proof}

\thmbiprogramsoundness*
\begin{proof}
Let $U,V$ be the traces and let $T$ be the biprogram trace given by embedding Lemma~\ref{lem:unary-to-bi}.  The judgment for $CC$ is applicable to $T$, 
so cases (b), (c), and (d) in the Lemma are ruled out---$T$ cannot fault.
The remaining case is (a), that is, $T$ covers every step of $U$ and $V$.
If $U$ and $V$ are terminated then so is $T$, whence the postcondition holds,
and the Write condition holds, by validity of the judgment.
Regardless of termination, we also get the unary Safety and Encap conditions for $U$ and $V$,
by definitions since every step is covered by $T$.
\end{proof}

\section{Appendix: Relational logic and its soundness (re Sect.~\ref{sec:rellog})}
%\ab{Move soundness statement restatement to Sect.~D; Besides Sects.~D3-D11, have another subsection that talks about limitations. Move nested linking to that subsection.} \dn{I moved the restatement; we decided not to move nested linking}

\thmsound*

Sect.~\ref{sec:app:rules} presents relational proof rules omitted from the body of the paper.
Sect.~\ref{sec:app:lemrloceq} proves the crucial lockstep alignment lemma.
The soundness proofs comprise subsections~\ref{sec:rLocEqSound}--\ref{sec:rWeaveSound};
these are largely independent and need not be read in any particular order. 

\subsection{Additional rules}\label{sec:app:rules}

\begin{figure}[t!]
\begin{footnotesize}
\begin{mathpar}
\inferrule*[left=rEmbS]{ 
   \Phi_0\HPflowtr{}{}{P}{A}{Q}{\eff}\\
   \Phi_1\HPflowtr{}{}{P'}{A}{Q'}{\eff'}\\
}{ \Phi\rHPflowtr{}{}{\leftF{P}\land\rightF{P'}}{\syncbi{A}}{\leftF{Q}\land\rightF{Q'}}{\eff|\eff'}
}

\inferrule*[left=rSeq]
{\Phi\rHPflowtr{}{}{\P}{CC_1}{\P_1}{\eff_1|\eff'_1} \\ 
\Phi\rHPflowtr{}{}{\P_1}{CC_2}{\Q}{\eff_2|\eff'_2} \\
\eff_2 \mbox{ is } \Left{\P}/\eff_1\mbox{-immune} \\
\eff'_2 \mbox{ is } \Right{P}/\eff'_1\mbox{-immune} 
}
{ \Phi\rHPflowtr{}{}{\P}{\seqc{CC_1}{CC_2}}{\Q}{\eff_1, \eff_2|\eff'_1,\eff'_2}
}

\rulerIf

\rulerWhile

\inferrule*[left=rIf4]
{
\Phi \rHPflowtr{}{M}{\P\land\leftF{E}\land\rightF{E'}}{\splitbi{C}{C'}}{\Q}{\eff|\eff'} \\
\Phi \rHPflowtr{}{M}{\P\land\leftF{E}\land\rightF{\neg E'}}{\splitbi{C}{D'}}{\Q}{\eff|\eff'} \\
\Phi \rHPflowtr{}{M}{\P\land\leftF{\neg E}\land\rightF{E'}}{\splitbi{D}{C'}}{\Q}{\eff|\eff'} \\
\Phi \rHPflowtr{}{M}{\P\land\leftF{\neg E}\land\rightF{\neg E'}}{\splitbi{D}{D'}}{\Q}{\eff|\eff'} \\
\delta=\unioneff{N\in\Phi,N\neq M}{\bnd(N)} \\
\ind{\delta}{\rTow(\ftpt(E))} \\
\ind{\delta}{\rTow(\ftpt(E'))} 
}{
\Phi\rHPflowtr{}{M}{\P}{\Splitbi{\ifc{E}{C}{D}}{\ifc{E'}{C'}{D'}}}{\Q}{\eff,\ftpt(E)|\eff',\ftpt(E')}
}

\inferrule*[left=rVar]
{ \Phi \proves^{\Gamma,x\scol T|\Gamma',x'\scol T'} 
          CC: \rflowty{\P}{\Q}{\eff|\eff'} 
}{ \Phi \proves^{\Gamma|\Gamma'} \varblockbi{x\scol T\smallSplitSym x'\scol T'}{CC} :
         \rflowty{\P\land\leftF{x=\Default{T}}\land\rightF{x'=\Default{T'}}}{\Q}{\eff|\eff'} 
}

\end{mathpar}
\end{footnotesize}
\caption{Relational proof rules omitted from Fig.~\ref{fig:proofrulesR}.
}
\label{fig:proofrulesRapp}
\end{figure}

Figure~\ref{fig:proofrulesRapp} presents the proof rules omitted in the body of the article.

Rule \rn{rIf} is typical of relational Hoare logics, with the addition of side conditions to ensure 
encapsulation.
Similarly, rules \rn{rSeq} and \rn{rWhile} have the same immunity conditions as their unary counterparts.
Rules \rn{rWhile} and \rn{rSeq} are slightly simplified from the general rules, for clarity.
The general rules should include an initial snapshot $r=\lloc$,
and region $H$ and field list $\ol{f}$,
with conditions to ensure that $H$ contains only freshly allocated objects so writes of $H\Img \ol{f}$ can be omitted from the frame condition.
This caters for writes to locations allocated in the first command of a sequence,
or previous iterations of a loop, just as it is done in the unary \rn{Seq} and \rn{While} rules (Fig.~\ref{fig:proofrules}).
(The details are justified in \RLI, though in \RLI\ the rules are slightly more succinct owing to use of freshness effect notation.)

%% % Ok but omit for TOPLAS
%% For \rn{rWhile} The soundness argument uses an induction on the number of iterations, where an iteration may execute the biprogram $CC$ or just one of its projections $\Left{CC}$ or $\Right{CC}$.
%% We show that $\Q$ holds after iteration, and the writes are within the frame condition, 
%% using the premises for $CC$, $\Left{CC}$, or $\Right{CC}$ as needed.  
%% The preconditions of these judgments correspond to the enabling conditions for the 
%% transition rules of bi-while.  
%% The side condition
%% $\Q\imp E\eqbi E' \lorbi (\P\land\leftF{E}) \lorbi (\P'\land\rightF{E'})$
%% ensures that the three premises cover all cases.
%% It also ensures that the faulting transition \rn{bWhX} is never taken,
%% which together with the premises takes care of Safety.

%% For Encap and R-safety the proof uses the projection Lemma~\ref{lem:bi-to-unary}
%% (which in turn relies on compatibility between the bi-model and the unary context models, Def.~\ref{def:interpRel}) ??and compatibility of the unary/relational frame conditions 
%% as per Def.~\ref{def:wfhypRel}?? 
%% The side conditions 
%% $\ind{\unioneff{N\in\Phi,N\neq M}{\bnd(N)}}{\rTow(\ftpt(E))}$ 
%% and 
%% $\ind{\unioneff{N\in\Phi,N\neq M}{\bnd(N)}}{\rTow(\ftpt(E'))}$ 
%% are taken from the unary \rn{While} rule; they ensure that the loop guards read outside the relevant boundaries.  

\begin{remark}
\upshape
As in the unary \rn{While}, the frame condition in \rn{rWhile} needs to include the footprint of the loop tests ($\ftpt(E)$, $\ftpt(E')$) as the behavior depends on them.
Given that the alignment guards $\P$ and $\P'$ influence the bi-while transitions,
one may expect that their footprints should also be included.
But the dependency of r-respect (Encap) is about execution on one side.
The value of $E$ (resp.\ $E'$) determines the control state (i.e., unfold the loop body or terminate) at the unary level.
By contrast, the value of $\P$ (resp.\ $\P'$) determines the biprogram control state.
This is reflected in the unary control state, but during a one-sided iteration the other side stutters; and stuttering transitions are removed (by projection, see Lemma~\ref{lem:bi-to-unary}) 
according to the definition of Encap in Def.~\ref{def:validR}. 
\qed\end{remark}

\begin{remark}
\upshape
Rule \rn{rWhile} can be slightly strengthened to take into account that in our semantics, 
to ensure quasi-determinacy, a right iteration only happens when the left guard or test is false.
We prefer the more symmetric phrasing of the rule: what matters is that one-sided 
executions under their designated alignment guard maintain the invariant.
The deterministic scheduling is a technical artifact, just like the specific details of the dovetailed execution of the bi-com construct are not important for reasoning.  
\qed\end{remark}

\subsection{Proof of lockstep alignment lemma}\label{sec:app:lemrloceq}

\lemsnap*
\begin{proof}
\begin{sloppypar}
Assume $\tau\models\snap(\eff)$ and $\tau\allowTo\upsilon\models\eff$.
The equality
$\wlocs(\tau,\eff)\setminus\rlocs(\upsilon,\delta^\oplus)=
\rlocs(\upsilon,\Asnap(\eff)\setminus\delta)$
is between sets of locations, i.e., variables and heap locations.
We consider the two kinds of location in turn.
\end{sloppypar}

For variables, we have 
$x\in\wlocs(\tau,\eff)\setminus\rlocs(\upsilon,\delta^\oplus)$
iff $\wri{x}$ is in $\eff$ and $\rd{x}$ is not in $\delta^\oplus$,
by definitions.
On the other hand, by definition of $\Asnap$, we have 
$x\in\rlocs(\upsilon,\Asnap(\eff)\setminus\delta)$
iff $\rd{x}$ is not in $\delta$ and $\wri{x}$ is in $\eff$ and $x\nequiv\lloc$.
The conditions are equivalent. 

For a heap locations, w.l.o.g.\ we assume $\eff$ and $\delta$ are in normal form
and have exactly one read and one write effect for each field.
We are only concerned with writes in $\eff$ and reads in $\delta$.
Consider any field name $f$ and suppose $\eff$ contains $\wri{G\Img f}$ 
and $\delta$ contains $\rd{H\Img f}$ for some $G,H$.   
Now for location $o.f$ we have
\[\begin{array}{ll}
     & o.f \in\wlocs(\tau,\eff)\setminus\rlocs(\upsilon,\delta^\oplus) \\
\iff & o \in \tau(G) \setminus \upsilon(H) 
       \quad\mbox{by defs $\wlocs,\rlocs$ and normal form}\\
\iff & o \in \tau(s_{G,f}) \setminus \upsilon(H) \quad
     \mbox{by $\tau\models\snap(\eff)$ we have $\tau(s_{G,f})=\tau(G)$} \\
\iff & o \in \upsilon(s_{G,f}) \setminus \upsilon(H) \quad
     \mbox{by $\tau\allowTo\upsilon\models\eff$ and $\wri{s_{G,f}}\notin\eff$ have $\tau(s_{G,f})=\upsilon(s_{G,f})$ } \\
\iff & o \in \upsilon(s_{G,f}\setminus H) \quad\mbox{by semantics of subtraction}
\end{array}
\]
On the other hand,
\[\begin{array}{ll}
     & o.f\in\rlocs(\upsilon,\Asnap(\eff)\setminus\delta) \\
\iff & o.f\in\rlocs(\upsilon, (\rd{s_{G,f}}\Img f \setminus\rd{H\Img f})) 
     \quad\mbox{by def $\Asnap$ and assumption about $G,H$}\\
\iff & o.f\in\rlocs(\upsilon, \rd{(s_{G,f} \setminus H)\Img f}) 
     \quad\mbox{by effect subtraction} \\
\iff & o \in \upsilon(s_{G,f}\setminus H) \quad \mbox{by def $\rlocs$} 
\end{array}
\]
The conditions are equivalent.
\end{proof}

\lemrloceq* % proof of lem:rloceq
\begin{proof}
As usual write $\hat{\sigma},\hat{\sigma}'$ for the extensions of $\sigma,\sigma'$ for the spec only variables of the precondition, as per (ii).

We show that the conditions (v--vii) hold at every step within $T$,
by induction on steps.\footnote{We are glossing over the local variables introduced by local blocks.
To be precise, the initial states are both for $\Gamma$ and have no extra variables.  
The Lemma should have additional conclusion that $\Vars{\tau}=\Vars{\tau'}$,
which becomes part of the induction hypothesis,
to account for possible addition of locals,
which will be in $\freshLocs$.
}
One might expect that the lemma could be simplified to simply say the conditions hold at every reachable step, without mentioning traces, 
but we are assuming rather than proving that the r-safety and r-respect conditions hold, 
so the present formulation seems more clear.

\textbf{Base Case.} For initial configuration $\configr{\Syncbi{C}}{\sigma}{\sigma'}{\_}{\_}$, we have 
$\freshLocs(\sigma,\sigma)=\emptyset=\freshLocs(\sigma',\sigma')$ and 
$\written(\sigma,\sigma)=\emptyset=\written(\sigma',\sigma')$. 
From hypothesis (ii) of the Lemma, % precondition (\ref{eq:loceqpre})
and the semantics of the agreement formulas in the precondition,
we get $\agree(\sigma,\sigma',\pi,\eff^\leftarrow_\delta)$ and
$\agree(\sigma',\sigma,\pi^{-1},\eff^\leftarrow_\delta)$. 
Unfolding definitions, we have proved the claim with $\rho,\tau,\tau':=\pi,\sigma,\sigma'$.

\textbf{Induction case.}
Suppose $\configr{\Syncbi{C}}{\sigma}{\sigma'}{\_}{\_} \biTranStar{\phi} 
\configr{BB}{\tau}{\tau'}{\mu}{\mu'} \biTrans{\phi} 
\configr{DD}{\upsilon}{\upsilon'}{\nu}{\nu'}$
as a prefix of $T$.
By induction hypothesis we have $\mu=\mu'$, 
$BB=\Syncbi{B}$ for some $B$ and 
for some $\rho\supseteq\pi$ we have 
\begin{equation}\label{eq:indhyp}
\begin{array}{c}
\Lagree(\tau,\tau',\rho, (\freshLocs(\sigma,\tau) \union \rlocs(\sigma,\eff)\union 
\written(\sigma,\tau))\setminus\rlocs(\tau,\delta^\oplus))\\
\Lagree(\tau',\tau,\rho^{-1}, (\freshLocs(\sigma',\tau') \union \rlocs(\sigma',\eff)\union 
\written(\sigma',\tau'))\setminus\rlocs(\tau',\delta^\oplus))\\
\end{array}
\end{equation}
Without loss of generality, we assume that $\Syncbi{B}\equiv\Syncbi{B_0};\Syncbi{B_1}$, where 
$\Active(B)\equiv B_0$. 
(Recall by Lemma~\ref{lem:activeBi} that $\Active{\Syncbi{B}} = \Syncbi{\Active{B}}$.)

To find $D$ and an extension of $\rho$, such that the agreements for $\upsilon|\upsilon'$ and other conditions hold for the step 
$\configr{BB}{\tau}{\tau'}{\mu}{\mu'} \biTrans{\phi} 
\configr{DD}{\upsilon}{\upsilon'}{\nu}{\nu'}$,
we go by cases on the possible transition rules.
The fault rules are not relevant.

\underline{Cases \rn{bComL}, \rn{bComR}, \rn{bComR0}, \rn{bWhL}, and \rn{bWhR}} are not applicable to $\Syncbi{B}$.

%\underline{Case \rn{bIfX} and \rn{bWhX}}, the non-agreeing branch transitions, may be applicable to configurations with $\Syncbi{B}$, but they go to $\Fault$ and so are not applicable to the trace under consideration.

\underline{Case \rn{bSync}}.
% NOTE: in the following is where it's not viable to require the alternate step to come from a complete execution.
So $B_0$ is an atomic command other than a method call and there are 
unary transitions $\configm{B_0}{\tau}{\mu}\trans{\phi_0}\configm{\skipc}{\upsilon}{\mu}$
and $\configm{B_0}{\tau'}{\mu'}\trans{\phi_1}\configm{\skipc}{\upsilon'}{\mu'}$.
The successor configuration has $DD\equiv\Syncbi{B_1}$ and $\nu=\mu=\mu'=\nu'$.
Because the step is not a method call, the same transitions can be taken via the other models, 
i.e.,
we have $\configm{B_0}{\tau}{\mu}\trans{\phi_1}\configm{\skipc}{\upsilon}{\mu}$
and $\configm{B_0}{\tau'}{\mu'}\trans{\phi_0}\configm{\skipc}{\upsilon'}{\mu'}$.
Moreover, owing to the agreements, we can instantiate the left and right trace's respect condition
(hypothesis (iv) of this Lemma). 
As we are considering a non-call command, 
the collective boundary for r-respect is 
$\dot{\delta} = \unioneff{N\in(\Psi,\mu),N\neq \topm(B,M)}{\bnd(N)} $.
By hypothesis (iii) of the Lemma, $C$ is let-free. So $\mu$ is empty.
Moreover, there is no $\Endcall$ in $B$, there being no environment calls (and as always the starting command has no end markers), so $\topm(B,M) = M$.
So the collective boundary for r-respect is the $\delta$ assumed in the Lemma, i.e.,
$\delta = \unioneff{N\in\Psi,N\neq M}{\bnd(N)}$.
Both steps satisfy w-respect, i.e., do not write inside the boundary, owing to hypothesis (iv) of the Lemma.
Instantiating r-respect twice (with $\tau,\tau',\phi_0,\rho$ and with $\tau',\tau,\phi_1,\rho^{-1}$),
we have the allowed dependences 
$\AllowedDep{\tau}{\tau'}{\upsilon}{\upsilon'}{\eff}{\sigma}{\delta}{\rho}$ and 
$\AllowedDep{\tau'}{\tau}{\upsilon'}{\upsilon}{\eff}{\sigma'}{\delta}{\rho^{-1}}$. 
Even more, r-respects applied to (\ref{eq:indhyp}) gives some $\dot{\rho}$ and $\dot{\rho}'$ with $\dot{\rho}\supseteq\rho$ 
and $\dot{\rho}'\supseteq\rho^{-1}$ and the following four conditions:
\begin{equation}\label{eq:rembeq1}
\begin{array}{l}
\Lagree(\upsilon,\upsilon',\dot{\rho},(\freshLocs(\tau,\upsilon) \union 
     \written(\tau,\upsilon))\setminus\rlocs(\upsilon,\delta^\oplus))\\
\dot{\rho}(\freshLocs(\tau,\upsilon)\setminus\rlocs(\upsilon,\delta))\subseteq
\freshLocs(\tau',\upsilon')\setminus\rlocs(\upsilon',\delta)
\\
\Lagree(\upsilon',\upsilon,\dot{\rho}',(\freshLocs(\tau',\upsilon') \union 
     \written(\tau',\upsilon'))\setminus\rlocs(\upsilon',\delta^\oplus))\\
\dot{\rho}'(\freshLocs(\tau',\upsilon')\setminus\rlocs(\upsilon',\delta))\subseteq
\freshLocs(\tau,\upsilon)\setminus\rlocs(\upsilon,\delta)
\end{array}
\end{equation}
\begin{sloppypar}
By balanced symmetry Lemma~\ref{lem:fresh-sym}, we get 
\[ %\label{eq:rembeq2}
\begin{array}{l}
\Lagree(\upsilon',\upsilon,\dot{\rho}^{-1},(\freshLocs(\tau',\upsilon') \union 
\written(\tau',\upsilon'))\setminus\rlocs(\upsilon',\delta^\oplus))\\
\dot{\rho}(\freshLocs(\tau,\upsilon)\setminus\rlocs(\upsilon,\delta))=
\freshLocs(\tau',\upsilon')\setminus\rlocs(\upsilon',\delta)
\end{array}
\]
We can use preservation Lemma~\ref{lem:selfframing-agreement2}
for these three sets of locations (which are subsets of $\locations(\tau)$): 
$\rlocs(\sigma,\eff)\setminus\rlocs(\tau,\delta^\oplus)$,
$\written(\sigma,\tau)\setminus\rlocs(\tau,\delta^\oplus)$, and
$\freshLocs(\sigma,\tau)\setminus\rlocs(\tau,\delta^\oplus)$.
By Lemma~\ref{lem:selfframing-agreement2} we get 
\end{sloppypar}
\[\Lagree(\upsilon,\upsilon',\dot{\rho},((\freshLocs(\sigma,\tau) \union \rlocs(\sigma,\eff)\union 
\written(\sigma,\tau))\setminus\rlocs(\tau,\delta^\oplus))
\setminus\rlocs(\upsilon,\delta^\oplus))\]
So by the boundary monotonicity condition of Encap we have
$\rlocs(\tau,\delta^\oplus)\subseteq\rlocs(\upsilon,\delta^\oplus)$.
Now from this and (\ref{eq:rembeq1}), using 
$\freshLocs(\sigma,\upsilon) = \freshLocs(\sigma,\tau)\union\freshLocs(\tau,\upsilon)$
and $\written(\sigma,\upsilon) \subseteq \written(\sigma,\tau)\union\written(\tau,\upsilon)$,
we can combine the agreements together to get 
\[\Lagree(\upsilon,\upsilon',\dot{\rho},(\freshLocs(\sigma,\upsilon) \union \rlocs(\sigma,\eff) \union 
\written(\sigma,\upsilon))\setminus\rlocs(\upsilon,\delta^\oplus))\]
With a similar argument we obtain the symmetric condition
\[\Lagree(\upsilon',\upsilon,\dot{\rho}^{-1},(\freshLocs(\sigma',\upsilon') \union 
\rlocs(\sigma',\eff) \union\written(\sigma',\upsilon'))\setminus\rlocs(\upsilon',\delta^\oplus))\]
which finishes this case for the induction step.

\begin{sloppypar}
\underline{Case \rn{bCallS}}.
So $B_0$ is $m()$ for some $m$, and $(\upsilon|\upsilon')\in\phi_2(m)(\tau|\tau')$.
The successor configuration has $DD\equiv\Syncbi{B_1}$ and $\nu=\mu=\mu'=\nu'$.
Suppose $\Psi(m)$ is $\flowty{R}{S}{\effe}$.
By the assumed r-safe condition (hypothesis (iv) of the Lemma), we have 
$\rlocs(\tau,\effe)\subseteq\freshLocs(\sigma,\tau)\union\rlocs(\sigma,\eff)$. 
Since $\phi_2(m)(\tau|\tau') \neq \Fault$,
there must be values for the spec-only variables $\ol{t}$ of $m$'s spec
for which $\tau|\tau'$ satisfy the method's precondition, which by hypothesis (i) of the lemma implies
the precondition of $\locEq_\delta(\Psi(m))$.
That is, there are $\ol{u}$ and $\ol{u}'$ such that
$\hat{\tau}|\hat{\tau}'\models_\rho  
\Both{R}\land\Agr(\reads(\effe)\setminus\delta^\oplus)\land\Both{(s_{\lloc}^m=\lloc \land \snap^m(\effe))}$,
where $\hat{\tau} = \extend{\tau}{\ol{t}}{\ol{u}}$
and $\hat{\tau}' = \extend{\tau'}{\ol{t}}{\ol{u}'}$.
(Apropos the identifier $s_{\lloc}^m$ see Footnote~\ref{fn:msnapshot}.)
Since $\phi\models\Phi$ and 
$(\upsilon|\upsilon')\in\phi_2(m)(\tau|\tau')$, we get the postcondition of $\Phi(m)$, which implies that of 
$\locEq_\delta(\Psi(m))$. 
Hence $\hat{\upsilon}|\hat{\upsilon}'\models_\rho{\later(\Both{Q}\land\Agr\effe^\rightarrow_\delta)}$,
where $\hat{\upsilon} = \extend{\upsilon}{\ol{t}}{\ol{u}}$,
$\hat{\upsilon}' = \extend{\upsilon'}{\ol{t}}{\ol{u}'}$,
and 
\begin{equation}\label{eq:rloceqEta}
\effe^\rightarrow_\delta\equiv(\rd{(\lloc\setminus s_{\lloc}^m)\Img\allfields},
\Asnap^m(\effe))\setminus\delta
\end{equation}
So by semantics of $\later$ and $\Agr$
there is $\dot{\rho}\supseteq\rho$ with 
$\agree(\hat{\upsilon},\hat{\upsilon}',\dot{\rho},\effe^\rightarrow_\delta)$ and
$\agree(\hat{\upsilon}',\hat{\upsilon},\dot{\rho}^{-1},\effe^\rightarrow_\delta)$. 
We have $\freshLocs(\tau,\upsilon)=\rlocs(\upsilon,\rd{(\lloc\setminus s_{\lloc}^m)\Img\allfields})$ and 
$\freshLocs(\tau',\upsilon')=\rlocs(\upsilon',\rd{(\lloc\setminus s_{\lloc}^m)\Img\allfields})$.
We also have $\written(\tau,\upsilon)\subseteq\wlocs(\tau,\effe)$ and 
$\written(\tau',\upsilon')\subseteq\wlocs(\tau',\effe)$, from $\tau\allowTo\upsilon\models\effe$ and
$\tau'\allowTo\hat\upsilon'\models\effe$. Furthermore, by Lemma~\ref{lem:snap}, we have 
\[ \begin{array}{l}
  \wlocs(\tau,\effe)\setminus\rlocs(\upsilon,\delta^\oplus)=
\rlocs(\upsilon,\Asnap^m(\effe)\setminus\delta)\subseteq\rlocs(\upsilon,\effe^\rightarrow_\delta) \\
\wlocs(\tau',\effe)\setminus\rlocs(\upsilon',\delta^\oplus)=
\rlocs(\upsilon',\Asnap^m(\effe)\setminus\delta)\subseteq\rlocs(\upsilon',\effe^\rightarrow_\delta)
\end{array}
\]
So we have
\begin{equation}\label{eq:callrembeq1} 
\Lagree(\upsilon,\upsilon',\dot{\rho},(\freshLocs(\tau,\upsilon) \union 
\written(\tau,\upsilon))\setminus\rlocs(\upsilon,\delta^\oplus))
\end{equation}
\begin{equation}\label{eq:callrembeq2}
\Lagree(\upsilon',\upsilon,\dot{\rho}^{-1},(\freshLocs(\tau',\upsilon') \union 
\written(\tau',\upsilon'))\setminus\rlocs(\upsilon',\delta^\oplus))
\end{equation}
Thus we have 
\mbox{$\AllowedDep{\tau}{\tau'}{\upsilon}{\upsilon'}{\effe}{\sigma}{\delta}{\rho}$}
and \mbox{$\AllowedDep{\tau'}{\tau}{\upsilon'}{\upsilon}{\effe}{\sigma'}{\delta}{\rho^{-1}}$}.
Since $\rlocs(\sigma,\eff)\setminus\rlocs(\tau,\delta^\oplus)$, 
$\written(\sigma,\tau)\setminus\rlocs(\tau,\delta^\oplus)$ and 
$\freshLocs(\sigma,\tau)\setminus\rlocs(\tau,\delta^\oplus)$ are subsets of 
$\locations(\tau)$, using 
Lemma~\ref{lem:selfframing-agreement2}, from (\ref{eq:indhyp}) we get 
\[\Lagree(\upsilon,\upsilon',\dot{\rho},((\freshLocs(\sigma,\tau) \union \rlocs(\sigma,\eff)\union 
\written(\sigma,\tau))\setminus\rlocs(\tau,\delta^\oplus))\setminus\rlocs(\upsilon,\delta^\oplus))\]
By hypothesis (iv) of the Lemma, the steps satisfy boundary monotonicity, i.e., 
$\rlocs(\tau,\delta)\subseteq\rlocs(\upsilon,\delta)$,
which implies 
$\rlocs(\tau,\delta^\oplus)\subseteq\rlocs(\upsilon,\delta^\oplus)$.
Combining this with the agreements of (\ref{eq:callrembeq1}), we get 
\[\Lagree(\upsilon,\upsilon',\dot{\rho},(\freshLocs(\sigma,\upsilon) \union \rlocs(\sigma,\eff) \union 
\written(\sigma,\upsilon))\setminus\rlocs(\upsilon,\delta^\oplus))\]
With a similar argument using (\ref{eq:callrembeq2}), we get the symmetric condition
\[\Lagree(\upsilon',\upsilon,\dot{\rho}^{-1},(\freshLocs(\sigma',\upsilon') \union 
\rlocs(\sigma',\eff) \union 
\written(\sigma',\upsilon'))\setminus\rlocs(\upsilon',\delta^\oplus))\]
which completes this case.
\end{sloppypar}

\underline{Case \rn{bCall0}}.
So $B_0$ is a context call $m()$ that stutters because the $\phi_2(m)$ is empty.
The agreements are maintained, as nothing changes.

\underline{Case \rn{bVar}}.
This relies on the additional condition that $\Vars{\tau}=\Vars{\tau'}$,
which can be included in the induction hypothesis but is omitted for readability.
We have that 
$B_0$ is $\varblock{x\scol T}{B_2}$ for some $x,T,B_2$,
so 
$\Syncbi{B_0} \equiv 
\varblockbi{x\scol T \smallSplitSym x\scol T}{\Syncbi{B_2}}$.
Because $\Vars{\tau}=\Vars{\tau'}$, and using the assumption that
$\varfresh$ depends only on $\Vars$ of the state
(Eqn.~(\ref{eq:varfresh})),
we have some $w$ with $w = \varfresh(\tau) = \varfresh(\tau')$.
This ensures $\Vars{\upsilon}=\Vars{\upsilon'}$, justifying the omitted induction hypothesis; the only other change to variables is by dropping them,
by \rn{bSync} transition for $\syncbi{\Endvar(w)}$.
The step from 
$\varblockbi{x\scol T \smallSplitSym x\scol T}{\Syncbi{B_2}}$
goes to 
\( \configr{\subst{\Syncbi{B_2}}{x,x}{w,w};\syncbi{\Endvar(w)};\Syncbi{B_1}}
{\upsilon}{\upsilon'}{\mu}{\mu'} \)
where 
$\upsilon = \extend{\tau}{w}{\Default{T}}$ and 
$\upsilon' = \extend{\tau'}{w'}{\Default{T'}}$.
We get the agreements because nothing changes except the addition of $w$ with default value.
We get the code alignment because 
$\subst{\Syncbi{B_2}}{x,x}{w,w} \equiv 
\Syncbi{\subst{B_2}{x,x}{w,w}}$ by definitions.

\underline{Cases \rn{bIfTT} and \rn{bIfFF}}.
So $B_0$ has the form $\ifc{E}{B_2}{B_3}$ and the successor configuration 
has the form either $\Syncbi{B_2};\Syncbi{B_1}$
or $\Syncbi{B_3};\Syncbi{B_1}$.
Nothing else changes so the agreements are maintained.

%% \underline{Case \rn{bIfY}}.
%% So $B_0$ has the form $\ifc{E}{B_2}{B_2}$ and the successor configuration 
%% has the form  $\Syncbi{B_2};\Syncbi{B_1}$.
%% Nothing else changes so the agreements are maintained.

\underline{Cases \rn{bWhTT} and \rn{bWhFF}}.
So $B_0$ has the form $\whilec{E}{B_2}$ and the successor configuration has the form 
either $\Syncbi{B_2};\Syncbi{B_0};\Syncbi{B_1}$ (for \rn{bWhTT})
or $\Syncbi{B_1}$.  
Nothing else changes so the agreements are maintained.

\underline{Case \rn{bCallE}} does not occur, because $C$ is let-free.

\underline{Case \rn{bLet}} does not occur, because $C$ is let-free.
\end{proof}

\subsection{Soundness of \rn{rLocEq}}\label{sec:rLocEqSound}

\[\rulerLocEq\]

Let $\eff^\leftarrow_\delta\eqdef\reads(\eff)\setminus\delta^\oplus$
as in Def.~\ref{def:loceq} of $\locEq_\delta(\flowty{P}{Q}{\eff})$. 
Let $\phi$ be a $\LocEq_\delta(\Phi)$-model, 
i.e., $\phi_0$ and $\phi_1$ are $\Phi$-models and 
$\phi_2$ satisfies $\Phi_2$ which is given by applying the $\locEq_\delta$ construction to each spec in $\Phi$ as per Def.~\ref{def:loceq}.
In symbols: $(\phi_0,\phi_1,\phi_2)\models (\Phi,\Phi,\locEq_\delta(\Phi))$.
Suppose $\ol{s}$ are the spec-only variables of $\flowty{P}{Q}{\eff}$,
and suppose $\sigma,\sigma'$ satisfy the precondition, for the unique snapshot values $\ol{v}$ and $\ol{v}'$ of $\ol{s}$ on left and right (cf.\ Lemma~\ref{lem:uniquespeconlyR}).
That is, 
\begin{equation}\label{eq:loceqpre1}
\hat{\sigma}|\hat{\sigma}'\models_\pi \Both{P}\land\Agr\eff^\leftarrow_\delta
\land \Both{(r=\lloc\land\snap(\eff))} 
\mbox{ where }
\hat{\sigma} = \extend{\sigma}{\ol{s}}{\ol{v}} \mbox{ and }
\hat{\sigma}' = \extend{\sigma'}{\ol{s}}{\ol{v}'}
\end{equation}
Notice that these assumptions entail hypotheses (i) and (ii) of Lemma~\ref{lem:rloceq},
to which we will appeal repeatedly.
We instantiate $\Phi$ in the Lemma by $\LocEq_\delta(\Phi)$,
and the initial states $\sigma|\sigma'$ satisfy the requisite precondition.

\medskip
\emph{Encap}.
Consider any trace $T$ from $\configr{\Syncbi{C}}{\sigma}{\sigma'}{\_}{\_}$.
Recall that $(\LocEq_\delta(\Phi))_0 = \Phi$ and $(\LocEq_\delta(\Phi))_1 = \Phi$.
So according to Def.~\ref{def:validR},
we must prove that the projections $U$ (resp.\ $V$) of $T$ 
(by projection Lemma~\ref{lem:bi-to-unary})
satisfy r-safe for $(\Phi,\eff,\sigma)$ (resp.\ $(\Phi,\eff,\sigma')$),
and respect for $(\Phi,M,\phi_0,\eff,\sigma)$ (resp.\ $(\Phi,M,\phi_1,\eff,\sigma')$).
These are both traces of $C$ from $P$-states, 
and $\phi_0,\phi_1$ are $\Phi$-models, so we get r-safe and respect 
by two instantiations of the premise.

\medskip
\emph{Write}.
A terminated trace via $\phi$ provides terminated unary traces via $\phi_0$ and $\phi_1$
The initial states satisfy the precondition $P$ of the premise, 
and we get the Write property directly from two instantiations of the premise.

\medskip
\emph{Safety}.
Suppose
$\configr{\Syncbi{C}}{\sigma}{\sigma'}{\_}{\_}
\biTranStar{\phi}
\configr{BB}{\tau}{\tau'}{\mu}{\mu'}
\biTrans{\phi}
\Fault$.
We can apply Lemma~\ref{lem:rloceq} to the trace ending in $BB$.
The lemma requires the trace to satisfy exactly the r-safe and respects conditions that are established above for Encap.
By Lemma~\ref{lem:rloceq} 
there are $B,\rho$ with 
$BB\equiv\Syncbi{B}$, $\rho\supseteq\pi$, $\mu=\mu'$,
\begin{equation}\label{eq:safety1}
\begin{array}{l}
\Lagree(\tau,\tau',\rho, (\freshLocs(\sigma,\tau) \union \rlocs(\sigma,\eff)\union 
\written(\sigma,\tau))\setminus\rlocs(\tau,\delta^\oplus))\\
\Lagree(\tau',\tau,\rho^{-1}, (\freshLocs(\sigma',\tau') \union \rlocs(\sigma',\eff)\union 
\written(\sigma',\tau'))\setminus\rlocs(\tau',\delta^\oplus))\\
\end{array}
\end{equation}
We show that  $\configr{BB}{\tau}{\tau'}{\mu}{\mu'}$ does not fault,
by contradiction, going by cases on the possible transition rules that yield fault. 
\begin{itemize}
\item \rn{bSyncX} would give a unary fault via $\phi_0$ or $\phi_1$, contrary to the premise.

\item 
\begin{sloppypar}
\rn{bCallX} applies if $\Fault$ is returned by $\phi_2(m)$,
and because $\phi_2$ is a context model, that means
$\tau|\tau'$ falsifies the precondition for $m$. Suppose that $\Phi(m) = \flowty{R}{S}{\effe}$.
The precondition includes $\Both{(s_{\lloc}^m=\lloc \land\snap^m(\effe))}$,
which uses spec-only variables that do not occur in $R$, $\delta$, or $\effe$,
and which can be satisfied by values determined by $\tau|\tau'$.
So for the precondition to be false there must be 
no $\rho,\ol{u},\ol{u}'$ such that 
$\rho\supseteq\pi$
and $\hat{\tau}|\hat{\tau}'\models_\rho\Both{R}\land\Agr\reads(\effe)\setminus\delta^\oplus$
where 
$\hat{\tau} = \extend{\tau}{\ol{t}}{\ol{u}}$
and $\hat{\tau}' = \extend{\tau'}{\ol{t}}{\ol{u}'}$.
From fault and relational compatibility (Def.~\ref{def:interpRel}) we have 
\[\Fault\in\phi_0(m)(\tau)\lor\Fault\in\phi_1(m)(\tau')\lor(\upsilon\in\phi_0(m)(\tau)\land\upsilon'\in\phi_1(m)(\tau'))\]
From the premise, it is not the case that $\Fault\in\phi_0(m)(\tau)$ or $\Fault\in\phi_1(m)(\tau')$, so 
there must be $\ol{u}$ and $\ol{u}'$ such that 
$\hat{\tau}\models R\land\hat{\tau}'\models R$ (with $\hat{\tau},\hat{\tau},$ as above).
(Note that $\ol{u},\ol{u}'$ are uniquely determined, by Lemma~\ref{lem:uniquespeconly}.)
Thus there is no $\rho\supseteq\pi$ with  
$\hat{\tau}|\hat{\tau}'\models_\rho\Agr\reads(\effe)\setminus\delta^\oplus$.
But from R-safe condition of the premise we know that 
$\rlocs(\tau,\effe)\subseteq\freshLocs(\sigma,\tau)\union\rlocs(\sigma,\eff)$
and $\rlocs(\tau',\effe)\subseteq\freshLocs(\sigma',\tau')\union\rlocs(\sigma',\eff)$.
So (\ref{eq:safety1}) implies $\agree(\tau,\tau',\rho,\effe\setminus(\delta,\rd\lloc))$ and 
$\agree(\tau',\tau,\rho^{-1},\effe\setminus(\delta,\rd\lloc))$ which is a contradiction. 
\end{sloppypar}

\item In case \rn{bIfX},
$B$ has the form $(\ifc{E}{D_0}{D_1});D_2$ for some $D_0,D_1,D_2$.

To show that \rn{bIfX} does not apply, 
%we go by cases on whether $D_0\equiv D_1$.
%\underline{Case $D_0\not\equiv D_1$.}
we show that $\tau(E)\neq\tau'(E)$ cannot happen, by contradiction.
Suppose $\tau(E)=\True$ and $\tau'(E)=\False$ (a symmetric argument handles the case
$\tau(E)=\False$ and $\tau'(E)=\True$).
By unary semantics we have
$\configm{\ifc{E}{D_0}{D_1};D_2}{\tau}{\mu} \trans{\phi_0} \configm{D_0;D_2}{\tau}{\mu}$
and 
$\configm{\ifc{E}{D_0}{D_1};D_2}{\tau'}{\mu} \trans{\phi_1} \configm{D_1;D_2}{\tau'}{\mu}$.
The latter step can also be taken via $\phi_0$ as it is not a call.
By  (\ref{eq:safety1}) we have
\[ \Lagree(\tau,\tau',\rho, (\freshLocs(\sigma,\tau) \union \rlocs(\sigma,\eff^\leftarrow_\delta))\setminus\rlocs(\tau,\delta^\oplus)) \]
The r-respects condition for the left step is for the collective boundary
$\unioneff{N\in(\Phi,\mu),N\neq \topm(B,M)}{\bnd(N)}$,
but because $C$ is let-free, $\mu$ is empty and $\topm(B,M)$ is $M$, so this simplifies to $\delta$.
So we have the agreement in the antecedent for r-respects,
and the other antecedent is $\agree(\tau',\tau',\delta)$ which holds.
So by r-respect from the premise, and instantiating the alternate step as the one from $\tau'$,
we can obtain $D_0;D_2\equiv D_1;D_2$.
%, contrary to our assumption $D_0\not\equiv D_1$. 
This is false, because we assume all subcommands are uniquely labeled
and thus the label on $D_0$ is distinct from the one on $D_1$.
(See footnote~\ref{fn:label} in Def.~\ref{def:wfjudge}.)

%% \underline{Case $D_0\equiv D_1$.}
%% This refutes the other condition of rule \rn{bIfX}, 
%% because in this case $\Syncbi{B}$ is 
%% $(\ifcbi{E|E}{\Syncbi{D_0}}{\Syncbi{D_0}}) ; \Syncbi{D_2}$.

%% This reasoning is why we have rule \rn{bIfY}.  

\item For \rn{bWhX}, 
$B$ has the form $\whilec{E}{D_0};D_1$
so $\Syncbi{B}$ is $\whilecbiA{E|E}{\False|\False}{D_0};\Syncbi{D_1}$.
As the alignment guards are false, rule \rn{bWhX} applies  just if 
$\tau(E)\neq \tau'(E)$.
We can show this contradicts the premise for the same reasons
as in the argument above for \rn{bIfX} in the case $D_0\not\equiv D_1$ i.e.\ the conditional branches 
differ.
We do not have to consider the situation where the branches go different 
ways but the code is the same:
if $\tau(E)=\True$ and $\tau'(E)=\False$ then
$\configm{\whilec{E}{D_0};D_1}{\tau}{\mu}
\trans{\phi_0}
\configm{D_0;\whilec{E}{D_0};D_1}{\tau}{\mu}$
and 
$\configm{\whilec{E}{D_0};D_1}{\tau'}{\mu}
\trans{\phi_1} \configm{D_1}{\tau'}{\mu}$
---the code is different, as needed to contradict r-respects in the premise.
\end{itemize}

\medskip
\emph{Post}.
Consider terminated trace 
$\configr{CC}{\sigma}{\sigma'}{\_}{\_} \biTranStar{\phi} 
\configr{\syncbi{\skipc}}{\tau}{\tau'}{\_}{\_}$, for states $\tau,\tau'$. 
We must prove $\hat{\tau},\hat{\tau}'\models_\pi\later(\Both{Q}\land\Agr\eff^\rightarrow_\delta)$, where
$\eff^\rightarrow_\delta\eqdef (\rd{(\lloc\setminus r)\Img\allfields},\Asnap(\eff))\setminus\delta$
with $\hat{\tau} = \extend{\tau}{\ol{s}}{\ol{v}}$
and $\hat{\tau}' = \extend{\tau'}{\ol{s}}{\ol{v}'}$
(with $\ol{v},\ol{v}'$ as defined following (\ref{eq:loceqpre1})).

Recall that we have $\hat{\sigma}|\hat{\sigma}'\models_\pi \Both{P}\land\Agr\eff^\leftarrow_\delta\land \Both{(s_{\lloc}=\lloc\land\snap(\eff))}$,
where $\eff^\leftarrow_\delta\eqdef\reads(\eff)\setminus\delta^\oplus$
(see (\ref{eq:loceqpre1})).
From (\ref{eq:safety1})  we get allowed dependences 
\begin{equation}\label{eq:rEmbEqP}
\AllowedDep{\sigma}{\sigma'}{\tau}{\tau'}{\eff}{\sigma}{\delta}{\pi}
\mbox{ and }\AllowedDep{\sigma'}{\sigma}{\tau'}{\tau}{\eff}{\sigma'}{\delta}{\pi^{-1}}
\end{equation}
Also, from Lemma~\ref{lem:bi-to-unary} (projection lemma), we get two terminated traces of the premise. 
Thus we have $\hat{\tau}\models Q$ and $\hat{\tau}'\models Q$. From 
$\hat{\sigma}|\hat{\sigma}'\models_\pi \Agr\eff^\leftarrow_\delta$ and
$\hat{\sigma}|\hat{\sigma}'\models_\pi \Both{P}$ and
side condition $P\models \wTor(\eff)\leq\reads(\eff)$ we
get $\hat{\sigma}|\hat{\sigma}'\models_\pi\Agr\wTor(\eff)\setminus\delta^\oplus$. 
This means, by semantics of $\Agr$ and definitions (noting that spec-only variables are not among the agreeing locations) that 
\[
\begin{array}{l}
\Lagree(\sigma,\sigma',\pi, \wlocs(\sigma,\eff)\setminus\rlocs(\sigma,\delta^\oplus))\\
\Lagree(\sigma',\sigma,\pi^{-1}, \wlocs(\sigma',\eff)\setminus\rlocs(\sigma',\delta^\oplus))
\end{array}
\]
Now using (\ref{eq:rEmbEqP}), by preservation Lemma~\ref{lem:selfframing-agreement2}, we get
\[
\begin{array}{l}
\Lagree(\tau,\tau',\rho, 
\wlocs(\sigma,\eff)\setminus\rlocs(\sigma,\delta^\oplus)\setminus\rlocs(\tau,\delta^\oplus))\\
\Lagree(\tau',\tau,\rho^{-1},
\wlocs(\sigma',\eff)\setminus\rlocs(\sigma',\delta^\oplus)\setminus\rlocs(\tau',\delta^\oplus))
\end{array}
\]
From Encap boundary monotonicity condition of the premise 
we get 
$\rlocs(\sigma,\delta)\subseteq\rlocs(\tau,\delta)$ and $\rlocs(\sigma',\delta)\subseteq\rlocs(\tau',\delta)$.
Thus the preceding agreements simplify to 
\[
\begin{array}{l}
\Lagree(\tau,\tau',\rho,\wlocs(\sigma,\eff)\setminus\rlocs(\tau,\delta^\oplus))\\
\Lagree(\tau',\tau,\rho^{-1},\wlocs(\sigma',\eff)\setminus\rlocs(\tau',\delta^\oplus))
\end{array}
\]
Furthermore, by Lemma~\ref{lem:snap}, we have
$\wlocs(\sigma,\eff)\setminus\rlocs(\tau,\delta^\oplus)=
\rlocs(\tau,\Asnap(\eff)\setminus\delta)$ and also
$\wlocs(\sigma',\eff)\setminus\rlocs(\tau',\delta^\oplus)=
\rlocs(\tau',\Asnap(\eff)\setminus\delta)$.
Thus we get 
\[
\begin{array}{l}
\Lagree(\tau,\tau',\rho, \rlocs(\tau,\Asnap(\eff)\setminus\delta)) \\
\Lagree(\tau',\tau,\rho^{-1}, \rlocs(\tau',\Asnap(\eff)\setminus\delta))
\end{array}
\]
This means $\hat{\tau}|\hat{\tau}'\models_\rho\Agr\Asnap(\eff)\setminus\delta$.

Since $\freshLocs(\tau,\upsilon)=\rlocs(\upsilon,\rd{(\lloc\setminus r)\Img\allfields})$ and
$\freshLocs(\tau',\upsilon')=\rlocs(\upsilon',\rd{(\lloc\setminus r)\Img\allfields})$, 
we can use the agreements on fresh locations given by 
(\ref{eq:rEmbEqP}) to get 
$\hat{\tau}|\hat{\tau}'\models_\rho\Agr (\rd{(\lloc\setminus r)\Img\allfields})\setminus\delta$. 

Combining what is proved above and using $\rho$ as witness of the 
existential in the semantics of $\later$, 
we conclude the proof of Post:
$\hat{\tau}|\hat{\tau}'\models_\pi\later(\Both{Q}\land\Agr (\rd{(\lloc\setminus r)\Img\allfields}, 
\Asnap(\eff)\setminus\delta))$.

\medskip
\emph{R-safe}. By projection Lemma~\ref{lem:bi-to-unary}(c) there are unary executions that take the same unary steps. The R-safe condition from the premise applies on both sides and yields
R-safety for the conclusion.

\subsection{Soundness of \rn{rSOF}}

\newcommand{\hacky}{^{\ol{t}|\ol{t}}_{\ol{u}|\ol{u}'}} % ALERT - notation for a specific substitution

\[ \rulerSOF \]

Before studying the following, 
readers are advised to be familiar with Sections~\ref{sec:app:lemrloceq} 
and~\ref{sec:rLocEqSound}. 

To show soundness of \rn{rSOF}, suppose the side conditions hold and the premise of the rule is valid:
\begin{equation}\label{eq:rSOFpremise}
\LocEq_\delta(\Phi,\Theta) \models_M \Syncbi{C} : \locEq_\delta(\flowty{P}{Q}{\eff})
\end{equation}
We must prove validity of the conclusion:
\begin{equation}\label{eq:rSOFconcl}  \LocEq_\delta(\Phi), (\LocEq_\delta(\Theta)\conjInv\N)  \models_M 
    \Syncbi{C} \: : \: \locEq_\delta (\flowty{P}{Q}{\eff}) \conjInv \N
\end{equation}
To that end, consider an arbitrary model $\phi^+$ of
the relational context $\LocEq_\delta(\Phi),\LocEq_\delta(\Theta)\conjInv\N$.
To make use of the premise we define a model, $\phi^-$, of $\LocEq_\delta(\Phi,\Theta)$.

For $m$ in $\Phi$, the definition is unchanged: $\phi^-_i(m)=\phi^+_i(m)$ for $i\in\{0,1,2\}$.
For methods $m$ of $\Theta$, we first define $\phi_2^-(m)$.
For that, we need some notation.
Suppose $\Theta(m) = \flowty{R}{S}{\effe}$. Let $\R$ be the local equivalence precondition
\begin{equation}\label{eq:rSOFloceq}
   \R\eqdef\Both{R}\land\Agr\reads(\effe)\setminus\delta^\oplus\land 
\Both{(s_{\lloc}^m=\lloc \land\snap^m(\effe))} 
\end{equation}
Let $\ol{t}$ be the spec-only variables, including $s_{\lloc}^m$ and the $\snap^m$ ones.
Note that $\N$ depends on no spec-only variables, 
by the side condition that it is framed by dynamic boundary $\bnd(N)$.
For any states $\tau$ and $\tau'$, define  
\[
\phi_2^-(m)(\tau|\tau')\eqdef\left\{
\begin{array}{ll}
\{\Fault\}&\all{\pi,\ol{u},\ol{u}'}{\tau|\tau'\models_\pi\neg\R\hacky} \\
\emptyset&
(\some{\pi,\ol{u},\ol{u}'}{\tau|\tau'\models_\pi\R\hacky}) 
    \land(\all{\pi,\ol{u},\ol{u}'}{\tau|\tau'\models_\pi\R\hacky \imp \tau|\tau'\not\models_\pi\N})
\\ 
\phi_2^+(m)(\tau|\tau') & \some{\pi,\ol{u},\ol{u}'}{\tau|\tau'\models_\pi\R\hacky\land\N}
\end{array}
\right.
\]
One might hope that $(\phi_0^+,\phi_1^+,\phi_2^-)$ is a model for $\LocEq_\delta(\Phi,\Theta)$ but this may fail 
for $m$ in $\Phi$ 
if $\phi_0^+(m)(\tau)$ or $\phi_1^+(m)(\tau')$ is non-empty for $\tau|\tau'$ that satisfy $\R$ but not 
$\N$---because then the relational compatibility condition for pre-model fails (Definition~\ref{def:interpRel},
which is a pre-requisite for Definition~\ref{def:ctxinterpRel}).

To solve this problem, we define $\phi_0^-(m)$ and $\phi_1^-(m)$ like $\phi_0^+(m)$ and $\phi_1^+(m)$ but yielding empty outcome sets for such $\tau,\tau'$.
To see why this works we make the following observations about the definitions of pre-model and model for unary specs.
For any pre-model $\phi(m)$ and states $\tau,\sigma$, 
if $\tau\in\phi(m)(\sigma)$ and $\phi'(m)$ is defined identically to $\phi(m)$ except that
$\phi'(m)(\sigma) = (\phi(m)(\sigma))\setminus\{\tau\}$, then $\phi'$ is a pre-model.  
Moreover, if $\phi(m)$ is a context model for some spec and $\sigma$ satisfies the precondition, then $\phi'$ is a context model.
Now, for any $\tau$, define 
$\phi_0^-(m)(\tau) \eqdef \emptyset$ if there is $\tau'$ such that
the conditions of the second case for $\phi_2^-$ hold for $\tau|\tau'$, that is:
\[ (\some{\pi,\ol{u},\ol{u}'}{\tau|\tau'\models_\pi\R\hacky})
\mbox{ and }
(\all{\pi,\ol{u},\ol{u}'}{\tau|\tau'\models_\pi\R\hacky \imp\tau|\tau'\not\models_\pi\N})
\]
Otherwise define $\phi_0^-(m)(\tau) \eqdef \phi_0(m)(\tau)$.
The displayed condition implies that $\tau$ satisfies the unary precondition $R$, 
so $\phi_0^-(m)$ is a model for $\Theta(m)$ as observed above.
Define $\phi_1^-(m)$ the same way but existentially quantifying the left state:
$\phi_1^-(m)(\tau) \eqdef \emptyset$ if there is $\tau$ such that
$(\some{\pi,\ol{u},\ol{u}'}{\tau|\tau'\models_\pi\R\hacky})$
and $(\all{\pi,\ol{u},\ol{u}'}{\tau|\tau'\models_\pi\R\hacky \imp \tau|\tau'\not\models_\pi\N})$;
otherwise define $\phi_1^-(m)(\tau) \eqdef \phi_1(m)(\tau)$.
We leave it to the reader to check that $(\phi^-_0,\phi^-_1,\phi_2^-)$ satisfies all the conditions to be a 
relational pre-model and to be a context model of $\LocEq_\delta(\Phi,\Theta)$.
The latter means $\phi^-_0$ and $\phi^-_1$ are $(\Phi,\Theta)$-models,
and $\phi_2^-(m)$ models $\locEq_\delta(\Phi,\Theta)(m)$ for all $m$.

Now we return to the proof of validity of the conclusion, (\ref{eq:rSOFconcl}). 
Having fixed an arbitrary context model $\phi^+$ we now consider any
$\sigma,\sigma',\pi$ that satisfy the precondition of the conclusion,
i.e., the precondition of $\locEq_\delta (\flowty{P}{Q}{\eff}) \conjInv \N$.
That is, we assume
\begin{equation}\label{eq:RSOF0}
 \hat{\sigma}|\hat{\sigma}'\models_\pi 
\Both{P}\land\Agr\reads(\eff)\setminus\delta^\oplus\land 
\Both{(s_{\lloc}=\lloc\land\snap(\eff))}\land\N 
\end{equation}
where $\ol{s}$ are the spec-only variables (which are the same on both sides of these specs),
$\hat{\sigma} = \extend{\sigma}{\ol{s}}{\ol{v}}$,
$\hat{\sigma}' = \extend{\sigma'}{\ol{s}}{\ol{v}'}$ for some $\ol{v},\ol{v}'$.
(Recall that $\ol{v},\ol{v}'$ are uniquely determined, by Lemma~\ref{lem:uniquespeconlyR}.)

To finish the soundness proof, we need the following claim involving $\sigma,\sigma',\pi$ and the context model $\phi^-$ derived from $\phi^+$.
\begin{quote}
\textbf{Claim.} 
If $\configr{\Syncbi{C}}{\sigma}{\sigma'}{\_}{\_}
\biTranStar{\phi^+}
\configr{BB}{\tau}{\tau'}{\mu}{\mu'}$
then there are $B$ and $\rho$ such that 
\begin{list}{}{}
\item[(a)] 
$\configr{\Syncbi{C}}{\sigma}{\sigma'}{\_}{\_} 
\biTranStar{\phi^-}\configr{BB}{\tau}{\tau'}{\mu}{\mu'}$
\item[(b)] $\tau|\tau'\models_\rho\N$
\item[(c)] $\rho\supseteq\pi$ and $BB\equiv\Syncbi{B}$  and $\mu=\mu'$ 
\item[(d)] $\Lagree(\tau,\tau',\rho, (\freshLocs(\sigma,\tau) \union \rlocs(\sigma,\eff)\union 
\written(\sigma,\tau))\setminus\rlocs(\tau,\delta^\oplus))$, and 
\item[(e)]
$\Lagree(\tau',\tau,\rho^{-1},(\freshLocs(\sigma',\tau') \union \rlocs(\sigma',\eff)\union 
\written(\sigma',\tau'))\setminus\rlocs(\tau',\delta^\oplus))$.
\end{list}
\end{quote}
Item (a) says a trace via the conclusion's $\phi^+$ can be taken via the premise's $\phi^-$.
Item (b) says $\N$ holds at every step (outside context calls).
Items (c), (d), and (e) are the same as the conclusions (v), (vi), and (vii) of the 
lockstep alignment Lemma~\ref{lem:rloceq}, for refperm $\rho$ that additionally truthifies $\N$ according to item (b).  

We do not directly apply Lemma~\ref{lem:rloceq} in the following argument,
%% \footnote{In
%%   a previous version of this paper (\url{https://arxiv.org/abs/1910.14560v3}),
%%   we did directly apply the lemma to prove a variation of \rn{rSOF}.  There, the Lemma
%%   establishes $\later\N$ separately from establishing $\later$ of other conjunctions, and
%%   the rule relies on intricate agreement compatibility conditions connecting $\later\N$
%%   with method preconditions and with the main postcondition.   Those agreement 
%%   compatibility conditions are difficult to satisfy for many coupling relations.

%% TODO - drop the reference and somewhere in the paper give quick hints about later/and:
%% agree compat is hard to achieve, and separation of heap locs isn't sufficient,
%% sep of refperms can suffice but too strong for practical cases;
%% can do a version of rSOF using it but for framing better to use agree mono.
%%   } % footnote
because it gives us no good way to establish $\tau|\tau'\models_\rho\N$.
However, we will establish (c)--(e) by similar arguments to the proof (Section~\ref{sec:app:lemrloceq}) of Lemma~\ref{lem:rloceq},
in which the conclusions (v)--(vii) are proved by induction on a given trace.
In short, we will apply the induction step of that proof.
Whereas the lemma connects an initial $\pi$ with a refperm $\rho\supseteq\pi$ for 
a given reachable configuration,
the proof of the induction step of the lemma does exactly what we need: Given a current $\rho$ with $\rho\supseteq\pi$, it yields a $\dot{\rho}$ with $\dot{\rho}\supseteq\rho$, for the next step of the trace.
We can reason the same way, for (c)--(e), but also add that $\dot{\rho}$ satisfies $\N$.

One could factor out the induction step of the lemma as a separate result, and then apply
it directly here.  We refrain from spelling that out explicitly, but we do need to be 
clear how we are instantiating the assumptions of Lemma~\ref{lem:rloceq}.
For the unary spec $\Psi$ in the Lemma we take $(\Phi,\Theta)$.
For the relational spec $\Phi$ in the Lemma we take $(\LocEq_\delta(\Phi),\LocEq_\delta(\Theta))$,
which is the same as $\LocEq_\delta(\Phi,\Theta)$.
For the context model $\phi$ we take $\phi^-$.
So we have assumption (i) of the Lemma.
We also have (ii), as direct consequence of (\ref{eq:RSOF0}).
For (iii), we will consider a trace via $\phi^-$ given by (a) in the Claim.
For (iv), i.e., r-safety and respect for that trace, we will appeal
to the premise (\ref{eq:rSOFpremise}).

\emph{Proof of Claim,}
by induction on steps.

\textbf{Base Case.} For initial configuration $\configr{\Syncbi{C}}{\sigma}{\sigma'}{\_}{\_}$, 
take $\rho:=\pi$.
We have 
$\sigma|\sigma'\models_\pi\N$ by assumption
(\ref{eq:RSOF0}); the rest follows.

\textbf{Induction Case.}
Suppose 
\begin{equation}\label{eq:rSOFtrace}
\configr{\Syncbi{C}}{\sigma}{\sigma'}{\_}{\_} 
\biTranStar{\phi^+} 
\configr{BB}{\tau}{\tau'}{\mu}{\mu'}
\biTrans{\phi^+}\configr{DD}{\upsilon}{\upsilon'}{\nu}{\nu'}
\end{equation}
By induction hypothesis there is $\rho$ 
such that the conditions (a)--(e) of the Claim hold for the configuration with $\tau,\tau'$ ---including 
$\rho\supseteq\pi$, 
$\tau|\tau'\models_\rho\N$, 
$BB$ has the form $\Syncbi{B}$ for some $B$,
and 
$\configr{\Syncbi{C}}{\sigma}{\sigma'}{\_}{\_} \biTranStar{\phi^-} 
\configr{BB}{\tau}{\tau'}{\mu}{\mu'}$.
We must show there is $\dot{\rho}$ such that $\dot{\rho}\supseteq\pi$,
$\upsilon|\upsilon'\models_{\dot{\rho}}\N$,
$\configr{\Syncbi{B}}{\tau}{\tau'}{\mu}{\mu'}\biTrans{\phi^-} 
\configr{DD}{\upsilon}{\upsilon'}{\nu}{\nu'}$,
and the other conditions of the Claim for $\dot{\rho},\upsilon,\upsilon'$.
We write (\.{a}), (\.{b}) etc. to indicate those conditions instantiated for 
$\dot{\rho},\upsilon,\upsilon'$.

%% \dn{this para gets folded into the cases} 
%% Let $T$ be the trace $\configr{\Syncbi{C}}{\sigma}{\sigma'}{\_}{\_} \biTranStar{\phi^-} 
%% \configr{\Syncbi{B}}{\tau}{\tau'}{\mu}{\mu'}$.
%% The premise (\ref{eq:rSOFpremise}) applies to $T$, because $\sigma,\sigma'$ satisfy its precondition
%% and $\phi^-$ is a $\LocEq_\delta(\Phi,\Theta)$-model.  
%% By the premise we have that the left and right projections of $T$ 
%% satisfy r-safe for $(\Phi,\Theta),\eff$ and $\sigma$ (resp.\ $\sigma'$).
%% (Note that context $\Phi,\Theta$ is both the left and right unary part of 
%% $\LocEq_\delta(\Phi,\Theta)$.)
%% The left and right projections also satisfy respect for $((\Phi,\Theta),M,\phi_0^-,\eff,\sigma)$ and 
%% $((\Phi,\Theta),M,\phi_1^-,\eff,\sigma')$ respectively.
%% Thus we have the assumption (iv) of Lemma~\ref{lem:rloceq}.
%% (Aside: By direct application of Lemma~\ref{lem:rloceq} 
%% for $T$ we get $BB\equiv\Syncbi{B}$ and $\mu=\mu'$, 
%% so those need not be included in the induction hypothesis (c) but we keep them for clarity.)

To find $\dot{\rho}$ and show the conditions of the Claim for $\upsilon,\upsilon'$
we distinguish three cases:

\underline{Case $\Active(B)$ is not a context call.} 
Because the step is not a call, it is independent of model, so we have
\begin{equation}\label{eq:RSOF2}
\configr{\Syncbi{B}}{\tau}{\tau'}{\mu}{\mu'}\biTrans{\phi^-}\configr{DD}{\upsilon}{\upsilon'}{\nu}{\nu'}
\end{equation}
which takes care of part (\.{a}) of the Claim.
Moreover, this together with (\ref{eq:RSOF0}) 
lets us instantiate the premise (\ref{eq:rSOFpremise})
so (by Encap) we have that the left and right projections 
of the whole trace (\ref{eq:rSOFtrace}) 
satisfy respect for 
$((\Phi,\Theta),M,\phi_0^-,\eff,\sigma)$ and 
$((\Phi,\Theta),M,\phi_1^-,\eff,\sigma')$ respectively.
Thus we have the assumption (iv) of Lemma~\ref{lem:rloceq}
applied to the trace (\ref{eq:rSOFtrace}). 
By direct application of the Lemma we get that $\nu=\nu'$ 
and there is some $D$ with $DD\equiv\Syncbi{D}$.
Direct application would also yield agreements for some $\dot{\rho}\supseteq\pi$,
but that is not enough.  Instead we apply the induction step of the Lemma's proof,
which yields $\dot{\rho}$ such that $\dot{\rho}\supseteq\rho$ 
and (\.{d}) and (\.{e}) hold.  
Finally, from the Encap condition of premise of the rule, we also know that unary steps on left and right of (\ref{eq:RSOF2}) w-respect $\bnd(N)$, so we get $\agree(\tau,\upsilon,\bnd(N))$ and $\agree(\tau',\upsilon',\bnd(N))$. 
So from side condition $\models \fra{ \bnd(N) | \bnd(N) }{ \N }$, 
by Def.~\ref{def:frmAgreeRel} 
of the relational framing judgment,
using (b), we get $\upsilon|\upsilon'\models_\rho\N$.
By $\dot{\rho}\supseteq\rho$ and the side condition $\N\imp\always\N$ of \rn{rSOF},
we get $\upsilon|\upsilon'\models_{\dot{\rho}}\N$, proving (\.{b})
and concluding the induction step for this case.

Note that the induction step in the proof of Lemma~\ref{lem:rloceq} goes by cases on transition rules. The preceding paragraph covered all the transition rules except for context call.

\underline{Case $\Active(B)$ is a context call to some $m$ in $\Phi$.}
The step can be taken via $\phi^-$ because $\phi^{-}_{2}(m)$ is defined to be $\phi^{+}_{2}(m)$, so we have (\.{a}). 
As in the preceding case, 
we can apply the induction step of Lemma~\ref{lem:rloceq} to get $\dot{\rho}\supseteq\rho$ with (\.{c})--(\.{e}).
As in the preceding case, 
we appeal to w-respect for premise (\ref{eq:rSOFpremise}),
and $\models \fra{ \bnd(N) | \bnd(N) }{ \N }$, to get (\.{b}).

In our appeal to the proof of Lemma~\ref{lem:rloceq},
we are here using the cases of transition rules \rn{bCallS} and \rn{bCall0}.

\underline{Case $\Active(B)$ is a context call to some $m$ in $\Theta$.}
So $B$ has the form $B\equiv m();B_2$ for some $B_2$.
The transition can go by either \rn{bCall0} or \rn{bCallS}.  In the case of \rn{bCall0}, we get the Claim directly from the induction hypothesis: taking $\dot{\rho}:=\rho$
we get (\.{a})--(\.{e}) from (a)--(e).
  
Now consider the case of \rn{bCallS}.
Suppose $\Theta(m) = \flowty{R}{S}{\effe}$
and $\ol{t}$ is spec-only variables of $R$ and of the snapshot variables 
of $\locEq_\delta( \flowty{R}{S}{\effe})$ tagged for $m$.
Since we are in the case \rn{bCallS}, the precondition of $m$ for $\phi^+$ holds, for some refperm; $\phi^-(m)$ is defined the same way (last case in its definition) and the transition can be taken via $\phi^-$, so we have (\.{a}).  
It remains to find some $\dot{\rho}\supseteq\pi$ satisfying (\.{b})--(\.{e}) for $\upsilon,\upsilon'$. 
For (\.{c}), by \rn{bCallS} the method environments are unchanged and
$DD$ has the form $\Syncbi{B_2}$.

Let us spell out what it means that the precondition of $m$ for $\phi^+$
(i.e., the precondition of $\locEq_\delta( \flowty{R}{S}{\effe})$)
holds for some $\rho_1$: we have 
 \begin{equation}\label{eq:RSOFpre}
%\hat{\tau}|\hat{\tau}'\models_{\rho_1}(\Both{R}\land\Agr\reads(\effe)\setminus\delta^\oplus\land ---- equivalent but more clear:
\hat{\tau}|\hat{\tau}'\models_{\rho_1}(\Both{R}\land\Agr\effe^\leftarrow_\delta \land
\Both{(s_{\lloc}^m=\lloc\land\snap^m(\effe))})\hacky  \land \N  
\end{equation}
where $\hat{\tau} \eqdef \extend{\tau}{\ol{s}}{\ol{v}}$ 
and  $\hat{\tau}' \eqdef \extend{\tau}{\ol{s}}{\ol{v}'}$ 
where $\ol{v},\ol{v}'$ are the unique values for the spec-only variables $\ol{s}$ defined in connection with (\ref{eq:RSOF0}),
and  $\ol{u},\ol{u}'$ are the unique values for the spec-only variables $\ol{t}$
for $\Theta(m)$.
We can write $\N$ outside the substitutions, because it has no spec-only variables,
but this is not important.
What is important is that  $\ol{v},\ol{v}',\ol{u},\ol{u}'$ are uniquely determined, independent of the refperm, by Lemma~\ref{lem:uniquespeconlyR}.
Let $\widehat{\tau} \eqdef \extend{\hat{\tau}}{\ol{t}}{\ol{u}}$ 
and $\widehat{\tau}' \eqdef \extend{\hat{\tau}}{\ol{t}}{\ol{u}'}$.
So (\ref{eq:RSOFpre}) can be written 
\begin{equation}\label{eq:RSOFpre1}
\widehat{\tau}|\widehat{\tau}'\models_{\rho_1}
   \Both{R}\land\Agr \effe^\leftarrow_\delta \land \Both{(s_{\lloc}^m=\lloc\land\snap^m(\effe))} \land \N 
\end{equation}
Now, $\Both{R} \land \Both{(s_{\lloc}^m=\lloc\land\snap^m(\effe))}$
 is refperm independent.  So using induction hypothesis (b) we have
$\widehat{\tau}|\widehat{\tau}'\models_{\rho} \Both{R} \land \Both{(s_{\lloc}^m=\lloc\land\snap^m(\effe))} \land\N$.
We can we get 
$\widehat{\tau}|\widehat{\tau}'\models_{\rho} \Agr \effe^\leftarrow_\delta $
from induction hypothesis (d) and (e), as follows.
First, we have Encap and r-safety for the trace up to $\tau,\tau'$,
by induction hypothesis (a) and the premise.  
Now $\effe^\leftarrow_\delta $ is $\reads(\effe)\setminus\delta^\oplus$,
i.e., $\reads(\effe)\setminus(\delta,\rd{\lloc})$.
By r-safety we have 
$\rlocs(\tau,\effe^\leftarrow_\delta)\subseteq 
(\freshLocs(\sigma,\tau) \union \rlocs(\sigma,\eff))\setminus\rlocs(\tau,\delta^\oplus))$
and 
$\rlocs(\tau',\effe^\leftarrow_\delta)\subseteq 
(\freshLocs(\sigma',\tau') \union \rlocs(\sigma',\eff))\setminus\rlocs(\tau',\delta^\oplus))$.
So by semantics of $\Agr \effe^\leftarrow_\delta$ and induction hypothesis (d) and (e)
we get $\widehat{\tau}|\widehat{\tau}'\models_{\rho} \Agr \effe^\leftarrow_\delta $.

Having established that the precondition (\ref{eq:RSOFpre1}) holds for $\rho_1:=\rho$,
we can instantiate the spec of $m$ with $\rho$ and obtain the postcondition
(in accord with Def.~\ref{def:ctxinterpRel} of relational context model):
\[  \widehat{\upsilon}|\widehat{\upsilon}'\models_{\rho} 
   \later( \Both{S}\land \Agr\effe^\rightarrow_\delta \land \N )
\] 
By semantics, this implies there is $\dot{\rho}\supseteq\rho$ with $\upsilon|\upsilon'\models_{\dot{\rho}} \Both{S}\land \Agr\effe^\rightarrow_\delta \land \N$.
So we have (\.{b}) and (\.{c}). 
Finally, $\dot{\rho}$ satisfies the agreements of (\.{d}) and (\.{e}); this follows from $\upsilon|\upsilon'\models_{\dot{\rho}} \Agr\effe^\rightarrow_\delta$
for reasons that are spelled out in detail in proving the induction step of Lemma~\ref{lem:rloceq} in the case of \rn{bCallS},
starting around the displayed formula (\ref{eq:rloceqEta}).

\medskip

Having proved the Claim, we prove validity of the conclusion (\ref{eq:rSOFconcl}) of \rn{rSOF}.

\medskip
\emph{Safety}.
Suppose 
$\configr{\Syncbi{C}}{\sigma}{\sigma'}{\_}{\_}
\biTranStar{\phi^+}
\configr{BB}{\tau}{\tau'}{\mu}{\mu'}$.
We show by contradiction the latter configuration cannot fault.

\underline{Case: fault by a non-call step.}
Then the faulting step can also be taken via $\phi^-$, 
and it is reached via $\phi^-$ owing to the Claim (a), 
but a faulting trace via $\phi^-$ contradicts the premise (\ref{eq:rSOFpremise}).  

\underline{Case: fault by a context call to some $m$ in $\Phi$.}  
Then the step can also be taken via $\phi^-$, again contradicting the premise. 

\underline{Case: fault by a context call to some $m$ in $\Theta$.}  
Let the spec of $m$ be $\flowty{R}{S}{\effe}$, so the relational precondition  
is $\R\land \N$ where $\R$ is given by (\ref{eq:rSOFloceq}).
Because $\phi^+$ is a context model, the call only faults 
if there are no $\dot{\rho},\ol{u},\ol{u}'$ such that 
$\tau|\tau'\models_{\dot{\rho}} \R\hacky \land \N$
(see transition rule \rn{bCallX}).
By the snapshot uniqueness Lemma~\ref{lem:uniquespeconlyR},
values $\ol{u},\ol{u}'$ exist and are uniquely determined by $\tau,\tau'$.  
So the call only faults if there is no $\dot{\rho}$ such that $\hat{\tau}|\hat{\tau}'\models_{\dot{\rho}} \R\land \N$
where $\hat{\tau},\hat{\tau}'$ are the states extended with 
$\ol{u},\ol{u}'$ for the snapshot variables.
But we have $\rho$ and 
can show $\hat{\tau}|\hat{\tau}'\models_\rho \R\land \N$ as follows.
We have $\hat{\tau}|\hat{\tau}'\models_\rho \N$ by Claim (b).
We have 
$\hat{\tau}|\hat{\tau}'\models_\rho \Both{(s_{\lloc}^m=\lloc \land\snap^m(\effe))}$ 
in accord with our choice of the correct snapshot values.
To show the conjunct 
$\hat{\tau}|\hat{\tau}'\models_\rho  \Both{R}$,
we can apply the premise, in particular Safety:
there must be some refperm for which 
$\hat{\tau}|\hat{\tau}'$ satisfy $\Both{R}$,
because otherwise the call would fault via $\phi^-$, contrary to the premise (\ref{eq:rSOFpremise}).
Now we get $\hat{\tau}|\hat{\tau}'\models_\rho \Both{R}$
because $\Both{R}$ is refperm independent.  
It remains to show the conjunct 
$\hat{\tau}|\hat{\tau}'\models_\rho \Agr \effe^\leftarrow_\delta$, 
that is, $\hat{\tau}|\hat{\tau}'\models_\rho \Agr\reads(\effe)\setminus\delta^\oplus$.
We have r-safety for the trace up to $\tau,\tau'$, by Claim (a) and the premise.  
By r-safety we have 
$\rlocs(\tau,\effe^\leftarrow_\delta)\subseteq 
(\freshLocs(\sigma,\tau) \union \rlocs(\sigma,\eff))\setminus\rlocs(\tau,\delta^\oplus))$
and 
$\rlocs(\tau',\effe^\leftarrow_\delta)\subseteq 
(\freshLocs(\sigma',\tau') \union \rlocs(\sigma',\eff))\setminus\rlocs(\tau',\delta^\oplus))$.
So by Claim (d) and (e)
we get $\hat{\tau}|\hat{\tau}'\models_{\rho} \Agr \effe^\leftarrow_\delta $.

\medskip
\emph{Post}.
For all $\tau,\tau'$ such that 
$\configr{\Syncbi{C}}{\sigma}{\sigma'}{\_}{\_} \biTranStar{\phi^+} \configr{\syncbi{\skipc}}{\tau}{\tau'}{\_}{\_}$,
we must show
$\tau|\tau'\models_\pi \later(\Both{Q}\land\Agr\eff^\rightarrow_\delta\land\N)$.
Applying the Claim to this trace we obtain $\rho$ such that conditions (a)--(e) 
hold for $\tau,\tau'$.
We will show
$\tau|\tau'\models_\rho \Both{Q}\land\Agr\eff^\rightarrow_\delta\land\N$;
our obligation then follows by semantics of $\later$,
using $\rho\supseteq\pi$ from (b).

We have $\tau|\tau'\models_\rho \N$ by (b).  
By (a) we can instantiate the premise (\ref{eq:rSOFpremise}) which yields 
$\tau|\tau'\models_\pi \later(\Both{Q}\land\Agr\eff^\rightarrow_\delta)$.
This implies 
$\tau|\tau'\models_\rho \Both{Q}$
because $\Both{Q}$ is refperm independent.
Finally, we get 
$\tau|\tau'\models_\rho \Agr\eff^\rightarrow_\delta$ 
as a consequence of (d) and (e) by essentially the same argument as 
the one spelled out in the proof of Post for rule \rn{rLocEq} (Sect.~\ref{sec:rLocEqSound}).

\medskip
\emph{Write, R-safe, and Encap}.  These are obtained directly from the premise,
using the Claim. 
Note that $\Phi,\Theta\conjInv \N$ has the same methods, and thus the same modules, as $\Phi,\Theta$ has, so the Encap conditions have exactly the same meaning for the conclusion of the rule as for the premise.

\subsection{Soundness of \rn{rPoss}, \rn{rDisj}, and \rn{rConj}}

For \rn{rPoss}, 
assume validity of the premise: $\Phi\rHVflowtr{}{M}{\P}{CC}{\Q}{\eff|\eff'}$.
To prove validity of the conclusion
$\Phi\rHVflowtr{}{M}{\later\P}{CC}{\later\Q}{\eff|\eff'}$,
consider any $\Phi$-model $\phi$.
Consider any $\sigma,\sigma',\pi$ such that $\sigma|\sigma'\models_\pi\later\P$.
By formula semantics, there is $\rho\supseteq\pi$ such that $\sigma|\sigma'\models_\rho\P$.
The Safety, Write, and Encap conditions now follow by instantiating the premise
with $\phi$ and $\rho$.
For Post, the premise yields that for terminal state pair $\tau|\tau'$ we have
$\tau|\tau'\models_\rho \Q$.  
This implies $\tau|\tau'\models_\pi\later \Q$ since $\rho\supseteq\pi$.  

For \rn{rDisj}, 
suppose $\phi$ is a $\Phi$-model and 
suppose 
$\sigma|\sigma'\models_\pi \P_0\lor\P_1$.
By semantics of formulas, 
either $\sigma|\sigma'\models_\pi \P_0$
or $\sigma|\sigma'\models_\pi \P_1$,
so we can instantiate one of the premises using $\phi$.
It is straightforward to check that the conditions of Def.~\ref{def:validR} 
for the conclusion follow directly from the premise.
Note that the propositional connectives have classical semantics in relational formulas, as they do in unary formulas.

For \rn{rConj} the argument is similar.

\subsection{Soundness of \rn{rFrame}}

All conditions except Post are easy consequences of the premise.
For Post, suppose $\sigma|\sigma'\models_\pi \P\land\R$ 
and $\configr{CC}{\sigma}{\sigma'}{\_}{\_} \biTranStar{\phi}
\configr{\syncbi{\skipc}}{\tau}{\tau'}{\_}{\_}$.
By Write we have $\sigma\allowTo\tau\models\eff$ and 
$\sigma'\allowTo\tau'\models\eff'$ (as well as $\sigma\successorTo\tau$
and $\sigma'\successorTo\tau'$ of course).
By the rule's condition
$\P\land\R \imp \leftF{\ind{\effe}{\eff}} \land \rightF{\ind{\effe'}{\eff'}}$,
we can use fact (\ref{eq:sepagree}) to get 
$\agree(\sigma,\tau,\effe)$ and $\agree(\sigma',\tau',\effe')$.
So by $\P\models \fra{\effe|\effe'}{\R}$ and
% Def.~\ref{def:frmAgreeRel}
semantics of this judgment
we get 
$\tau|\tau'\models_\pi \R$.  
We have $\tau|\tau'\models_\pi \Q$ by Post  for the premise.

\subsection{Soundness of \rn{rEmb} and \rn{rEmbS}}

We prove \rn{rEmb} (Fig.~\ref{fig:proofrulesR}).  The argument for \rn{rEmbS} (Fig.~\ref{fig:proofrulesRapp}) is similar.

Suppose $\Phi_0\HVflowtr{}{M}{P}{C}{Q}{\eff}$
and $\Phi_1\HVflowtr{}{M}{P'}{C'}{Q'}{\eff'}$.
To show validity of the conclusion,  $\Phi\rHVflowtr{}{M}{\leftF{P}\land\rightF{P'}}{\splitbi{C}{C'}}{\leftF{Q}\land\rightF{Q'}}{\eff|\eff'}$,
consider any $\Phi$-model $\phi$ and any $\sigma,\sigma',\pi$ such that
$\sigma|\sigma'\models_\pi \leftF{\subst{P}{\bar{s}}{\bar{v}}}\land\rightF{\subst{P'}{\bar{s'}}{\bar{v'}}}$.
By biprogram semantics, $\splitbi{C}{C'}$ goes by dovetailed steps of $C$ via $\phi_0$ (rule \rn{bComL}) and steps of $C'$ via $\phi_1$ (rules \rn{bComR} and \rn{bComR0}).
All reached configurations are in the bi-com form. 
For Safety, observe that if fault is reached it is by \rn{bComLX} or \rn{bComRX},
so by projection we obtain a faulting trace either of $C$ or of $C'$, contrary to the premises.
For Post and Write, 
suppose $\configr{\splitbi{C}{C'}}{\sigma}{\sigma'}{\_}{\_} \biTranStar{\phi}
\configr{\syncbi{\skipc}}{\tau}{\tau'}{\_}{\_}$.
Then by projection we obtain terminated traces (via $\phi_0$ and $\phi_1$ respectively) to which the premises apply. This yields $\sigma\allowTo\tau \models \eff$ and $\sigma'\allowTo\tau' \models \eff'$ (proving Write) and $\tau\models \subst{Q}{\bar{s}}{\bar{v}}$ and $\tau'\models \subst{Q'}{\bar{s'}}{\bar{v'}}$ so that $\tau|\tau'\models_\pi \leftF{\subst{Q}{\bar{s}}{\bar{v}}}\land\rightF{\subst{Q'}{\bar{s'}}{\bar{v'}}}$ (proving Post).
For every trace from $\configr{\splitbi{C}{C'}}{\sigma}{\sigma'}{\_}{\_}$ consider its projections which are unary traces from $\configm{C}{\sigma}{\_}$ via $\phi_0$ and $\configm{C'}{\sigma'}{\_}$ via $\phi_1$. Then both R-safe and Encap follow using R-safe and Encap for the unary traces to which the premises apply.

%% For Read(i) (in Def.~\ref{def:validR1}),
%% suppose in addition to the above we have
%% $\configr{\splitbi{C}{C'}}{\dot{\sigma}}{\sigma'}{\_}{\_} \biTranStar{\phi}
%% \configr{\syncbi{\skipc}}{\dot{\tau}}{\tau'}{\_}{\_}$, where $\dot{\sigma}\models P$ 
%% (we can ignore the refperm because the precondition 
%% $\leftF{P}\land\rightF{P'}$ is refperm independent).
%% By projection we have 
%% $\configm{C}{\dot{\sigma}}{\_} \tranStar{\phi_0} \configm{\skipc}{\dot{\tau}}{\_}$.
%% Instantiating the premise for $C$, the Read condition gives
%% $\dot{\sigma},\sigma \allowDep \dot{\tau},\tau \models\eff$ as required.
%% (Indeed, we also have
%% $\sigma,\dot{\sigma} \allowDep \tau,\dot{\tau}\models\eff$ from the premise.)
%% For Read(ii), the argument is symmetric.  

\subsection{Soundness of \rn{rCall}}

%% \[
%% \inferrule*[left=rCall]{
%%   \Phi_2(m) = \rflowty{\P}{\Q}{\eff|\eff'} 
%% }{ 
%%   \Phi \rHPflowtr{}{\emptymod}{\P}{\syncbi{m()}}{\Q}{\eff|\eff'} 
%% }
%% \]

\[\rulerCall\]
Let the current module be $N$ in all three judgments. 

Suppose $\Phi_2(m)$ is $\rflowtr{\P}{m}{\Q}{\eff}$.
Let $\phi$ be a $\Phi$-model and suppose $\sigma,\sigma'\models_\pi \P$.
Because $\phi$ is a $\Phi$-model (Def.~\ref{def:ctxinterpRel}),
$\phi_2(m)(\sigma|\sigma')$ does not contain $\Fault$.
Moreover, execution from $\configr{\syncbi{m()}}{\sigma}{\sigma'}{\_}{\_}$
either goes by \rn{bCallS} to a terminated state, or by \rn{bCall0} repeating the configuration
$\configr{\syncbi{m()}}{\sigma}{\sigma'}{\_}{\_}$ unboundedly.
So Safety holds.
We also get Post and Write by definition of context model.
R-safety requires $\rlocs(\sigma,\effe)\subseteq\rlocs(\sigma,\effe)$ and 
$\rlocs(\sigma',\effe')\subseteq\rlocs(\sigma',\effe')$ which hold.

Encap is more interesting, as it is not a direct consequence of $\phi$ being a context model.
Encap imposes conditions on the unary projections of every trace from 
$\configr{\syncbi{m()}}{\sigma}{\sigma'}{\_}{\_}$.
By projection Lemma~\ref{lem:bi-to-unary}, or indeed by unary compatibility of the context model,
the premises of \rn{rCall} apply to these traces---and yield all the Encap conditions.

\subsection{Soundness of \rn{rIf}}

\[\rulerIf\]

As in the unary rule \rn{If}, the separator
$\ind{\unioneff{N\in\Phi,N\neq M}{\bnd(N)}}{\rTow(\ftpt(E))}$ and its counterpart simplify 
to true or false.
In virtue of condition $\P\imp E\eqbi E'$, every biprogram trace from states satisfying $\P$
begins with a step going to $CC$ via \rn{bIfT} or a step going to $DD$ via \rn{bIfF}; it cannot fault via \rn{bIfX} %or step via \rn{bIfY} 
which is for tests that disagree.
Subsequent steps satisfy all the conditions Safety, Post, Write, R-safe because these are the same as the conditions for the premises $CC$ and $DD$.  
Encap for the conclusion is almost the same condition as for the premise, the only difference being that the frame condition $\eff|\effe'$ for the premise is a subeffect of the one for the conclusion.
So Encap for the conclusion follows from the premises by an argument like that for soundness of rule \rn{rConseq}.  

The first step clearly satisfies Safety, Post, Write, and R-safe.
To show the first step satisfies Encap, boundary monotonicity and w-respect are immediate because the step does not change the state.  For r-respect, we need that alternate executions follow the same control path---and this is ensured by separator conditions, for reasons spelled out in detail in the proof of \rn{If}.

\subsection{Soundness of \rn{rLink}}\label{sec:app:rLink}

\begin{mathpar}
\rulerLink
\end{mathpar}

The rule caters for different specs on left and right, subject to the constraints of Def.~\ref{def:wfhypRel}.
For \rn{rMLink}, we instantiate $\Theta_2(m)$ to something of the form $locEq(...)\conjInv \M$,
for coupling relation $\M$, and the operation $\conjInv \M$ conjoins $\Left{\M}$ and $\Right{\M}$ to the unary specs.
Some unary ingredients appear in the premises and side conditions but are not directly used in the conclusion: $P$, $Q$, and $\dot{\Phi}$ and $\dot{\Theta}$.
These ensure that the specs are strengthenings of a local equivalence spec.

\begin{remark}\upshape
This version of the rule includes unary premises for $B$ and $B'$.  
These are used only to obtain unary models (of $\Theta_0(m)$ and $\Theta_1(m)$), which are formally required 
in order to define a full context model of $\Theta$ (using Lemma~\ref{lem:denotBiprog}).
As the proof shows, execution of $\Syncbi{C}$ remains fully aligned (except during environment calls to $m$) and all calls are sync'd, 
so the unary models have no influence on the traces used in the proof.
In future work we expect to eliminate these unary premises by revisiting the definitions of compatibility for context models (Def.~\ref{def:interpRel}),
and adjusting the well-formedness conditions for contexts (Def~\ref{def:wfhypRel}) 
and definition of covariant implication (Def.~\ref{def:prePostImply})
for a better fit with compatibility.
\qed\end{remark}

In the following proof of \rn{rLink} we assume there are no recursive calls in $B$ or $B'$.
To allow recursion, one should use a fixpoint construction for 
the denotational semantics (as in proof of linking for impure methods in \RLIII)
and an extra induction on calling depth (as in the linking proofs in \RLII\ and \RLIII).
This adds complication but does not shed light; and there are plenty other complications that do deserve to be spelled out carefully.

As in the unary semantics, we say a biprogram trace is \dt{$m$-truncated} iff the last configuration does not contain $\Endcall(m)$.  
In general, there may be unary environment calls
and $\Endcall(m)$ may occur inside a bi-com, as in $\splitbi{\skipc}{B;\Endcall(m);C};DD$.

Consider any $\Phi$-model $\phi$.   
Let $\theta_0(m)$ and $\theta_1(m)$ be the models of $\Theta_0(m)$ and $\Theta_1(m)$
from the denotations of $B$ and $B'$, by Lemma~\ref{lem:denotComm},
using the unary premises for $B$ and $B'$, and side conditions about imports.  
Let $\theta$ be the bi-model of $m$ given by Lemma~\ref{lem:denotBiprog}(i)
for the denotation of $\splitbi{B}{B'}$ in $\phi$, for which we use that 
each method's relational precondition implies its unary preconditions (which holds because $\Phi$ is wf, see Def.~\ref{def:wfhypRel}).
% formerly this was the $\SCompat$ condition).  
Owing to validity of $\Phi,\Theta \proves_N \splitbi{B}{B'} : \Theta_2(m)$, we have that 
$(\phi,\theta)$ is a $(\Phi,\Theta)$-model by Lemma~\ref{lem:denotBiprog}(ii).

In the rest of the proof, no further use is made of the unary premises for $B$ and $B'$.

To introduce identifiers for the relational spec of $m$, 
suppose $\Phi_2(m)$ is $\rflowty{\R}{\S}{\eta|\eta'}$.
For clarity we follow a convention also used the in proof of unary \rn{Link}: environments that contain $m$ have dotted names like $\dot{\mu}$ and the corresponding environment without 
$m$ has the same name without dot.

\textbf{Claim:} Let $\sigma, \sigma', \pi$ be such that
$\hat{\sigma}|\hat{\sigma}'\models_\pi \P$, where $\hat{\sigma}$  
is $\extend{\sigma}{\ol{s}}{\ol{v}}$ and $\hat{\sigma}'$  
is $\extend{\sigma}{\ol{s}'}{\ol{v}'}$ for the unique values $\ol{v},\ol{v}'$ determined by $\sigma,\sigma'$ for the spec-only variables $\ol{s},\ol{s}'$ of $\P$.
Suppose
\[\configr{\Syncbi{C}}{\sigma}{\sigma'}{[m\scol B]}{[m\scol B']} \biTranStar{\phi} \configr{DD}{\tau}{\tau'}{\dot{\mu}}{\dot{\mu'}}\]
is $m$-truncated (for some $DD,\tau,\tau',\dot{\mu},\dot{\mu}'$).  
Then
\(\configr{\Syncbi{C}}{\sigma}{\sigma'}{\_}{\_} \biTranStar{\phi\theta} \configr{DD}{\tau}{\tau'}{\mu}{\mu'}\), 
where $\mu = \drop{\dot{\mu}}{m}$ and $\mu' = \drop{\dot{\mu}'}{m}$, and
$DD = \Syncbi{D}$ for some $D$.
Moreover, if $D \equiv m();D_0$ for some $D_0$ then
there is $\rho$ such that 
$\tau|\tau'\models_\rho \R$.
% un-abbreviated version:
%% $\hat{\tau}|\hat{\tau'}\models_\rho \R$
%% where $\hat{\tau}$ is $\extend{\tau}{\ol{t}}{\ol{u}}$ and
%% $\hat{\tau'}$ is $\extend{\tau}{\ol{t}}{\ol{u}'}$ 
%% for the values $\ol{u},\ol{u}'$ of the spec-only variables $\ol{t},\ol{t}'$ of $\R$ determined by $\tau,\tau'$.

\emph{Proof of Claim:}
by induction on the number of completed top-level calls of $m$.
(Since we are not considering recursion, all calls are top level.)
The steps taken in code of $\Syncbi{C}$ can be taken via $\biTrans{\phi\theta}$ because the two transition relations are identical except for calls to $m$.  
By induction hypothesis, any call is in sync'd form, and a completed call from $\syncbi{m()}$ amounts to a terminated execution of $\splitbi{B}{B'}$.
Thus a completed call gives rise to a single step via $(\phi,\theta)$ with the same outcome, because $\theta_2(m)$ is defined to be the denotation of $\splitbi{B}{B'}$, which is defined directly in terms of executions of $\splitbi{B}{B'}$---provided that the precondition $\R$ of $m$ holds.  
The premise for $\Syncbi{C}$ is applicable to the trace via $\phi,\theta$, 
so the precondition $\R$ must hold---because otherwise that trace could fault, contrary to the premise for $\Syncbi{C}$.
It remains to show that at $DD$ is $\Syncbi{D}$ for some $D$.  
For this we appeal to lockstep alignment Lemma~\ref{lem:rloceq}.  
Let $U$ and $V$ be the unary projections of this trace.
By validity of the premise for $\Syncbi{C}$ we get that 
$U$ (resp.\ $V$) satisfies r-safe for $((\Phi_0,\Theta_0),\eff,\sigma)$
(resp.\ $((\Phi_1,\Theta_1),\eff,\sigma')$)
and respect for $((\Phi_0,\Theta_0), \emptymod, (\phi_0,\theta_0), \eff, \sigma)$ 
(resp.\ $((\Phi_1,\Theta_1), \emptymod, (\phi_1,\theta_1), \eff, \sigma')$).
By side condition of \rn{rLink}, $C$ is let-free.
Thus the assumptions are satisfied for the instantiation 
$\Phi\eqdef(\Phi,\Theta)$ of Lemma~\ref{lem:rloceq}, which yields that 
$DD$ is $\Syncbi{D}$ for some $D$.  
The Claim is proved.

\medskip
\emph{Post}.
Consider any $\phi,\sigma,\sigma',\pi$ 
with $\hat{\sigma}|\hat{\sigma}'\models_\pi \P$ 
(where $\hat{\sigma}$ is $\extend{\sigma}{\ol{s}}{\ol{v}}$ 
and $\hat{\sigma}'$ is $\extend{\sigma}{\ol{s}'}{\ol{v}'}$ for the unique values $\ol{v},\ol{v}'$ determined by $\sigma,\sigma'$ for the spec-only variables $\ol{s},\ol{s}'$ of $\P$).
A terminated trace of the linked program has the form
\[
\begin{array}{ll}
\configr{\letcombi{m}{\splitbi{B}{B'}}{\Syncbi{C}}}{\sigma}{\sigma'}{\_}{\_}
& \biTrans{\phi}
\configr{\Syncbi{C};\syncbi{\Endlet(m)}}{\sigma}{\sigma'}{[m\scol B]}{[m\scol B']} \\
& \biTranStar{\phi}
\configr{\syncbi{\Endlet(m)}}{\tau}{\tau'}{[m\scol B]}{[m\scol B']} \\
& \biTrans{\phi}
\configr{\syncbi{\skipc}}{\tau}{\tau'}{\_}{\_}
\end{array}
\]
By semantics we obtain 
\(
\configr{\Syncbi{C}}{\sigma}{\sigma'}{[m\scol B]}{[m\scol B']} 
\biTranStar{\phi}
\configr{\syncbi{\skipc}}{\tau}{\tau'}{[m\scol B]}{[m\scol B']} 
\).
This is $m$-truncated.  By the Claim, we have
\(
\configr{\Syncbi{C}}{\sigma}{\sigma'}{\_}{\_} 
\biTranStar{\phi\theta}
\configr{\syncbi{\skipc}}{\tau}{\tau'}{\_}{\_} 
\).
By the premise for $\Syncbi{C}$ we get $\hat{\tau}|\hat{\tau}'\models_\pi \Q$
where $\hat{\tau},\hat{\tau}'$ are the extensions using $\ol{v},\ol{v}'$.

\medskip
\emph{Write}.
Very similar to the argument for Post.

\medskip
\emph{Safety}.
As the steps for $\keyw{let}$ and $\Endlet$ do not fault, a faulting execution gives one of the form
\[
\configr{\Syncbi{C}}{\sigma}{\sigma'}{[m\scol B]}{[m\scol B']} 
\biTranStar{\phi}
\configr{DD}{\tau}{\tau'}{\dot{\mu}}{\dot{\mu}'}
\biTrans{\phi} \Fault
\]
We show this contradicts the premises, by cases on whether the trace up to $DD$ is $m$-truncated.

\underline{Case m-truncated.} 
The active command of $D$ (equivalently, of $\Syncbi{D}$) is not a call to $m$ because an environment call does not fault on its first step; it goes by rule \rn{bCallE}.
By the Claim, we have
$\configr{\Syncbi{C}}{\sigma}{\sigma'}{\_}{\_} 
\biTranStar{\phi\theta}
\configr{DD}{\tau}{\tau'}{\mu}{\mu'}
$.
Because the active command is not a call to $m$, the step
$\configr{DD}{\tau}{\tau'}{\dot{\mu}}{\dot{\mu}'}\biTrans{\phi}\Fault$
can also be taken via $\biTrans{\phi\theta}$.
But then we have a faulting trace that contradicts the premise for $\Syncbi{C}$.

\underline{Case not m-truncated.} 
A trace with an incomplete call of $m$ has the following form.
(Here we rely on the Claim to write parts in fully aligned form.)
\[
\begin{array}{ll}
\configr{\Syncbi{C}}{\sigma}{\sigma'}{[m\scol B]}{[m\scol B']} 
&\biTranStar{\phi}
\configr{\syncbi{m()};\Syncbi{D_0}}{\tau_0}{\tau'_0}{\dot{\mu}}{\dot{\mu}'} \\
& \biTrans{\phi}
\configr{\splitbi{B}{B'};\Syncbi{D_0}}{\tau_0}{\tau'_0}{\dot{\mu}_0}{\dot{\mu}'_0} 
 \biTranStar{\phi}
\configr{BB_0;\Syncbi{D_0}}{\tau}{\tau'}{\dot{\mu}}{\dot{\mu}'} 
 \biTrans{\phi}
\Fault 
\end{array}
\]
with $BB_0\nequiv \syncbi{\skipc}$. 
Applying the Claim to the $m$-truncated prefix we get $\tau_0|\tau'_0\models_\rho\R$ for some $\rho$.  
By semantics we get
\(
\configr{\splitbi{B}{B'}}{\tau_0}{\tau'_0}{\dot{\mu}_0}{\dot{\mu}'_0} 
\biTranStar{\phi}
\configr{BB_0}{\tau}{\tau'}{\dot{\mu}}{\dot{\mu}'} 
\biTrans{\phi}
\Fault 
\).
Now, $\splitbi{B}{B'}$ has no calls to $m$---because we are proving soundness assuming there is no recursion.
So the same transitions can be taken via $\trans{\phi\theta}$.
But then we get a faulting trace that contradicts the premise for $\splitbi{B}{B'}$.

\medskip\emph{R-safety}.
For any trace $T$ of $\letcombi{m}{\splitbi{B}{B'}}{\Syncbi{C}}$
from $\sigma,\sigma'$ satisfying $\P$, we must show that the left projection $U$ and right projection $V$ is r-safe for $(\Phi_0,\eff,\sigma)$ and $(\Phi_1,\eff,\sigma')$ respectively.  
Observe that the premises for $\Syncbi{C}$ and for $\splitbi{B}{B'}$ give r-safety of their left projections,
for $((\Phi_0,\Theta_0),\eff,\sigma)$, and r-safety of their right projection for $((\Phi_1,\Theta_1),\eff,\sigma')$).
For methods of $\Phi$, by definition of r-safety, these are the same conditions as 
r-safety for $(\Phi_0,\eff,\sigma)$ and for $(\Phi_1,\eff,\sigma')$.
Let us consider $U$, as the argument for $V$ is symmetric.
We must show the r-safety condition for any configuration, say $U_i$.
Let $\dot{T}$ the prefix of $T$ such that $U_i$ is aligned (by projection Lemma) with the last configuration of $\dot{T}$.  
Now go by cases on whether $\dot{T}$ is $m$-truncated.  

\textbf{case} $\dot{T}$ is $m$-truncated.  If the last configuration is calling $m$ there is nothing to prove.  Otherwise, that configuration is not within a call of $m$, so by the Claim we get from $\dot{T}$ a trace
$\ddot{T}$ of $\Syncbi{C}$ via $\biTrans{\phi}$
that ends with the same configuration.
Now can appeal to r-safety from the premise for $\Syncbi{C}$ and we are done.
(The claim does not address the first step of 
$\letcombi{m}{\splitbi{B}{B'}}{\Syncbi{C}}$, but that satisfies r-safety by definition.)

\textbf{case} $\dot{T}$ is not $m$-truncated.  
So a suffix of $\dot{T}$ is an incomplete environment call of $m$, say at position $j$.
By the Claim, the call is sync'd (and $m$'s relational precondition holds),
so the code of $\dot{T}_j$ has the form $\syncbi{m()};DD$ for some continuation code $DD$,
and the following steps execute starting from $\splitbi{B}{B'};DD$
(by transition rule \rn{bCallE}).
By dropping ``$;DD$'' from each configuration we obtain a trace of $\splitbi{B}{B'}$ 
that includes configuration $\dot{T}_j$.  Now we can appeal to r-safety from the premise for $\splitbi{B}{B'}$ and we are done.

\medskip\emph{Encap}.
For any trace of $\letcombi{m}{\splitbi{B}{B'}}{\Syncbi{C}}$
from $\sigma,\sigma'$ satisfying $\P$, we must show that the left projection 
respects $(\Phi_0,\emptymod,\phi_0,\eff,\sigma)$ and the right 
respects $(\Phi_1,\emptymod,\phi_1,\eff,\sigma')$.
The proof is structured similarly to the proof of R-safe, 
though it is a bit more intricate.

Observe that the premises yield respect of 
$((\Phi_0,\Theta_0),\emptymod,(\phi_0,\theta_0),\eff,\sigma)$ and
$((\Phi_1,\Theta_1),\emptymod,(\phi_1,\theta_1),\eff,\sigma')$.
By contrast with the argument above for r-safety, where the meaning of the condition for the conclusion is very close to its meaning for the premises,
for respect there are two significant differences.
First, the respect condition depends on the current module $\emptymod$, and the judgment for $\splitbi{B}{B'}$ is for a possibly different module.  
Second, respect depends on the modules in context, and by side conditions of the rule 
the modules of $\Phi$ are not the same as those of $(\Phi,\Theta)$.
Fortunately, these differences are exactly the same in the setting of rule \rn{Link}.
The proof Encap for \rn{Link} (Sect.~\ref{sec:app:link})
shows in detail how respect, for traces of $\letcom{m}{B}{C}$,
follows from respect for traces of $B$ and for traces of $C$ in which calls to $m$ are context calls.

Now we proceed to prove Encap.
For any trace $T$ of $\letcombi{m}{\splitbi{B}{B'}}{\Syncbi{C}}$
from $\sigma,\sigma'$ satisfying $\P$, 
consider its left projection $U$ (the right having a symmetric proof),
which is a trace of $\letcom{m}{B}{C}$.
Consider any step in $U$, say $U_{i-1}$ to $U_i$.

If the step is an environment call to $m$,
i.e., the call is the active command of $U_{i-1}$, it satisfies respect
of $(\Phi_0,\emptymod,\phi_0,\eff,\sigma)$ by definitions and semantics. 
If the active command is $\Endcall(m)$ then again we get respect by definitions and semantics.
Otherwise, let $\dot{T}$ be the prefix of $T$ such that the last configuration corresponds with $U_i$,
and go by cases on whether $\dot{T}$ is $m$-truncated.

\textbf{case} $\dot{T}$ is $m$-truncated.
So the step is not within a call of $m$,
and is present in the trace $\ddot{T}$ given by the Claim.
So we can appeal to the premise for $\Syncbi{C}$. 
We get that the step respects $(\Phi_0,\emptymod,\phi_0,\eff,\sigma)$, using the arguments in the \rn{Link} proof
to connect with respect of $((\Phi_0,\Theta_0),\emptymod,(\phi_0,\theta_0),\eff,\sigma)$ in accord with the premise for $\Syncbi{C}$.

\textbf{case} $\dot{T}$ is not $m$-truncated.
As in the r-safety argument, 
we obtain a trace of $\splitbi{B}{B'}$ that includes the step in question, and it respects
$(\Phi_0,\emptymod,\phi_0,\eff,\sigma)$, using the arguments in the \rn{Link} proof to connect 
with respect of $((\Phi_0,\Theta_0),\mdl(m),(\phi_0,\theta_0),\eff,\sigma)$ in accord with the premise for $\splitbi{B}{B'}$.

\subsection{Soundness of \rn{rWeave}}\label{sec:rWeaveSound}

\[ % copied from figure with rules 
\inferrule*[left=rWeave]{
  \Phi \rHPflowtr{}{}{\P}{DD}{\Q}{\eff|\eff'} \\
  CC \weave^* DD 
}{ 
  \Phi \rHPflowtr{}{}{\P}{CC}{\Q}{\eff|\eff'} 
}
\]

\begin{remark}\upshape
In general $\Phi \rHVflowtr{}{}{\P}{DD}{\Q}{\eff}$ and $DD \weave CC$ do not imply 
$\Phi \rHVflowtr{}{}{\P}{CC}{\Q}{\eff}$, for one reason:
$CC$ may assert additional test agreements that do not hold.
\qed\end{remark}

The crux of the soundness proof for rule \rn{rWeave} is soundness for a single weaving step,
$CC\weave DD$, which is Lemma~\ref{lem:weave-one} below.  Using the lemma, we can prove soundness of \rn{rWeave} by induction on the number of weaving steps $CC\weave^* DD$.  
In case of zero steps, $CC\equiv DD$ and the result is immediate.
In case of more than one steps, apply Lemma~\ref{lem:weave-one} and the induction hypothesis.

Before proving Lemma~\ref{lem:weave-one} we prove preliminary results.

\begin{lemma}[weave and project]\label{lem:weave-project}
\upshape
If $CC\weave DD$ then
$\Left{CC}\equiv \Left{DD}$ and 
$\Right{CC}\equiv \Right{DD}$.
\end{lemma}
\begin{proof}
By induction on the rules for $\weave$ (Fig.~\ref{fig:weave}), 
making straightforward use of the definitions of the syntactic projections.
As an example, 
for the if-else axiom we have\\
$\Left{ \Splitbi{ \ifc{E}{C}{D} }{ \ifc{E'}{C'}{D'} } }
\equiv \ifc{E}{C}{D}
\equiv \ifc{E}{\Left{\splitbi{C}{C'}}}{\Left{\splitbi{D}{D'}}}
\equiv \Left{ \ifcbi{E\smallSplitSym E'}{\splitbi{C}{C'} }{ \splitbi{D}{D'} }}$.
As an example inductive case, 
for the rule from $BB\weave CC$ infer $BB;DD \weave CC;DD$, we have
$\Left{BB;DD} \equiv \Left{BB};\Left{DD} \equiv \Left{CC};\Left{DD} \equiv \Left{CC;DD}$
where the middle step is by induction hypothesis.
\end{proof}

\begin{lemma}[trace coverage]\label{lem:coverage}
\upshape 
Suppose $\Phi \rHVflowtr{}{}{\P}{DD}{\Q}{\eff}$ and let $\phi$ be a $\Phi$-model.
Consider any $\pi$ and any $\sigma,\sigma'$ such that $\sigma|\sigma'\models_\pi \P$.
Let $U$ and $V$ be traces from $\configm{\Left{DD}}{\sigma}{\_}$
and $\configm{\Right{DD}}{\sigma'}{\_}$ respectively.
Then there is a trace $T$ from $\configr{DD}{\sigma}{\sigma'}{\_}{\_}$,
with projections $W,X$ such that $U\leq W$ and $V\leq X$.
\end{lemma}
\begin{proof}
Apply embedding Lemma~\ref{lem:unary-to-bi}
to $U,V$ to obtain $T,W,X$ satisfying one of the conditions (a), (b), (c), or (d)
in that Lemma.
Conditions (b), (c), and (d) contradict the premise, specifically Safety for $DD$.
That leaves condition (a) which completes the proof.
\end{proof}

\begin{lemma}[weave and trace]\label{lem:weave-trace}
\upshape
Suppose
$\Phi \rHVflowtr{}{}{\P}{DD}{\Q}{\eff}$
and $CC \weave DD$ or $DD \weave CC$.
Consider any $\Phi$-model $\phi$.
Consider any $\pi$ and any $\sigma,\sigma'$ such that $\sigma|\sigma'\models_\pi \P$.
Consider any trace $S$ from
$\configr{CC}{\sigma}{\sigma'}{\_}{\_}$
and let $U,V$ be the projections of $S$ according to the projection Lemma~\ref{lem:bi-to-unary}.
Then there is a trace $T$ from $\configr{DD}{\sigma}{\sigma'}{\_}{\_}$,
with projections $W,X$ such that $U\leq W$ and $V\leq X$.
\end{lemma}
\begin{proof}
\begin{sloppypar}
Using $CC \weave DD$ or $DD \weave CC$, 
by Lemma~\ref{lem:weave-project} we have 
$\Left{\configr{DD}{\sigma}{\sigma'}{\_}{\_}} = \configm{\Left{CC}}{\sigma}{\_}$ and
$\Right{\configr{DD}{\sigma}{\sigma'}{\_}{\_}} = \configm{\Right{CC}}{\sigma}{\_}$
so we get the result by Lemma~\ref{lem:coverage}.
\end{sloppypar}
\end{proof}

Finally, we proceed to prove soundness for a single weaving step.
The hard case is Safety, for reasons explained in the proof.

\begin{lemma}[one weave soundness]\label{lem:weave-one}
\upshape
Suppose 
$\Phi \rHVflowtr{}{}{\P}{DD}{\Q}{\eff}$ and $CC \weave DD$.
Then $\Phi \rHVflowtr{}{}{\P}{CC}{\Q}{\eff}$.
\end{lemma}
\begin{proof}
Suppose
$\Phi \rHVflowtr{}{}{\P}{DD}{\Q}{\eff}$ and $CC \weave DD$.
To show the conclusion $\Phi \rHVflowtr{}{}{\P}{CC}{\Q}{\eff}$,
consider any $\Phi$-model $\phi$.
Consider any $\pi$ and any $\sigma,\sigma'$ such that $\sigma|\sigma'\models_\pi \P$.

\emph{R-safe}.
Consider any trace $S$ of $CC$ from $\sigma,\sigma'$.
By Lemma~\ref{lem:weave-trace}, there is a trace $T$ of $DD$ such that 
every unary step in $S$ is covered by a step in $T$. So r-safety
follows from r-safety of the premise.

\emph{Encap}. Similar to R-safe.

\emph{Write and Post}
By Lemma~\ref{lem:weave-trace}, a terminated trace of $CC$ gives rise to one of $DD$ with the same final states, to which the premise applies.

\emph{Safety}. This requires additional definitions and results.
Faults by $CC$ may be alignment faults (rules \rn{bCallX}, \rn{bIfX}, \rn{bWhX})
or due to unary faults (\rn{bSyncX}, \rn{bComLX}, \rn{bComRX}).
The latter can be ruled out by reasoning similar to the above,
but alignment faults pose a challenge, because weaving rearranges 
the alignment of execution steps.
We proceed to develop some technical notions about alignment faults, 
and use them to prove Safety.

In most of this paper, we only need to consider traces from 
initial configurations $\configr{CC}{\sigma}{\sigma'}{\mu}{\mu'}$
where the environments are empty (written $\_$) and the code has no endmarkers.  
In the following definitions, we need to consider non-empty initial environments, and $CC$ may be an extended bi-program; in particular, $CC$ may include endmarkers.  (It turns out that we will not have occasion 
to consider an initial biprogram $CC$ that contains a right-bi-com.)
This is needed because, in the proof of Lemma~\ref{lem:weave-sync} below,
specifically the case of weaving the body of a bi-let, we apply the induction hypothesis to a trace in which the initial environments are non-empty.
The initial configuration of a trace must still be well formed: free variables in 
$CC$ should be in the states, and methods called in $CC$ must be in either the context or the environment and not in both.

Define a \dt{sync point} in a biprogram trace $T$ to be a position $i$,
$0\leq i < len(T)$,
such that one of the following holds:
\begin{itemize}
\item $i=0$ (i.e., $T_i$ is the initial configuration)
\item The configuration $T_i$ is terminal, i.e., has code $\syncbi{\skipc}$
\item $\Active(T_i)$ is not a bi-com, i.e., neither $\splitbi{-}{-}$ nor $\splitbir{-}{-}$.
Thus $\Active(T_i)$ may be $\syncbi{-}$, bi-if, bi-while, bi-let, or bi-var.
(By definition, the active biprogram is not a sequence.)

\item $i>0$ and the step from $T_{i-1}$ to $T_i$ completed the first part of a biprogram sequence.
That is, the code in $T_{i-1}$ has the form $CC;DD$ with $CC$ the active command, 
and the code in $T_i$ is $DD$.  
Such a transition is a transition from $CC$ to $\syncbi{\skipc}$ that is lifted to $CC;DD$ by
rule \rn{bSeq}.\footnote{One could make this more explicit by dropping the identification of $\syncbi{\skipc};DD$
with $DD$ and instead having a separate transition from $\syncbi{\skipc};DD$ to $DD$, but this would make extra cases in other proofs.}
Later we refer to this kind of step as a ``semi-colon removal''.
\end{itemize}

A \dt{segment} of a biprogram trace is just a list of configurations that occur contiguously in the trace.
A \dt{segmentation} of trace $T$ is a list $L$ of nonempty segments,
the catenation of which is $T$.
Thus, indexing the list $L$ from 0, the configuration 
$(L_i)_j$ is $T_{n+j}$ where $n = \Sigma_{0\leq k < i} len(L_k)$.
An \dt{alignment segmentation} of $T$ is a segmentation $L$ 
such that each segment in $L$ begins with a sync point of $T$.

For an example, using abbreviations
$A0\eqdef x:=0$, 
$A1\eqdef x:=1$, 
$A2\eqdef x:=2$
and omitting states/environments from the configurations, 
here is a trace with one of its alignment segmentations depicted by boxes:
\[
\begin{array}[t]{l}
\fbox{$\begin{array}{l} 
  \configc{\splitbi{A0}{A0};\ifcbi{x>0|x>0}{\splitbi{A1}{A1}}{\splitbi{A2}{A2}}} \\
  \configc{\splitbir{\skipc}{A0};\ifcbi{x>0|x>0}{\splitbi{A1}{A1}}{\splitbi{A2}{A2}}} \\
  \end{array}
$}\\
\fbox{$\begin{array}{l}
  \configc{\ifcbi{x>0|x>0}{\splitbi{A1}{A1}}{\splitbi{A2}{A2}}} \\
  \configc{\splitbi{A2}{A2}} \\
  \configc{\splitbir{\skipc}{A2}}
  \end{array}
$}\\
\fbox{$
  \configc{\syncbi{\skipc}}$}
\\
\end{array}
\]
Every trace has a minimal-length alignment segmentation consisting of the trace itself---a single segment---and also a maximal-length alignment segmentation (which has a segment for each sync point).  
(Keep in mind that we define traces to be finite.)
The above example, with three segments, is maximal.

As another example, here is a trace that faults next (because $x>0$ is false on the left but true on the right), with its maximal alignment segmentation.
\[ 
\begin{array}{l}
\fbox{$\begin{array}{l}
      \configc{\splitbi{x:=0}{x:=1};\ifcbi{x>0|x>0}{\syncbi{A1}}{\syncbi{A2}}} \\
      \configc{\splitbir{\skipc}{x:=1};\ifcbi{x>0|x>0}{\syncbi{A1}}{\syncbi{A2}}} 
     \end{array}$}
\\
\fbox{
$ \configc{\ifcbi{x>0|x>0}{\syncbi{A1}}{\syncbi{A2}}}$
}
\end{array}
\]
Note that a segment can begin with a configuration that contains end-markers whose beginning was in a previous segment.  
For example, 
\[
\begin{array}{l}
\fbox{$\begin{array}{l}
  \configc{ \varblockbi{x:T|x':T'}{ \splitbi{a}{b}; \splitbi{c}{d} } } \\
  \configc{ \splitbi{a}{b}; \splitbi{c}{d} ; \splitbi{\Endvar(x)}{\Endvar(x')} } \\
  \configc{ \splitbir{\skipc}{b}; \splitbi{c}{d} ; \splitbi{\Endvar(x)}{\Endvar(x')} } 
  \end{array}$} \\
\fbox{$\begin{array}{l}
   \configc{ \splitbi{c}{d} ; \splitbi{\Endvar(x)}{\Endvar(x')} } \\
\configc{ \splitbir{\skipc}{d} ; \splitbi{\Endvar(x)}{\Endvar(x')} } \\
\configc{ \splitbi{\Endvar(x)}{\Endvar(x')} } \\ 
\configc{ \splitbir{\skipc}{\Endvar(x')} } \\
\configc{ \syncbi{\skipc} } 
  \end{array}
$}
\end{array}
\]

In the following we sometimes refer to the left and right sides of a weaving as lhs and rhs.
A weaving $lhs\weave rhs$ introduces sync points in the biprogram's traces, but it does not remove sync points of $lhs$.   
Moreover, though it rearranges the order in which the underlying unary steps are taken, 
it does not change the states that appear at sync points.  
This is made precise in the following lemma which gives a sense in which 
weaving is directed (i.e., not commutative).

\begin{lemma}[weaving preserves sync points]\label{lem:weave-sync}
\upshape
%LEMMA (C): 
Consider any pre-model $\phi$.
Consider any biprograms $CC$ and $DD$ such that $CC\weave DD$.
Let $S$ be a trace (via $\phi$) of $CC$ from some initial states and environments.
(No assumption is made about the initial states, and non-empty method environments are allowed.)
Let $L$ be the maximal alignment segmentation of $S$.  
Then there is a trace $T$ of $DD$ from the same states and environments, 
such that either 
\begin{list}{}{}
\item[(i)] the last configuration of $T$ can fault next, by alignment fault; or
\item[(ii)] there is an alignment segmentation $M$ of $T$ 
such that $M$ has the same length as $L$ and 
for all $i$, 
segment $M_i$ and segment $L_i$ begin with the same states, same environments, 
and same underlying unary programs, that is: 
\begin{equation}\label{eq:sameProj}
\Left{(L_i)_0} = \Left{(M_i)_0} \mbox{ and } \Right{(L_i)_0} = \Right{(M_i)_0} 
\end{equation}
\end{list}
\end{lemma}
% It seems that we do not need to say $CC^L_i\weave^*DD^M_i$ where $CC^L_i$ is the code of $(L_i)_0$ and $DD^M_i$ is the code of $(M_i)_0$.
Note that $M$ in Lemma~\ref{lem:weave-sync} need not be the maximal segmentation.
Typically $T$ will have additional sync points, but these are not relevant to the conclusion of the lemma.  What matters is that $T$ covers the sync points of $S$.
(Note that $T$ need not cover all the steps of $S$.)
As an example of the lemma, 
consider a biprogram of the form 
$\configc{\splitbi{A0}{A0};\ifcbi{x>0|x>0}{\splitbi{A1}{A1}}{\splitbi{A2}{A2}}}$.
It relates by $\weave$ to
\( \configc{\splitbi{A0}{A0};\ifcbi{x>0|x>0}{\splitbi{A1}{A1}}{\syncbi{A2}}} \) 
(by an axiom and the congruence rules for sequence and conditional).
From the same initial states (and empty environments), the latter biprogram 
has a shorter trace (owing to sync'd execution of $A2$) but that trace can still be segmented in accord with the lemma.
Its second segment has three configurations:
\[ \configc{\ifcbi{x>0|x>0}{\splitbi{A1}{A1}}{\syncbi{A2}}}
\configc{\syncbi{A2}}
\configc{\syncbi{\skipc}}
\]

We defer the proof of Lemma~\ref{lem:weave-sync} and use it to finish 
the proof of Lemma~\ref{lem:weave-one} by completing the proof of Safety.
As before, we assume
$\Phi \rHVflowtr{}{}{\P}{DD}{\Q}{\eff}$ and $CC \weave DD$.
To show the Safety condition for $\Phi \rHVflowtr{}{}{\P}{CC}{\Q}{\eff}$,
consider any $\Phi$-model $\phi$.
Consider any $\pi$ and any $\sigma,\sigma'$ such that $\sigma|\sigma'\models_\pi \P$.
Suppose $CC$ has a trace $S$ from $\sigma,\sigma'$ (and empty environments).
If $S$ faults next by a unary fault,
let its unary projections be $U,V$ (one of which faults next).
Then by Lemma~\ref{lem:weave-trace} the trace $T$ from $U,V$ must also fault next---and this contradicts the assumed judgment for $DD$.

Finally, suppose $S$ faults next by alignment fault.
Consider the maximal alignment segmentation of $S$ and let $T$ be the trace from $DD$ given by
Lemma~\ref{lem:weave-trace}.
By Lemma~\ref{lem:weave-sync} there is a segmentation of $T$ that covers each sync point of $S$, including the last configuration of $S$ which faults.  
But then $T$ faults next, contrary to the premise for $DD$.

This concludes the proof of Lemma~\ref{lem:weave-one} and thus soundness of \rn{rWeave}.
\end{proof}

\begin{proof} (Of Lemma~\ref{lem:weave-sync}.)
By induction on the derivation of the weaving relation $CC\weave DD$,
and by cases on the definition of $\weave$     starting with the axioms
(Fig.~\ref{fig:weave}).

\textbf{Case} weaving axiom $\splitbi{A}{A}\weave\syncbi{A}$.
For most atomic commands $A$, a trace $S$ of the lhs consists of an initial configuration $\configr{\splitbi{A}{A}}{\sigma}{\sigma'}{\mu}{\mu'}$,
possibly a second one with code
$\splitbir{\skipc}{A}$, and possibly a third one that is terminated (i.e., has code $\syncbi{\skipc}$).  
However, because the lemma allows non-empty environments, 
there is also the case that $A$ is an environment call to some $m$ in the domain of $\mu$ and of $\mu'$.
In that case, if $\mu(m)= B$ and $\mu'(m)=B'$,
then there are traces of the form 
$\configc{ \splitbi{m()}{m()} }
\configc{ \splitbir{B;\Endcall(m)}{m()} }
\configc{ \splitbi{B;\Endcall(m)}{B';\Endcall(m)} }
\ldots$.
Traces of the $\syncbi{m()}$ can have the form
$\configc{ \syncbi{m()} } \configc{ \splitbi{B}{B'} }\ldots$
but also, if $B'\equiv B$, the form 
$\configc{ \syncbi{m()} } \configc{ \Syncbi{B} }\ldots$ (see rule \rn{bCallE}
and Fig.~\ref{fig:fullAlign}).
The latter is susceptible to alignment faults.

In any case, the only sync points in $S$ are the initial configuration and, if present, the terminated one.
If $S$ is not terminated then it has only the initial sync point, 
so $L$ has only a single segment.  This can be matched
by the trace $T$ consisting of the one configuration $\configr{\syncbi{A}}{\sigma}{\sigma'}{\mu}{\mu'}$ which also serves as the single segment for $T$.  (The lemma does not require $T$ to cover all steps of $S$, only 
the sync points of $S$.)

If $S$ terminated, then by projection and then embedding Lemma~\ref{lem:unary-to-bi},
$\configr{\syncbi{A}}{\sigma}{\sigma'}{\mu}{\mu'}$ has a trace $T$ that either terminates, covering the steps of $S$, or faults.  It cannot have a unary fault because $S$ did not.
If it has an alignment fault, which would be via context call transition \rn{bCallX}
or by some step of an environment call executing $\Syncbi{B}$,
we are done.  Otherwise $T$ can be segmented to match the segmentation $L$:
One segment including all of $T$ except the last configuration, followed by that configuration as a segment.

\textbf{Case} weaving axiom 
$\Splitbi{C;D}{C';D'} \weave \splitbi{C}{C'};\splitbi{D}{D'}$.
A trace $S$ of the lhs may make several steps, and may eventually terminate.
If terminated, it has two sync points, initial and final; otherwise only the initial configuration is a sync point.  
If not terminated, the initial configuration for $\splitbi{C}{C'};\splitbi{D}{D'}$
provides the trace $T$ and its single segment.  
If $S$ terminated, then by projection and embedding we obtain a trace $T$ that either
terminates in the same states or has an alignment fault.  
So we either get a matching segmentation of $T$ or an alignment fault.

\textbf{Cases} the other weaving axioms.
The argument is the same as above, in all cases.  The rhs of weaving has additional sync points which are of no consequence except that they can give rise to alignment faults.  
Like the preceding cases, bi-if and bi-while introduce the possibility of alignment fault;
bi-let and bi-var weavings do not. 

Having dispensed with the base cases, we turn to the inductive cases which each have as premise that $BB\weave CC$ (Fig.~\ref{fig:weave}).
The induction hypothesis is that for any trace $S$ of $BB$
and any alignment segmentation $L$ of $S$,
there is a trace $T$ of $CC$ 
such that either its last configuration can alignment-fault
or there is a segmentation $M$ of $T$ that covers the segmentation of $S$.

\textbf{Case}
$BB;DD \weave CC;DD$.

A trace $S$ of $BB;DD$ may include only execution of $BB$ or may continue to execute $DD$. 
\begin{itemize}
\item
In case $S$ never starts $DD$, the trace $S$ determines a trace $S^+$ of $BB$ by removing the trailing ``$;DD$'' 
from every configuration.  (In the special case that $CC$ is run to completion in $S$, i.e., its last configuration has exactly the code $DD$, then the last configuration of $S^+$ has $\syncbi{\skipc}$.)
(Note that $S$ may have sync points besides the initial one, as $BB$ is an arbitrary biprogram.)
By induction we obtain trace $T$ of $CC$ and either alignment fault or segmentation of $T$ that covers the segmentation of $S$.  Adding $;DD$ to every configuration of $T$ yields the requisite segmentation of $S$.
\item 
Now consider the other case: $S$ includes at least one step of $DD$, 
so there is some $i>0$ such that $S_{i-1}$ has code $BB';DD$ for some $BB'$ 
that steps to $\syncbi{\skipc}$, and $S_i$ has code $DD$.  
Because $L$ is the maximal segmentation of $S$, it includes 
a segment that starts with the configuration $S_i$.
Now we can proceed as in the first bullet, to obtain trace $T$ of $CC$ and either alignment fault or segmentation for the part of $S$ up to but not including position $i$.  Catenating this segmentation with the one for the trace of $DD$ from $i$ yields the result.
\end{itemize}

\textbf{Case}
$DD;BB \weave DD;CC$.
For a trace $S$ that never reaches $BB$, the result is immediate by taking $T:=S$ and $M:=L$.
Otherwise, the given trace $S$ can be segmented into an execution of $DD$ that terminates, followed by a terminating execution of $BB$.
By maximality, the segmentation breaks at the semicolon, and we obtain the result using the induction hypothesis similarly to the preceding case.

\textbf{Case} 
$\ifcbi{E\smallSplitSym E'}{BB}{DD} \weave \ifcbi{E\smallSplitSym E'}{CC}{DD}$.
If the given trace $S$ has length one, we immediately obtain a 
length-one trace and segmentation that satisfies the same-projection condition (\ref{eq:sameProj}).

If $len(S)>1$ then the first step does not fault, i.e., the tests agree.
Let $S^+$ be the trace starting at position 1, which is a trace of $BB$ or of $DD$
depending on whether the tests are initially true or false.  
If the tests are false then 
catenating the initial configuration for $\ifcbi{E\smallSplitSym E'}{CC}{DD}$
with $S^+$ provides the requisite $T$, and also its segmentation.
If the tests are true, then apply the induction hypothesis to obtain a trace $T$ for $CC$, and segmentation (if not alignment fault);
and again, prefixing the initial configuration to $T$ and to its first segment yields the result.

\textbf{Case} 
$\ifcbi{E\smallSplitSym E'}{DD}{BB} \weave \ifcbi{E\smallSplitSym E'}{DD}{CC}$.
Symmetric to the preceding case.

\textbf{Case} 
$\whilecbiA{E\smallSplitSym E'}{\P\smallSplitSym \P'}{BB} \weave
  \whilecbiA{E\smallSplitSym E'}{\P\smallSplitSym \P'}{CC}$.
A trace $S$ of lhs can be factored into a series of zero or more iterations possibly followed by an incomplete iteration of left/right/both.
Note that a completed iteration ends with a ``semi-colon removal'' step (the left-, right-, or both-sides loop body finishes and was followed by the bi-loop).  
Because the segmentation $L$ is maximal, it has a separate segment for each iteration.

Now the argument goes by induction on the number of iterations.
The inner induction hypothesis yields segmentation for rhs up to the last iteration,
which in turn ensures that lhs and rhs agree on whether the last iteration is left-only, right-only, or both-sides.  In the one-sided cases there are no sync points.
In the both-sides case, the main induction hypothesis for $BB\weave CC$ 
can be used in a way similar to the argument for sequence weaving above.

%TODO it wouldn't hurt to spell out more details,
%or give an example to make clear that during one-side iterations there are no alignment points, but upon termination of one iteration we're back to the bi-while and thus (by maximality) an alignment point.   

\begin{sloppypar}
\textbf{Case} 
$\letcombi{m}{\splitbi{B}{B'}}{ BB } \weave \letcombi{m}{\splitbi{B}{B'}}{ CC }$.
Suppose $S$ is a trace from 
$\configr{\letcombi{m}{\splitbi{B}{B'}}{ BB }}{\sigma}{\sigma'}{\mu}{\mu'}$,
with segmentation $L$.
If $S$ has length one the rest is easy.
Otherwise, $S$ takes at least one step, to
$\configr{BB;\syncbi{\Endlet(m)}}{\sigma}{\sigma'}{\hat{\mu}}{\hat{\mu}'}$ where $\hat{\mu}$ and $\hat{\mu}'$
extend $\mu,\mu'$ with $m\scol B$ and $m\scol B'$ respectively.  
We obtain trace $S^+$ of $\configr{BB;\syncbi{\Endlet(m)}}{\sigma}{\sigma'}{\hat{\mu}}{\hat{\mu}'}$
by omitting the first configuration of $S$---and here we use a trace where the initial 
environments are non-empty.  
Applying the induction hypothesis, we obtain trace $T^+$ for $S^+$,
and either alignment fault or matching segmentation $M^+$.
Prefixing the configuration 
$\configr{\letcombi{m}{\splitbi{B}{B'}}{ CC }}{\sigma}{\sigma'}{\mu}{\mu'}$
yields the requisite trace $T$.  
If there is alignment fault, we are done.
Otherwise, if $BB$ begins with an aligning bi-program,
i.e., if $S_1$ is a sync point in $S$, then
let segmentation $M$ consist of the singleton 
$\configr{\letcombi{m}{\splitbi{B}{B'}}{ CC }}{\sigma}{\sigma'}{\mu}{\mu'}$
followed by the elements of $M^+$. 
Finally, if $S_1$ is not a sync point in $S$,
we  obtain $M$ by prefixing 
$\configr{\letcombi{m}{\splitbi{B}{B'}}{ CC }}{\sigma}{\sigma'}{\mu}{\mu'}$
to the first segment in $M^+$.
\end{sloppypar}

\textbf{Case}  
$\varblockbi{x\scol T \smallSplitSym x'\scol T'}{BB} 
\weave \varblock{x\scol T\smallSplitSym x'\scol T'}{CC}$.
By semantics and induction hypothesis, similar to the preceding case for bi-let.
\end{proof}

\section{Appendix: Guide to identifiers and notations} \label{sec:app:guide}

The prime symbol, like $\sigma'$, is consistently used for right side in a pair of commands, states, etc.
Other decorations, like $\dot{\sigma}$ and $\ddot{\tau}$, are used for fresh identifiers in general.
%Subscripts $0,1,2$ are used as projections of relational hypothesis contexts and models, and otherwise freely.

\begin{small}
\begin{table}[h!]
\begin{footnotesize}
\begin{tabular}{lll}
$A$ & atomic command & Fig.~\ref{fig:bnf} \\ 
$B,C,D$ & command & Fig.~\ref{fig:bnf} \\ 
$BB,CC,DD$ & biprogram & Fig.~\ref{fig:bnf} \\ 
$E$ & program expression & Fig.~\ref{fig:bnf} \\ 
$G,H$ & region expression & Fig.~\ref{fig:bnf} \\ 
$F$ & either program or region expression & Fig.~\ref{fig:bnf} \\ 
$f,g$ & field name & Fig.~\ref{fig:bnf}, Eqn.~(\ref{eq:effects}) \\
$K$ & reference type & Fig.~\ref{fig:bnf} \\ 
$M,N,L$ & module name \\
$T$       & data type & Fig.~\ref{fig:bnf} \\ 
$T,U,V,W$ & trace (unary or biprogram) \\
$P,Q,R$ & formula & Fig.~\ref{fig:preds} \\
$\P,\Q,\R,\M,\N$ & relation formula & Fig.~\ref{fig:relFormulas} \\
$x,y,z,r,s$ & program variable \\
$\eff,\effe,\delta$ & effect expression & Eqn.~(\ref{eq:effects}) \\
$\Gamma$ & typing context \\
$\Phi,\Theta,\Psi,$ & unary or relational hypothesis context & Sects.~\ref{sec:ucorr} and~\ref{sec:relspec}  \\
$\phi,\theta,\psi$ & unary or relational context model & Sects.~\ref{sec:ctxmodel} and~\ref{sec:rel-corr}    \\
$\Phi_0,\Phi_1,\Phi_2$ & components of relational context & see preceding Def.~\ref{def:wfjudgeRel} \\
$\sigma,\tau,\upsilon$ & state & Sect.~\ref{sec:states} \\
$\hat{\sigma}$ & state with spec-only vars &  \\
$\pi,\rho$ & refperm & Sect.~\ref{sec:effect} \\
\end{tabular}
\end{footnotesize}
\medskip
\caption{Use of identifiers}
\label{tab:ident}
\end{table}
\end{small}

\begin{small}
\begin{table}[h!]               
\begin{footnotesize}
\begin{tabular}{lll}
$\ind{}{}$ & separator function & Eqn.~(\ref{eq:sepagree}) \\
$\emptymod$ & default/main module & Sect.~\ref{sec:modules} \\
$\emptyeff$ & empty effect & Eqn.~(\ref{eq:effects}) \\
$\eff\setminus\effe$ & effect subtraction & following Def.~\ref{def:framedreadsDyn} \\ 
$\unioneff{-}{-}$ & combination of effects & following Def.~\ref{def:framedreadsDyn} \\ 
$\Img f$ & image in region expression or effect & Fig.~\ref{fig:bnf}, Eqn.~(\ref{eq:effects}) \\
$\disj{}{}$ & disjoint regions & Fig.~\ref{fig:preds} \\
$\imports \;\; \irimports$ & module import & Sect.~\ref{sec:modules} \\ 
$\eqbi$ & equal reference or region, modulo refperm & Fig.~\ref{fig:relFormulas},
Fig.~\ref{fig:relFmlaSem}
 \\
$\Agr x \;\; \Agr G\Img f$ & agreement formulas & Fig.~\ref{fig:relFormulas}, Fig.~\ref{fig:relFmlaSem} \\
$\leftF{-},\; \rightF{-},\; \Both{-} $ & embed unary formula (left, right, both)  & Fig.~\ref{fig:relFormulas}, Fig.~\ref{fig:relFmlaSem}  \\
$\leftex{-},\;  \rightex{-} $ & embed unary expression  & Fig.~\ref{fig:relFormulas}, Fig.~\ref{fig:relFmlaSem}  \\
$\later$ & possibly (in an extended refperm)   & Fig.~\ref{fig:relFormulas}, Fig.~\ref{fig:relFmlaSem} \\

$\conjInv$ & conjoin invariant & %Eqn.~(\ref{eq:conjInv}), 
Def.~\ref{def:conjInv} \\
$\Syncbi{-}$ & full alignment of command & Fig.~\ref{fig:fullAlign} \\

$\weave$ & weave biprogram & Fig.~\ref{fig:weave} \\

$\extend{\sigma}{x}{v}$ & extend state to map $x$ to $v$ & Sect.~\ref{sec:states} \\
$\update{\sigma}{x}{v}$ & update value of $x$ & Sect.~\ref{sec:states} \\
$\drop{\sigma}{x}$ & drop variable $x$ from state & Sect.~\ref{sec:states} \\
$\successorTo$ & can succeed & Sect.~\ref{sec:effect} \\ 

$\delta^\oplus$ & abbreviates effect $\delta,\rd{\lloc}$ & preceding Def.~\ref{def:valid} \\

$\stackrel{\pi}{\sim}$ & equiv modulo refperm & Sect.~\ref{sec:effect} \\
$\RprelT{\pi\smallSplitSym \pi'}{}{} \;\; \RprelTS{\pi\smallSplitSym \pi'}{}{}$ &
   state pair isomorphism & Def.~\ref{def:state-iso-rel} \\
$\RprelT{\pi}{}{} \;\; \RprelTS{\pi}{}{}$ & state isomorphism, outcome equivalence & Def.~\ref{def:state-iso} \\

$\trans{\phi} \;\; \tranStar{\phi}$ & unary transitions & Figs.~\ref{fig:transSel} and~\ref{fig:trans} \\
$\biTrans{\phi} \;\; \biTranStar{\phi}$ & biprogram transitions & Figs.~\ref{fig:biprogTrans} and~\ref{fig:biprogTransU} \\

$\splitbir{C}{C'}$ & r-bi-com biprogram & Sect.~\ref{sec:bitrans} \\

$\sigma\allowTo\tau\models\eff$ & allows change &  Sect.~\ref{sec:effect} \\
$\AllowedDep{\tau}{\tau'}{\upsilon}{\upsilon'}{\eff}{\sigma}{\delta}{\pi}$ & allowed dependence & Def.~\ref{def:allowedDepEResB}\\
$P\models \eff\leq\effe $ & subeffect judgment & Eqn.~(\ref{eq:subeffect}) \\
$P\models\fra{P}{\eff}$ & framing of a formula & Eqn.~(\ref{eq:frmAgree}) \\
$\P\models \fra{\effe|\effe'}{\Q}$ & framing of a relation & Sect.~\ref{sec:biprog} \\ % Def~\ref{def:frmAgreeRel} \\

$\Phi \proves_M  C: \flowty{P}{Q}{\eff}$ & correctness judgment & Def.~\ref{def:wfjudge}, Def.~\ref{def:valid} \\
$\Phi\proves_M CC: \rflowty{\P}{\Q}{\eff|\eff'} $ & relational correctness judgment & Def.~\ref{def:wfjudgeRel}, Def.~\ref{def:validR} \\

$\locEq_\delta(\flowty{P}{Q}{\eff})$ \;\; $\LocEq_\delta(\Phi)$  & local equivalence specs & Def.~\ref{def:loceq}  \\

$\prePostImply$ & covariant spec implication & Def.~\ref{def:prePostImply} \\

\end{tabular}
\end{footnotesize}
\medskip
\caption{Use of symbols}
%\label{tab:symb}
\end{table}
\end{small}

\bibliographystyle{ACM-Reference-Format}
\bibliography{paper}

\ifarxiv % INDEX
\newpage

\begin{tiny}
%\tableofcontents
\printindex
\end{tiny}

% Note: although it seems like this line should precede \printindex, 
% but that causes the toc line to link to wrong page    
\addcontentsline{toc}{section}{\protect\numberline{}{Index of defined terms}}

\fi % INDEX

\end{document}